\documentclass[11pt]{article}
\usepackage[letterpaper,margin=0.92in]{geometry}
\usepackage[T1]{fontenc}
\usepackage{lmodern}
\usepackage{microtype}
\usepackage{amsmath,amssymb,amsthm,mathtools,mathrsfs}
\usepackage[dvipsnames]{xcolor}
\usepackage[colorlinks=true,linkcolor=MidnightBlue,urlcolor=MidnightBlue,citecolor=MidnightBlue]{hyperref}
\usepackage{enumitem,booktabs,tabularx,needspace}
\usepackage{tikz}
\usetikzlibrary{arrows.meta,positioning}
\hypersetup{pdftitle={The Standard Model on a finite bosonic lattice with continuum Weyl fermions},pdfsubject={Fixed-spacing bosonic lattice and spectrally regulated continuum Weyl fermions}}
\setlist[itemize]{topsep=4pt,itemsep=2pt,parsep=1pt,leftmargin=1.5em}
\newtheorem{theorem}{Theorem}
\newtheorem{lemma}{Lemma}
\newtheorem{proposition}{Proposition}
\newtheorem{stageproposition}{Proposition}

\newtheorem{corollary}{Corollary}
\newtheorem*{unnumberedlemma}{Lemma}
\newcommand{\SM}{G_{\mathrm{SM}}}
\newcommand{\calU}{\mathcal U}

\newcommand{\Lat}{\mathbb L}
\newcommand{\cZ}{\mathcal Z}
\newcommand{\ind}{\operatorname{ind}}
\newcommand{\Pf}{\operatorname{Pf}}

\title{The Standard Model on a finite bosonic lattice\\with continuum Weyl fermions}
\author{Stephen D.~H.~Hsu\thanks{Email: \href{mailto:hsusteve@gmail.com}{hsusteve@gmail.com}.}\\Department of Physics and Astronomy\\Michigan State University}
\date{6 November 2026}
\begin{document}
\maketitle
\vspace{-1.5em}
\begin{abstract}
We investigate a formulation of the Standard Model with gauge and Higgs fields regulated on a finite bosonic lattice, while quarks and leptons remain continuum Weyl fermions. This retains the familiar continuum treatment of anomalies, spectral flow, and fermion zero modes. Standard Model anomaly cancellation permits a consistent, gauge-invariant phase for the massless chiral determinant. Yukawa interactions are incorporated with a gauge-covariant spectral cutoff acting on the interaction vertices, not on the fermion propagator. Under the stated analytic conditions, the cutoff can be removed by the local subtractions expected from one-loop power counting, yielding a finite, exactly gauge-invariant functional integral at fixed bosonic lattice spacing and spacetime volume, including topologically nontrivial backgrounds. At zero Yukawa coupling and on smooth weak-field backgrounds, we establish agreement between the appropriately chirality-matched overlap and continuum Weyl fermion determinants up to local counterterms. We also establish reflection positivity at zero Yukawa coupling and zero theta for a modified gauge-field reconstruction; its continuum equivalence with Craig's original reconstruction is not proved. The bosonic continuum limit remains open.
\end{abstract}

\section{Motivation and scope}

A nonperturbative regulator for the Standard Model must accommodate chiral fermions while preserving gauge invariance. Conventional lattice discretizations make this difficult: the familiar continuum structure of chirality, anomalies, and fermion zero modes must be recovered from a finite-dimensional lattice description. An alternative is to discretize the bosonic fields while retaining continuum Weyl fermions. The gauge and Higgs fields then have finitely many degrees of freedom, but the fermions remain governed by continuum Dirac operators.

There is a conceptual advantage to this approach. Anomalies can be studied using the usual continuum Ward identities and index theorems. Fermion zero modes remain solutions of the continuum Dirac equation, and spectral flow of associated Hermitian operators retains the familiar interpretation in terms of level crossings. The electroweak $B+L$ anomaly and its selection rules can therefore be understood in essentially the same way as in continuum quantum field theory. Rather than reconstructing these phenomena from lattice fermion operators, we retain the framework in which they are ordinarily derived.

We consider the Standard Model on a finite four-dimensional spacetime torus. Gauge and Higgs fields are defined on a finite lattice, while quarks and leptons remain continuum Weyl fermions coupled to smooth backgrounds reconstructed from the lattice variables. Even with vanishing Yukawa couplings, this formulation requires several nontrivial steps. The lattice gauge fields must be converted into smooth continuum connections in a way that preserves exact gauge covariance and accommodates nontrivial gauge bundles. The chiral fermion determinant must be defined consistently, including its complex phase. Finally, the formulation must handle fermion zero modes and the associated anomalous selection rules without introducing singularities into physical correlation functions.

Our starting point is Craig's recent construction of smooth, gauge-covariant connections from admissible lattice links~\cite{craig}. The reconstructed connections can be defined on topologically nontrivial bundles, allowing continuum index theory to be applied directly. Standard Model anomaly cancellation, including the absence of global gauge anomalies, permits a consistent gauge-invariant phase convention for the massless chiral determinant using determinant-line methods, which track the determinant and its phase as gauge backgrounds vary and fermion zero modes appear. Fermionic correlation functions are described by regularized determinant minors, including when the zero-source determinant vanishes because of zero modes. These ingredients yield a finite-volume, gauge-invariant formulation of the zero-Yukawa theory at fixed bosonic lattice spacing. We also prove that every smooth Standard Model gauge bundle occurs in the reconstructed family on sufficiently fine lattices (Proposition~\ref{prop:topological_realization}). This realization result does not assert that every sector occurs at a single fixed lattice spacing.

The next challenge is to incorporate Yukawa interactions without sacrificing these properties. Unlike the massless determinant prescription, the Yukawa-dependent determinant introduces additional ultraviolet divergences that must be regulated and renormalized. We introduce Yukawa interactions through gauge-covariant spectral filters acting at the interaction vertices, while leaving the continuum fermion kinetic operator unchanged. This preserves the continuum description of chirality and zero modes and allows the fermion spectral cutoff to be removed after appropriate local subtractions.

The ultraviolet structure has a simple interpretation in Feynman-diagram language. A closed fermion loop with $n$ Yukawa insertions has superficial degree of divergence $4-n$. Thus only the first four orders require ultraviolet subtraction, and gauge invariance organizes the divergent terms into local counterterms. Mixed gauge--Yukawa effects are included because both the fermion propagator and the spectral filters depend covariantly on the background gauge field. The finite, momentum-dependent contributions of the loops are retained.

We also specify the normalization of the complete complex fermion weight, including backgrounds with nonzero index and phases associated with Yukawa matrices. The prescription distinguishes regulator-dependent local subtractions from physical parameters such as theta angles and CP-violating Yukawa phases. Under the analytic conditions stated below, Theorem~\ref{thm:renormalized_weight} establishes a finite, exactly gauge-invariant bosonic functional integral after removal of the Yukawa spectral cutoff, with the bosonic lattice spacing and spacetime volume held fixed. The bosonic continuum limit and reflection positivity with nonzero Yukawa couplings remain open.

\paragraph{Earlier work.}
Flume and Wyler used interpolated lattice backgrounds in the Schwinger model~\cite{flume_wyler}. Continuum gauge-field reconstruction from lattice links was developed by G\"ockeler and collaborators~\cite{goeckeler}. 't~Hooft advocated preregulating gauge fields while retaining continuum fermions, emphasizing chiral symmetry and the finite-volume functional integral~\cite{thooft_lattice}. Hsu's 1995 proposal used the same broad two-stage strategy and developed a chiral-determinant phase prescription in terms of a vector-like Dirac operator, including local and global anomalies and a comparison with zeta-function methods~\cite{hsu}. It also discussed the five-dimensional eta invariant as a representation of the continuum determinant phase. Kronfeld analyzed continuum fermions in interpolated backgrounds with finite-mode cutoffs and explicit gauge-restoring counterterms~\cite{kronfeld}. Montvay and Hern\'andez--Sundrum studied related chiral and two-cutoff constructions~\cite{montvay,hernandez_sundrum,hsinterp}. These works establish the historical motivation and much of the ultraviolet strategy. We use Craig's more recent smooth, gauge-covariant reconstruction to specify the Standard Model gauge bundles and their dependence on the lattice links~\cite{craig}; our fermion prescription is distinct from Craig's overlap measure.

\paragraph{Results and scope.}
Theorems~1--2 establish the finite, gauge-invariant zero-Yukawa functional integral and fermionic insertions, including in the presence of zero modes. Theorems~3--4 incorporate Yukawa interactions and remove their spectral cutoff after local subtractions and fixed finite matching conditions. Section~\ref{sec:overlap} and Appendix~\ref{app:weak_overlap} compare the overlap and continuum Weyl determinants on smooth weak backgrounds. That direct comparison is established at zero Yukawa coupling (Proposition~\ref{prop:D2_overlap}); matching to our reconstructed-field regulator also requires the separate reconstruction bound of Proposition~\ref{prop:D2_weak_match}.

Appendix~\ref{sec:collar_variant} proves site-reflection positivity for a \emph{modified} gauge-field reconstruction at zero Yukawa coupling and zero theta (Theorem~\ref{prop:collar_rp}). This construction preserves the bundle topology and continuum Weyl fermions, but its continuum equivalence to Craig's original reconstruction has not been established. Neither that theorem nor the preceding existence results establish the bosonic continuum limit.

During preparation of this revision, Edward Witten kindly shared with the author a draft of his independent work on lattice gauge theories with continuum Weyl fermions~\cite{witten_lattice_2026}. His five-dimensional eta-invariant formulation clarifies several issues also encountered here, including anomaly cancellation, topologically nontrivial gauge configurations, and fermion zero modes. We have added a discussion of these connections and show that, with matching reference phases and normalization conventions, the eta-invariant approach gives the same relative determinant phase as the determinant-line formulation used here (Proposition~\ref{prop:aps_bf_exact}). The regulator construction and fixed-lattice existence results of this paper do not depend on this comparison.

The construction defines a regulated functional integral but not yet a practical Monte Carlo algorithm: its fermion weight is generally complex, and evaluating it requires continuum spectral information.

\paragraph{Guide to the paper.}
Sections 2--6 build the zero-Yukawa theory and explain its anomalies and zero modes. Section~\ref{sec:yukawa_full} develops the Yukawa prescription through fermion-loop calculations and physical checks; Section~\ref{sec:overlap} compares it with overlap fermions. The main text explains the physical role of the mathematical constructions; the appendices give their detailed ultraviolet and zero-mode arguments, reconstruction geometry, and positivity proof. Appendix~\ref{app:terminology} supplies a glossary of the mathematical terms.

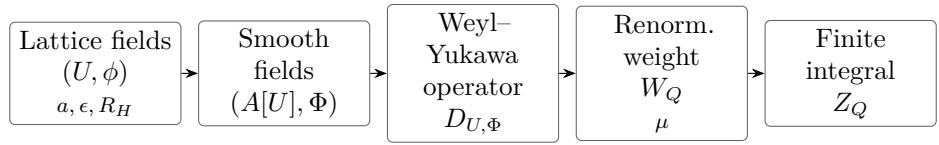
\begin{figure}[htbp]
\centering
\begin{tikzpicture}[>=Stealth, node distance=2.2mm, every node/.style={font=\small}]
\tikzstyle{stepbox}=[draw=black!60,rounded corners=2pt,align=center,minimum height=0.76cm,text width=2.05cm,inner sep=3pt]
\node[stepbox] (l) {Lattice fields\\$(U,\phi)$\\{\scriptsize $a,\epsilon,R_H$}};
\node[stepbox,right=of l] (r) {Smooth fields\\$(A[U],\Phi)$};
\node[stepbox,right=of r] (d) {Weyl--Yukawa\\operator $D_{U,\Phi}$};
\node[stepbox,right=of d] (w) {Renorm. weight\\$W_Q$\\{\scriptsize $\mu$}};
\node[stepbox,right=of w] (z) {Finite integral\\$Z_Q$};
\draw[->] (l) -- (r);
\draw[->] (r) -- (d);
\draw[->] (d) -- (w);
\draw[->] (w) -- (z);
\node[align=center,font=\scriptsize,text width=6cm,below=7mm of d] {Anomaly cancellation permits a consistent massless phase; the Yukawa cutoff and subtraction specify a renormalization scheme.};
\end{tikzpicture}
\caption{Structure of the fixed-bosonic-lattice regulator. The bosonic spacing $a$, admissibility threshold $\epsilon$ and Higgs support $R_H$ enter at the lattice stage. The fermion cutoff $\Lambda$ enters through the Yukawa vertices and is removed before the finite bosonic integral. With finite matching conditions fixed at the scale $\mu$, the normalized weight $W_Q$ is independent of the chosen Yukawa filter.}
\label{fig:overview}
\end{figure}

The fermion spectral cutoff $\Lambda$ and bosonic lattice spacing $a$ are independent regulators. Whenever we remove $\Lambda$, the finite spacetime volume, $a$, and the bosonic integration domain are held fixed. Our estimates do not assert uniform control as $a\to0$ or in the infinite-volume limit.

\section{Lattice variables and the Standard Model gauge group}

\subsection{Bosonic lattice configuration space}
Let Euclidean spacetime be the four-torus $X=T^4$, equipped with a fixed flat metric and spin structure (fermion boundary conditions around the four cycles). Introduce a finite periodic hypercubic lattice $\Lat$ with spacing $a>0$. We use the double-struck $\Lat$ for the lattice, reserving $\Lambda$ for the fermion-mode ultraviolet cutoff. Write $\Lat^0$ and $\Lat^1$ for the sets of sites and oriented links; the lattice gauge-transformation group $\SM^{\Lat^0}$ assigns an element of $\SM$ independently to each site.

Each oriented link $\ell$ carries $U_\ell\in\SM$, representing parallel transport between neighboring lattice sites. The ordered product $U_p$ around an elementary square is the \emph{plaquette holonomy}; its leading deviation from the identity is proportional to $a^2F_{\mu\nu}$ in the smooth-field limit. Each site $x$ also carries a complex Higgs doublet $\phi_x\in\mathbb C^2$.

\Needspace{10\baselineskip}
The gauge group is
\begin{equation}
\SM=S(U(3)\times U(2))\cong
\frac{SU(3)\times SU(2)\times U(1)}{\mathbb Z_6}\subset SU(5).
\label{eq:gaugegroup}
\end{equation}
The notation $S(U(3)\times U(2))$ means the subgroup of block-diagonal matrices $\operatorname{diag}(A,B)\in SU(5)$, where $A\in U(3)$ and $B\in U(2)$ satisfy $\det A\,\det B=1$. Its Lie algebra is precisely $\mathfrak{su}(3)\oplus\mathfrak{su}(2)\oplus\mathfrak{u}(1)$: the $SU(5)$ notation does not add grand-unified gauge bosons. To see the discrete quotient, consider the surjective homomorphism
\[
 (g_3,g_2,z)\ \longmapsto\
 \operatorname{diag}(z^{-2}g_3,\,z^3g_2),
 \qquad g_i\in SU(i),\quad z\in U(1).
\]
The source group is connected, so its quotient $\SM$ and the finite lattice gauge group $\SM^{\Lat^0}$ are connected. The displayed map has the six-element kernel generated by $(e^{2\pi i/3}\mathbf1_3,-\mathbf1_2,e^{i\pi/3})$. Thus the direct-product group and $S(U(3)\times U(2))$ have the same infinitesimal gauge transformations, but the latter identifies a $\mathbb Z_6$ subgroup. That generator acts trivially on every Standard Model quark, lepton, and Higgs representation. In the integer hypercharge convention $n=6Y$, that generator acts on a representation of color triality $t$ and weak isospin $I$ by $\exp[2\pi i(t/3+I+n/6)]$. Requiring this phase to be one gives $t/3+I+n/6\in\mathbb Z$; Appendix~\ref{app:geometry} records the corresponding bundle conventions~\cite{hucks,tong,craig}.

The global quotient matters for allowed representations, nontrivial gauge bundles, magnetic fluxes, and line operators, even though it does not change local perturbative Feynman rules. The alternatives $(SU(3)\times SU(2)\times U(1))/\Gamma$, with $\Gamma=1,\mathbb Z_2,\mathbb Z_3,\mathbb Z_6$, are all compatible with the observed Standard Model fields. We choose the full $\mathbb Z_6$ quotient because its block embedding into $SU(5)$ is faithful and particularly convenient for the bundle and anomaly arguments below, not because perturbative Standard Model physics uniquely requires it~\cite{hucks,tong,craig}.

In this integer hypercharge convention, one left-handed generation is
\begin{align}
\mathbf{10}&\longrightarrow
(\overline{\mathbf3},\mathbf1)_{-4}\oplus
(\mathbf3,\mathbf2)_1\oplus(\mathbf1,\mathbf1)_6,\nonumber\\
\overline{\mathbf5}&\longrightarrow
(\overline{\mathbf3},\mathbf1)_2\oplus
(\mathbf1,\mathbf2)_{-3}.
\label{eq:reps}
\end{align}
An overbar denotes the complex-conjugate representation. All fermion fields are written as left handed; in particular, $u^c$ denotes the left-handed conjugate of a right-handed up quark. The representation for the three generations is
\begin{equation}
\mathcal R=(\mathbf{10}\oplus\overline{\mathbf5})^{\oplus3}.
\label{eq:rep}
\end{equation}
The Higgs has charges $(\mathbf1,\mathbf2)_3$. This is the minimal Standard Model matter content expressed in $SU(5)$ representation notation: no right-handed neutrino is included, so there are no Dirac neutrino mass terms. The gauge group remains $\SM$.

\subsection{Admissibility and compact regulators}
A uniformly controlled continuum reconstruction requires the plaquette holonomies to remain sufficiently close to the identity. This restriction is the \emph{admissibility condition}. To implement it with compact support inside the reconstruction domain, introduce a nonnegative function
\begin{equation}
\chi_\epsilon(U)=\prod_p w_\epsilon(U_p),
\label{eq:admissibility}
\end{equation}
The product contains each unoriented plaquette once, with a fixed choice of orientation. We choose $w_\epsilon(V^{-1})=w_\epsilon(V)$, so reversing that orientation has no effect and the same cutoff can be used in the site-reflection argument. Here $w_\epsilon$ is continuous, invariant under conjugation $U_p\mapsto gU_pg^{-1}$ (it depends only on gauge-invariant plaquette data, such as $\operatorname{Re}\operatorname{tr}U_p$), strictly positive at the identity, and zero outside a sufficiently small neighborhood of the identity. Its \emph{support}, the closure of the region where it is nonzero, lies strictly within the \emph{geometric} admissibility region $\|\mathbf 1_5-U_p\|<\epsilon_g$ in the fundamental $\mathbf5$ representation, with $\epsilon_g$ fixed in Ref.~\cite[App.~A.5]{craig}. One can enforce this condition by taking the support sufficiently close to $\operatorname{Re}\operatorname{tr}_{\mathbf5}U_p=5$. The argument uses this geometric threshold, not the stronger admissibility assumptions needed to control the overlap measure in Ref.~\cite{craig}. We write $\calU_\epsilon$ for the open geometric link domain containing $\operatorname{supp}\chi_\epsilon$; the subscript labels the chosen regulator, while $\epsilon_g$ fixes the reconstruction threshold.

For the Higgs field choose another nonnegative function,
\begin{equation}
b_{R_H}(\phi)=\prod_x b(|\phi_x|^2/R_H^2),
\label{eq:higgscut}
\end{equation}
where $R_H>0$ is a fixed Higgs cutoff and $b$ is continuous and vanishes outside a bounded interval. The Higgs amplitude is therefore bounded at every site. This cutoff is an independent regulator; integrability over unbounded Higgs configurations is not assumed.

Gauge links are integrated with normalized \emph{Haar measure}, the invariant measure on the compact gauge group. Higgs variables carry the ordinary Lebesgue measure. Because the lattice is finite and $\chi_\epsilon b_{R_H}$ has compact support, the resulting bosonic integral is over a compact, finite-dimensional domain.

\section{Smooth reconstruction of continuum gauge fields}

The continuum Weyl operator requires a smooth gauge connection, not just link matrices. When the gauge bundle is topologically nontrivial, its potential $A_\mu(y)$ is defined in overlapping local frames rather than as one global matrix field; the transition functions identify the gauge indices between frames. This is the familiar need for gauge potentials related by gauge transformations on overlapping spacetime patches.

Craig's Appendix~A~\cite{craig} constructs the continuum bundle and connection used here, in addition to the paper's reconstruction of the overlap fermion measure. The relevant smoothness is joint dependence on spacetime and the original links: the family differential is $d_{\rm tot}=d_y+d_U$. We use the following consequences of that explicit construction.

\begin{lemma}[Smooth gauge-field family; Ref.~\cite{craig}]\label{lem:smooth_reconstruction}
On the geometric admissible set $\calU_\epsilon$, Craig's construction supplies one smooth Hermitian rank-five bundle
\[
 \mathscr E\longrightarrow\calU_\epsilon\times T^4,
 \qquad \mathscr E|_{\{U\}\times T^4}=E[U],
\]
with a smooth unitary family connection whose spacetime restriction is $A[U]$. It preserves the determinant trivialization and the splitting $\mathscr E=\mathscr E_3\oplus\mathscr E_2$. Derivatives of every fixed order are bounded on compact regulated families. Lattice gauge transformations act smoothly by bundle isometries intertwining the spacetime connections. The fibers need not have trivial topology.
\end{lemma}
The inputs are Craig's Eqs.~(A.12)--(A.14) for the bundle and frames, Eqs.~(A.15)--(A.18) for the family connection and bounds, and Eqs.~(A.21)--(A.23) with the end of A.3 for covariance, torus descent and the splitting. Appendix~\ref{app:family_reconstruction} explains the family-bundle statement. The reconstructed parallel transports reproduce the input link matrices to $O(\epsilon)$, not exactly. This is sufficient here: the fermions couple to the \emph{reconstructed} connection. A ``compact smooth family'' below means a set of these backgrounds confined away from the boundary of the admissible region, so that the needed field derivatives are uniformly bounded at fixed lattice spacing. The topology statement follows from the local frames and torus identifications of Craig's construction; it is not presented as a separate theorem in Ref.~\cite{craig}. Further references, curvature conventions, and interpolation comparisons are given in Appendix~\ref{app:geometry}.

The Higgs doublet must also become a smooth section of the weak bundle $E_2$. Craig does not reconstruct scalar fields; we extend the local-frame construction by transporting site doublets and averaging with a fixed smooth partition of unity (overlapping nonnegative weights that sum to one):
\begin{equation}
\Phi[U,\phi](y)=\sum_i\rho_i(y)C_i(y)j_2\phi_i\in(E_2)_y.
\label{eq:higgsinterp}
\end{equation}
Here $C_i$ are Craig's smooth local frames and $j_2:\mathbb C^2\hookrightarrow\mathbb C^5$ includes the weak block. This interpolation is smooth, gauge covariant, and linear in the site doublets: $\Phi[U,t\phi]=t\Phi[U,\phi]$. The verification, including its uniform family choice, is in Appendix~\ref{app:geometry}.

For the subsequent arguments, the only required geometric facts are the smooth gauge-covariant families $(A[U],\Phi[U,\phi])$, preservation of the Standard Model representations, and consistent nontrivial bundle topology. We write $D_U$ for the massless chiral Dirac operator on $A[U]$. The notation $D_{U,\Phi}$ in Figure~\ref{fig:overview} is schematic; Section~\ref{sec:yukawa_full} specifies the coupled operator $\mathcal D_Y$ and its fermion field ordering.

\section{Continuum Weyl fermions and determinant lines}

\subsection{The chiral Dirac operator}
For a reconstructed connection $A[U]$, the massless Weyl action takes the form
\begin{equation}
S_F=\int_X d^4y\;\bar\psi\,D_A^+\psi,
\qquad
D_A^+:\Gamma(S^+\otimes E_{\mathcal R})
\longrightarrow\Gamma(S^-\otimes E_{\mathcal R}).
\label{eq:dirac}
\end{equation}
Here $S^\pm$ denote the positive- and negative-chirality spinor bundles, $E_{\mathcal R}$ carries gauge indices in the quark--lepton representation $\mathcal R$, and $\Gamma$ denotes smooth fields of the indicated type (sections). The domain consists of fields with a left-handed spinor index and a gauge index in $\mathcal R$; $D_A^+$ maps them into fields with the opposite chirality and the same gauge representation. The field $\bar\psi$ is an independent conjugate Grassmann variable in Euclidean signature.

For finitely many Grassmann variables, Gaussian integration produces a determinant:
\begin{equation}
\int d\bar\psi\,d\psi\;e^{-\bar\psi M\psi}=\det M.
\label{eq:matrixdet}
\end{equation}
The continuum operator $D_A^+$ has infinitely many modes, so its formal determinant requires regularization. Moreover, a phase assigned independently at each background need not vary smoothly or respect gauge transformations. A consistent definition must address both regularization and phase dependence.

\subsection{Determinant lines and their geometric data}
Fix one spectral scale $\mu>0$ for all regularized determinants below; it sets the common subtraction convention.

For a chiral fermion, the Dirac operator maps between distinct spinor spaces. A numerical determinant therefore depends on the choice of bases in those spaces, and a consistent choice of its \emph{phase} as the background gauge field varies is not automatic. The natural basis-independent object belongs to a one-dimensional complex space called the \emph{determinant line}. The lines at different backgrounds fit together as a bundle, and the regularized fermion determinant is a \emph{section} $s_D$ of that bundle. Anomaly cancellation is what allows us to identify these lines consistently with $\mathbb C$ while preserving gauge invariance.

The construction remains meaningful when the Dirac operator develops zero modes: the determinant line records their kernel and cokernel, while the ordinary unsaturated determinant vanishes. Choosing its phase is analogous to choosing phases for a continuously varying family of quantum states. The finite-dimensional basis-change formula and the precise exterior-power definition of the line are recalled in Appendix~\ref{app:terminology}; the inverse-line convention sometimes used in the literature is related by dualization. Throughout this paper we keep the convention for $\mathcal L_U$ fixed.

Three associated structures enter the construction:
\begin{itemize}
\item The \emph{determinant section} $s_D$ represents the regularized fermion Gaussian integral before a global phase convention is chosen. It is defined even when zero modes appear, although it vanishes when an unsaturated zero mode is present.
\item The \emph{Quillen metric} is a rule for measuring the length (absolute magnitude) of the complex determinant-line element, without choosing its phase. In physics terms it is the positive spectral determinant of the fermion kinetic operator, defined by analytic continuation of a spectral sum. On an invertible background, in the normalization adopted here,
\[
\|s_D(A)\|_Q^2=\det_{\zeta}\bigl((D_A^+)^\dagger D_A^+/\mu^2\bigr).
\]
Here $\det_\zeta P=\exp[-\zeta_P'(0)]$ with $\zeta_P(z)=\operatorname{Tr}(P^{-z})$ analytically continued to $z=0$ (zero eigenvalues are omitted and treated in the finite low-mode factor). For a free fermion on a torus the eigenvalues of $P$ are essentially $p^2$, so $\zeta_P$ is a sum of inverse powers of squared momenta: it converges at sufficiently large $\operatorname{Re}z$ and is then analytically continued. In that sense it plays a role analogous to analytic or dimensional regularization of a loop integral, although the prescriptions are not identical. This defines the gauge-only determinant using all modes, rather than deleting modes above a fixed energy. It produces a finite positive magnitude from the infinite product of nonzero singular values. The determinant-line convention supplies the separate phase.
\item The \emph{Bismut--Freed connection} transports determinant phases when the gauge background varies. It is the analogue of a \emph{Berry connection} on the space of gauge fields. Its curvature is determined by the local anomaly polynomial; its \emph{holonomy}, the phase accumulated on a closed loop, is analogous to a Berry phase and can detect a global anomaly.
\end{itemize}
The determinant section answers \emph{what is being integrated}; the Quillen metric specifies its \emph{magnitude}; and the Bismut--Freed connection describes how a \emph{phase} changes with the background.

\noindent\emph{Five-dimensional interpretation.}
Witten's formulation gives a particularly simple picture of these structures~\cite[Secs.~2.1--2.3]{witten_lattice_2026}. A massive nonchiral fermion in five dimensions can support a massless chiral boundary mode. The \emph{eta invariant} measures the imbalance between positive and negative eigenvalues of a five-dimensional Dirac operator, with a prescribed treatment of its zero modes. The relevant boundary conditions are of Atiyah--Patodi--Singer (APS) type: they select boundary modes using the spectrum of the boundary Dirac operator. In an index-zero sector where the four-dimensional fields bound a suitable five-manifold $Y$, and with compatible boundary and regulator conventions, the chiral weight can be written schematically as
\begin{equation}
 W_{\rm Weyl}(A)=|\det D_A^+|\,\exp[-i\pi\eta_{\rm W}(D_Y)].
 \label{eq:witten_aps_summary}
\end{equation}
Here $\eta_{\rm W}$ includes the zero-mode contribution, as in Witten's convention. The positive spectral determinant supplies the magnitude; the spectral asymmetry of the five-dimensional Dirac operator supplies the phase. For an anomalous representation, this expression can depend on the auxiliary filling; cancellation of both local and global anomalies removes that dependence. Equation~\eqref{eq:witten_aps_summary} is not by itself a formula for a nonzero vacuum amplitude in a sector with unsaturated zero modes. Witten's prescription there uses suitable fiducial manifolds and zero-mode-saturating correlators. In determinant-line language the same data are encoded in line-valued sections and their source coefficients. The eta invariant fixes the phase \emph{relative to} chosen reference data; it does not determine independent physical theta angles. Appendix~\ref{app:phase_references}, Proposition~\ref{prop:aps_bf_exact}, gives the precise identification of the relative APS eta element with Bismut--Freed phase transport, including zero-mode crossings and nontrivial gauge bundles, after matching the Quillen and fiducial conventions. This equivalence does not establish convergence of the distinct finite-eigenmode phase prescription proposed in Ref.~\cite{hsu}; that question is discussed in Appendix~\ref{app:hsu_comparison}.\par
 These structures arise from the determinant-line construction initiated in Ref.~\cite{quillen} and from the subsequent theory of smooth families of elliptic operators \cite{bf1,bf2}. For a Dirac operator, \emph{ellipticity} means that the highest-derivative kinetic matrix, proportional to $\not\!p$, is invertible for $p\neq0$. On compact spacetime this ensures that the kernel of $D_A^+$ and the kernel of its adjoint are finite-dimensional; this is the relevant \emph{Fredholm} property. The nonnegative self-adjoint operators $D_A^{+\dagger}D_A^+$ and $D_A^+D_A^{+\dagger}$ have discrete spectra.

The \emph{index} is the difference between the dimensions of the kernel and cokernel:
\begin{equation}
\ind D_A^+=\dim\ker D_A^+-\dim\ker((D_A^+)^\dagger).
\label{eq:index}
\end{equation}
\noindent\emph{Physical interpretation.} The chiral determinant behaves like a complex amplitude whose phase can change as the gauge background varies. The Quillen metric defines its magnitude. Gauge invariance also requires the phase to return consistently when a background is transported around a closed loop.\par

The index counts the difference between the two types of chiral zero modes: the zero modes of $D_A^+$ and those of its adjoint (often called the kernel and cokernel). It is invariant under continuous deformations of the background within a fixed topological sector. If it is nonzero, $D_A^+$ is noninvertible throughout that sector, and the zero-source determinant section vanishes identically there. Fermionic insertions can nevertheless saturate the zero modes and yield nonvanishing correlators (Section~\ref{sec:corr}).

\subsection{Anomaly cancellation and global phase consistency}
A gauge \emph{anomaly} is an obstruction to maintaining gauge invariance in the quantum fermion measure. Two distinct obstructions must be considered:
\begin{itemize}
\item A \emph{local} (perturbative) anomaly occurs under infinitesimal gauge transformations. In determinant-line theory, it appears as curvature of the phase connection, obstructing path-independent transport.
\item A \emph{global} anomaly may persist when the curvature vanishes, through a nontrivial phase acquired around a closed loop in gauge-field space.
\end{itemize}
In four dimensions, the local gauge anomaly is obtained by the descent procedure from a six-form \emph{anomaly polynomial}. Its gauge terms are proportional to $\operatorname{tr}_{\mathcal R}F^3$, with a mixed gauge--gravitational term proportional to $\operatorname{tr}_{\mathcal R}F\,\operatorname{tr}R^2$. The descent of the cubic term gives the familiar perturbative triangle anomaly. For the $SU(5)$ representation $\mathbf{10}\oplus\overline{\mathbf5}$, the cubic gauge anomaly cancels,
\begin{equation}
\operatorname{tr}_{\mathbf{10}}F^3+
\operatorname{tr}_{\overline{\mathbf5}}F^3=0,
\label{eq:anomalypoly}
\end{equation}
and the mixed gauge--gravitational term vanishes because the representation is traceless. The invariant-polynomial identity also applies to the full curvature of the family connection, including its mixed spacetime--parameter components. With the spacetime geometry fixed, the Bismut--Freed curvature therefore vanishes and the phase-transport connection is flat.

A flat connection need not have trivial holonomy: vanishing local curvature still allows a nontrivial Berry phase on a noncontractible closed loop. The Dai--Freed theorem relates the possible phase acquired around a closed loop to a regularized spectral asymmetry, the \emph{eta invariant}, of a five-dimensional Dirac operator \cite{daifreed}. For a loop of four-dimensional connections $A(t)$, the fifth coordinate is the loop parameter $t$. Identifying its endpoints produces a five-dimensional \emph{mapping torus}; if the endpoint fields differ by a gauge transformation, the same transformation is used in this identification. The spectrum of the five-dimensional operator need not be symmetric between positive and negative eigenvalues. Its regulated asymmetry measures the corresponding determinant-line holonomy. The relation to the spectral-phase prescription of Ref.~\cite{hsu} is described below.

A familiar example of the distinction is Witten's global $SU(2)$ anomaly~\cite{witten_su2}. A single left-handed fermion doublet has no ordinary perturbative $SU(2)$ triangle anomaly, but transport along a path from $A$ to a gauge-equivalent background $A^g$ can reverse the sign of its fermion Pfaffian. The associated group is $\Omega_5^{\mathrm{Spin}}(BSU(2))\cong\mathbb Z_2$: the discrete sign is invisible to the local anomaly calculation.

Once the perturbative anomaly polynomial vanishes, the phase $\exp(2\pi i\bar\eta_5)$ associated with the five-dimensional operator depends only on the \emph{spin-bordism class} of the five-dimensional spacetime and its gauge bundle. In physical terms, the residual phase depends only on discrete topological information, not on a smooth deformation of the background. Two configurations are bordism equivalent when together they form the boundary of a six-dimensional spin manifold with compatible gauge fields. For the present representation the relevant result is
\begin{equation}
\Omega_5^{\mathrm{Spin}}(BSU(5))=0.
\label{eq:bordism}
\end{equation}
Here $BSU(5)$ is the mathematical space used to classify $SU(5)$ gauge bundles. Equation~\eqref{eq:bordism} means that every closed five-dimensional spin configuration with such a bundle bounds an appropriate six-dimensional configuration. Since the local anomaly also vanishes, the exponentiated eta phase is then unity.\footnote{Formally, $\Omega_5^{\mathrm{Spin}}(BSU(5))$ is the spin-bordism group of pairs consisting of a closed five-manifold with spin structure and an $SU(5)$ principal bundle. The anomaly defines a $U(1)$-valued bordism invariant once the local anomaly polynomial vanishes.} The Standard Model fermion representation is the restriction of $\mathbf{10}\oplus\overline{\mathbf5}$. More explicitly, a principal $\SM$ bundle $P$ induces the $SU(5)$ bundle $P\times_{\SM}SU(5)$ under the faithful inclusion $\SM\hookrightarrow SU(5)$; the associated fermion bundles and their Dirac operators are the same after restriction. Hence the Standard Model anomaly phase is the pullback of the trivial $SU(5)$ anomaly phase. This argument does not require a separate computation of the Standard Model spin-bordism group. Independently, the direct calculation of Davighi, Gripaios, and Lohitsiri~\cite{davighi_gripaios_lohitsiri} gives
\begin{equation}
\Omega_5^{\mathrm{Spin}}\bigl(BS(U(3)\times U(2))\bigr)=0.
\label{eq:sm_bordism}
\end{equation}
Unlike the representation-specific $SU(5)$ pullback, this vanishing rules out residual spin-bordism global gauge anomalies for any perturbatively anomaly-free fermion content representing the same faithful gauge group. It depends on the global form of the gauge group: some other groups with the same Standard Model Lie algebra have nontrivial fifth spin-bordism. The holonomy is therefore trivial on every loop in the reconstructed family, including loops not induced by a gauge transformation~\cite{daifreed,gem}.

\begin{lemma}[Gauge-invariant determinant phase]\label{lem:determinant_phase}
For $\mathcal R=(\mathbf{10}\oplus\overline{\mathbf5})^{\oplus3}$, the determinant line over admissible lattice backgrounds has a smooth, gauge-equivariant flat phase convention. In particular, this convention gives continuous scalar representatives of the determinant section and its source coefficients. The vacuum coefficient is gauge invariant; source coefficients transform with their external gauge indices.
\end{lemma}
\begin{proof}
The Bismut--Freed curvature vanishes by~\eqref{eq:anomalypoly}. With this local anomaly absent, Dai--Freed holonomy is a spin-bordism invariant and is trivial by~\eqref{eq:bordism}. This applies both to ordinary loops and to paths closed by gauge transformations. Parallel transport from a unit reference is therefore path independent. For a path $\gamma:U\to U^g$ and induced line map $\Phi_g$, gauge-closed holonomy is $\Phi_g^{-1}\operatorname{PT}_\gamma=1$, which gives equivariance, including at stabilizers. Appendix~\ref{app:phase_references} spells out this step and fixes the references by ambient $SU(5)$ transport.
\end{proof}
\noindent\emph{Physical interpretation.} The chiral determinant can be assigned a phase continuously and gauge invariantly throughout the admissible configuration space, without a perturbative or global gauge-anomaly obstruction.

\subsection{The fermion weight and the eta-invariant prescription}
\label{sec:hsu_eta}
The preceding lemma ensures that the determinant can be represented by a scalar rather than by a vector in a background-dependent complex line. Let $\tau_U:\mathcal L_U\to\mathbb C$ be the \emph{phase convention} specified below: a rule for converting the determinant-line element into a complex number without changing its magnitude (a unitary \emph{trivialization} with respect to the Quillen metric). The regulated chiral fermion weight is
\begin{equation}
\boxed{W_0(U)=\tau_U\bigl(s_D(A[U])\bigr).}
\label{eq:weight}
\end{equation}
The section $s_D$ contains the full regularized determinant; $\tau_U$ fixes how its phase is expressed as a complex number. At an invertible background one may write $W_0(U)=|W_0(U)|e^{i\varphi(U)}$. In the chosen Quillen normalization, $|W_0(U)|$ is the norm $\|s_D(A[U])\|_Q$, obtained by spectral regularization of the associated positive operator $D_A^{+\dagger}D_A^+$. The phase $\varphi(U)$ need not vanish or be constant, even though the anomaly connection is flat. Flatness asserts that the chosen phase convention is globally consistent; it does not make the chiral determinant real.

Gauge covariance of the reconstruction and equivariance of $\tau$ give
\begin{equation}
W_0(U^g)=\tau_{U^g}\bigl(s_D(A[U^g])\bigr)
=\tau_U\bigl(s_D(A[U])\bigr)=W_0(U).
\label{eq:gaugeinvariant}
\end{equation}
A gauge transformation maps both the zero modes and the nonzero modes at $U$ to those at $U^g$. This map identifies their determinant lines, and it is precisely the identification used in the middle equality; the phases are not chosen independently at the two backgrounds. The weight extends continuously to noninvertible backgrounds, where the zero-source section vanishes. These are the precise properties required in the finite-volume theorem.

\paragraph{Relation to the original spectral proposal.}
Ref.~\cite{hsu} defined a chiral phase using overlaps of four-dimensional eigenmodes and discussed its expected relation to a five-dimensional eta invariant. The precise overlap conventions and the unproved identification of the two prescriptions are reviewed in Appendix~\ref{app:hsu_comparison}. Neither identification is needed to define $W_0$ or to prove the finite-volume existence theorem.

\paragraph{Relative phases between topological sectors.}
We compare phases by transport over all smooth connections on the induced $SU(5)$ bundle, then restrict to the reconstructed family. One unit determinant-line reference is fixed per induced bundle class. Disconnected admissible components in the same class therefore share a convention, rather than receiving independent phases. With antiperiodic temporal fermions, we use the positive real twelfth-power references on the weak-$SU(2)$ locus specified in Appendix~\ref{app:phase_references} and checked in F.5. For other spin structures, reference vectors remain specified input data of the regulator. Theta terms are additional choices relative to this baseline: \textquotedblleft zero theta\textquotedblright\ means that their coefficients vanish in the chosen convention.

Only at zero Yukawa coupling does the vacuum selection rule restrict all potentially contributing backgrounds to the trivial induced $SU(5)$ bundle, leaving one common vacuum phase. Fermionic insertions can reach other classes. With nonzero Yukawa couplings, pure-color sectors of nonzero charge can also contribute to the vacuum, so their reference phases matter there as well. Section~\ref{sec:phase_checks} distinguishes this baseline from the conventional positive-mass normalization used to identify the physical strong-CP angle.

\section{Finite-lattice theorem at zero Yukawa coupling}

The regulated partition function is
\begin{equation}
\boxed{\begin{aligned}
Z_{a,\epsilon,R_H}={}&\int\!\prod_{\ell\in\Lat^1}dU_\ell
\int\!\prod_{x\in\Lat^0}d\phi_x\;\chi_\epsilon(U)b_{R_H}(\phi)\\[-2pt]
&\hspace{24mm}\times e^{-S_{gH}(U,\phi)}W_0(U).
\end{aligned}}
\label{eq:Z}
\end{equation}
The action $S_{gH}$ is any continuous, real, gauge-invariant gauge--Higgs lattice action, such as the Wilson--Higgs action. The fermion weight $W_0$ does not include Yukawa interactions. In~\eqref{eq:Z}, $\chi_\epsilon W_0$ is extended by zero outside $\calU_\epsilon$. This extension is continuous because $\operatorname{supp}\chi_\epsilon$ is compact and contained in that open reconstruction domain.

\begin{theorem}[Zero-Yukawa existence and exact gauge invariance]\label{thm:zero_yukawa_existence}
Fix the finite lattice, $a>0$, admissibility cutoff $\epsilon$, Higgs bound $R_H$, spin structure, smooth reconstruction, and anomaly-free determinant phase convention. The integral~\eqref{eq:Z} converges absolutely to a finite complex number. Its integrand is invariant under all lattice gauge transformations. The same holds after inserting any bounded, continuous gauge-invariant function of the bosonic variables. At zero Yukawa coupling, a nonzero index of any separate Weyl species forces the vacuum weight to vanish.
\end{theorem}
\begin{proof}
There are finitely many links and sites. Each link belongs to a compact matrix group. The Higgs cutoff makes every Higgs integration range bounded and closed after taking the support. Therefore the combined region on which the integrand can be nonzero is compact. Lemmas~\ref{lem:smooth_reconstruction}--\ref{lem:determinant_phase} and determinant-line theory make the physical chiral weight $W_0(U)$ continuous on this region, so its absolute value has a finite maximum. The action and cutoff functions are also continuous and bounded there. The integral of the absolute value is therefore bounded by the supremum of the integrand times the finite measure of its support. Gauge invariance follows from invariance of Haar measure, Higgs Lebesgue measure under the unitary gauge action, the cutoffs and the bosonic action, together with~\eqref{eq:gaugeinvariant}. Nonzero index forces the vacuum determinant to vanish, as explained after~\eqref{eq:index}.
\end{proof}
\noindent\emph{Physical interpretation.} The substantive inputs are the smooth family of Weyl operators and its genuine gauge-invariant chiral determinant; compactness then gives integrability. Fermionic insertions can saturate zero modes in sectors excluded from the vacuum integral, as treated by Theorem~\ref{thm:smeared_insertions}. Vanishing species indices alone do not prove that the vacuum weight is nonzero.

\paragraph{Scope of the theorem.}
The theorem establishes the existence of the integral but does not determine its value. In particular, it does not establish $Z\ne0$, reality or positivity of the integrand, locality of the induced fermionic action in the original lattice links, or a controlled $a\to0$ limit. The infinitely many fermion modes are treated by the specified continuum determinant-line regularization, not by a finite matrix cutoff. Appendix~\ref{sec:collar_variant} transfers these existence and source-finiteness results to a variant gauge field reconstruction and proves its site-reflection positivity at zero Yukawa coupling and zero theta.

\section{Fermion correlators and zero modes}\label{sec:corr}

\subsection{Fixed-background propagators}
Fermion correlation functions are obtained by introducing Grassmann sources $\eta,\bar\eta$. For a background at which the chiral Dirac operator $D_U=D^+_{A[U]}$ is invertible, Gaussian integration gives
\begin{equation}
\cZ_F[\bar\eta,\eta;U]=W_0(U)\exp\!\left(
\langle\bar\eta,D_U^{-1}\eta\rangle\right).
\label{eq:sources}
\end{equation}
The bracket denotes contraction of the source fields with the Green kernel of $D_U$, including spin and gauge indices. Functional differentiation generates the propagators and Wick contractions at fixed background. Physical expectation values require subsequent integration over the bosonic configurations.

To avoid singularities associated with coincident-point insertions, we smear fermion fields against smooth test functions: in momentum language, the sources have controlled high-momentum tails rather than probing a field at a mathematical point. Take $f_U$ in the positive-chirality space $L^2(S^+\otimes E_{\mathcal R}[U])$ and $g_U$ in the negative-chirality space $L^2(S^-\otimes E_{\mathcal R}[U])$. The inverse maps $g_U$ back to positive chirality, where it pairs with $f_U$; our inner product is antilinear in its first slot. The profiles are smooth in spacetime and continuous as the links vary, after local unitary identifications of the bundles (Section~\ref{sec:source_families}). Define the \emph{unnormalized} two-point insertion on invertible backgrounds by
\begin{equation}
I_{f,g}(U)=W_0(U)\,\langle f_U,D_U^{-1}g_U\rangle.
\label{eq:minor}
\end{equation}
The determinant factor in $I_{f,g}$ is essential. The bare inverse $D_U^{-1}$ is undefined at zero modes, whereas the determinant-weighted insertion admits a continuous extension. This is the familiar cancellation between a vanishing fermion determinant and singular propagator factors when external fields saturate a zero mode; the resulting insertion is not obtained by dividing two separately vanishing amplitudes.

The finite-dimensional analogue is the adjugate identity. For invertible $M$,
\begin{equation}
(\det M)(M^{-1})_{ij}=(\operatorname{adj}M)_{ij}.
\label{eq:adj}
\end{equation}
The adjugate matrix is polynomial in the entries of $M$ and remains defined when $\det M=0$. For several source pairs, Wick contractions give $\det M$ times a determinant of propagator entries, rather than an unantisymmetrized product. Their continuum analogues are regularized \emph{determinant minors}. Finite-dimensional reduction defines them directly at zero modes and at nonzero index; on invertible backgrounds they agree with the usual Wick expressions. Like the determinant itself, they initially take values in a determinant line.

For a gauge-invariant insertion, the bosonic average is
\begin{equation}
\mathcal C_{f,g}=\int d\mu_B(U,\phi)\;I_{f,g}(U),
\quad d\mu_B=\prod_\ell dU_\ell\prod_xd\phi_x\,
\chi_\epsilon b_{R_H}e^{-S_{gH}}.
\label{eq:correlator}
\end{equation}
For $Z\ne0$, the normalized expectation value is $\mathcal C_{f,g}/Z$. An individual quark propagator transforms covariantly under gauge transformations. Gauge-invariant observables instead require appropriate gauge-index contractions or, for separated charged insertions, a Wilson line; gauge-fixed correlators must specify the gauge choice.

\subsection{Zero modes and topological sectors}
A zero mode is a normalizable solution of $D_Uu=0$, so $D_U^{-1}$ does not exist on that background. Nevertheless, Gaussian Grassmann integrals with suitable fermionic insertions remain defined. For a single Grassmann variable,
\begin{equation}
\int d\xi\;1=0,\qquad \int d\xi\;\xi=1
\label{eq:grassmann}
\end{equation}
An unsaturated zero mode (one not paired with a fermionic insertion) makes the vacuum integral vanish; a fermionic insertion containing the corresponding Grassmann variable can give a nonzero result.

Let $\{u_i\}$ span $\ker D_U$ and $\{v_j\}$ span $\ker D_U^\dagger$. Define $G_U$ as the inverse restricted to their orthogonal complements. After choosing local bases for the zero modes, the source integral has the schematic form
\begin{equation}
\cZ_F[\bar\eta,\eta;U]
=\mathcal N_U\,e^{\langle\bar\eta,G_U\eta\rangle}
\prod_i\langle\bar\eta,u_i\rangle
\prod_j\langle v_j,\eta\rangle.
\label{eq:zeroformula}
\end{equation}
The normalization $\mathcal N_U$ includes the regularized nonzero-mode determinant, the orientation of the zero-mode spaces, and Grassmann-ordering conventions. Different choices of how many low modes to isolate, or of bases in those low-mode spaces, must give the same result. The determinant-line transition rules enforce that independence. Consequently, sectors with nonzero index can contribute to correlators whose insertions saturate the zero modes, although their vacuum partition-function contribution vanishes.

As an example, take weak topological charge $k_2=+1$, with no additional relevant flux. In the all-left-handed convention used for $D_Q$ and $D_L$, the indices per generation are $\operatorname{ind}D_Q=-3$ and $\operatorname{ind}D_L=-1$. Thus the index-required unpaired modes lie in the \emph{cokernels} (zero modes of the adjoint), and the Grassmann insertions that saturate them are conjugate fields. Extra kernel--cokernel pairs may occur at special backgrounds and require additional source saturation; the index fixes the net excess and its minimal source factor. For three generations the source factor is schematically
\begin{equation}
\prod_{g=1}^3(\bar Q\bar Q\bar Q\bar L)_g,
\label{eq:tHooft}
\end{equation}
with gauge and spinor indices contracted to a singlet. The appearance of these conjugate Euclidean fields does not mean that physical antiquarks and antileptons are emitted. In the corresponding Minkowski matrix elements, conjugate fermion fields contain creation operators for outgoing quarks and leptons; the mapping to external particle states is fixed by the usual continuation and LSZ reduction, together with the orientation of the topological sector. Thus the physical transition produces $QQQL$ quanta per generation and has $\Delta B=\Delta L=+3$ for three generations. The conjugation in the Euclidean source factor reflects the cokernel and source conventions, not a reversal of the physical $B+L$ selection rule. The coupled-operator index used later is positive because it uses the opposite source/target orientation; it describes the same anomaly. A charge-one \emph{sector} on $T^4$ is meant; no periodic self-dual charge-one $SU(2)$ instanton is required~\cite{braam}.

\noindent\emph{Physical interpretation.} A nonzero index does not mean the regulator has lost the fermions. It means the vacuum amplitude vanishes until enough external fermionic operators absorb the unpaired Grassmann integrations, exactly as in the usual instanton selection rule~\cite{thooft}.

\subsection{Continuity and finiteness of fermionic insertions}
The continuum counterpart of~\eqref{eq:adj} is defined by \emph{finite-dimensional reduction}: isolate the modes below a fixed local spectral gap, invert the complementary operator, and evaluate the low-mode Grassmann integral exactly. The latter is polynomial in the sources and reduced matrix entries, including when that matrix is rectangular. Source continuity means $L^2$ continuity after local smooth unitary identifications of the reconstructed bundles, with the external reference indices transformed under gauge transformations. Appendix~\ref{sec:source_families} specifies these identifications and gives a covariant source example.

\begin{theorem}[Finite smeared fermion insertions]\label{thm:smeared_insertions}
Fix finitely many smooth smeared source families with the $L^2$ continuity and gauge covariance just specified. Finite-dimensional reduction defines their fermion integrals as continuous determinant-line source coefficients in every index sector. Wherever $D_U$ is invertible, these coefficients equal the determinant times the corresponding antisymmetrized Wick contractions. After applying the same gauge-equivariant phase convention as for $W_0$, any continuous gauge-invariant contraction with the bosonic fields has an absolutely convergent finite-volume integral against $d\mu_B$. The ambient references of Section~\ref{sec:hsu_eta} and Appendix~\ref{app:phase_references} are used also in nonzero-index sectors.
\end{theorem}
\begin{proof}
Appendix~\ref{sec:source_families} identifies the varying Hilbert spaces and proves continuity of the complementary inverse and source projections. The low-mode Grassmann integral is polynomial, so every source coefficient is continuous. The exact Gaussian identities~\eqref{eq:rect_sources} and the shell transitions~\eqref{eq:window_isometry}--\eqref{eq:naturality_shell} in Appendix~\ref{sec:quillen_gluing} identify the expressions on overlaps, including at nonzero index. Applying the fixed equivariant phase convention gives continuous covariant scalar coefficients. Their continuous gauge-invariant contractions are bounded on the compact regulated support and hence absolutely integrable.
\end{proof}
\noindent\emph{Physical interpretation.} Multiplying propagators by the corresponding determinant replaces singular inverses by regularized minors, so the bosonic average of a smooth gauge-invariant fermionic observable remains finite even at backgrounds with zero modes.

The theorem applies to \emph{smooth smeared insertions}. Coincident-point composite operators can have ultraviolet singularities and require a separate renormalization prescription. For numerical work, a finite spectral truncation may approximate the continuum operator, but its cutoff removal and phase convention must converge to those of the specified determinant-line prescription. An arbitrary finite truncation need not preserve exact gauge invariance.

\section{Yukawa couplings and the determinant of the coupled fermion operator}\label{sec:yukawa_full}

The Yukawa sector adds no new Grassmann integration variables: for fixed gauge and Higgs backgrounds it remains quadratic in the fermions. Integrating them out is therefore a Gaussian calculation. Its logarithm is the sum of one-fermion-loop vacuum diagrams with arbitrarily many external gauge and Higgs insertions; this statement is exact for fixed bosonic backgrounds, not a weak-coupling approximation to the subsequent bosonic integral. The main questions are whether the resulting Pfaffian can be treated by ordinary determinant-line theory, whether its complex phase is gauge invariant, and which ultraviolet prescription makes that phase physically meaningful. We answer them in that order, with the common cutoff as the operative regulator.

\subsection{Yukawa interactions and reduction to an ordinary determinant}
Let $r,s=1,2,3$ label generations, $\alpha$ color indices, and $i,j$ weak indices. With the antisymmetric two-spinor contraction understood, the action is
\begin{equation}
\begin{aligned}
S_Y=\int_Xd^4y\;\big[&(Y_u)_{rs}\epsilon_{ij}Q_r^{\alpha i}\Phi^j u^c_{s\alpha}
 +(Y_d)_{rs}Q_r^{\alpha i}\Phi_i^*d^c_{s\alpha}\\
 &+(Y_e)_{rs}L_r^i\Phi_i^*e^c_s+\text{conjugate terms}\big].
\end{aligned}
\label{eq:yukawa}
\end{equation}
The $3\times3$ complex matrices $Y_u,Y_d,Y_e$ are arbitrary \emph{finite input parameters} of the regulator. Renormalized Yukawa couplings are extracted from three-point functions at a specified scale and receive gauge/Higgs corrections, including QCD corrections to the quark Yukawas~\cite{machacek_vaughn}. The coupling renormalization is distinct from the task of defining the fermion determinant. Gauge invariance of the interactions follows from hypercharge and bundle pairings; detailed representation identities are collected in Appendix~\ref{app:geometry}.

Write $D_U=D^+_{A[U]}$ and introduce independent Euclidean conjugate variables. The fermion action can be written in antisymmetric Nambu--Gorkov form. Nambu--Gorkov notation collects the independent Euclidean fields and conjugate variables into one column, so that the kinetic and Yukawa terms form a single skew quadratic form. The corresponding Gaussian integral is a Pfaffian, whose square is the determinant of the skew matrix:
\begin{equation}
\Xi=\begin{pmatrix}\psi\\\bar\psi^T\end{pmatrix},\qquad
S_F=\tfrac12\int_X\Xi^T\mathcal K_Y\Xi,\qquad
\mathcal K_Y=\begin{pmatrix}M_\Phi&-D_U^T\\D_U&\overline M_\Phi\end{pmatrix},
\quad \mathcal K_Y^T=-\mathcal K_Y.
\label{eq:nambu}
\end{equation}
The two skew blocks $M_\Phi$ and $\overline M_\Phi$ are built respectively from the Yukawa/Higgs data and their complex conjugates; the bar labels the conjugate couplings and does not identify the independent Euclidean Grassmann variables. Each block is skew because of the antisymmetric spinor contraction, not because any generation matrix is symmetric. The finite-dimensional sign and block verification appear in Appendix~\ref{app:grushin}.

Every kinetic and Yukawa term pairs fields of opposite $B+L$ charge. Group the fields of negative $B+L$ charge $(u^c,d^c,e^c,\bar Q,\bar L)$ first (the source of $\mathcal D_Y$), and those of positive charge $(Q,L,\bar u^c,\bar d^c,\bar e^c)$ second (its target). In this ordered decomposition the lower-left block maps the negative group to the positive group, and the skew matrix becomes
\begin{equation}
\mathcal K_Y\simeq
\begin{pmatrix}0&-\mathcal D_Y^T\\\mathcal D_Y&0\end{pmatrix},
\qquad \Pf\mathcal K_Y=\pm\det\mathcal D_Y.
\label{eq:yukawa_determinant}
\end{equation}
The overall sign is a fixed Grassmann-ordering convention. At zero Yukawa coupling the lower-left block is $\mathcal D_0$, assembled from $D_{u^c},D_{d^c},D_{e^c}$ and the transposed $D_Q,D_L$ blocks in the indicated source-to-target order. Thus $\operatorname{ind}\mathcal D_q=\operatorname{ind}D_{u^c}+\operatorname{ind}D_{d^c}-\operatorname{ind}D_Q$, and similarly for leptons; transposition reverses the Fredholm index.

The zero-Yukawa determinant line of $\mathcal D_0$ is identified with the original Weyl determinant line: transposition preserves nonzero singular values and hence the Quillen norm. With the stated Grassmann ordering fixed, the determinant-line dualities are natural and gauge equivariant, so they carry the convention $W_0$ across without a further phase choice. We use this identification throughout the Yukawa section. The source/target duality implicit in the bilinear notation is stated in Appendix~\ref{app:grushin}.

The Yukawa operator $\mathcal D_Y=\mathcal D_0+\mathcal V_Y$ is first-order elliptic because its Yukawa part $\mathcal V_Y$ is zeroth order. Thus the determinant-line construction applies directly. The decomposition is a statement about the \emph{classical bilinear}; the quantum $B+L$ current remains anomalous. Majorana terms, which would spoil this decomposition, are not included.\par
\noindent\emph{Physical interpretation.} With no Majorana masses, the two groups of Grassmann variables pair through an ordinary determinant. The anomalous $B+L$ selection rule is still present: it follows from the index and zero-mode saturation, not from an explicit $B+L$-violating Yukawa term.

\subsection{Gauge-invariant Yukawa phase conventions}
There are two distinct questions. The first is whether a smooth, gauge-invariant fermion weight \emph{exists} once Yukawa couplings are switched on. The answer is yes at fixed bosonic lattice spacing, but gauge invariance alone permits nonlocal, background-dependent phase choices. The second question is which weight has the expected \emph{local ultraviolet renormalization}. To address it, Section~\ref{sec:yukawa_cutoff} specifies a single gauge-covariant filter for the Yukawa interactions and fixes the finite matching conditions. The existence theorem below should not be confused with this more restrictive prescription.

We first establish a general existence result: Yukawa interactions introduce no new gauge-anomaly obstruction. This alone does not select a unique physical weight. The common spectral prescription in Sections~\ref{sec:yukawa_cutoff}--\ref{sec:quillen_matching} selects the normalized weight $W_Q$ used in Theorem~\ref{thm:renormalized_weight}.

The canonical determinant section $s_{\mathcal D_Y}$ has a gauge-invariant Quillen norm. At invertible backgrounds,
\begin{equation}
\|s_{\mathcal D_Y}\|_Q=
\left[\det_\zeta(\mathcal D_Y^\dagger\mathcal D_Y/\mu^2)\right]^{1/2}
=\left[\det_\zeta(\mathcal K_Y^\dagger\mathcal K_Y/\mu^2)\right]^{1/4}.
\label{eq:pf_norm}
\end{equation}
At zero modes, this formula is interpreted through finite-dimensional determinant-line frames. The norm fixes the magnitude, not the complex phase.

Because $\mathcal V_Y(U,t\phi)=t\mathcal V_Y(U,\phi)$, the family
\begin{equation}
\mathcal D_t(U,\phi)=\mathcal D_0(U)+t\mathcal V_Y(U,\phi),\qquad 0\le t\le1,
\label{eq:yukawa_homotopy}
\end{equation}
contracts equivariantly to the zero-Yukawa theory, including through zero-mode crossings.
\begin{lemma}[Gauge-equivariant determinant-line extension]
The zero-Yukawa gauge-equivariant trivialization extends smoothly and unitarily to the determinant line of $\mathcal D_Y$ for arbitrary finite Standard Model Yukawa matrices. This proves that a gauge-invariant extension \emph{exists}; it does not choose a unique extension, and it supplies no locality bound on its phase.
\end{lemma}
\begin{proof}
Choose a smooth unitary connection on the determinant line and average it over the compact lattice gauge group. Radial parallel transport $T_{U,\phi}$ from $(U,0)$ to $(U,\phi)$ is then gauge equivariant. Composing the massless trivialization $\tau_U$ with its inverse gives
\begin{equation}
\tau^Y_{U,\phi}=\tau_U\circ T_{U,\phi}^{-1},\qquad
W_Y(U,\phi)=\tau^Y_{U,\phi}\bigl(s_{\mathcal D_Y}(U,\phi)\bigr).
\label{eq:pf_weight}
\end{equation}
This transport is defined on the \emph{line} even when the determinant section vanishes; at $\phi=0$ it is the identity. Consequently
\begin{equation}
W_Y(U^g,\phi^g)=W_Y(U,\phi),\qquad W_Y(U,0)=W_0(U).
\label{eq:pf_gauge}
\end{equation}
\end{proof}

For a source $J$ coupled to $\Xi$, the unnormalized Gaussian generating functional on invertible backgrounds is
\begin{equation}
\mathcal Z_F[J;U,\phi]=W_Y(U,\phi)\exp[-\tfrac12 J^T\mathcal K_Y^{-1}J].
\label{eq:pf_sources}
\end{equation}
At zero modes its coefficients are defined by the finite-dimensional source construction of Theorem~\ref{thm:smeared_insertions}. The explicit reduction is in Appendix~\ref{sec:source_families}.

\begin{theorem}[Fixed-regulator gauge--Higgs--Yukawa integral]
For the fixed compact bosonic regulators, smooth reconstructed fields, arbitrary finite Yukawa matrices, Quillen metric and any smooth gauge-equivariant phase convention just constructed, the integral
\begin{equation}
\boxed{Z_Y=\int d\mu_B(U,\phi)\;W_Y(U,\phi)}
\label{eq:Z_yukawa}
\end{equation}
is absolutely convergent and exactly gauge invariant. The same holds for bounded gauge-invariant bosonic insertions and fixed finite smooth-smeared gauge-invariant fermionic insertions. The statement does not require small Yukawa couplings or a nonzero Higgs background.
\end{theorem}
\begin{proof}
The spectral-window definition of Theorem~\ref{thm:smeared_insertions} applies to $\mathcal D_Y$ with the same principal symbol, including at nonzero index. The determinant section and its source coefficients are continuous; the equivariant trivialization makes them gauge covariant. They are bounded on the compact bosonic support, so their integrals converge absolutely.
\end{proof}
This theorem establishes existence and gauge invariance for the specified input data; by itself it does not distinguish the intended Standard Model from every other gauge-invariant weight. For example, multiplying $W_Y$ by $e^{i\chi(U,\phi)}$, with any smooth gauge-invariant $\chi$ satisfying $\chi(U,0)=0$, preserves the theorem even if $\chi$ is nonlocal. The common spectral prescription restricts differences between admissible regulators to local finite renormalizations under the analytic conditions below. The remaining freedom to choose physical theta angles and complex Yukawa matrices is the familiar freedom to specify the Standard Model couplings; it is not an obstruction from gauge anomalies.

\subsection{Zero modes and the anomalous \texorpdfstring{$B+L$}{B+L} selection rule}
The quark and lepton sectors remain separate in the Yukawa operator: $\mathcal D_Y=\mathcal D_q\oplus\mathcal D_\ell$. Put $\mathfrak c=c_1(E_2)=-c_1(E_3)$ and $k_2=\int_Xc_2(E_2)$. The index theorem gives, for each generation,
\begin{equation}
\ind\mathcal D_q=\ind D_{u^c}+\ind D_{d^c}-\ind D_Q=3k_2,
\qquad \ind\mathcal D_\ell=\ind D_{e^c}-\ind D_L=k_2.
\label{eq:baryon_idx}
\end{equation}
For three generations these are $9k_2$ and $3k_2$. Yukawa couplings are zeroth-order and cannot change the indices. When $k_2\ne0$, the vacuum determinant sections vanish, but external sources can saturate the unpaired variables. For $k_2=+1$, the index requires a net twelve-field saturation, with the $B+L$ charges described in Section~\ref{sec:corr}. At zero Yukawa coupling its minimal all-left-handed source factor is $(\bar Q\bar Q\bar Q\bar L)^3$. Yukawa couplings can mix the external components that saturate the modes, and accidental paired zero modes can require extra insertions, but the net baryon and lepton number is unchanged. Counting these charges gives the integrated Ward identity
\begin{equation}
\Delta B=\Delta L=3k_2,\qquad \Delta(B-L)=0.
\label{eq:BL_selection}
\end{equation}
This holds at each finite fermion cutoff because the regulator preserves the $B,L$ grading. These signs refer to outgoing physical charges for the stated orientation of topological charge. The conjugate Euclidean source fields in Eq.~\eqref{eq:tHooft} transform oppositely because they saturate cokernel zero modes; the choice is fixed by the operator convention, not by a relabeling of sources. The electroweak theta angle is removable by an anomalous baryon-number redefinition, while other theta angles and Yukawa phases can affect the vacuum theory. Indeed, for each generation the index-weighted baryon and lepton charges satisfy $\sum_i B_i\operatorname{ind}D_i=\tfrac13(\operatorname{ind}D_Q-\operatorname{ind}D_{u^c}-\operatorname{ind}D_{d^c})=-k_2$ and $\sum_i L_i\operatorname{ind}D_i=\operatorname{ind}D_L-\operatorname{ind}D_{e^c}=-k_2$. The $c_1^2$ and color-charge terms cancel in both sums, so these rephasings shift only the electroweak topological angle. The rectangular low-mode Grassmann calculation is given in Appendix~\ref{app:grushin}.\par
\noindent\emph{Physical interpretation.} Changing the Yukawa matrices cannot eliminate the index. The same $B+L$ selection rule therefore holds at every fermion cutoff and for any finite Yukawa couplings.

\subsection{A common spectral prescription for the Yukawa phase}\label{sec:yukawa_cutoff}
The massless determinant already has a global phase convention. Instead of assigning a new phase independently at every Higgs background, we define the Yukawa theory as a regulated \emph{change} of that massless determinant. The cutoff acts only on the Yukawa multiplication operator, not on the massless kinetic term. Define
\begin{equation}
H_+=\sigma^2+\mathcal D_0^\dagger\mathcal D_0,\qquad H_-=\sigma^2+\mathcal D_0\mathcal D_0^\dagger,\qquad
\mathcal V_{Y,\Lambda}=f(H_-/\Lambda^2)\,\mathcal V_Y\,f(H_+/\Lambda^2),
\label{eq:common_cutoff}
\end{equation}
where $\sigma>0$ is a fixed mass scale and $f\in C_c^\infty([0,\infty))$ is fixed, real, and equal to one near zero. The operator $f(H_\pm/\Lambda^2)$ is a \emph{spectral filter}: on an eigenmode of $H_\pm$ with eigenvalue $h$ it multiplies the mode by $f(h/\Lambda^2)$, exactly one on the low-energy plateau and zero above the profile support. Compact support, not merely rapid decay, makes each filter finite rank: only finitely many eigenvalues of an elliptic operator on compact spacetime lie below a fixed cutoff. The auxiliary infrared shift $\sigma$ is independent of the zeta subtraction scale $\mu$; its effect is removed by the finite matching below. The filters are gauge covariant at finite $\Lambda$. They never require $\mathcal D_0^{-1}$ or a preferred Higgs direction, even in a sector of nonzero index. In momentum-space language the filter suppresses high-energy modes \emph{at each Yukawa vertex}, rather than replacing the massless fermion propagator by a finite lattice matrix. Using the covariant kinetic operator to define the cutoff preserves background gauge invariance, analogous to a covariant momentum cutoff. A \emph{nested} filter means that a loop contains one factor on each side of each Yukawa insertion; when the Higgs field varies, these factors do not commute through the Yukawa vertex. The finite-cutoff Yukawa theory is a finite-rank perturbation of the established massless determinant-line theory, so the complete source-dependent phase is defined before any choice of low-mode coordinates.

\paragraph{Two kinds of fermion vacuum loops.}
The full weight consists of the gauge-only determinant $W_0(U)$ and the \emph{additional} vacuum loops containing Yukawa vertices. On a background where $D_U$ is invertible, and suppressing determinant-line conventions, the exact finite-rank relative-determinant identity reads
\[
 W_{Y,\Lambda}(U,\Phi)=W_0(U)\,
 \det_F\!\left(1+\mathcal D_0^{-1}\mathcal V_{Y,\Lambda}\right).
\]
Here $\det_F$ is the determinant of the identity plus a finite-rank operator; it is an ordinary finite-dimensional determinant on the subspace affected by the Yukawa interaction. When $D_U$ has zero modes or nonzero index, the determinant-line construction with explicit low-mode Grassmann factors replaces the displayed inverse, without changing the physical separation into gauge-only and Yukawa-dependent contributions. This is not a second cutoff on the momentum integration: the gauge-only part already includes \emph{all} continuum fermion modes, and the interaction part is regulated by modifying the Yukawa vertices.

For example, the flat-space momentum integral for a loop with two Yukawa insertions carrying external Higgs momentum $q$ has the schematic form below. It displays the local ultraviolet structure; on our finite torus the corresponding complete amplitude uses a discrete momentum sum.
\[
\begin{aligned}
 \mathcal A_{2,\Lambda}(q)\ \propto\ &\int\!\frac{d^4p}{(2\pi)^4}\,
 \operatorname{tr}\!\left[S(p)\mathcal V_Y(q)
                 S(p+q)\mathcal V_Y(-q)\right]\\[-2pt]
 &\quad\times f\!\left(\frac{p^2+\sigma^2}{\Lambda^2}\right)^2
 f\!\left(\frac{(p+q)^2+\sigma^2}{\Lambda^2}\right)^2 .
\end{aligned}
\]
$S(p)$ is the \emph{unmodified} continuum propagator; spinor and species indices and an overall sign have been suppressed. The two momenta probe different parts of the cutoff profile, which matters when extracting the Higgs wavefunction counterterm. On general gauge backgrounds, eigenvalues of the gauge-covariant kinetic operator replace free momentum labels. Thus filtering vertices retains the continuum propagation and its familiar loop interpretation, while regulating the \emph{change} caused by Yukawa interactions.

\Needspace{11\baselineskip}
The number of required subtractions follows from an elementary one-loop Feynman diagram. A closed loop with $n$ external Higgs/Yukawa insertions contains $n$ fermion propagators, each behaving as $1/\not\!p$ at large Euclidean loop momentum. For fixed external momenta its superficial degree of divergence is
\[
\mathcal A_n\sim\int^{\Lambda}d^4p\,p^{-n}
\sim\int^{\Lambda}dp\,p^{3-n},\qquad \omega_n=4-n.
\]
Thus loops with $n\ge5$ Yukawa insertions are ultraviolet convergent. Gauge-invariant local bosonic terms contain equal numbers of $\Phi$ and $\Phi^*$, by hypercharge neutrality, so their Yukawa degree is even. The quadratic order can renormalize $\Phi^\dagger\Phi$ and, at second order in external momentum, $(D_\mu\Phi)^\dagger D_\mu\Phi$; the quartic order gives the logarithmically divergent $(\Phi^\dagger\Phi)^2$ term. Each external momentum lowers $\omega_n$ by one.

The exact symmetry behind this counting is $g_*=\operatorname{diag}(\mathbf1_3,-\mathbf1_2)\in S(U(3)\times U(2))$. It leaves every link unchanged and reverses both the Higgs doublet and the weak-doublet fermions. Gauge invariance of the full Gaussian integral gives $W_Y(U,-\phi)=W_Y(U,\phi)$. With fermionic insertions, the sources must be transformed too; individual coefficients at fixed sources need not be even. Appendix~\ref{app:yukawa_uv} expresses this symmetry as an operator grading. Auxiliary low-mode coordinates may obscure it in individual traces, so the general construction retains the first four orders before their physical cancellations.

As in ordinary renormalization of one-loop Green functions, only the ultraviolet-sensitive \emph{local} polynomials in fields and external momenta are subtracted. The finite, generally nonpolynomial momentum dependence of the Higgs amplitudes must remain. With one fixed filter, the renormalized source functional is
\begin{equation}
\boxed{\mathcal Z_Y^{\rm ren}[J;U,\phi]
=\lim_{\Lambda\to\infty}
 e^{-C_\Lambda^{\rm div}(U,\phi)}\,\mathcal Z_{\mathcal D_0+\mathcal V_{Y,\Lambda}}[J],}
\label{eq:renormalized_source}
\end{equation}
where $C_\Lambda^{\rm div}$ contains the divergent local parts of the one- through four-point Yukawa amplitudes; finite local terms are fixed by explicit renormalization conventions. The multiplicative factor $e^{-C_\Lambda^{\rm div}}$ acts on the fermion \emph{weight}; its sign is reversed when expressed as a counterterm in the Euclidean effective action $-\log W$. In practice this means subtracting powers and logarithms of $\Lambda$ while retaining the finite momentum dependence of each amplitude, just as in conventional perturbative renormalization. The limit is coefficientwise for fixed smooth-smeared sources. The filter itself has cutoff-scale nonlocality: the locality assertion concerns the \emph{subtractions}, not the complete integrated-out effective action.\par

\paragraph{Electroweak Ward identities and cutoff-scale nonlocality.}
The regulator leaves the massless kinetic operator unchanged, while its Yukawa vertices are built from gauge-covariant spectral functions. If $R_+(g)$ and $R_-(g)$ denote the gauge actions on the source and target spaces, then the complete vertex obeys
\[
\mathcal V_{Y,\Lambda}[A^g,\Phi^g]
=R_-(g)\mathcal V_{Y,\Lambda}[A,\Phi]R_+(g)^{-1}.
\]
Together with the anomaly-free determinant phase and gauge-invariant subtractions, this yields exact background-field electroweak Ward identities for the source-dependent fermion weight at every finite $\Lambda$. Variations of the covariant filters contribute gauge-field vertices in addition to those obtained by varying the propagator; both contributions are required. The filters have rapidly decreasing coordinate-space kernels on distances set by $\Lambda^{-1}$, not strictly compact coordinate-space support. As $\Lambda\to\infty$ at fixed smooth bosonic background, the filter effects vanish for ultraviolet-convergent fermion amplitudes; for divergent amplitudes, the remaining scheme dependence is local and is removed by the specified gauge-invariant matching terms. These results concern the regulated fermion determinant and its cutoff-removed limit. They do not establish reflection positivity or Minkowski unitarity for finite $\Lambda$, nor the full gauge-fixed interacting Slavnov--Taylor identities after bosonic integration and removal of the independent bosonic regulator.
\noindent\emph{Physical interpretation.} This is the familiar separation between local ultraviolet counterterms and finite momentum-dependent loop amplitudes. With the zero-Yukawa reference held fixed, the relevant ultraviolet terms renormalize Higgs mass, kinetic and quartic operators. The nonlocal part of the amplitude is common to admissible renormalization prescriptions and must be retained.

The only analytic estimate needed for the low-order subtraction has a simple physical meaning. The difference between two ways of placing smooth spectral cutoffs probes momenta $p\sim\Lambda$. Its first omitted term has five inverse powers of $\Lambda$; summing over the $O(\Lambda^4)$ modes in a four-dimensional momentum shell leaves an error $O(\Lambda^{-1})$. This is the usual large-loop-momentum expansion in external fields and momenta, made gauge covariant by using the Dirac kinetic operator. Appendix~\ref{sec:uniform_shell_bound} gives the operator version of this argument, including derivatives with respect to the reconstructed background. The zero-mode factors are handled separately by finite-dimensional algebra.\par

The cutoff can now be removed in two stages.
\begin{enumerate}
\item \emph{Convergent higher-order loops.} The relative Yukawa operator $K$ is schematically a massless propagator times a Yukawa vertex. When zero modes are present, an invertible auxiliary extension defines this comparison without inverting the original chiral operator. In four dimensions $K$ belongs to the Schatten class $\mathcal S_5$: the sum of the fifth powers of its singular values is finite. This is the operator form of ultraviolet convergence for loops with five or more Yukawa insertions. The regularized determinant $\det_5$ collects these contributions.
\item \emph{Renormalized lower-order loops.} For the first four orders, a Mellin transform rewrites the smooth cutoff as a superposition of powers of the kinetic operator, much as Schwinger parameters rewrite propagators. The ultraviolet-shell estimate compares the filtered loops with complex-zeta regularization, which analytically continues the same divergent traces. Their difference is local on compact smooth families; the finite nonlocal responses are retained.
\end{enumerate}
Appendices~\ref{sec:mellin_anchor} and~\ref{sec:uniform_shell_bound} establish the comparison. Appendix~\ref{sec:quillen_gluing} then shows that the high-mode determinant and finite zero-mode source factors fit together consistently. Divergent subtractions do not depend on the auxiliary low-mode coordinates. Finite local matching fixes the remaining profile dependence and normalization as explained below. They use smooth, uniformly elliptic reconstructed backgrounds on compact support inside the admissible domain, local invertible operator charts, and fixed subtraction conventions. Appendix~\ref{app:yukawa_uv} organizes this standing setup and the analytic tools as (A1)--(A4), and derives the required uniform estimates at fixed bosonic lattice spacing.

\subsection{Mixed gauge--Yukawa divergences and the two-stage regulator}\label{sec:mixed_renorm}
\emph{Does separating the gauge-only determinant from the Yukawa-dependent correction miss diagrams containing both interactions?} No. The separation is by the \emph{number of Yukawa insertions}, not by the presence or absence of gauge fields. At fixed smooth bosonic background, write schematically
\[
 \Gamma_{\rm ferm,\Lambda}[A,\Phi]
 =\underbrace{-\log W_0[A]}_{\Gamma_0[A]}
  +\underbrace{\Delta\Gamma_{Y,\Lambda}[A,\Phi]}_{\text{at least one Yukawa insertion}}.
\]
For an invertible chiral operator, this is the logarithm of the relative-determinant identity displayed above. On noninvertible backgrounds it denotes the corresponding determinant-line and source-functional construction; no ratio of vanishing vacuum determinants is used. The \emph{reference} determinant $W_0[A]$ has already been defined by continuum spectral regularization and includes the fermion vacuum loops with external gauge fields but no Yukawa vertices. The second term contains the additional loops generated by the Yukawa interaction.

\paragraph{Where the mixed diagrams are.}
The fermion propagator in the relative determinant is $S_A=D_A^{-1}$, the propagator in the \emph{complete background gauge field}. In perturbation theory around $A=0$, one may expand it as
\[
 S_A=S_0-S_0\not\!A\,S_0
       +S_0\not\!A\,S_0\not\!A\,S_0-\cdots,
\]
where the gauge couplings and representation matrices are understood to be included in $\not\!A$. Inserting this series into a fermion loop with Yukawa vertices generates diagrams with any number of external gauge-boson legs. There are also gauge-field insertions from expanding the \emph{covariant filters} $f(H_\pm[A]/\Lambda^2)$ themselves. Both types are part of the regulated amplitude and must be kept together to satisfy the finite-cutoff gauge Ward identities. The filter suppresses high-energy contributions \emph{at the Yukawa vertices}; it is not a separate cutoff on the gauge vertices or on the momentum integration.

Ordinary large-momentum power counting for a closed fermion loop with $n_Y$ Yukawa insertions and $n_A$ explicit gauge vertices gives, before Ward-identity cancellations or expansion in external momenta,
\[
 \omega_{n_Y,n_A}=4-n_Y-n_A.
\]
Each extra power of external momentum lowers this degree by one. For $n_Y=0$, the gauge-only determinant carries the usual fermionic contribution to gauge-field renormalization. For $n_Y>0$, electroweak gauge invariance eliminates odd Higgs powers in the complete weight and organizes the potentially divergent terms as follows:
\begin{center}
\small
\renewcommand{\arraystretch}{1.16}
\begin{tabularx}{\linewidth}{@{}p{.33\linewidth}X@{}}
\toprule
Fermion loop & Allowed local ultraviolet contribution \\
\midrule
Gauge insertions, no Yukawa vertex & Gauge kinetic operators $\operatorname{tr}F_{\mu\nu}F^{\mu\nu}$; included in $W_0$ \\
Two Yukawa vertices & Higgs mass term $\Phi^\dagger\Phi$ and covariant kinetic term $(D_\mu\Phi)^\dagger D_\mu\Phi$ \\
Two Yukawa vertices with external gauge bosons & The terms with one or two gauge fields obtained by expanding $(D_\mu\Phi)^\dagger D_\mu\Phi$ \\
Four Yukawa vertices & Higgs quartic operator $(\Phi^\dagger\Phi)^2$ \\
Five or more Yukawa vertices & Ultraviolet convergent; no new Yukawa counterterm \\
\bottomrule
\end{tabularx}
\end{center}
In particular, the mixed two-Higgs--one-gauge and two-Higgs--two-gauge divergences are \emph{not} missing counterterms. They are the gauge-field pieces of the same covariant Higgs kinetic operator, since $D_\mu\Phi$ contains the electroweak gauge connection. Conversely, the Yukawa-dependent determinant does not introduce a new divergent pure $F_{\mu\nu}^2$ term at this fermion-loop order: such a term has no Higgs insertions and belongs to the reference determinant. Finite nonlocal amplitudes involving both Higgs and gauge backgrounds are retained.

\paragraph{No double counting or independent ultraviolet limits.}
At zero Yukawa interaction the finite-rank perturbation vanishes and its relative determinant is one on an invertible chart; equivalently, the source-dependent construction reduces to the massless one in any index sector. Thus the Yukawa subtraction terms vanish with the Yukawa interaction and cannot subtract a pure-gauge loop a second time. The reference determinant has its own spectral scale and finite normalization, which must be matched to the bare bosonic gauge couplings. The Yukawa loops have additional local Higgs counterterms, fixed in a compatible subtraction convention. Changing schemes may shift these local coefficients, as in ordinary renormalization, without changing the nonlocal physics. The relation $\bar\theta_{\rm QCD}=\theta_{\rm QCD}+\arg\det(Y_uY_d)$ is a separate, familiar relation between physical phases, not a mixed ultraviolet divergence.

\paragraph{What this does not yet renormalize.}
For \emph{fixed} gauge and Higgs backgrounds, integrating over continuum fermions is a Gaussian integral: its logarithm sums closed \emph{one-fermion-loop} diagrams, with arbitrarily many background-field insertions but no internal dynamical gauge- or Higgs-boson propagators. Those propagators arise only when the bosonic fields themselves are integrated. At fixed lattice spacing $a$, that integration has a finite bosonic ultraviolet regulator. In the conventional continuum theory, gauge-boson exchange renormalizes Yukawa vertices, and the running gauge, Yukawa and scalar couplings mix at higher orders. Recovering this complete multiloop renormalization as $a\to0$ is \emph{not} established by the present removal of the fermion spectral cutoff $\Lambda$. Theorem~\ref{thm:renormalized_weight} below establishes integrability over the complete regulated bosonic domain at fixed $a$; it makes no uniform claim as $a\to0$. The two-stage separation therefore organizes the fermion ultraviolet problem without assuming that gauge and Yukawa physics decouple.

\subsection{Matching the cutoff modulus to the Quillen metric}\label{sec:quillen_matching}
Removing divergent loop terms does not yet specify the \emph{finite} normalization of the determinant. This is the same issue encountered when translating a Feynman amplitude between subtraction schemes: the nonanalytic dependence on external momenta agrees, but finite local coefficients need a convention or a matching condition. Two regulators can therefore agree on all nonlocal amplitudes while assigning different coefficients to local Higgs operators. We use the Quillen metric, introduced in Section 4.2, as a reference for the determinant's magnitude. The physical prescription fixes both the local complex scheme difference and the positive spectral norm. These require two distinct matching operations.

\emph{First match the complex prescription.} Smooth spectral filtering and complex-zeta regularization generally assign different finite values to ultraviolet-sensitive fermion loops even after all divergences are subtracted. The generally complex local functional $p_f$ corrects this finite scheme difference, including its phase.

\emph{Then match the magnitude.} The squared magnitude of a complex-zeta determinant need not equal the positive zeta determinant of $\mathcal D_Y^\dagger\mathcal D_Y$. The real local correction $\delta_N$ accounts for this remaining squared-norm defect, so the factor $e^{-\delta_N/2}$ makes the magnitude agree with the Quillen metric. Positive here describes the norm of the determinant-line element; the fermion weight can still be complex and need not define a positive Euclidean measure. The two functionals are defined precisely in Appendix~\ref{sec:cutoff_quillen_app}, where $\delta_N$ is the exact degree-at-most-four local polynomial determined by the small-Yukawa norm comparison, not a truncation of higher-order norm corrections. We define the single complex weight:
\begin{equation}
 \boxed{W_Q=e^{-p_f(U,\Phi)-\delta_N(U,\Phi)/2}\,W_f^{\rm ren},}
 \qquad \rho_f:=-\operatorname{Re}p_f-\tfrac12\delta_N.
 \label{eq:main_rho}
\end{equation}
We suppress the auxiliary shift in the notation: $p_f\equiv p_{f,\sigma}$ and $\rho_f\equiv\rho_{f,\sigma}$. With the same reference scale and finite matching conditions, $W_Q$ is independent of both $f$ and $\sigma$. Appendix~\ref{sec:cutoff_quillen_app} makes this independence explicit by writing it as a convergent $\det_5$ factor times the fixed renormalized low-order terms and the physical low-mode factors. The factor $1/2$ converts the squared-norm correction into a correction to the magnitude. Both $p_f$ and $\delta_N$ are local under the analytic conditions of Appendix~\ref{app:yukawa_uv}.

The local imaginary part $\operatorname{Im}p_f$ is subtracted with the rest of the finite scheme difference; we do not claim that it vanishes on every background. This fixes a regulator convention, not a physical CKM or theta parameter. The reference complex-zeta phase obeys conjugation under CP reflection of the backgrounds together with conjugation of the Yukawa matrices. Only a spectral cut near zero Yukawa scaling is needed to define the comparison coefficients.

The finite-dimensional zero-mode factors remain part of the physical norm. Appendix~\ref{sec:mellin_anchor} derives the local comparison, while Appendix~\ref{sec:quillen_gluing} identifies these low-mode factors under changes of spectral window, using the Bismut--Freed isometry~\cite{bf1,freed_87}. The construction extends through vanishing vacuum amplitudes by source coefficients, rather than a quotient of vanishing determinants.

\paragraph{Why the low modes matter.}
A \emph{Grushin extension} adds finitely many auxiliary source and target variables to turn a possibly noninvertible or rectangular chiral operator into an invertible square operator. It is a computational device, not an extra physical fermion. Its zeta determinant gives a convenient \emph{high-mode} normalization, not the entire fermion amplitude. Enlarging the low-mode spectral window transfers finitely many nonzero modes out of that zeta determinant. Their determinant reappears in the low-mode exterior-algebra factor; the Quillen norm of the complete determinant-line section is unchanged. The result is valid even when the unsaturated vacuum amplitude vanishes.

\begin{theorem}[Renormalized Yukawa weight and fixed-$a$ bosonic integral]\label{thm:renormalized_weight}
Fix the finite bosonic lattice and spacetime volume, the admissibility and Higgs cutoffs of Section~2, a common real spectral profile, and the anomaly-free zero-Yukawa phase convention. Use the smooth reconstruction, the standing compact-family setup (A1)--(A4) of Appendix~\ref{app:yukawa_uv}, and one globally consistent choice of divergent and finite local subtractions and determinant-line transitions. Then:
\begin{enumerate}
\item The source-dependent, locally subtracted common-filter weight has a continuous gauge-equivariant limit, uniformly on the full compact bosonic support.
\item The phase-matched weight $W_Q$ of Eq.~\eqref{eq:main_rho} is continuous, gauge invariant, and independent of the filter and its auxiliary shift at fixed matching conditions. Its magnitude is the physical Quillen norm for arbitrary finite Yukawa backgrounds, including the finite low-mode factors.
\item Its bosonic integral and the integrals of any fixed finite smooth-smeared gauge-invariant fermionic insertions are absolutely convergent at fixed $a$.
\end{enumerate}
No uniformity as $a\to0$, reflection positivity for the Yukawa theory, or locality of the complete weight in separated original lattice links is asserted.
\end{theorem}
\begin{proof}
By the strict containment of $\operatorname{supp}\chi_\epsilon$ in the geometric admissibility region and the bounded support of $b_{R_H}$, the bosonic support $\mathcal K$ is compact inside the smooth reconstruction domain. The reconstructed $(A[U],\Phi[U,\phi])$ therefore form a compact smooth family with uniformly bounded derivatives of each fixed order. Cover $\mathcal K$ with finitely many neighborhoods admitting invertible Grushin charts. Complex spectral cuts enter only in the neighborhood of $t=0$ used to determine the local coefficients $p_f$ and $\delta_N$; no compatible cut is assumed for the full Yukawa operator at $t=1$.

On each chart, Appendix~\ref{app:yukawa_uv} gives uniform $\mathcal S_5$ convergence for the fifth-order determinant and uniform $O(\Lambda^{-1})$ remainders, after the same globally defined local subtractions, for the lower Yukawa orders. The finite-cutoff determinant sections and smooth-smeared source minors agree on overlaps by their canonical finite-rank determinant-line transitions. The ultraviolet subtraction densities are gauge invariant and independent of the auxiliary chart; Appendix~\ref{sec:quillen_gluing} identifies the low-mode exterior factors and the Quillen normalization under changes of spectral window. With common reference phases and finite normalization choices, the local limits hence glue to a continuous gauge-invariant complex weight; where the vacuum determinant vanishes, the source-minor formulation gives the continuous extension rather than a quotient of vanishing determinants. Uniform convergence on a finite cover is uniform on $\mathcal K$.

The single local normalization $e^{-p_f-\delta_N/2}$ preserves continuity and gauge invariance. The analyticity lemma following Eq.~\eqref{eq:rho_app} proves the regularity needed to continue the squared-norm identity to full Yukawa strength through zero crossings. The finite-part overlap calculation in Appendix~\ref{sec:quillen_gluing} and the regularized-minor construction in Appendix~\ref{sec:source_families} justify the normalized source limits.

Compactness bounds the limiting weights and fixed smooth-smeared insertion coefficients, while the bosonic measure of the support is finite, proving absolute convergence of their integrals.
\end{proof}
\noindent\emph{Physical interpretation.} With the bosonic lattice held fixed, every fermion vacuum-loop contribution has a consistent ultraviolet subtraction throughout the integration domain. The theorem concerns the actual renormalized Yukawa prescription, not just an arbitrary gauge-invariant extension of the massless determinant. It does not establish removal of the independent bosonic lattice regulator.

\paragraph{An elementary two-point example.}
To display the local matching coefficients, take a free Euclidean Dirac fermion on $\mathbb R^4$ with a slowly varying real mass $m(x)$ and the same filter at both Yukawa vertices. The full finite-torus amplitude also contains finite volume- and spin-structure-dependent terms, common to the compared prescriptions and not displayed here. For external momentum $q$, the loop integrand contains $f((p^2+\sigma^2)/\Lambda^2)^2 f(((p+q)^2+\sigma^2)/\Lambda^2)^2$; these factors depend on different momenta. At $q=0$ the flat-space quadratic coefficient is
\begin{equation}
 B_{\Lambda,f}(0)=\frac{\Lambda^2}{8\pi^2}
       \int_0^\infty f(u)^4\,du-\frac{\sigma^2}{8\pi^2}.
 \label{eq:main_twopoint}
\end{equation}
This identity is exact once $f(u)=1$ for $0\le u\le\sigma^2/\Lambda^2$, as is true for sufficiently large $\Lambda$. The first term is the familiar quadratic contribution to the scalar mass and must be subtracted. The power $\Lambda^2$ comes from the four-dimensional measure $d^4p/p^2$ after the spin trace, rather than from any extra factor in the definition of the cutoff. The second survives in $W_f^{\rm ren}$ because the spectral filter contains $\sigma^2$; the finite matching to $W_Q$ cancels it. It is not a prediction of a physical Higgs mass. The $q^2$ term is the Higgs two-point function's wavefunction-renormalization term: it has a universal logarithmic divergence and a finite profile-dependent coefficient. Because the two propagators carry momenta $p$ and $p+q$, differentiating the two \emph{different} filters contributes to that finite coefficient. Replacing both by one filter at momentum $p$ would therefore change the finite subtraction scheme. Appendix~\ref{app:counterterms} evaluates these pieces directly.

For three Standard Model generations without right-handed-neutrino Yukawas, define the positive trace combination
\begin{equation}
 \mathcal T_Y=3\operatorname{tr}(Y_u^\dagger Y_u+Y_d^\dagger Y_d)
               +\operatorname{tr}(Y_e^\dagger Y_e).
 \label{eq:SM_Ytrace}
\end{equation}
The quadratic \emph{real normalization} is then
\begin{equation}
 \boxed{\rho_{2,f}=\frac{\sigma^2\mathcal T_Y}{8\pi^2}
                  \int_X\!\Phi^\dagger\Phi\,d^4x
 -\frac{\mathcal T_Y\kappa_f^{\rm kin}}{16\pi^2}
                  \int_X\!(D_\mu\Phi)^\dagger(D_\mu\Phi)\,d^4x.}
 \label{eq:rho2_main}
\end{equation}
The dimensionless profile moment $\kappa_f^{\rm kin}$ is given in Eq.~\eqref{eq:Cf}. These are the finite counterterms matching the filtered weight to $W_Q$ at the fixed zeta scale $\mu$. A quartic correction is determined by a local residue (Eq.~\eqref{eq:quartic_p}) together with the positive-zeta norm defect; the residue alone is not the entire Standard Model quartic coefficient.

\noindent\emph{Physical interpretation.} Different ultraviolet subtraction conventions can assign different finite coefficients to the Higgs mass and kinetic terms, even after the divergent contributions are removed. These shifts are absorbed by finite renormalization conditions. The Quillen metric fixes a determinant-magnitude convention, but it neither predicts the Higgs mass nor determines the physical theta angles and Yukawa phases.

\subsection{Physical phase checks and finite normalization choices}\label{sec:phase_checks}
The following calculation checks the \emph{change} in the fermion phase produced by complex quark masses. It does not by itself identify the reference phase in every color-topology sector with the conventional positive-mass normalization; that separate comparison is stated below. Set the electroweak connections to zero, take a constant nonzero Higgs field, and let the color field have charge $k_3\ne0$. The six massive quarks form a vector-like color theory: paired nonzero Dirac modes contribute positive determinant factors, while chiral zero modes yield
\begin{equation}
\arg W_{\rm quark}=\nu\arg\det(Y_uY_d)\pmod{2\pi},\qquad \nu=-k_3,
\label{eq:massphase}
\end{equation}
relative to positive real masses and the chosen reference phase. For the convention $k_3=\int_X c_2(E_3)$, a pure-color background has $\operatorname{ind}D_{u^c}=\operatorname{ind}D_{d^c}=-k_3$ and $\operatorname{ind}D_Q=-2k_3$. For $k_3>0$ these negative indices put the unpaired chiral modes in the relevant cokernel blocks, so the zero-mode mass insertions use the conjugate mass matrix in our Euclidean Grassmann convention. Hence the net quark-mass factor has phase $-k_3\arg\det(Y_uY_d)$; reversing the definition of $k_3$ reverses the displayed sign. The filter preserves the nonzero-mode pairing and the zero-mode mass factors, so the result holds once the cutoff includes the zero modes. Appendix~\ref{app:geometry} writes the explicit color factor as $e^{i\theta_3k_3}$. Combining that factor with the mass-dependent change gives $e^{ik_3[\theta_3-\arg\det(Y_uY_d)]}$. In the conventional positive-mass color normalization, $\theta_{\rm QCD}=-\theta_3$ and the familiar combination is $\bar\theta_{\rm QCD}=\theta_{\rm QCD}+\arg\det(Y_uY_d)$.

Equation~\eqref{eq:massphase} determines how quark-mass phases change the weight. An absolute identification of our baseline with that conventional color normalization additionally requires comparing its transported weak-$SU(2)$ reference with the positive vector-like quark determinant in each color sector. This orientation comparison is not supplied by the mass-phase calculation. Appendix~\ref{app:phase_references} states the required comparison and how an additive reference offset would enter the physical theta angle. The existence theorems and the zero-Yukawa positivity theorem use their specified references and do not require this extra identification.

\paragraph{Physical meaning of the phase convention.}
Anomaly cancellation allows the chiral determinant to be defined consistently, but does not set the values of CP-violating couplings. For example, multiplying the fermion weight in a color sector labeled by $k_3=\int_Xc_2(E_3)$ by $e^{i\alpha k_3}$ preserves gauge invariance and the determinant magnitude. It is equivalent to shifting $\theta_3$ by $\alpha$, or the matched $\theta_{\rm QCD}$ by $-\alpha$. After the baseline is matched to the positive-mass color convention, the physical strong-CP parameter is the combination $\bar\theta_{\rm QCD}$ above: a change in the determinant phase convention must be compensated by a corresponding change of the theta parameter to describe the same theory. The CKM phase is instead encoded in rephasing-invariant combinations of the Yukawa matrices; it is not an arbitrary extra phase attached to the determinant. Thus the regulator must preserve the chosen theta and Yukawa data, not predict them.

Gauge invariance alone also permits finite-cutoff phase modifications beyond the renormalizable Standard Model. For example, at fixed lattice spacing the local gauge-invariant factor
\begin{equation}
\exp\left[ic a^2\int_X(\Phi^\dagger\Phi)\operatorname{tr}(F_2\wedge F_2)\right]
\label{eq:irrelevant_phase}
\end{equation}
changes the phase in generic weak backgrounds while preserving the norm and Ward identities. It is a dimension-six interaction suppressed by $a^2$; it is not a new unsuppressed CP-violating parameter of the renormalizable Standard Model. Restricting finite counterterms to engineering dimension at most four leaves the familiar theta-angle choices and associated phase conventions. The electroweak theta angle can be removed in the minimal Standard Model by an anomalous baryon- or lepton-number redefinition, whereas $\bar\theta_{\rm QCD}$ can be physical. Appendix~\ref{app:grushin} treats the norm comparison. Majorana masses, reflection positivity with nonzero Yukawa couplings, and removal of the bosonic regulator remain open. Appendix~\ref{sec:collar_variant} proves reflection positivity for a variant gauge field reconstruction at zero Yukawa coupling and zero theta.

\section{Comparison with overlap fermions}\label{sec:overlap}

Both constructions regulate Standard Model chiral fermions, but use different ultraviolet variables and different conventions for naming the chirality eigenspaces. In this paper $D_A^+$ acts on $\gamma_5=+1$ spinors, whereas Craig calls $\gamma_5=-1$ spinors left-handed. We therefore compare with the complementary-chirality overlap weight; the exact relation and its orientation-reversal interpretation are given in Appendix~\ref{app:weak_overlap}. Historically, continuum-fermion regulators were explored with interpolation and finite-mode/Pauli--Villars schemes~\cite{thooft_lattice,kronfeld,montvay}, while overlap regulators descend from Narayanan and Neuberger~\cite{narayanan_neuberger}. This section gives a free massive-fermion benchmark and a smooth weak-field comparison of the interacting complex fermion determinants. Equivalence for arbitrary admissible backgrounds remains open.

Grabowska and Kaplan also use a smoothed gauge field in a chiral regulator~\cite{grabowska_kaplan}. Their construction extends lattice gauge fields by gradient flow into an extra dimension, with domain-wall fermions and a mirror sector. Here the kinetic operator acts on continuum Weyl fields in four dimensions. The shared use of smooth backgrounds does not establish equivalence of the phase prescriptions or continuum limits.

\subsection{Established differences}
The continuum-Weyl approach retains the ordinary chiral operator on a smooth bundle. The kinetic operator acts on infinitely many continuum modes, so the Nielsen--Ninomiya obstruction~\cite{nielsen_ninomiya} to a local finite-dimensional Weyl discretization does not apply to this fermionic sector. Exact continuum index theory and determinant-line geometry are available directly. By contrast, overlap fermions~\cite{neuberger} are finite lattice spinors whose modified chiral-symmetry relation (the Ginsparg--Wilson relation) permits a lattice chiral projector and an exact lattice index. The abelian anomaly-free chiral measure was constructed nonperturbatively in Ref.~\cite{luscher_abelian}. L\"uscher constructed non-Abelian anomaly-free lattice chiral theories to all orders in perturbation theory~\cite{luscher_perturbative}; Craig's non-Abelian refinement construction~\cite{craig} is a later, distinct proposal. Neither description is intrinsically superior; overlap methods have substantially more developed finite-matrix algorithms.

Smooth local reconstruction of a continuum connection from lattice links does not by itself establish locality of the resulting fermion measure under variations of distant links. A massless fermion's \emph{integrated-out effective action} is physically nonlocal (for example through vacuum-polarization logarithms); this is different from asking whether regulator-dependent ultraviolet terms and measure variations admit suitable local bounds. Under the compact-family conditions (A1)--(A4), we establish locality of the continuum \emph{subtraction difference}. This is distinct from locality of the complete measure under variations of separated original lattice links. Craig's overlap construction~\cite{craig} supplies bounds for its measure current, though the construction involves an infinite refinement series and does not yet give a demonstrated practical algorithm for the complete phase. Both approaches face complex weights and potentially severe numerical sign problems.

\begin{table}[htbp]
\centering\small
\renewcommand{\arraystretch}{1.23}
\begin{tabularx}{\linewidth}{@{}p{0.23\linewidth}XX@{}}
\toprule
 & Continuum-Weyl construction & Overlap construction~\cite{craig} \\
\midrule
Fermion variables & Continuum Weyl fields; infinitely many modes & Finite overlap matrices at each lattice refinement \\
Chirality and index & Ordinary chiral projector and continuum index theorem & Ginsparg--Wilson projection and exact overlap index \\
Complex weight & Determinant/Pfaffian section and chosen spectral normalization & Overlap Gaussian plus its non-Abelian measure phase \\
Locality results & Local smooth reconstruction; locality of induced fermion action not proved & Overlap operator and non-Abelian measure-current bounds; Secs.~5--7 and App.~C of Ref.~\cite{craig} \\
Computation & Continuum spectrum for every bosonic background & Finite matrices per refinement, but an infinite convergent series for the measure current \\
Continuum matching & Conditional on reconstruction bound (Prop.~\ref{prop:D2_weak_match}) & Weak-field overlap $\to$ continuum (Prop.~\ref{prop:D2_overlap}) \\
\bottomrule
\end{tabularx}
\caption{Comparison at fixed bosonic lattice spacing and finite spacetime volume; the weak-field comparison uses the chirality-matched overlap measure and separates its continuum limit from the conditional reconstruction step.}
\label{tab:comparison}
\end{table}

\subsection{An explicit free Yukawa comparison}\label{sec:free_overlap}

There is a useful comparison that can already be proved without the non-Abelian measure problem. Work at flat gauge background, take a constant Higgs field, and consider one vector-like massive Dirac fermion. Craig's construction supplies a \emph{massless} chiral measure, not a Yukawa vertex; for this test we extend the overlap operator by the standard Ginsparg--Wilson mass prescription~\cite{luscher_gw,giusti}.

Write $s=m^2$. For the standard massive overlap extension, the momentum-space spin determinant obeys
\begin{equation}
\frac{\det_{\rm spin}D_{m,\rm ov}(p)}{\det_{\rm spin}D_{0,\rm ov}(p)}
=\left[1+sG_a(p)\right]^2,\qquad G_a(p)\longrightarrow\frac1{p^2},
\label{eq:overlap_free_mass}
\end{equation}
for every fixed nonzero momentum as $a\to0$. Define the subtracted effective response by removing the terms linear and quadratic in $s=m^2$, corresponding to quadratic and quartic Higgs counterterms. Diagrammatically these are the two- and four-mass-insertion vacuum bubbles. We impose identical subtraction conditions on the overlap and continuum expressions before comparing their \emph{sixth}- and higher-order mass dependence. On a physical four-torus with an antiperiodic fermion direction (hence no zero momentum), the remaining response is
\begin{equation}
\mathcal R_{a,\mathrm{sub}}(s)=2\sum_{p\in{\rm BZ}_a}
\left[\log(1+sG_a(p))-sG_a(p)+\tfrac12s^2G_a(p)^2\right].
\label{eq:overlap_free_response}
\end{equation}
Here $\mathcal R$ is the subtracted change in $\log W$. The corresponding change in the fermion effective action $\Gamma_{\rm ferm}=-\log W$ has the opposite sign.
Its summand falls as $p^{-6}$ in the continuum limit. The sum of this bound is finite in four dimensions, so the lattice-spacing limit can be taken inside the momentum sum; this is the dominated-convergence step used to establish agreement with the continuum spectral result after identical subtractions. In particular, the coefficient of $s^3$, equivalently the sixth-order constant-Higgs response, converges to $\tfrac23\sum_p p^{-6}$. This preserves the finite nonpolynomial mass dependence, rather than discarding all lower-order loops. An explicit $16^4$--$64^4$ lattice calculation illustrates the convergence; details and the exact formula for $G_a$ appear in Appendix~\ref{app:free_overlap}.

This is an agreement theorem for a free \emph{vector-like} subsector, not for Craig's full chiral overlap weight. The interacting comparison still requires the chiral Yukawa vertex, nontrivial gauge backgrounds, lower-point amplitude matching and the complete measure phase. The free test provides a benchmark and fixes the appropriate relative renormalization convention.\par
\noindent\emph{Physical interpretation.} The large-momentum subtractions make the response summable before the overlap lattice is refined. Agreement of this one finite observable tests the scheme but cannot establish equality of interacting chiral phases.

\subsection{A precise continuum comparison problem}\label{sec:continuum_comparison}
The comparison has two logically separate stages:
\begin{center}
\small
\begin{tabularx}{\linewidth}{@{}p{.22\linewidth}X@{}}
\toprule
\textbf{Stage 1: established} & Compare the chirality-matched overlap and continuum Weyl determinants on the \emph{same} smooth weak gauge background, allowing local counterterms (Proposition~\ref{prop:D2_overlap}).\\
\textbf{Stage 2: conditional} & Compare the continuum Weyl determinant on the reconstructed connection with that on the original smooth background; this requires an additional differentiated reconstruction estimate (Proposition~\ref{prop:D2_weak_match}).\\
\bottomrule
\end{tabularx}
\end{center}
Let $A$ lie in a compact family $\mathcal S$ of smooth connections on a fixed bundle and $U_a(A)$ denote exact sampled edge holonomies. Write $W_{\rm CW}^{(a)}[U_a(A)]$ for the continuum-Weyl weight evaluated on the smooth connection $\mathcal R_a(A)$ reconstructed from the lattice links. For the interacting chiral comparison below, $W_{{\rm ov},+}^{(a)}$ denotes the \emph{complementary-chirality} overlap weight defined from Craig's original negative-chirality weight by conjugation and common reference normalization on the invertible weak-field chart. This choice matches the positive-chirality domain $D_A^+$ used in the present paper; it is not an identification of the unmodified Craig weight with its complex conjugate.

The physical question is whether two different ultraviolet regulators generate the same low-energy fermion effective action. Varying its logarithm with respect to the background gauge potential gives the induced fermion current (up to the conventional sign). Expanding this current in $A$ generates vacuum-polarization, triangle, box and higher-point fermion loops. The nonlocal momentum dependence of these diagrams is physical: for example, the massless-fermion logarithm in the vacuum polarization must agree between regulators. A difference proportional to a polynomial in the external momenta instead corresponds to a local operator, such as a finite shift of the gauge kinetic term. Both the real and imaginary parts matter; the complex phase can encode parity-odd information invisible in the modulus.

Two comparisons must be distinguished. First, the chirality-matched overlap determinant is compared with the continuum BF/Quillen Weyl determinant at the \emph{same} smooth background $A$,
\begin{equation}
W_{{\rm ov},+}^{(a)}[U_a(A)]\quad\longrightarrow\quad e^{C_a^{\rm ov}[A]}W_{\rm BF}[A].
\label{eq:overlap_stageA}
\end{equation}
This overlap-to-continuum step involves the exact overlap vertices, ultraviolet power counting, the complex measure phase, and compatible chiral conventions; it does not use gauge-field reconstruction. Second, the continuum-Weyl construction evaluates $W_{\rm BF}$ on $\mathcal R_a(A)$, rather than on $A$. Its comparison with $W_{\rm BF}[A]$ requires a separate differentiated reconstruction estimate.

The desired complete-regulator equivalence is that appropriate local gauge-invariant counterterms $C_a[A]$ and reference phases give
\begin{equation}
\sup_{A\in\mathcal S}\left|e^{-C_a[A]}
\frac{W_{{\rm ov},+}^{(a)}[U_a(A)]}{W_{\rm CW}^{(a)}[U_a(A)]}-1\right|\longrightarrow0
\quad(a\to0).
\label{eq:overlap_target}
\end{equation}
Near determinant zeros the appropriate objects are determinant-line sections and fermionic minors, not ratios of scalar weights. Proposition~3 in Appendix~D.1 (the determinant-variation criterion) supplies the final step.

Appendix~\ref{app:weak_overlap} proves the finite-grid overlap vertex bounds and the ultraviolet-subtracted convergence of its complex one-loop amplitudes. The $n\ge5$ terms converge without subtractions; for $n\le4$ the difference between overlap subtraction and the BF/Quillen finite-part prescription is local. Craig's smooth-background estimates, together with the explicit moving-frame argument in Lemma~\ref{lem:D2_current_uniform}, suppress the measure-current contribution. The exact complementary-chirality identity aligns the dynamical chiral trace with $D_A^+$, while the common vacuum reference fixes the normalized comparison on the invertible chart.

The appendix gives two separate results: Proposition~\ref{prop:D2_overlap} compares the overlap weight directly with $W_{\rm BF}[A]$ on the smooth weak-field chart; Proposition~\ref{prop:D2_weak_match} adds the differentiated reconstruction estimate. In particular, for sufficiently high $r$ the second gives
\begin{equation}
\sup_{A\in\mathcal S}\sup_{\|v\|_{C^r}\le1}
\left|\delta_v\log W_{{\rm ov},+}^{(a)}-\delta_v\log W_{\rm CW}^{(a)}-\delta_v C_a[A]\right|
\le\eta(a)+Ca^2(1+|\log a|)+Ca^6\longrightarrow0,
\label{eq:overlap_weak_bound}
\end{equation}
where $\eta(a)\to0$ is an error bound for ultraviolet-subtracted overlap convergence and both weights are evaluated on the exact holonomies $U_a(A)$. These three errors have different origins: $\eta(a)$ measures loop convergence, the $a^2|\log a|$ term measures sensitivity to reconstruction, and the $a^6$ term bounds the smooth-field measure current. No overall $O(a^2|\log a|)$ rate is claimed. Integrating the variation bound with matched reference values gives \eqref{eq:overlap_target}.

The established first-stage result is restricted to smooth weak fields on the trivial bundle at zero Yukawa coupling. The second stage remains conditional, and neither result controls arbitrary admissible backgrounds or the bosonic continuum limit.

\section{Conclusion}

We have developed a Standard Model regulator in which gauge and Higgs fields live on a finite lattice, while the fermions retain their continuum description. The motivation is to preserve the familiar physical structure of chiral gauge theories: anomalies, spectral flow, zero modes, and the associated selection rules. Standard Model anomaly cancellation permits a gauge-invariant definition of the chiral determinant, and continuum index theory remains directly applicable. This approach builds on earlier hybrid lattice/continuum proposals, including Hsu (1995)~\cite{hsu}, using a smooth gauge-field reconstruction and a specified treatment of the fermion determinant.

The main technical problem is incorporating Yukawa interactions into this framework. We do so by regulating the Yukawa vertices rather than truncating the fermion spectrum. The resulting ultraviolet divergences have the expected one-loop structure and can be removed by local gauge-invariant counterterms. Mixed gauge--Yukawa effects are retained, and admissible changes of subtraction scheme leave the nonlocal physical contributions of fermion loops unchanged. We also construct a consistent normalization of the complex chiral determinant, including nontrivial topological sectors and fermion zero modes. On the compact smooth reconstructed family, the renormalized weight defines a finite, exactly gauge-invariant bosonic functional integral with the fermion cutoff removed and the bosonic lattice held fixed.

Appendix~\ref{sec:collar_variant} proves site-reflection positivity for a modified gauge-field reconstruction at zero Yukawa coupling and zero theta, with specified reference phases. It preserves the input holonomies and bundle topology, but differs from Craig's interpolation. Effective-field-theory considerations motivate a common long-distance limit; they do not establish that equivalence or reflection positivity for Craig's original reconstruction.

We have also compared the overlap and continuum Weyl regulators in the zero-Yukawa gauge sector. On smooth weak backgrounds, the chirality-matched overlap weight reproduces the nonlocal continuum Weyl effective action, up to local counterterms. No extension to unrestricted admissible fields or nontrivial topological sectors is claimed.

The result is a finite-regulator formulation of the Standard Model, not yet a complete continuum construction. Locality in separated lattice links, the bosonic continuum limit, and reflection positivity for Craig's reconstruction or nonzero Yukawa couplings remain open. The unrestricted interacting-overlap comparison and infinite-volume limit remain open. Nevertheless, the formulation separates these questions from the problem of defining the chiral fermion determinant and its Yukawa interactions. Its principal advantage is that the continuum understanding of chiral fermions is retained throughout, rather than reconstructed after discretization.

\section*{Acknowledgments}
The author used GPT and Claude extensively in this work, but takes responsibility for the results. The technical details in the Appendix are entirely AI-led results (mostly GPT Astra with verification by Claude and other models). He thanks Tim Raben for calling his attention to the paper by Craig.

\clearpage
\appendix
\numberwithin{equation}{section}
\section{Spectral regularization and determinant-line matching}\label{app:yukawa_uv}
The purpose of this appendix is to prove the \emph{locality of the difference between regularization prescriptions}. We first isolate the loop orders that converge without subtraction. We then compare nested spectral filters with zeta regularization and finally obtain the extra correction needed for the Quillen magnitude. This concerns the limit $\Lambda\to\infty$ at \emph{fixed bosonic lattice spacing}; it is not a proof of a continuum limit in $a$. Appendix~\ref{app:counterterms} evaluates the corresponding finite Higgs coefficients, and Appendix~\ref{app:grushin} handles the physical low-mode factors. A nonzero-index chiral operator has no inverse, so we use an invertible auxiliary extension and retain the physical low-mode source integral.

\paragraph{Standing setup and analytic inputs.}
The following items distinguish the background restrictions from the tools and conventions used in the proof. (A1) specifies the compact smooth family supplied by the regulated bosonic support. (A2) identifies the external functional calculus; its application-specific uniform remainder is proved in this appendix. (A3) records local operator charts and spectral cuts, and (A4) fixes normalization choices. Together they specify the fixed-$a$ setting of Theorem~\ref{thm:renormalized_weight}; the required uniform estimates are established below using these tools. A ``spectral cut'' below is a chosen ray in the complex eigenvalue plane used to define a branch of the logarithm, needed when an elliptic operator is not positive. Like fixing a branch of a complex logarithm in a loop calculation, it is a phase convention for analytic continuation, \emph{not} an energy cutoff:
\begin{description}[leftmargin=2.4em,itemsep=3pt,topsep=3pt]
\item[(A1)] The gauge--Higgs backgrounds form a compact smooth family inside the admissible domain, with uniform ellipticity and sufficiently many bounded field derivatives; no uniform claim is made at the admissibility boundary or as $a\to0$.
\item[(A2)] We use the standard semiclassical functional calculus for a positive Dirac Laplacian on a compact manifold and its trace-norm estimate for symbols confined to a finite momentum shell~\cite{kuster}. The bounds needed for the \emph{particular nested products in this paper} are derived below from (A1), with the explicit $O(\Lambda^{-1})$ remainder in Eq.~\eqref{eq:trace_remainder}. Thus locality of the cutoff difference is a conclusion, not a separate assumption. No bound uniform in bosonic lattice spacing or at the admissibility boundary is claimed.
\item[(A3)] Finite-dimensional source and target borders give an invertible rectangular Grushin extension on a neighborhood of each background. At zero Yukawa scaling its comparison weight is positive and invertible, so a common spectral cut persists for sufficiently small scaling on a smaller chart. Compatible cuts are needed only there, where the Taylor coefficients are defined. The full finite-Yukawa weight uses regularized Fredholm determinants and low-mode source factors; no single complex spectral cut at arbitrary Yukawa strength is assumed.
\item[(A4)] The cutoff, complex-zeta and positive-zeta prescriptions share the same reference scale, finite-part subtraction conditions, and finite-dimensional determinant-line transition convention. This fixes the comparison normalization; the theta angles and Yukawa phases are independently specified physical input parameters.
\end{description}
The $\mathcal S_5$ limit and finite-rank determinant identities follow from the stated elliptic/Fredholm setup. The direct low-order comparison below derives the needed uniform remainder from (A1) and the standard calculus identified in (A2), with the chart and normalization choices (A3)--(A4). Whenever ``local'' is used below it means an integral of a gauge-covariant function of the fields and finitely many of their derivatives, not locality of the complete effective action in the lattice links.

\paragraph{Notation used in Appendices A--B.}
\begin{center}\small
\begin{tabularx}{\linewidth}{@{}p{0.23\linewidth}X@{}}\toprule
Symbol & Meaning \\\midrule
$f,\Lambda;\sigma,\mu$ & Spectral profile and fermion cutoff; independent positive infrared shift and zeta subtraction scale. \\
$D,\mathcal G_0,K$ & Chiral kinetic operator, invertible auxiliary (Grushin) extension, and relative Yukawa perturbation. \\
$P_0,P_t,A_t$ & Order-two elliptic weights (all on the augmented chart). $P_0=\mathcal G_0^\dagger\mathcal G_0$. \\
$p_f,\delta_N,\rho_f$ & Cutoff-to-complex-zeta mismatch, squared-norm defect, and total real normalization. \\
$\kappa_f^{(4)},\kappa_f^{\rm kin}$ & Dimensionless quartic and kinetic cutoff moments (Appendix~\ref{app:counterterms}). \\
$\mathcal T_Y,\mathfrak c$ & Sum of Standard Model Yukawa squares; first Chern class $c_1(E_2)$, distinct from filter $f$. \\
\bottomrule\end{tabularx}
\end{center}

\subsection{Ultraviolet power counting and the \texorpdfstring{$\det_5$}{det5} remainder}
The diagrammatic reason for the fifth-order subtraction has already been given in Section~\ref{sec:yukawa_cutoff}: a one-loop $n$-vertex Higgs graph has superficial degree of divergence $4-n$. The following operator formulation makes that argument precise even when the background gauge field is nontrivial and the chiral operator has zero modes. A singular value is an eigenvalue of $|T|=(T^\dagger T)^{1/2}$, the nonnegative analogue of an absolute eigenvalue for a non-Hermitian operator. Summability of their powers controls whether loop operators can be traced without further subtraction. The Yukawa interaction multiplies a field without adding derivatives, so it has order zero and does not alter the leading (Dirac) kinetic term. An inverse Dirac operator contributes one inverse power of high momentum. Consequently, its product $T$ with the Yukawa interaction has singular values $s_j(T)=O(j^{-1/4})$ and belongs to the Schatten class $\mathcal S_p$ for every $p>4$. Here $\mathcal S_p$ means that $\sum_j s_j(T)^p$ is finite; $\mathcal S_1$ is the usual trace class. The operators $T^n$ are trace class for $n\ge5$. The convergent fifth-order determinant is defined \emph{without} invoking an undefined ordinary determinant of $1+T$. Here $\det_F$ denotes the ordinary Fredholm determinant of an identity-plus-trace-class operator, the infinite-dimensional extension of a convergent product of eigenvalues:
\begin{equation}
\det_5(1+T)=\det_F\!\left[(1+T)
 \exp\!\left(-T+\tfrac12T^2-\tfrac13T^3+\tfrac14T^4\right)\right].
\label{eq:det5}
\end{equation}
The operator inside the Fredholm determinant is identity plus trace class. For small $\|T\|$,
\begin{equation}
\log\det_5(1+T)=\sum_{n=5}^\infty\frac{(-1)^{n+1}}{n}\operatorname{Tr}(T^n).
\label{eq:det5exp}
\end{equation}
The subtraction in $\det_5$ is bookkeeping, not an instruction to discard four Feynman amplitudes. The \emph{complete} relative determinant restores the \emph{renormalized} one- through four-point terms, including their finite and generally nonlocal parts. The physical evenness discussed in Section~\ref{sec:yukawa_cutoff} has a precise grading: let $\Gamma_\pm$ be the actions of $g_*$ on the source and target spaces of $D=\mathcal D_0$. Then
\[
 \Gamma_\pm^2=1,\qquad
 \Gamma_-D=D\Gamma_+,\qquad
 \Gamma_-V=-V\Gamma_+.
\]
The kinetic filters preserve this grading. A grading-compatible inverse extension therefore makes the relative Yukawa operator off diagonal, so its odd traces vanish. An arbitrary auxiliary Grushin chart need not display that block structure; retaining all first four orders is a uniform bookkeeping bound, and the complete gauge-invariant bosonic coefficients have the required evenness. Source-dependent coefficients instead obey covariance with the sources transformed.

\subsection{Common finite-rank regulator}

This appendix gives the operator argument behind Eqs.~\eqref{eq:common_cutoff} and~\eqref{eq:renormalized_source}. It distinguishes exact Fredholm identities from the smooth-family analytic conditions needed to pass to the unfiltered continuum Weyl operator. The bosonic lattice spacing $a$ and the admissibility cutoff remain fixed throughout the $\Lambda\to\infty$ limit.

\subsubsection*{Determinant-line perturbation}
Write $D=\mathcal D_0$ and $V=\mathcal V_Y$, with no assumption on $\ind D$. Let $C_{\pm,\Lambda}=f(H_\pm/\Lambda^2)$, where $H_\pm$ are defined in~\eqref{eq:common_cutoff}. Elliptic compactness implies that $C_{\pm,\Lambda}$ have finite rank. Their dependence on a smooth reconstructed background is smooth for a $C^\infty$ cutoff $f$ and a locally uniformly elliptic family. Gauge transformations conjugate $H_\pm$, so $V_\Lambda=C_{-,\Lambda}VC_{+,\Lambda}$ is exactly gauge covariant. The compact perturbation $D+V_\Lambda$ has the same Fredholm index as $D$. This finite-rank statement uses the compact support of $f$; a merely rapidly decreasing profile would be smoothing/trace class but generally not finite rank.

A finite-rank change in a Fredholm operator canonically relates the corresponding determinant lines, including for nonzero index~\cite{kaadnest,zinger}. The canonical determinant section may vanish, but fixed smooth-smeared source minors can be defined in finite-dimensional exterior algebra. At each $\Lambda$ this construction defines the \emph{complete} source-dependent object before selecting low-mode or Grushin coordinates. This fact is what makes phase gluing exact at finite cutoff; a high-mode determinant considered without the low-mode source polynomial has no analogous chart independence.

\subsection{Schatten convergence and loop subtractions}
Choose a local rectangular Grushin extension $\mathcal G_0$ of $D$ as in Appendix~\ref{app:grushin}, and set
\begin{equation}
K_\Lambda=\mathcal G_0^{-1}\,\operatorname{diag}(V_\Lambda,0),\qquad
K=\mathcal G_0^{-1}\,\operatorname{diag}(V,0).
\label{eq:grushin_K}
\end{equation}
The infinite-dimensional inverse block has pseudodifferential order $-1$; other blocks are finite rank. In dimension four, $K\in\mathcal S_p$ for $p>4$. For completeness, write $G=\mathcal G_0^{-1}$ and extend $C_{\pm,\Lambda}$ by the identity on the finite auxiliary spaces. The difference decomposes as
\[
K_\Lambda-K=G(C_{-,\Lambda}-1)VC_{+,\Lambda}+GV(C_{+,\Lambda}-1).
\]
For the first term put $A_\Lambda=C_{-,\Lambda}-1$. Since $f$ is real, $A_\Lambda$ is self-adjoint and tends strongly to zero. Finite-rank approximation gives $A_\Lambda G^\dagger\to0$ in $\mathcal S_5$; taking adjoints gives $GA_\Lambda\to0$ in the same norm. Uniform boundedness of $V$ and $C_{+,\Lambda}$ proves the assertion. The second tends to zero by the same trace-ideal approximation applied to $GV\in\mathcal S_5$. For a quantitative uniform bound, choose $u_0>0$ such that $f(u)=1$ for $u\le u_0$. Four-dimensional uniform Weyl counting on the compact reconstructed family gives $N_{H_\pm}(L)\le C_{\mathcal K}L^2$ for sufficiently large $L$. Therefore, by a dyadic spectral sum,
\[
 \left\|H_\pm^{-1/2}(C_{\pm,\Lambda}-1)\right\|_{\mathcal S_5}^{5}
 \le C_f\sum_{\lambda_j(H_\pm)\ge u_0\Lambda^2}\lambda_j(H_\pm)^{-5/2}
 \le C_{f,\mathcal K}\Lambda^{-1}.
\]
The factors $GH_-^{1/2}$ and $GVH_+^{1/2}$ are uniformly bounded order-zero operators (with finite-rank auxiliary blocks), so applying the Schatten ideal inequality to the two summands gives
\[
 \sup_{(U,\phi)\in\mathcal K}\|K_\Lambda-K\|_{\mathcal S_5}
 \le C_{\mathcal K}\Lambda^{-1/5}.
\]
This controls the convergent fifth- and higher-order loop terms. The first four traces still require the separate ultraviolet-shell expansion for their local subtractions. Thus
\begin{equation}
\|K_\Lambda-K\|_{\mathcal S_5}\longrightarrow0,
\qquad \det_5(1+K_\Lambda)\longrightarrow\det_5(1+K).
\label{eq:grushin_S5}
\end{equation}
Here convergence in the $\mathcal S_5$ norm means that the fifth root of the sum of fifth powers of singular values tends to zero, not merely that each matrix element converges. The argument does not need the intertwining relation in Eq.~\eqref{eq:dirac_intertwine}; that identity is instead used below to compare the \emph{nested} filtered traces and their local ultraviolet coefficients. It does not require an inverse of the original chiral $D$. For finite-rank $K_\Lambda$ the exact identity is
\begin{equation}
\det_F(1+K_\Lambda)=\det_5(1+K_\Lambda)
\exp\!\left[\sum_{n=1}^{4}\frac{(-1)^{n+1}}n\operatorname{Tr}K_\Lambda^n\right].
\label{eq:fredholm5split}
\end{equation}
The remaining terms at order $n\ge5$ are trace class in four dimensions, while the first four loop traces require ultraviolet subtractions. The renormalized \emph{finite} parts of those four traces must be retained: they contain generally nonlocal physical responses. The grading-compatible Nambu ordering makes odd traces vanish; arbitrary auxiliary Grushin choices need not preserve that block structure, although the complete gauge-invariant weight remains even in the Higgs field.

\subsection{Locality, family dependence and the scope of the limit}
The symbol language used below is a systematic large-momentum expansion. In local coordinates a differential operator is represented by a matrix-valued function of position and momentum; inverse operators have corresponding asymptotic expansions in inverse momentum. The leading homogeneous term is its \emph{principal symbol}. A \emph{classical pseudodifferential operator} admits a complete expansion in homogeneous momentum orders. Two such operators with the same complete symbol differ by a \emph{smoothing} operator, whose kernel is smooth; on compact spacetime it is trace class. Smoothing is therefore stronger than agreement of principal symbols. These definitions translate the usual expansion of a high-momentum loop into operator language.

Changing the Grushin chart changes how finitely many low-energy modes are handled, but not the ultraviolet behavior of the original Dirac operator. More precisely, two generalized inverses of the same $D$ have the same complete pseudodifferential symbol and differ by finite-rank smoothing terms. In the difference of the first four traces in~\eqref{eq:fredholm5split}, telescoping shows that every term contains a smoothing factor and has a finite $\Lambda\to\infty$ limit. Hence the \emph{divergent} coefficients are chart independent. The same coefficients can be calculated using the gauge-covariant positive elliptic weights $H_\pm$, which remain invertible even when $D$ has zero modes.

\paragraph{Nested-filter symbolic expansion.}
The trace is not of the simple form $\operatorname{Tr}(f(H/\Lambda^2)A)$: every factor of $K_\Lambda$ contains a filter on each side of a Yukawa vertex. To see what is required, put $h=\Lambda^{-1}$, $\eta=h\xi$ with ordinary momentum $\xi$, and work in a local trivialization on a compact smooth gauge--Higgs parameter patch $\mathcal P$. The semiclassical functional calculus for the positive elliptic $H_\pm$ gives symbols of $f(h^2H_\pm)$ with leading term $f(|\eta|^2)$, equivalently $f(|\xi|^2/\Lambda^2)$ in ordinary momentum and successive terms obtained from the symbols of $H_\pm$ and derivatives of $f$~\cite{kuster}. Compose all $2n$ filter symbols, the $n$ Yukawa symbols and the $n$ generalized-inverse symbols using the pseudodifferential product. This composition is necessary: commuting the filters past a spatially varying Higgs vertex would omit derivative counterterms.

For $n\le4$, the high-frequency terms that can affect a divergent or finite local contribution have homogeneous degrees $-n-j$ with $0\le j\le4-n$. Their traces reduce to integrals over the cotangent fibers of local symbol expressions. A radial high-momentum integration gives
\begin{equation}
\operatorname{Tr}K_\Lambda^n
=\sum_{j=0}^{3-n}c_{n,j}\Lambda^{4-n-j}
+b_n\log\Lambda+F_{n,f}+o(1),\qquad 1\le n\le4,
\label{eq:loop_asymptotic}
\end{equation}
where the sum is empty for $n=4$. The $c_{n,j}$ and $b_n$ are integrals of local symbol coefficients; the constant $F_{n,f}$ can be nonlocal. The indicated proof scheme uses a low/high-momentum partition, symbol composition through degree $-4$, and an order-below-$-4$ trace-class remainder; the required ultraviolet-shell remainder bound is derived in Appendix~\ref{sec:uniform_shell_bound}. To obtain uniform $o(1)$ bounds and parameter derivatives, one needs uniformly controlled $C^r$ seminorms of the reconstructed fields and the corresponding parameter-dependent symbol remainders. The constants in these standard bounds are uniform only on compact smooth patches inside the admissibility domain, not at its boundary or as $a\to0$.

\begin{proposition}[Local change of spectral-filter profile on a smooth family]\label{prop:profile}
Let $f,g\in C_c^\infty([0,\infty))$ be real and equal to one near zero, with the same massless determinant-line reference. On a compact smooth patch $\mathcal P$ satisfying (A1), use the standard elliptic calculus (A2), with local chart and normalization choices (A3)--(A4). The nested-product expansion and trace-class remainder are derived in Appendix~\ref{sec:uniform_shell_bound}. Subtract from each filter prescription its own divergent local one- through four-Yukawa amplitudes. Then, for fixed smooth-smeared fermion sources, the two limiting source functionals differ by multiplication by $\exp L_{f,g}(U,\Phi)$, where $L_{f,g}$ is a finite local gauge-invariant functional. Every term has Yukawa degree at most four and no more than $4-n$ derivatives at degree $n$ (ordinary four-dimensional power counting). In particular, profile changes do not modify the convergent fifth- and higher-order Yukawa responses.
\end{proposition}
\begin{proof}
Both relative operators approach the same $K$ in $\mathcal S_5$, so their $\det_5$ limits agree. The difference between the first four loop traces contains at least one factor $f(h^2H_\pm)-g(h^2H_\pm)$; its symbol vanishes in the low-momentum region relative to $\Lambda$. For fixed physical momentum its contribution disappears. In the shell $|\xi|\asymp\Lambda$, replace each factor by its parameter-dependent symbol and expand in external derivatives. The $-n-j$ terms contribute powers $\Lambda^{4-n-j}$ for $n+j<4$, finite local terms when $n+j=4$, and vanish when $n+j>4$. The trace remainder bound in Eq.~\eqref{eq:trace_remainder} justifies these limits and interchanges. Subtracting the two sets of divergent terms leaves only the $n+j=4$ local coefficients. The one- through four-point source-independent terms multiply all smooth-smeared source coefficients; additional source factors are smoothing and have a profile-independent trace-class limit. Gauge covariance of the filters ensures invariance of the resulting local difference.
\end{proof}
For example, the principal-symbol contribution at $n=4$ changes by a local angular trace multiplied by the finite radial moment
\begin{equation}
\int_0^\infty\frac{dr}{r}
\left[f(r^2)^8-g(r^2)^8\right].
\label{eq:profile_radial}
\end{equation}
The eighth power counts two filters at each of four Yukawa vertices. This illustrative term does not replace the derivative-symbol contributions at lower loop order. The proposition compares \emph{renormalization classes}: the two finite results need not coincide until these local finite terms are fixed by common renormalization conditions. The ultraviolet-shell bound in Appendix~\ref{sec:uniform_shell_bound} is what justifies profile equivalence; $\mathcal S_5$ convergence alone would not do so.

The divergent counterterms are independent of the Grushin chart: two generalized inverses differ by smoothing terms, so telescoping a loop-trace difference gives finite limits and no divergent coefficients. For fixed smooth-smeared sources, rank-one insertions add no determinant ultraviolet divergence. Hence the complete source coefficients admit a common subtraction on the stated smooth patches, and their finite-cutoff chart agreement passes to the limit. This does not assert exponential locality in the lattice links or a bosonic continuum limit.

\subsection{Direct cutoff-to-zeta comparison by Mellin transformation}\label{sec:mellin_anchor}
This is the central analytic step. Comparing two smooth cutoff profiles is not enough: we must compare the filtered determinant to the \emph{specific} complex-zeta prescription used to define the Quillen normalization. For a HEP reader, this is analogous to checking that a smooth momentum cutoff and an analytic regularization differ by a choice of local counterterms, rather than by a different physical one-loop amplitude. The risk is that the constant term left after subtraction could contain a different, nonlocal Higgs response. We show instead that the nonlocal part agrees, and that the discrepancy is a local counterterm. The calculation has three stages: use the Dirac intertwining identity to put paired filters on one side; keep their noncommutativity in explicit ultraviolet-shell commutators; and evaluate the remaining single-filter trace by a Mellin transform. In the last step the smooth cutoff is expressed as an integral of complex powers of a positive kinetic operator, so the constant term left after removing the ultraviolet powers and logarithm can be compared directly with zeta regularization. No filter is simply commuted past a position-dependent Higgs field. We use standard elliptic-family and weighted-trace calculus~\cite{kuster,paycha_scott,dang}.

\paragraph{Setting and fixed-weight trace.}
On a compact smooth reconstructed family satisfying (A1)--(A4), write $D=\mathcal D_0$, $V=\mathcal V_Y$, and take the physical spectral weights
\begin{equation}
 Q_+=\sigma^2+D^\dagger D,\qquad Q_-=\sigma^2+DD^\dagger,
 \qquad C_{\pm,\Lambda}=f(Q_\pm/\Lambda^2).
 \label{eq:mellin_weights}
\end{equation}
Use the local gauge-only Grushin extension $\mathcal G_0$ and the operators $K=\mathcal G_0^{-1}\widetilde V$ and $K_{f,\Lambda}=\mathcal G_0^{-1}\widetilde V_{f,\Lambda}$ from Eq.~\eqref{eq:grushin_K}. On the augmented source space extend $Q_+$ by a fixed positive finite-dimensional weight. A finite block contributes an ordinary trace but no ultraviolet residue. For an operator $T$ with a classical symbol, set
\begin{equation}
 \operatorname{tr}^{Q_+}(T)
 =\operatorname*{FP}_{z=0}\operatorname{Tr}(TQ_+^{-z}),\qquad
 \operatorname{res}(T)=2\operatorname*{Res}_{z=0}\operatorname{Tr}(TQ_+^{-z}).
 \label{eq:weighted_anchor}
\end{equation}
The weighted trace $\operatorname{tr}^{Q_+}$ is the finite part obtained by inserting a complex power of the positive operator $Q_+$: $\operatorname{FP}_{z=0}$ means the constant coefficient in the Laurent expansion at $z=0$, while $\operatorname{Res}_{z=0}$ is the coefficient of $1/z$. The \emph{Wodzicki residue} $\operatorname{res}$ extracts the coefficient of momentum order $-4$ in four dimensions and integrates it locally over spacetime and momentum directions. It is the analogue of the coefficient of a logarithmic ultraviolet divergence. Neither operation discards a finite nonlocal amplitude. All spectral weights use the same scale convention.

\begin{proposition}[Local cutoff-to-zeta comparison]\label{prop:mellin_anchor}
For $n=1,2,3,4$ for the compact smooth families (A1), using the standard calculus (A2) and chart choices (A3)--(A4), the fixed-regulator finite part satisfies
\begin{equation}
\boxed{\begin{aligned}
 \operatorname*{FP}_{\Lambda\to\infty}\operatorname{Tr}K_{f,\Lambda}^n
 &=\operatorname{tr}^{Q_+}(K^n)
 +\frac{\kappa_{n,f}}2\operatorname{res}(K^n)
 +\mathscr C_{n,f},\\
 \kappa_{n,f}&=\int_0^\infty\frac{du}{u}
       \bigl[f(u)^{2n}-\mathbf1_{u<1}\bigr],
\end{aligned}}\label{eq:mellin_anchor}
\end{equation}
where $\mathscr C_{n,f}$ is a local gauge-covariant functional homogeneous of Yukawa degree $n$, with at most $4-n$ derivatives. In an auxiliary chart chosen as above $\mathscr C_{1,f}=\mathscr C_{4,f}=0$. The local divergent powers, logarithm, and finite coefficient have uniform asymptotics under any fixed finite number of reconstructed-family derivatives. The statement is at fixed bosonic cutoff and does not include the admissibility boundary.
\end{proposition}
\begin{proof}
\emph{Step 1: combine the filters without commuting them through the Yukawa field.}
First use the \emph{exact} Dirac spectral intertwining identities
\begin{equation}
 D^\dagger f(Q_-/\Lambda^2)=f(Q_+/\Lambda^2)D^\dagger,
 \qquad Df(Q_+/\Lambda^2)=f(Q_-/\Lambda^2)D.
 \label{eq:dirac_intertwine}
\end{equation}
At high frequency an elliptic generalized inverse is $D^\dagger$ multiplied by a smooth inverse of $DD^\dagger$, with a fixed low spectral window removed. It satisfies the same intertwining identities; its difference from the Grushin inverse is smoothing. As $f$ equals one near zero, inserting $1-C_{\pm,\Lambda}$ next to such a smoothing operator is $O(\Lambda^{-N})$ in trace norm for every $N$, uniformly on the stated smooth patch. Thus, with $P_\Lambda=f(Q_+/\Lambda^2)^2$ extended as the identity on auxiliary modes,
\begin{equation}
 \operatorname{Tr}K_{f,\Lambda}^n
 =\operatorname{Tr}(KP_\Lambda)^n+O(\Lambda^{-\infty}).
 \label{eq:mellin_reduction}
\end{equation}
\emph{Step 2: isolate the effects of noncommutativity.}
No Yukawa insertion has been commuted with a filter. Instead, expand the nested product algebraically:
\begin{equation}
 \operatorname{Tr}(KP_\Lambda)^n
 =\operatorname{Tr}[K^nP_\Lambda^n]+\mathscr E_{n,\Lambda}.
 \label{eq:commutator_shell}
\end{equation}
Every term of $\mathscr E_{n,\Lambda}$ contains a commutator $[P_\Lambda,K]$. In rescaled cotangent momentum $\eta=\xi/\Lambda$, its symbol is supported away from $\eta=0$, since the profile is constant there. Scalar principal symbols give one additional power of $\Lambda^{-1}$ for this commutator. Semiclassical composition then yields
\begin{equation}
 \mathscr E_{n,\Lambda}=\sum_{j=1}^{4-n}
 \Lambda^{4-n-j}\mathscr C_{n,j,f}+O(\Lambda^{-1}),
 \qquad n=2,3,4,
 \label{eq:commutator_asymptotics}
\end{equation}
with local coefficients and no ultraviolet logarithm. The sum is empty at $n=4$; the $n=1$ commutator correction is identically zero. For example, $\operatorname{Tr}(KP_\Lambda)^2-\operatorname{Tr}(K^2P_\Lambda^2)=\operatorname{Tr}(K[P_\Lambda,K]P_\Lambda)$, which contains the finite derivative-of-filter contribution to the Higgs kinetic term in Appendix~\ref{sec:quadratic_quillen_app}.

\emph{Step 3: compare the remaining single-filter trace with zeta regularization.}
For the first term in Eq.~\eqref{eq:commutator_shell}, let $F_n(u)=f(u)^{2n}$. Its Mellin transform obeys $\widehat F_n(z)=1/z+\kappa_{n,f}+O(z)$ near $z=0$, while
$\operatorname{Tr}(K^nQ_+^{-z})=\operatorname{res}(K^n)/(2z)+\operatorname{tr}^{Q_+}(K^n)+O(z)$. Mellin inversion gives
\begin{equation}
 \operatorname{Tr}[K^nF_n(Q_+/\Lambda^2)]
 =\frac1{2\pi i}\int_{\Re z=c}
 \widehat F_n(z)\Lambda^{2z}\operatorname{Tr}(K^nQ_+^{-z})\,dz.
 \label{eq:mellin_inverse}
\end{equation}
Shifting the contour through $z=0$ with the standard elliptic-zeta bounds, the double pole contributes the finite part $\operatorname{tr}^{Q_+}(K^n)+\kappa_{n,f}\operatorname{res}(K^n)/2$ and a local logarithm. Positive poles generate local powers of $\Lambda$. Standard elliptic complex-power estimates, uniform on the compact smooth family, control the shifted Mellin contour by $O(\Lambda^{-1})$; the separate nested-commutator remainder is proved in Appendix~\ref{sec:uniform_shell_bound}. Combining these bounds proves the proposition.
\end{proof}

The Mellin calculation uses the fixed positive weight $Q_+$. The complex-zeta determinant, however, is naturally expressed with the augmented weight $P_0=\mathcal G_0^\dagger\mathcal G_0$. Changing between these two weights produces a second, \emph{local} contribution, governed by the weighted-trace defect formula:
\begin{equation}
\boxed{\begin{aligned}
 \operatorname*{FP}_{\Lambda\to\infty}\operatorname{Tr}K_{f,\Lambda}^n
 -\operatorname{tr}^{P_0}(K^n)
 =&\;\frac{\kappa_{n,f}}2\operatorname{res}(K^n)+\mathscr C_{n,f}\\
 &-\frac12\operatorname{res}\bigl[K^n(\log Q_+-\log P_0)\bigr].
\end{aligned}}\label{eq:cutoff_zeta_defect}
\end{equation}
The residue is local because the leading logarithmic symbols cancel. Its finite-dimensional blocks use ordinary traces. The remaining change from $\operatorname{tr}^{P_0}$ to the \emph{varying} weight in $\log\det_\zeta A_t$ is another local residue integral, displayed in Eq.~\eqref{eq:zeta_taylor_v24}. Therefore every low-order cutoff-to-complex-zeta coefficient is local; for $n\ge5$ the ordinary trace-class Taylor coefficients already agree. The general locality-of-renormalization theorem in Ref.~\cite{dang} is the conceptual precedent; the reduction here concerns the specific noncommuting common filters.

\noindent\emph{Physical interpretation.} The cutoff dependence comes from loop momenta of order the cutoff. When external momenta and reconstructed background-field gradients are held fixed, expanding the high-momentum region gives polynomials in those external momenta: in position space these are local operators with finitely many derivatives. The physical nonanalytic momentum dependence is shared by the cutoff and zeta prescriptions. This is a renormalization-scheme statement, not an assertion that the whole fermion effective action is local. Unlike a filter-profile comparison, the Mellin calculation explicitly establishes the required reference to complex-zeta regularization within the stated smooth-family class.

\subsection{Uniform ultraviolet-shell remainder: explicit estimate}\label{sec:uniform_shell_bound}
This short estimate supplies the only uniform bound used in the preceding Mellin argument. It is a four-dimensional version of ordinary momentum power counting, with ``number of modes'' replaced by a trace-norm estimate. We give the steps because they identify exactly what is taken from standard elliptic calculus and what is proved for our nested regulator.

Write $h=\Lambda^{-1}$ and $\eta=h\xi$, where $\xi$ is the ordinary momentum variable. The notation $\operatorname{Op}_h(a)$ below means the operator obtained from the momentum-dependent expression $a(x,\eta)$; it packages the usual large-momentum expansion, with $h$ counting inverse powers of the cutoff. The positive Dirac Laplacians satisfy
\begin{equation}
 h^2Q_\pm=\operatorname{Op}_h(|\eta|^2\mathbf1+hq_{\pm,1}+h^2q_{\pm,2}),
 \qquad q_{\pm,j}\text{ have uniformly bounded symbol seminorms on }\mathcal P.
 \label{eq:semiclassical_laplacian}
\end{equation}
The symbol coefficients depend on finitely many derivatives of the reconstructed gauge connection. Compact smoothness of $\mathcal P$ in (A1) makes their relevant symbol seminorms uniform. Ordinary momentum derivatives are uniformly bounded on each fixed annulus used below; $q_{\pm,1}$ itself need not be bounded as $|\eta|\to\infty$. The standard functional calculus for $f\in C_c^\infty$, with $f=1$ on $[0,u_0]$, then gives
\begin{equation}
 f(h^2Q_\pm)\sim f(|\eta|^2)\mathbf1+\sum_{j\ge1}h^j c_{\pm,j}(x,\eta),
 \qquad \operatorname{supp}_{\eta}c_{\pm,j}\subset\{c<|\eta|<C\}.
 \label{eq:shell_symbols}
\end{equation}
The $c_{\pm,j}$ involve only finitely many derivatives of $f$ and the background fields. The annular support follows because every correction contains a derivative of $f$, which vanishes both where the profile equals one and beyond its compact support. It is \emph{not} true that a whole filter commutes with a spatially varying Higgs field.

Let $P_h=f(h^2Q_+)^2$ and $K=\mathcal G_0^{-1}\widetilde V$. On the same annulus $K$ has an expansion beginning with $h k_{-1}$; its finite-dimensional Grushin correction is smoothing.

\emph{Excluding the inner momentum region.} The ultraviolet-shell expansion needs one infrared qualification. The symbol $hk_{-1}$ is not uniform at $\eta=hp=0$, but the commutator has no contribution from that region. Choose an inner spectral multiplier $\chi_L=\chi(h^2Q_+)$ supported strictly inside the range where $f=1$, with a positive gap between $\operatorname{supp}\chi$ and the region where $f^2\ne1$. Then $(P_h-1)\chi_L=0$ exactly.

Although the position-dependent Yukawa field can mix momenta, the off-diagonal piece $(P_h-1)K\chi_L$ is $O_{\mathcal S_1}(h^N)$ for every $N$: between the two separated momentum regions, smooth multiplication has rapidly decreasing Fourier matrix elements, equivalently the disjoint-support semiclassical composition is $O(h^\infty)$. The smoothing part of the Grushin inverse satisfies the same bound. A partition with this inner cutoff and an outer multiplier therefore confines all non-negligible commutator symbols to $|\eta|\ge c>0$, where the inverse Dirac symbol has its usual expansion. 

The principal symbol of $P_h$ is the scalar function $f(|\eta|^2)^2$. The leading $O(h)$ matrix commutator therefore cancels, and symbolic composition gives
\begin{equation}
 [P_h,K]\sim\sum_{j\ge2}h^j b_j(x,\eta),\qquad
 \operatorname{supp}_{\eta}b_j\subset\{c<|\eta|<C\}.
 \label{eq:commutator_shell_symbols}
\end{equation}
Every term in $\operatorname{Tr}[(KP_h)^n]-\operatorname{Tr}(K^nP_h^n)$ contains one such commutator and $n-1$ factors beginning at $O(h)$, so its first possible power is $h^{n+1}$. Expanding this ordered product \emph{without moving $P_h$ past $K$} through total symbol order $h^4$ yields local coefficients. The omitted symbol is $h^5r_{n,h}$, with $r_{n,h}$ supported in a fixed annulus and uniformly bounded together with enough of its derivatives, including any prescribed finite number of smooth-family derivatives.

For a compact-shell symbol $a_h$ on a four-dimensional compact manifold, the trace norm $\|T\|_1=\operatorname{Tr}\sqrt{T^\dagger T}$ is the sum of singular values and bounds the absolute value of an ordinary trace. The semiclassical shell estimate is
\begin{equation}
 \|\operatorname{Op}_h(a_h)\|_1\le C h^{-4}
 \sum_{|\alpha|\le M}\sup_{x,\eta}|\partial_{x,\eta}^{\alpha}a_h(x,\eta)|.
 \label{eq:trace_shell_ideal}
\end{equation}
Here $M$ is a fixed sufficiently large integer. K\"uster's local compact-support estimate~\cite[Eq.~(2.7)]{kuster} bounds the trace norm by $h^{-4}$ times finitely many $L^1$ symbol derivatives. A fixed shell converts these to the displayed suprema. Summing over the finite matrix entries in bundle charts and a finite coordinate partition gives this bundle version, with constants uniform on $\mathcal P$. Applying it to $h^5r_{n,h}$ gives the promised bound, uniformly with any fixed finite number of parameters differentiated:
\begin{equation}
 \boxed{\left|\partial_\theta^\beta\!\left(\mathscr E_{n,h}
 -\sum_{j=1}^{4-n}h^{n+j-4}\mathscr C_{n,j,f}\right)\right|
 \le C_{\beta,\mathcal P}h,\qquad 1\le n\le4.}
 \label{eq:trace_remainder}
\end{equation}
For $n=1$ the commutator term vanishes; for $n=4$ its entire trace is $O(h)$. Since the momentum shell has fixed size in $\eta$, its expansion produces powers of $h$, not $\log h$. A smoothing operator $S$ produces no extra ultraviolet term: for any $N$, $(1-f(h^2Q_\pm))Q_\pm^{-N}$ has norm $O(h^{2N})$, whereas $Q_\pm^NS$ is trace class. Thus $(1-f(h^2Q_\pm))S=O_{\mathcal S_1}(h^{2N})$, and likewise with the factors reversed.

\noindent\emph{What has been used.} We have used standard functional calculus and the trace-ideal estimate, not proved those general theorems anew. Their hypotheses hold for the smooth uniformly elliptic Dirac family in (A1). The nontrivial application-specific steps are the exact Dirac intertwining, the scalar leading symbol, the commutator starting at $h^2$, and the $h^5\times h^{-4}=h$ remainder bound. This is enough for the local cutoff-to-zeta comparison. It says nothing about locality in separated lattice links or about the $a\to0$ limit.

\subsection{From the common cutoff to the Quillen magnitude}\label{sec:cutoff_quillen_app}
We now turn the trace comparison into a statement about complete determinants. There are two distinct questions: (i) does the complex cutoff determinant agree with the complex-zeta determinant up to a local factor, and (ii) does its \emph{magnitude} agree with the positive spectral norm? The first is answered by $p_f$ and the second by $\delta_N$. A complex spectral cut defines their comparison coefficients near zero Yukawa scaling. The Fredholm and source construction defines the full finite-strength weight; real-analyticity of the squared norms then extends the local norm identity along real scaling, including zeros. Under (A1)--(A4), the argument identifies the finite local counterterm on an invertible auxiliary Grushin chart with compatible spectral cuts. The physical determinant line, rather than the augmented operator alone, is considered at the end.

Let $\mathcal G_0$ denote such an invertible gauge-only Grushin extension of $\mathcal D_0$, and let $\widetilde V$ be the Yukawa perturbation extended by zero on the auxiliary rows and columns. Write
\begin{equation}
\begin{aligned}
 K&=\mathcal G_0^{-1}\widetilde V,\qquad
 K_\Lambda=\mathcal G_0^{-1}\widetilde V_\Lambda,\qquad
 P_0=\mathcal G_0^\dagger\mathcal G_0,\\
 A_t&=\mathcal G_0^\dagger(\mathcal G_0+t\widetilde V)=P_0(1+tK),\\
 P_t&=(\mathcal G_0+t\widetilde V)^\dagger(\mathcal G_0+t\widetilde V).
\end{aligned}
 \label{eq:quillen_operators}
\end{equation}
The parameter $t$ counts Yukawa insertions; it need only be small while defining Taylor coefficients. Choose compatible Agmon cuts for $A_t$ and use the same zeta scale $\mu$ throughout. At finite cutoff the change from $\mathcal G_0$ is finite rank, so the ordinary Fredholm determinant is defined. After subtracting its divergent first four terms, the finite cutoff result $R_f^{\rm ren}$ may be compared with
\begin{equation}
 R_\zeta(t)=\frac{\det_\zeta A_t}{\det_\zeta P_0},\qquad
 \beta_{n,f}=\frac{(-1)^{n+1}}n\operatorname*{FP}_{\Lambda\to\infty}\operatorname{Tr}K_\Lambda^n,
 \quad\gamma_n=[t^n]\log R_\zeta(t).
 \label{eq:beta_gamma}
\end{equation}
Here $\operatorname{FP}$ means the constant remaining after subtraction of all positive powers and logarithms of $\Lambda$ in the same chosen scale convention. It does \emph{not} mean dropping the finite, generally nonlocal part of a loop. Because $K_\Lambda\to K$ in $\mathcal S_5$, the entire $\det_5$ remainder agrees in the limit. Thus, as an analytic germ at $t=0$,
\begin{equation}
 \boxed{R_f^{\rm ren}(t)=R_\zeta(t)e^{p_f(t)},\qquad
 p_f(t)=\sum_{n=1}^4t^n(\beta_{n,f}-\gamma_n).}
 \label{eq:pf_comparison}
\end{equation}
The local polynomial $p_f(t)$ is defined by this Taylor comparison at $t=0$; it does not require a complex spectral cut at arbitrary $t$. The continuation of the \emph{squared norm} to arbitrary real $t$, including changes of invertible chart, is explained below after Eq.~\eqref{eq:rho_app}. In particular, the fifth- and higher-order terms agree \emph{exactly}. Proposition~\ref{prop:mellin_anchor} establishes that the remaining first four differences $\beta_{n,f}-\gamma_n$ are local; the finite nonlocal portions of the low-order loops also agree.

\paragraph{The matched weight is independent of the filter.}
Before introducing the complex-zeta comparison, the convergent chart expression is
\[
 R_f^{\rm ren}(t)=\det\nolimits_5(1+tK)
                   \exp\!\left(\sum_{n=1}^4t^n\beta_{n,f}\right).
\]
Substituting the polynomial $p_f$ from~\eqref{eq:pf_comparison} gives, at every finite $t$,
\[
 R_Q(t)=e^{-p_f(t)-\delta_N^{\rm loc}(t)/2}R_f^{\rm ren}(t)
       =\det\nolimits_5(1+tK)
          \exp\!\left(\sum_{n=1}^4t^n\gamma_n
                         -\tfrac12\delta_N^{\rm loc}(t)\right).
\]
This is an entire determinant times the exponential of a fixed polynomial; it requires no complex spectral cut away from the germ at zero. Neither $f$ nor the auxiliary filter shift $\sigma$ occurs on the right. The scale $\mu$, massless references and finite matching conditions remain fixed. The coefficients $\gamma_n$ generally contain nonlocal finite amplitudes; only $\beta_{n,f}-\gamma_n$ are local. This distinction also appears in the factorization theorem of Ref.~\cite[Thm.~2 and its footnote~8]{dang}. The displayed $R_Q$ is an augmented relative factor. The physical weight $W_Q$ includes the low-mode exterior polynomial and reference-line factor of Appendix~\ref{sec:quillen_gluing}, including in sectors with identically vanishing vacuum determinant. The finite counterterms in Appendix~\ref{app:counterterms} relate $W_f^{\rm ren}$ to this common $W_Q$; they are not residual filter parameters of $W_Q$.

\paragraph{Why the low-order differences are local.}
Proposition~\ref{prop:mellin_anchor} supplies the needed comparison rather than assuming the two renormalized finite parts coincide away from the diagonal. It proves that the first four differences between the filtered trace coefficients and fixed-weight zeta traces are local residues plus ultraviolet-shell commutators. The varying complex-zeta weight contributes the explicit local residue integral in Eq.~\eqref{eq:zeta_taylor_v24}. Therefore the function $p_f$ in Eq.~\eqref{eq:pf_comparison} is local, of Yukawa degree at most four. This is consistent with the independent position-space scaling-degree count $3n-4(n-1)=4-n$: only the $n\le4$ distributions can require local extensions. Dang's general determinant-renormalization theorem~\cite{dang} provides the broader framework; the Mellin argument verifies the specific nested-filter comparison on the smooth families described by (A1)--(A4). Grushin smoothing terms change no ultraviolet counterterm, although the finite-dimensional physical determinant-line factors must still be included in the norm.

\paragraph{A separate correction for the squared magnitude.}
For $t$ in a neighborhood of zero, where a compatible complex spectral cut exists, the ratio of two \emph{squared} magnitudes defines the germ
\begin{equation}
 \delta_N^{\rm germ}(t)=2\operatorname{Re}\log R_\zeta(t)
       -\log\frac{\det_\zeta P_t}{\det_\zeta P_0}
       \qquad (t\text{ near }0).
 \label{eq:q_def_app}
\end{equation}
This is not a second cutoff divergence. It is a finite normalization difference between two ways of taking a determinant's magnitude. Applying the order-two weighted-trace identity, with $X=\mathcal G_t^{-1}\delta\mathcal G_t$ and $\mathcal G_t=\mathcal G_0+t\widetilde V$, gives the germ identity
\begin{equation}
 \delta\delta_N^{\rm germ}=-\operatorname{Re}\operatorname{res}
 \bigl[X(\log A_t-\log P_t)\bigr]
 \qquad (t\text{ near }0).
 \label{eq:q_res_app}
\end{equation}
The Wodzicki residue $\operatorname{res}$ reads off the symbol coefficient of degree $-4$ and is consequently local. Its Yukawa-dependent Taylor coefficients vanish beyond degree four: higher insertions lower the symbol order below $-4$. Thus $\delta_N^{\rm germ}(t)$ agrees exactly with a polynomial $\delta_N^{\rm loc}(t)$ of degree at most four on its germ domain, with $\delta_N^{\rm loc}(0)=0$. Extending that polynomial to all $t$ involves no truncation error; set $\delta_N:=\delta_N^{\rm loc}(1)$. Both quantities are thus defined without making a complex spectral-cut choice at $t=1$. The factor $1/2$ associated with order-two weights cancels the factor $2$ arising from the squared modulus; it must \emph{not} be inserted again. Equations~\eqref{eq:pf_comparison} and~\eqref{eq:q_def_app} yield, initially for small $t$, the relative \emph{real} normalization. At $t=1$ its continuation established below gives
\begin{equation}
 \boxed{\rho_f=-\operatorname{Re}p_f(1)-\tfrac12\delta_N,\qquad
 |e^{\rho_f}R_f^{\rm ren}(1)|^2
 =\frac{\det_\zeta P_1}{\det_\zeta P_0}.}
 \label{eq:rho_app}
\end{equation}
\begin{unnumberedlemma}[Real-analytic squared norm]
Fix a smooth bosonic background and scale its Yukawa operator as $\mathcal D_t=\mathcal D_0+t\mathcal V_Y$, with $t$ real. The complete squared Quillen norm in real-analytic spectral-window frames is real analytic in $t$, including at changes of kernel dimension. The same holds for fixed-source coefficients defined by the finite Grassmann construction.
\end{unnumberedlemma}
\begin{proof}
The bundles and their Hermitian metrics are fixed along this scaling. Complexify the parameter by defining
\[
 \mathscr P_+(z)=(\mathcal D_0^\dagger+z\mathcal V_Y^\dagger)
                    (\mathcal D_0+z\mathcal V_Y),\qquad
 \mathscr P_-(z)=(\mathcal D_0+z\mathcal V_Y)
                    (\mathcal D_0^\dagger+z\mathcal V_Y^\dagger).
\]
These are holomorphic elliptic families on common $H^2$ domains, with fixed principal symbol, and equal the positive source/target Laplacians for real $z=t$. The daggered coefficients in this extension are fixed; no $\bar z$ is introduced. Near any real $t_0$, choose a positive window edge away from both spectra. The resolvent identity and contour formula make its finite-rank Riesz projectors analytic. Kato transport gives real-analytic orthonormal frames on the real interval~\cite{kato}. The complementary operators remain uniformly invertible there. Their positive zeta determinants are real analytic: the parameter-dependent resolvent/heat parametrix has holomorphic coefficients and uniform remainders near $t_0$, so the Mellin representation, after subtracting finitely many heat terms, is holomorphic in $z$ and meromorphic in the zeta variable. Its derivative at zero is holomorphic as well. Low-mode matrices and metric Gram factors are finite dimensional and real analytic; fixed-source contractions with the complementary inverse have the same regularity. Their finite Grassmann coefficients therefore give real-analytic complete squared norms when multiplied by the high-mode spectral factor. Zero modes remain inside the window and introduce no inverse. The isometries in Appendix~\ref{sec:quillen_gluing} identify these expressions on overlaps. Only analyticity in the scaling parameter is used; the bosonic background need only be smooth.
\end{proof}

The squared-norm equality~\eqref{eq:rho_app} is first established near $t=0$. Both its complete physical sides are real analytic on overlapping windows by the lemma: the renormalized side uses the entire $\det_5(1+tK)$, the restored polynomial finite parts, and the physical low-mode source factors; the positive side uses the Quillen metric just constructed. The correcting polynomials $p_f(t)$ and $\delta_N^{\rm loc}(t)$ are fixed by their degree-at-most-four jets at zero and are defined for every finite $t$. Real-analytic uniqueness propagates the equality along a finite cover of $0\le t\le1$. This proves the full-strength norm identification, including at zero crossings. It is the \emph{squared norm} that is analytic; its square root or logarithm can fail to be analytic at a zero. No global complex logarithm or single Agmon cut is needed along the path.

The original identity concerns the \emph{augmented} operators in one invertible chart; the displayed continuation uses the physical determinant-line metric, including the low-mode exterior factors and their $L^2$ norms. A real positive multiplication leaves the complex phase unchanged.

\paragraph{Complex phase matching and physical reality conditions.}
The complex comparison includes the local phase difference $\operatorname{Im}p_f$, which the Quillen norm alone does not determine. We use \emph{one} complex renormalized weight, $W_Q=e^{-p_f-\delta_N/2}W_f^{\rm ren}$, where $\delta_N=\delta_N^{\rm loc}(1)$, as in Eq.~\eqref{eq:main_rho}; no independent ``optional'' phase is left in the definition. The reference complex-zeta convention is fixed in a neighborhood of $t=0$, together with cuts compatible under complex conjugation and CP reflection of the backgrounds and Yukawa spurions. The polynomial $p_f$ carries this local convention to arbitrary Yukawa strength without requiring an admissible complex-zeta cut at $t=1$. This correction is gauge invariant, local, and the identity when $\Phi=0$.

The explicit parity-even quadratic and quartic terms have real flavor traces such as $\operatorname{tr}(Y^\dagger Y)$ and $\operatorname{tr}(Y_u^\dagger Y_uY_d^\dagger Y_d)$. Genuine CKM-odd polynomial invariants first occur at higher Yukawa degree; the strong-CP mass phase instead comes from topological zero-mode factors and is not a Taylor polynomial at $Y=0$. Although these observations support $\operatorname{Im}p_f=0$ in the usual CP-compatible subtraction convention, they are not used here as an unproved general vanishing theorem for the augmented complex-zeta comparison. Instead the local difference is removed by definition, while the theta angles and Yukawa matrices remain physical inputs. The chart-independent statement uses the common local counterterm density and massless reference phase.

\section{Explicit finite Higgs counterterms}\label{app:counterterms}
The previous appendix showed that the cutoff-to-Quillen difference is local. Here we ask \emph{which local Higgs operators it changes}. These are the same operators that would be fixed by renormalization conditions on the scalar two- and four-point functions; the point of this calculation is to determine their dependence on our particular covariant spectral regulator. In familiar language, the residues below extract logarithmic ultraviolet terms from one-loop fermion graphs; their finite counterparts are local scheme-dependent shifts, not new physical interactions. At quartic Yukawa order, power counting permits a local residue with no Higgs derivatives. At quadratic order, both the Higgs mass and kinetic operators appear, and the kinetic coefficient probes derivatives of the spectral profile. We keep the conventions for the cutoff $f$, infrared shift $\sigma$, and zeta scale $\mu$ from Section~\ref{sec:quillen_matching}.

\subsection{The explicit quartic residue}\label{sec:quartic_quillen_app}
At Yukawa order four, no external derivative is allowed by four-dimensional power counting. The principal symbol of $P_0$ is $|\xi|^2I$. There are eight copies of the spectral profile around the four vertices, and the resulting fixed-weight finite-part difference is
\begin{equation}
 \operatorname*{FP}_{\Lambda\to\infty}\operatorname{Tr}K_\Lambda^4
  -\operatorname{tr}^{P_0}(K^4)
 =\kappa_f^{(4)}\operatorname{res}(K^4),\qquad
 \kappa_f^{(4)}=\frac12\int_0^\infty\frac{du}{u}
 \bigl[f(u)^8-\mathbf 1_{u<1}\bigr],
 \label{eq:cf_quartic}
\end{equation}
where $\operatorname{tr}^{P_0}(T)=\operatorname*{FP}_{z=0}\operatorname{Tr}(TP_0^{-z})$. The step function fixes the reference radial subtraction and is \emph{not} the cutoff profile used in the regulator. The integral is convergent for $f=1$ near zero and compactly supported.

There is also a local correction because zeta determinants of products do not in general multiply as finite-dimensional determinants do. One cannot replace the first four coefficients of $\log\det_\zeta A_t$ with ordinary cyclic traces: the weighted trace is not generally cyclic. The correct Taylor identity is
\begin{equation}
 \log R_\zeta(t)=\sum_{n\ge1}\frac{(-1)^{n+1}t^n}{n}
     \operatorname{tr}^{P_0}(K^n)
 -\frac12\int_0^t ds\,\operatorname{res}
 \left[(1+sK)^{-1}K(\log A_s-\log P_0)\right].
 \label{eq:zeta_taylor_v24}
\end{equation}
At fourth order, only principal symbols enter the residue. Writing $\log A_t-\log P_0=tL_1+t^2L_2+t^3L_3+\cdots$, its leading pieces are $L_1\sim K$, $L_2\sim-K^2/2$, and $L_3\sim K^3/3$, modulo terms of lower pseudodifferential order. Their relative coefficients sum to $1+1/2+1/3=11/6$; integration in $t$, together with the factor $1/2$ in Eq.~\eqref{eq:zeta_taylor_v24}, produces $11/48$. Therefore
\begin{equation}
 \gamma_4=-\frac14\operatorname{tr}^{P_0}(K^4)
          -\frac{11}{48}\operatorname{res}(K^4),\qquad
 \boxed{p_{4,f}=\left(\frac{11}{48}-\frac{\kappa_f^{(4)}}{4}\right)
                    \operatorname{res}(K^4).}
 \label{eq:quartic_p}
\end{equation}
For reference, the $11/48$ arises from $(1/3+1/2+1)/8$: the three terms are the leading $K^3/3$, $-K^2/2$, and $K$ contributions to the logarithmic-residue expansion. Commutator-symbol terms lie at lower order and do not alter this coefficient. The \emph{full} quartic real counterterm is $\rho_{4,f}=-\operatorname{Re}p_{4,f}-\delta_{N,4}/2$ with $\delta_{N,4}$ determined by Eq.~\eqref{eq:q_res_app}; Eq.~\eqref{eq:quartic_p} alone is not that full counterterm. As a cross-check, in a free constant real Dirac mass background one finds $\delta_{N,4}=-5\operatorname{res}(K^4)/24$ and thus $\rho_{4,f}=(\kappa_f^{(4)}/4-1/8)\operatorname{res}(K^4)$. This special-background identity is not asserted as the complete Standard Model quartic coefficient.\par
\noindent\emph{Physical interpretation.} The residue is the coefficient of the logarithmic ultraviolet scaling of a four-vertex loop. Its locality explains why changing the cutoff profile shifts a quartic Higgs coupling rather than a nonlocal scattering amplitude.

\paragraph{Constant-mass check of the norm defect.}
For a free Dirac fermion of constant real mass $tm$, write $V_4=\int d^4x$. The order-$t^4$ Mellin coefficient in $-2\zeta_{A_t}(s)+\zeta_{P_t}(s)$ is
\[
 -\frac{s^2(s+1)(s+5)}{3}\,V_4m^4 I(s),\qquad
 I(s)=\int_{|p|\ge r_0}\!\frac{d^4p}{(2\pi)^4}(p^2)^{-s-2}
       =\frac{1}{16\pi^2s}+O(1).
\]
Here $r_0>0$ is any fixed infrared threshold; the ultraviolet pole is independent of it. The numerator follows by combining the fourth coefficient of $(p^2\pm itm|p|)^{-s}$ with that of $(p^2+t^2m^2)^{-s}$, including four spin components. Differentiation at $s=0$ gives $\delta_{N,4}=-5V_4m^4/(48\pi^2)$. Since $\operatorname{res}(K^4)=V_4m^4/(2\pi^2)$, the ratio is $-5/24$, as quoted above. This checks the constant-mass benchmark.

\subsection{Explicit quadratic Higgs normalization}\label{sec:quadratic_quillen_app}
This subsection computes local matching coefficients in flat $\mathbb R^4$, with continuous external momentum. Finite-volume nonlocal terms in the torus amplitude are common to the compared prescriptions and are not part of these counterterms. At quadratic order, unlike the quartic case, derivatives of the Higgs field affect the \emph{finite} local term. To show why the nested filters matter, start with one free Dirac fermion and a real position-dependent mass $m(x)$. We normalize the amplitude $B_{\Lambda,f}$ so that the quadratic contribution to $\log W$ is $\int d^4q\,(2\pi)^{-4}|m(q)|^2B_{\Lambda,f}(q)$; its sign is reversed in the effective action $-\log W$. The coefficients below have the factor $|m(q)|^2$ divided out. For $H_\pm=\sigma^2+D_0^\dagger D_0$ or $\sigma^2+D_0D_0^\dagger$, with $\sigma^2>0$, the two-point momentum integral is
\begin{equation}
 B_{\Lambda,f}(q)=\frac{2}{(2\pi)^4}
 \int d^4p\,\frac{p\cdot(p+q)}{p^2(p+q)^2}
 f\!\left(\frac{p^2+\sigma^2}{\Lambda^2}\right)^2
 f\!\left(\frac{(p+q)^2+\sigma^2}{\Lambda^2}\right)^2.
 \label{eq:quadratic_momentum}
\end{equation}
The two profiles depend on \emph{different} momenta. At zero external momentum, radial integration is elementary:
\begin{equation}
 B_{\Lambda,f}(0)=\frac{\Lambda^2}{8\pi^2}\int_0^\infty f(u)^4\,du
                     -\frac{\sigma^2}{8\pi^2}.
 \label{eq:mass_moment}
\end{equation}
This equality is exact whenever $\sigma^2/\Lambda^2$ lies in the interval where $f=1$; no asymptotic remainder is needed. The second term survives the removal of the fermion cutoff; it is due to the positive spectral shift, and is not a second quadratic divergence.

To obtain the finite kinetic normalization, compare two smooth profiles $f$ and $h$ at the same ultraviolet scale. Their difference has no infrared divergence and, once the common logarithm cancels, the angular expansion at order $q^2$ is
\begin{equation}
 \lim_{\Lambda\to\infty}[B_{\Lambda,f}(q)-B_{\Lambda,h}(q)]_{q^2}
 =-\frac{q^2}{16\pi^2}\int_0^\infty du\left[
  \frac{f(u)^4-h(u)^4}{u}
  +4u\bigl(f(u)^2f'(u)^2-h(u)^2h'(u)^2\bigr)\right].
 \label{eq:kin_profile_diff}
\end{equation}
The first integrand comes from the propagators; the second comes from derivatives of the \emph{nested} filters at momenta $p$ and $p+q$. Dropping the second term would give a wrong finite wavefunction counterterm.\par
The derivative term can be checked before taking the cutoff limit. Put $r=p^2$ and $g(r)=f((r+\sigma^2)/\Lambda^2)^2$. Angular averaging of the integrand of~\eqref{eq:quadratic_momentum} gives the $q^2$ coefficient
\[
 -\frac{g^2}{2r^2}+\frac{gg'}{2r}+\frac{gg''}{2}.
\]
Here primes differentiate $r$. Multiplication by the radial measure $r\,dr/(8\pi^2)$ and integration by parts use $\int(gg'+rgg'')\,dr=-\int r(g')^2\,dr$, with zero boundary terms. Taking the difference of two profiles and setting $u=r/\Lambda^2$ gives exactly~\eqref{eq:kin_profile_diff}, including its sign and factor four.
\par
\noindent\emph{Physical interpretation.} A position-dependent Higgs field transfers momentum through a fermion loop. The two propagators therefore probe slightly different parts of the cutoff profile, and the derivative of that profile shifts the finite coefficient of $|D_\mu\Phi|^2$.\par

Define the absolutely convergent moments
\begin{equation}
\begin{aligned}
 I_f^{\rm log}&=\int_0^\infty\frac{du}{u}\bigl[f(u)^4-e^{-2u}\bigr],\\
 I_f^{\rm der}&=4\int_0^\infty du\,u f(u)^2[f'(u)]^2,\\
 \kappa_f^{\rm kin}&=\gamma_E+\log2+\tfrac12-I_f^{\rm log}-I_f^{\rm der}.
\end{aligned}
 \label{eq:Cf}
\end{equation}
The Gaussian $f_G(u)=e^{-u/2}$ is a convenient \emph{calculation reference} (not a compactly supported finite-rank filter): Schwinger-parameter integration gives a finite kinetic coefficient $[\log|W_G^{\rm ren}|]_{q^2}=q^2[\log(q^2/\mu^2)+\gamma_E+\log2-7/4]/(16\pi^2)$ with $|m(q)|^2$ divided out. Its profile moments are $I_{f_G}^{\rm log}=0$ and $I_{f_G}^{\rm der}=1/4$. For an explicit check of the Gaussian constant, Schwinger parameters for the two propagators give
\[
 B_{\Lambda,G}(q)=\frac{e^{-2\sigma^2/\Lambda^2}}{8\pi^2}
 \int_0^1du\left[\frac{2e^{-a q^2L}}{L}
                      -3a q^2E_1(a q^2L)\right],\quad
 a=u(1-u),\quad L=\frac{\Lambda^{-2}}{\min(u,1-u)}.
\]
Here $E_1(z)=\int_z^\infty e^{-v}dv/v$. Using $E_1(z)=-\gamma_E-\log z+O(z)$, $\int_0^1a\,du=1/6$, and
\[
 \int_0^1 a\log\max(u,1-u)\,du=-\frac{13}{72}+\frac{\log2}{6}
\]
gives the stated finite constant $\gamma_E+\log2-7/4$ after the logarithmic subtraction. The prefactor supplies the finite mass term $-\sigma^2/(8\pi^2)$ in~\eqref{eq:mass_moment} and does not alter the limiting kinetic coefficient. The relative difference between a smooth plateau $f$ and this reference is computed in the high-momentum region, where it is local. At the fixed zeta scale $\mu$ the logarithm of the cutoff fermion weight has the quadratic ultraviolet subtraction
\begin{equation}
 C_{\Lambda,f}^{(2),\rm div}=
 \frac{\mathcal T_Y\Lambda^2}{8\pi^2}
 \left(\int_0^\infty f(u)^4du\right)\int_X\Phi^\dagger\Phi\,d^4x
 -\frac{\mathcal T_Y}{16\pi^2}\log\frac{\Lambda^2}{\mu^2}
 \int_X(D_\mu\Phi)^\dagger(D_\mu\Phi)\,d^4x.
 \label{eq:quadratic_subtraction}
\end{equation}
The counterterm in the Euclidean \emph{effective action} has the opposite sign: by definition $W\propto e^{-S_{\rm eff}}$. This explicit sign convention prevents confusion about the Higgs mass term.

An independent check comes from zeta-function perturbation theory. For a single Dirac fermion and a scalar Fourier amplitude $m(q_{\rm ext})$, the quadratic coefficients, after dividing out $|m(q_{\rm ext})|^2$, are $q_{\rm ext}^2[\log(q_{\rm ext}^2/\mu^2)-7/3]/(16\pi^2)$ for the complex-zeta logarithm and $q_{\rm ext}^2[\log(q_{\rm ext}^2/\mu^2)-2]/(16\pi^2)$ for the \emph{logarithm of the Quillen modulus}. These constants follow directly from the quadratic zeta coefficients. Let $\mathrm B_E(u,v)=\Gamma(u)\Gamma(v)/\Gamma(u+v)$ be Euler's beta function. Dividing out $|m(q)|^2$, the flat-space kinetic coefficients before differentiating the zeta variable are
\[
 \begin{aligned}
 \zeta_A^{(2)}(s;q)&=\frac{q^2}{8\pi^2}
    \left(\frac{q^2}{\mu^2}\right)^{-s}
    \frac{3-s}{1-s}\,\mathrm B_E(2-s,2-s),\\
 \zeta_P^{(2)}(s;q)&=\frac{q^2}{8\pi^2}
    \left(\frac{q^2}{\mu^2}\right)^{-s}
    \mathrm B_E(1-s,1-s).
 \end{aligned}
\]
Feynman parametrization gives these expressions after the spin trace; in $P_t$ the term quadratic in $m$ has no external-momentum dependence and does not affect this kinetic coefficient. Applying $-\partial_s$ to the first and $-\tfrac12\partial_s$ to the second at zero yields the constants $-7/3$ and $-2$, respectively. Combining them with the Gaussian and profile calculations also fixes $\kappa_f^{\rm kin}$ in~\eqref{eq:Cf}. The constants $-7/3$ and $-2$ need not coincide: the complex determinant is built from an operator with derivative Yukawa couplings, whereas the positive operator entering the Quillen metric treats those terms differently. The common nonlocal logarithm cancels from their difference. Correspondingly, the quadratic norm defect is
\begin{equation}
 \delta_{N,2}=-\frac{\mathcal T_Y}{24\pi^2}
 \int_X(D_\mu\Phi)^\dagger(D_\mu\Phi)\,d^4x.
 \label{eq:q2_SM}
\end{equation}
The Dirac computation extends to the Standard Model because, on the flat four-torus, gauge invariance and dimension-four power counting allow only $\Phi^\dagger\Phi$ and $(D_\mu\Phi)^\dagger D_\mu\Phi$ as quadratic Higgs invariants, modulo integration by parts. A constant neutral Higgs background gives $\sum|m_f|^2=\mathcal T_Y\Phi^\dagger\Phi$; the color multiplicity is three, and both up- and down-type Higgs contractions have the same kinetic invariant. Thus the \emph{finite real factor needed to match Quillen} at quadratic order is
\begin{equation}
 \boxed{\rho_{2,f}=
 \frac{\sigma^2\mathcal T_Y}{8\pi^2}\int_X\Phi^\dagger\Phi\,d^4x
 -\frac{\mathcal T_Y\kappa_f^{\rm kin}}{16\pi^2}
 \int_X(D_\mu\Phi)^\dagger(D_\mu\Phi)\,d^4x.}
 \label{eq:rho2_app}
\end{equation}
For numerical illustrations one may set $\sigma=1$ in chosen mass units; this convention is not dimensionless. The finite wavefunction coefficient depends on the chosen smooth spectral profile; the logarithmic coefficient in Eq.~\eqref{eq:quadratic_subtraction} does not. For curved backgrounds an additional $R\Phi^\dagger\Phi$ invariant may arise and is not calculated here. All statements about locality and convergence are uniform only on the specified compact smooth patches of reconstructed fields. Neither the quadratic nor the quartic correction determines a physical Higgs parameter without an independent renormalization condition.

\section{Zero modes, the Quillen metric, and physical phases}\label{app:grushin}
\subsection{Pfaffian reduction and global Yukawa contractions}
For the block notation of Section~\ref{sec:yukawa_full}, let $\mathscr V_-$ and $\mathscr V_+$ denote the spaces of the two ordered groups of Grassmann variables. A bilinear pairing is a map into the dual of the second space:
\[
 \mathcal D_Y:\mathscr V_-\longrightarrow\mathscr V_+^*,
 \qquad \mathcal D_Y^T:\mathscr V_+\longrightarrow\mathscr V_-^*.
\]
The spinor pairing and the canonical gauge-bundle dualities identify these dual targets with the operator bundles used in the main text. The independent Euclidean conjugate fields supply the dual variables. This explains the source/target terminology without identifying a field with its complex conjugate.

To verify the skew Nambu form without imposing symmetry on generation matrices, combine gauge, generation and species into indices $A,B$, and denote Weyl spinor indices by $a,b$. Then
\[
(M_\Phi)_{Aa,Bb}=\varepsilon_{ab}H_{AB}(\Phi),\qquad H^T=H.
\]
Insert $(Y_u)_{rs}\varepsilon_{ij}\Phi^j$ in the $(Q_r,u_s^c)$ block of $H$ and its transpose in the reverse block; do likewise for down quarks and leptons. The bosonic entries commute, so $M_\Phi^T=-M_\Phi$ regardless of the symmetry of the $Y$ matrices. For a finite-dimensional $D$, and only at finite dimension,
\begin{equation}
\Pf\begin{pmatrix}0&-D^T\\D&0\end{pmatrix}=(-1)^{N_{\rm mat}(N_{\rm mat}+1)/2}\det D,
\label{eq:pf_det_sign}
\end{equation}
with $N_{\rm mat}=\dim D$ and fixed Grassmann ordering~\cite[Remark~11.5]{freedlectures}. No infinite-dimensional limit is taken of this numerical sign: the continuum Pfaffian/determinant identification uses compatible determinant-line frames and fixes the ordering convention separately on each chart and index sector. The bundle pairings that make the Yukawa action gauge invariant are reviewed in Appendix~\ref{app:geometry}.

\subsection{Optional finite-part phase differential}
On an invertible chart one may consider the formal logarithmic derivative
\begin{equation}
\mathcal A^{\rm sp}(\dot{\mathcal D}_Y)=
\operatorname*{FP}_{z=0}\operatorname{Tr}\!\left[
(\mathcal D_Y^\dagger\mathcal D_Y/\mu^2)^{-z}\mathcal D_Y^{-1}\dot{\mathcal D}_Y\right].
\label{eq:pf_connection}
\end{equation}
It does not by itself specify a phase through zero modes or a global integration convention, and is \emph{not} used to define our common-cutoff weight. Compatibility with a chosen determinant-line connection can require local adjustments.

\subsection{Rectangular low modes and source insertions}\label{sec:source_families}
The reason for introducing an augmented operator is simple: $D^{-1}$ cannot be used at a zero mode, and a nonzero-index chiral operator is not invertible even away from accidental crossings. We therefore isolate a finite set of potentially singular modes, invert only the remaining high-energy block, and retain the low-mode Grassmann integral explicitly. 

\paragraph{Source families and local Hilbert-space identifications.}
Write $\mathscr H_\pm(U)=L^2(S^\pm\otimes E_R[U])$. On a compact parameter neighborhood, realize $E_R[U]$ as the range of a smooth projector $Q_U$ in a fixed finite ambient bundle over spacetime. A finite collection of bundle charts and a smooth partition of unity provides such an embedding. Its ambient bundle need not have a global frame; covering-space coordinates obey the torus deck identifications described in Appendix~\ref{app:family_reconstruction}. Near $U_0$, put $Q_0=Q_{U_0}$ and
\[
 B_U=Q_UQ_0+(1-Q_U)(1-Q_0),\qquad
 T_U=B_U(B_U^\dagger B_U)^{-1/2},\qquad T_UQ_0=Q_UT_U.
\]
For $\|Q_U-Q_0\|<1$, this is a smooth unitary bundle map, also smooth in the spacetime variables. It acts on the gauge factor and commutes with Clifford multiplication. Transporting by $T_U$ identifies the Sobolev spaces and changes $D_U$ to $D_{U_0}+V_U$, where $V_U$ is a zeroth-order multiplication operator continuous in operator norm on $L^2$. Thus the doubled self-adjoint Dirac family has common $H^1$ domain and a bounded norm-continuous perturbation. Resolvent identities give norm-continuous low spectral projectors and a norm-continuous complementary inverse, bounded by $\lambda^{-1/2}$ when the squared singular-value window has edge $\lambda>0$.

A source family means a section $f_U$ of $\mathscr H_+(U)$, or $g_U$ of $\mathscr H_-(U)$, whose transport to these fixed spaces is $L^2$ continuous. Smoothness in spacetime is retained in the theorem. For example, choose a smooth local frame $e_i^R(U;y_0)$ induced from a vertex frame, a smooth spinor test function $h$ supported in a contractible patch, and specified smooth paths from $y_0$. Then
\[
 f_U(y;v)=h(y)\,\mathcal P_U(y\leftarrow y_0)e_i^R(U;y_0)v
\]
is continuous, where $\mathcal P_U$ is gauge parallel transport. The reference vector $v$ carries the vertex representation: under a lattice gauge transformation, $f_{U^g}(y;R(g_i)v)=\Phi_g(y)f_U(y;v)$. Thus the profile and its external indices are transformed together. Gauge-invariant observables contract those indices with continuous bosonic data.

Choose a common spectral window locally. Each source coefficient is a finite sum of high contractions, low source projections and minors of a finite reduced matrix, multiplied by a continuous Quillen frame. The preceding identifications prove continuity of every factor. The shell identities in Appendix~\ref{sec:quillen_gluing} make this a definition independent of the window, even when the original operator is invertible nowhere in the component. No density of an invertible locus is assumed.

\paragraph{Exact finite-dimensional reduction.}
At any reference background split the source and target spaces into finite low-mode and invertible high-mode parts. In these bases write
\begin{equation}
D_{U,\phi}=\begin{pmatrix}M&B\\C&D_{\mathrm{hi}}\end{pmatrix},\qquad
D_{\mathrm{hi}}\text{ invertible},\qquad M:\mathbb C^k\to\mathbb C^c.
\label{eq:rect_block}
\end{equation}
The finite-dimensional defect has index $k-c=\ind D$. Add $R_-:\mathbb C^c\to\mathscr H_-$ and $R_+:\mathscr H_+\to\mathbb C^k$ to form the invertible map from $\mathscr H_+\oplus\mathbb C^c$ to $\mathscr H_-\oplus\mathbb C^k$, the \emph{Grushin extension}
\begin{equation}
\mathcal G_D=\begin{pmatrix}D&R_-\\R_+&0\end{pmatrix},
\qquad E_{-+}=-(M-BD_{\mathrm{hi}}^{-1}C).
\label{eq:rect_grushin}
\end{equation}
The extension exists locally for every index, even though $D$ itself cannot be made invertible by finite-rank corrections when its index is nonzero. The effective map $-E_{-+}:\mathbb C^k\to\mathbb C^c$ is generally rectangular. Its fermion integral is therefore a source polynomial rather than an ordinary determinant.

With low sources $\bar\eta_\ell,\eta_\ell$ and high sources $\bar\eta_h,\eta_h$, Gaussian integration of only the invertible high modes gives the exact algebraic identity
\begin{align}
\mathcal Z_D[\bar\eta,\eta]
={}&\det D_{\mathrm{hi}}\;e^{\bar\eta_h D_{\mathrm{hi}}^{-1}\eta_h}
\int[d\bar\chi\,d\chi]\,
\exp\!\left[-\bar\chi (M-BD_{\mathrm{hi}}^{-1}C)\chi
+\bar\eta_{\rm eff}\chi+\bar\chi\eta_{\rm eff}\right],
\label{eq:rect_sources}\\
\bar\eta_{\rm eff}={}&\bar\eta_\ell-\bar\eta_h D_{\mathrm{hi}}^{-1}C,
\qquad \eta_{\rm eff}=\eta_\ell-BD_{\mathrm{hi}}^{-1}\eta_h.\nonumber
\end{align}
The low Grassmann integral is polynomial and smooth when its effective matrix develops zero modes. Counting the unpaired low variables yields $N_{\bar\eta}-N_\eta=k-c$ for each separate quark or lepton block. The indices~\eqref{eq:baryon_idx} then give Eq.~\eqref{eq:BL_selection}. This is an integrated anomalous Ward identity, rather than an assertion of classical quantum conservation of $B+L$. Nontrivial gauge indices of the source zero modes must be contracted to form gauge-invariant observables.

\paragraph{Continuity of regularized source minors.}
The source statement also survives the $\mathcal S_5$ limit in Appendix~\ref{app:yukawa_uv}. To see this directly, for $T\in\mathcal S_5$ set
\[
 S_4(T)=-T+\tfrac12T^2-\tfrac13T^3+\tfrac14T^4,\qquad
 B_5(T)=(1+T)e^{S_4(T)}-1\in\mathcal S_1.
\]
The map $T\mapsto B_5(T)$ is analytic into trace class: its power series starts at degree five, and the Schatten ideal inequalities control the series on bounded sets. The ordinary Fredholm adjugate $\mathfrak A(B)=\det_F(1+B)(1+B)^{-1}$ has an analytic extension across noninvertibility. Hence
\[
 \mathfrak M_5(T)=e^{S_4(T)}\mathfrak A(B_5(T))
\]
is continuous in the $\mathcal S_5$ topology and agrees with $\det_5(1+T)(1+T)^{-1}$ wherever the inverse exists. For $r$ source pairs, use the entire Fredholm minor
\[
 \mathfrak A_r(B)=\det_F(1+B)\,\bigwedge\nolimits^r(1+B)^{-1},
 \qquad
 \mathfrak M_{5,r}(T)
 =\bigwedge\nolimits^r(e^{S_4(T)})\,\mathfrak A_r(B_5(T)),
\]
where the first expression extends analytically through noninvertibility. It gives one determinant times the antisymmetrized $r$-fold inverse, rather than the $r$th exterior power of the first adjugate. These are Fredholm minors, not separate products of inverse entries; see Ref.~\cite[Secs.~6--7, especially Thm.~7.1]{simon_det}. Contracting with finitely many $L^2$ sources and adjoining the finite low-mode polynomials preserves continuity and uniform convergence on compact families. Thus the source limits used in Theorems~3--4 have no additional singularity at a zero of the determinant.

\subsection{Determinant-line norm and the finite-rank comparison}
\paragraph{Fixed spectral scale.}
For a positive invertible operator $P$, the convention used throughout obeys
\[
 \zeta_{P/\mu^2}(z)=\mu^{2z}\zeta_P(z),\qquad
 \det_\zeta(P/\mu^2)=\mu^{-2\zeta_P(0)}\det_\zeta P.
\]
This follows by differentiating at $z=0$ in $\det_\zeta P=\exp[-\zeta'_P(0)]$. For a primed determinant apply the same identity to the positive spectrum, retaining the finite zero-mode factors separately. The scale $\mu$ is fixed in the existence theorem and in each comparison of normalizations.

The Quillen metric combines the positive spectral determinant of the nonzero modes with the ordinary metric of finite-dimensional kernel and cokernel spaces~\cite{bf1,freed_87}. The low-mode factors are essential in nonzero-index sectors. For example, when $D:\mathbb C^{n+1}\to\mathbb C^n$ is surjective, the kernel vector $\kappa_D$ formed from its signed maximal minors satisfies the Cauchy--Binet identity
\begin{equation}
D\kappa_D=0,\qquad
\|\kappa_D\|^2=\det(DD^\dagger).
\label{eq:cauchy_binet_kernel}
\end{equation}
The induced squared norm of the dual kernel frame is $\|\kappa_D\|^{-2}$, which cancels the corresponding finite nonzero singular-value product in the Quillen metric. At general index this is expressed by the exterior-algebra determinant-line exact sequence, not by dividing by a vacuum determinant that vanishes.

For an invertible augmented chart and a finite-rank Yukawa perturbation, ordinary trace-class differentiation gives
\begin{equation}
\frac{\det_\zeta(\mathcal G_\Lambda^\dagger\mathcal G_\Lambda)}
{\det_\zeta(\mathcal G_0^\dagger\mathcal G_0)}
=\left|\det_F\!\left(1+\mathcal G_0^{-1}
(\mathcal G_\Lambda-\mathcal G_0)\right)\right|^2.
\label{eq:finite_quillen}
\end{equation}
This identity is independent of the original index; it compares \emph{augmented} operators. The preceding low-mode metric identity is still required to identify it with a physical determinant-line norm.

The distinction between the augmented determinant and the physical determinant-line factor can already be seen in the one-mode example $D=t$ with
\begin{equation}
 \mathcal G_t=\begin{pmatrix}t&1\\1&0\end{pmatrix},\qquad
 \det\mathcal G_t=-1,\qquad E_{-+}=-t,\qquad
 \det D=(\det\mathcal G_t)(\det E_{-+})=t.
 \label{eq:grushin_one_mode}
\end{equation}
The augmented determinant alone has \emph{no} $t$-dependence; the low-mode factor restores all of it. At nonzero index, $E_{-+}$ is rectangular and its exterior minors, rather than a scalar determinant, play the same role. A physical spectral-window change adds an invertible finite shell: its determinant appears inversely in the low-mode transition and directly in the high-mode positive spectral norm, so the Quillen metric is preserved. This finite-dimensional gluing is independent of the local ultraviolet normalization computed in Appendix~\ref{app:yukawa_uv}.

\paragraph{Derivation and normalization of the residual norm defect.}
Write $\mathcal G=\mathcal G_\phi$ and set $A=\mathcal G_0^\dagger\mathcal G$, $P=\mathcal G^\dagger\mathcal G$, and $X=\mathcal G^{-1}\delta\mathcal G$. On an invertible augmented chart, for a sufficiently small scaling of $\phi$ so that compatible spectral cuts exist, define the \emph{relative squared-norm} germ
\[
\delta_N^{\rm germ}(\phi)=2\operatorname{Re}\log\frac{\det_\zeta A}{\det_\zeta P_0}
-\log\frac{\det_\zeta P}{\det_\zeta P_0},\qquad
P_0=\mathcal G_0^\dagger\mathcal G_0.
\]
For an admissible elliptic weight $\Omega$, denote by $\operatorname{tr}^{\Omega}(X)=\operatorname*{FP}_{z=0}\operatorname{Tr}(X\Omega^{-z})$ the weighted trace. Fixed-order zeta differentiation gives
\[
\delta\log\det_\zeta A=\operatorname{tr}^{A}(X),\qquad
\delta\log\det_\zeta P=2\operatorname{Re}\operatorname{tr}^{P}(X).
\]
For the second identity use $\delta P=PX+X^\dagger P$ and cyclicity of the ordinary trace against $P^{-z}$ before taking the meromorphic finite part. Both $A$ and $P$ have order two and the same positive principal symbol. The weighted-trace defect identity for elliptic weights of common order $m$~\cite{paycha_scott} reads
\[
\operatorname{tr}^{\Omega_1}(X)-\operatorname{tr}^{\Omega_2}(X)
=-\frac1m\operatorname{res}\left[X(\log\Omega_1-\log\Omega_2)\right]
\quad\text{if }\operatorname{ord}\Omega_1=\operatorname{ord}\Omega_2=m.
\]
Combining $2\operatorname{Re}$ with $1/m=1/2$ fixes the coefficient in the germ variation
\begin{equation}
\boxed{\delta_\phi\delta_N^{\rm germ}=-\operatorname{Re}\operatorname{res}\!\left[
\mathcal G_\phi^{-1}\delta\mathcal G_\phi
\bigl(\log(\mathcal G_0^\dagger\mathcal G_\phi)
-\log(\mathcal G_\phi^\dagger\mathcal G_\phi)\bigr)\right].}
\label{eq:quillen_residue}
\end{equation}
Thus the displayed normalization does \emph{not} require an additional factor $1/2$ in the final variation; it has already canceled against the factor of two in the squared norm. This identity is derived for the \emph{augmented} zeta determinants near zero Yukawa strength, with constant operator orders, compatible Agmon cuts and smooth logarithms. The germ is exactly a local polynomial of degree at most four, as explained after Eq.~\eqref{eq:q_res_app}. The physical correction $\delta_N$ evaluates that polynomial at full Yukawa strength without requiring a complex spectral cut there. A separate family identification still relates augmented frames to the physical determinant line.

The residue is local and unchanged by finite-rank smoothing changes of Grushin charts. Radial integration with $\delta_N^{\rm germ}(0)=0$ gives the exact local polynomial defining $\delta_N$. Four-dimensional residue power counting bounds its Yukawa degree by four. Moreover $g_*=\operatorname{diag}(\mathbf1_3,-\mathbf1_2)$ leaves the links fixed and reverses $\Phi$, so the physical correction is even, $\delta_N=\delta_{N,2}+\delta_{N,4}$. Multiplication by $e^{-\delta_N/2}$ corrects only the complex-zeta to positive-zeta norm defect; the \emph{additional} cutoff-to-complex-zeta finite factor $e^{-\operatorname{Re}p_f}$ from Eq.~\eqref{eq:pf_comparison} is required for the filtered determinant. Neither real correction sets the physical theta angles or the complex Yukawa matrices.

\subsection{Quillen-isometric spectral-window gluing at arbitrary index}\label{sec:quillen_gluing}
We now check that the convenient augmented determinant and the physical determinant-line norm describe the \emph{same} fermion amplitude. The issue is not another ultraviolet divergence. A spectral cutoff divides modes into a high-mode determinant and a finite low-mode Grassmann integral. Changing where that division is made must change neither the total amplitude nor its norm. The canonical spectral-window construction is due to Bismut--Freed~\cite{bf1,freed_87}. We spell out its factors because a bare Grushin determinant gives the wrong answer at a zero-mode crossing.

\paragraph{Physical spectral windows.}
Measure the Dirac operator in units of a fixed reference scale $\mu$ and write $d=D/\mu$. Choose $0<\lambda_1<\lambda_2$ away from the positive spectrum of $d^\dagger d$ and $dd^\dagger$. Let $E_\pm(\lambda_i)$ denote their low spectral subspaces at edge $\lambda_i$, and set
\begin{equation}
 L_1=(\det E_+(\lambda_1))^*\otimes\det E_-(\lambda_1),\qquad
 D_{\mathrm{hi},1}=d|_{E_+(\lambda_1)^\perp}:E_+(\lambda_1)^\perp\longrightarrow E_-(\lambda_1)^\perp.
 \label{eq:window_line}
\end{equation}
The high-mode operator $D_{\mathrm{hi},1}$ is invertible. The determinant line of the original chiral operator is canonically represented by $L_1$, including when its Fredholm index is nonzero. The squared Quillen norm of a low-mode frame $s_1\in L_1$ is
\begin{equation}
 \|s_1\|_{Q,1}^2=\|s_1\|_{L^2,1}^2\,
 \det_\zeta(D_{\mathrm{hi},1}^\dagger D_{\mathrm{hi},1}).
 \label{eq:window_metric}
\end{equation}
The intermediate shell $F_\pm=E_\pm(\lambda_2)\ominus E_\pm(\lambda_1)$ has equal source and target dimensions. With $T_{12}=d|_{F_+}:F_+\to F_-$ invertible, the canonical determinant-line transition is $s_1\mapsto s_1\otimes\det T_{12}$. Spectral products split exactly across this finite shell:
\begin{equation}
 \det_\zeta(D_{\mathrm{hi},1}^\dagger D_{\mathrm{hi},1})=|\det T_{12}|^2
  \det_\zeta(D_{\mathrm{hi},2}^\dagger D_{\mathrm{hi},2}),\qquad
 \|s_1\otimes\det T_{12}\|_{Q,2}=\|s_1\|_{Q,1}.
 \label{eq:window_isometry}
\end{equation}
This is an \emph{isometry}: the determinant gained by the explicit low-mode shell is exactly the determinant lost by the zeta-regularized high modes. The transferred shell has equally many source and target modes, even when the full low-mode operator is rectangular and the index is nonzero. No ratio of vanishing vacuum determinants is needed.

\paragraph{Adapted Grushin normalization.}
In an orthogonally adapted window, write $d=M_1\oplus D_{\mathrm{hi},1}$ with $M_1:E_+(\lambda_1)\to E_-(\lambda_1)$ of size $c\times k$. Define
\begin{equation}
 \mathcal G_1^{\rm ad}=D_{\mathrm{hi},1}\oplus
 N(M_1),\qquad N(M)=\begin{pmatrix}M&I_c\\I_k&0\end{pmatrix}.
 \label{eq:adapted_grushin}
\end{equation}
The finite block $N(M)$ is invertible for \emph{every} $M$, even if its rank changes. Direct column permutation gives $\det N(M)=(-1)^{kc}$; consequently,
\begin{equation}
\boxed{\det_\zeta\bigl[(\mathcal G_1^{\rm ad})^\dagger
 \mathcal G_1^{\rm ad}\bigr]
 =\det_\zeta(D_{\mathrm{hi},1}^\dagger D_{\mathrm{hi},1}).}\label{eq:adapted_metric}
\end{equation}
Thus the augmented determinant accounts \emph{only} for the high-mode positive spectral factor. It would be incorrect to identify it with the physical fermion weight by itself. The finite low-mode exterior polynomial supplies the omitted amplitudes, including their behavior when an unsaturated vacuum determinant vanishes and when sources saturate the zero modes. Nonorthonormal auxiliary volume maps $B,C$ replace the unit finite-block determinant by $(-1)^{kc}\det B\det C$; their squared norm is $\det(B^\dagger B)\det(CC^\dagger)$ and is canceled by the induced change of low-mode frame. Schur-complement determinants likewise belong to the finite low-mode factor, not to an ultraviolet counterterm.

For example, in the scalar crossing $d=t$ the augmentation
$\left(\begin{smallmatrix}t&1\\1&0\end{smallmatrix}\right)$ has constant determinant $-1$ although the physical index-zero determinant section is $t$. In the rectangular index-one example $d=(a\ b):\mathbb C^2\to\mathbb C$, the low source amplitude is the cofactor vector $(-b,a)$, whose squared norm is $|a|^2+|b|^2$, while the adapted augmented determinant has constant modulus one. These elementary examples expose precisely the factors one loses by equating a raw augmented determinant with a physical determinant-line section.

\paragraph{Why the chart transitions agree (finite-dimensional naturality).}
There are two ways to compute a source-saturated fermion amplitude: integrate all modes at once, or first integrate a chosen set of invertible high modes and then integrate the remaining Grassmann variables. These operations commute because finite-dimensional Berezin integration (ordinary Grassmann-variable integration) is associative. To see the factors, write a finite truncation of the operator as
\begin{equation}
 D=\begin{pmatrix}A&B\\ C&D_{\mathrm{hi}}\end{pmatrix},\qquad D_{\mathrm{hi}}\text{ invertible},\qquad
 S=A-BD_{\mathrm{hi}}^{-1}C.
 \label{eq:naturality_schur}
\end{equation}
Unit-determinant block-row and block-column eliminations reduce $D$ to $S\oplus D_{\mathrm{hi}}$. With sources, the same eliminations give exactly the shifts in Eq.~\eqref{eq:rect_sources}. Thus its full source functional satisfies
\begin{equation}
 \mathcal Z_D[\bar\eta,\eta]=\det D_{\mathrm{hi}}\;\mathcal Z_S[\bar\eta_{\rm eff},\eta_{\rm eff}]
 \,e^{\bar\eta_hD_{\mathrm{hi}}^{-1}\eta_h}.
 \label{eq:naturality_sources}
\end{equation}
This formula is valid when $S$ is rectangular or has zero modes: $\mathcal Z_S$ is then a finite Grassmann polynomial, not the determinant of $S$. For the orthogonal spectral windows used in the Quillen construction, $D$ maps the intervening source shell of $D^\dagger D$ exactly onto its target shell of $DD^\dagger$; there are no off-diagonal couplings between that shell and its complement. Thus the shell determinant factors on its own. For a general low/high decomposition those off-diagonal blocks need not vanish, and one must use the Schur complement $S=A-BD_{\mathrm{hi}}^{-1}C$ and the source shifts above rather than treating the shell as a separate block. In either case, the full fermion amplitude is unchanged by the order of Gaussian integration. 

Reassigning an invertible shell $T_{12}$ from $D_{\mathrm{hi}}$ to $S$ changes the factors by
\begin{equation}
 \det D_{\mathrm{hi},1}=(\det T_{12})\det D_{\mathrm{hi},2},\qquad
 \mathcal Z_{S_2}=(\det T_{12})\mathcal Z_{S_1},
 \label{eq:naturality_shell}
\end{equation}
after integrating the shell variables (with their sources set to zero in the displayed scalar factor); with nonzero shell sources the same identity includes the standard Gaussian source shifts from Eq.~\eqref{eq:rect_sources}. Consequently both orders of integration give the same full source functional. Equation~\eqref{eq:naturality_shell} is the commuting transition diagram in elementary form; the map of low-mode determinant-line frames is precisely $s_1\mapsto s_1\otimes\det T_{12}$ in Eq.~\eqref{eq:window_isometry}. Its squared norm transforms by $|\det T_{12}|^2$, canceling the inverse change of the high-mode positive determinant.

An adapted Grushin augmentation implements the same block elimination with auxiliary variables. Its finite block has constant determinant $(-1)^{kc}$; changing auxiliary bases multiplies this determinant by their finite Jacobians and divides the low-mode exterior frame by the same Jacobians. Hence it induces exactly the spectral-window determinant-line transition, not an additional normalization choice. The argument remains valid at a zero crossing: enlarge the low window so $D_{\mathrm{hi}}$ remains invertible, and never divide by the vanishing vacuum determinant. At finite fermion cutoff the Yukawa perturbation is finite rank, so its relative determinant obeys these finite-dimensional Gaussian identities exactly. The only remaining infinite-dimensional normalization is the massless high-mode spectral determinant, whose splitting over an invertible shell is the standard Bismut--Freed isometry. The cutoff-removed high-frequency finite part is compatible with these transitions because changes between Grushin inverses are smoothing and leave the local ultraviolet coefficients unchanged. The standard spectral-window metric isometry is that of Bismut--Freed~\cite{bf1,freed_87}.

\paragraph{Compatibility of the finite parts and the real correction.}
Compare two gauge-only Grushin charts $a,b$ at a common background. First stabilize their finite auxiliary spaces and use the determinant-functor identifications; their infinite-dimensional operator symbols then agree, while the remaining changes are smoothing. Put $a_n=(-1)^{n+1}/n$ and $K_j=\mathcal G_{0,j}^{-1}\widetilde V$, $j=a,b$. For $1\le n\le4$, telescoping expresses $K_a^n-K_b^n$ as a sum with one smoothing factor. Its filtered counterpart converges in trace norm, uniformly on compact chart overlaps. Since the divergent coefficients are common, it follows that
\[
 \beta^{(a)}_{n,f}-\beta^{(b)}_{n,f}
       =a_n\operatorname{Tr}(K_a^n-K_b^n).
\]
The augmented weights $P_{0,a}$ and $P_{0,b}$ also differ by smoothing terms. Weighted traces agree with ordinary traces on smoothing insertions, and the weight-change residue vanishes for their smoothing logarithmic difference. All residue corrections to the Taylor coefficients in Eq.~\eqref{eq:zeta_taylor_v24} are identical. Therefore
\[
 \gamma_n^{(a)}-\gamma_n^{(b)}
       =a_n\operatorname{Tr}(K_a^n-K_b^n),\qquad
 p_f^{(a)}=p_f^{(b)}.
\]
Likewise Eq.~\eqref{eq:quillen_residue} gives the same $\delta_N$ on both charts, with the common normalization at zero Yukawa strength. This proves agreement of the finite correcting factor, not merely of its divergent part.

For clarity, where $1+K_a$ and $1+K_b$ are invertible, the relative trace-class determinant obeys
\[
 \det_F\!\left[(1+K_a)(1+K_b)^{-1}\right]
 =\frac{\det_5(1+K_a)}{\det_5(1+K_b)}
    \exp\!\left(\sum_{n=1}^4a_n\operatorname{Tr}(K_a^n-K_b^n)\right).
\]
This follows first for finite-rank approximants from~\eqref{eq:fredholm5split}, then by $\mathcal S_5$ and relative trace-norm convergence. It displays the finite trace shifts accompanying the nonmultiplicative $\det_5$ factors. The auxiliary determinant transition and the low-mode exterior factors supply the corresponding finite-dimensional determinant, by~\eqref{eq:naturality_sources}. At noninvertibility the proof uses the complete finite-cutoff source identities instead of the displayed ratios: subtract the same local terms on both charts and pass coefficientwise to the uniform limits, using the preceding minor construction. Thus no density assertion about invertible backgrounds is needed. Multiplication by the common factor $e^{-p_f-\delta_N/2}$ preserves their equality. Together with~\eqref{eq:window_isometry} and~\eqref{eq:adapted_metric}, this proves physical Quillen gluing at every index and justifies the normalization used by the real-analytic continuation. The theta angles and Yukawa matrices remain independent input data.

\noindent\emph{Physical interpretation.} Splitting the fermion spectrum into explicit low modes and an integrated-out high-mode determinant is a calculational choice. A nonzero mode contributes the same determinant factor on either side of the split. True zero modes must instead be saturated by external fermionic insertions; this does not create a new ultraviolet divergence.

\subsection{Spectral-pair check of the mass phase}
The physical statement and finite phase freedom are explained in Section~\ref{sec:yukawa_full}. Here is the algebraic check behind Eq.~\eqref{eq:massphase}. In a pure color background with constant Higgs field, a nonzero Dirac pair $\pm i\lambda$ yields
\begin{equation}
\det\begin{pmatrix}M&i\lambda\\i\lambda&M^\dagger\end{pmatrix}
=\det(\lambda^2\mathbf1+M^\dagger M)>0,
\label{eq:mass_pair}
\end{equation}
where $M=(v/\sqrt2)\operatorname{diag}(Y_u,Y_d)$. Nonzero pairs have no mass phase; it comes entirely from chiral zero modes, giving Eq.~\eqref{eq:massphase}. The common filter preserves this pairing. The theta-angle convention and possible higher-dimensional CP-odd finite terms are described in Section~\ref{sec:yukawa_full}; the Quillen norm concerns only the magnitude.

\section{A free massive-overlap comparison}\label{app:free_overlap}

This is a deliberately limited benchmark, \emph{not} a test of the interacting chiral measure. It asks whether, in a free vector-like theory, the overlap and continuum determinants give the same finite mass response after the common divergent terms have been subtracted. Craig's non-Abelian chiral measure current is not involved. Consider one free vector-like Dirac fermion on a four-torus of fixed side length $\ell$, with antiperiodic temporal boundary conditions to exclude zero momentum. Use Wilson parameter and overlap mass parameter $r_W=\rho_{\rm ov}=1$. For $t_\mu=ap_\mu$ define
\begin{equation}
 b(t)=\sum_{\mu=1}^4(1-\cos t_\mu)-1,\quad
 S^2(t)=\sum_\mu\sin^2t_\mu,\quad
 \omega(t)=\sqrt{b(t)^2+S(t)^2}.
\end{equation}
The massless and massive overlap operators are~\cite{neuberger,luscher_gw,giusti}
\begin{equation}
D_{\rm ov}(p)=\frac1a\left[1+\frac{b(t)+i\gamma_\mu\sin t_\mu}{\omega(t)}\right],\qquad
D_{m,\rm ov}=\left(1-\frac{am}{2}\right)D_{\rm ov}+m.
\label{eq:free_overlap_def}
\end{equation}
The two spin eigenvalues, each repeated twice, are
\[
 z_\pm=\frac1a\left(1+\frac b\omega
                   \pm i\frac{\sqrt{S^2}}\omega\right).
\]
Since $\omega^2=b^2+S^2$, they obey $|z_\pm|^2=(2/a)\operatorname{Re}z_\pm$, the Ginsparg--Wilson circle relation. Consequently
\[
 \frac{|(1-am/2)z_\pm+m|^2}{|z_\pm|^2}
 =1+m^2\left(\frac1{|z_\pm|^2}-\frac{a^2}{4}\right).
\]
The spin multiplicity gives the square in Eq.~\eqref{eq:overlap_free_mass}, with
\begin{equation}
\boxed{G_a(p)=\frac{a^2}{4}\frac{\omega(t)-b(t)}{\omega(t)+b(t)}.}
\label{eq:G_free_overlap}
\end{equation}
For each fixed physical nonzero momentum $G_a(p)\to p^{-2}$. Put $r(x)=\log(1+x)-x+x^2/2$. Since $r'(x)=x^2/(1+x)$, for $x\ge0$,
\begin{equation}
0\le r(x)\le\frac{x^3}{3}.
\end{equation}
On the compact Brillouin torus, $|t|^2(\omega-b)/[4(\omega+b)]$ is bounded and extends continuously at the physical zero with limit one. Thus for allowed nonzero physical momenta $0\le G_a(p)\le C|p|^{-2}$ independently of $a$. After extending the finite lattice sum by zero to the full continuum momentum lattice, dominated convergence with the summable bound $C'|p|^{-6}$ proves
\begin{equation}
\boxed{\mathcal R_{a,\mathrm{sub}}(s)\longrightarrow
2\sum_{p\in{\rm AP}(T^4)}r(s/p^2),\qquad s=m^2\ge0.}
\label{eq:overlap_free_convergence}
\end{equation}
Convergence is uniform for $s$ on compact sets. Its first term not fixed by quadratic and quartic subtraction is
\begin{equation}
[s^3]\,\mathcal R_{a,\mathrm{sub}}=\frac23\sum_{p\in{\rm BZ}_a}G_a(p)^3
\longrightarrow\frac23\sum_{p\in{\rm AP}(T^4)}p^{-6}.
\label{eq:overlap_sixth}
\end{equation}
For $\ell=2\pi$, $m=0.9$, antiperiodic time and periodic space, use $n_\mu=-N/2,\ldots,N/2-1$, $p_0=n_0+1/2$, $p_j=n_j$, and $a=2\pi/N$. The directly evaluated log-weight response is
\begin{center}\small
\begin{tabular}{@{}rrrr@{}}\toprule
$N=\ell/a$&$\mathcal R_{a,\mathrm{sub}}$&continuum sum on same momenta&difference\\\midrule
16&15.77195&16.99593&1.22398\\
32&16.55093&17.02134&0.47041\\
64&16.86546&17.02767&0.16221\\\bottomrule
\end{tabular}
\end{center}
The difference column is the continuum sum minus the overlap response on the same finite set of momenta. The continuum sum itself changes with $N$, so it is not the full limiting continuum value. An independent heat-kernel evaluation gives the full continuum result
\[
\mathcal R_{\mathrm{cont,sub}}(0.9^2)=17.029780286352048\ldots\,,
\]
using $\Theta(t)=\sum_{\mathbf n\in\mathbb Z^3,\,j\in\mathbb Z}\exp[-t(|\mathbf n|^2+(j+1/2)^2)]$ and
\[
\mathcal R_{\mathrm{cont,sub}}(s)=\int_0^\infty\!dt\,\Theta(t)
\left[s^2t-2s+\frac{2(1-e^{-st})}{t}
\right].
\]
For stable numerical evaluation at small $st$, the square bracket must be evaluated as $s^3t^2/3-s^4t^3/12+s^5t^4/60+O(s^6t^5)$ rather than by subtracting nearly equal numbers. This integral evaluates the infinite momentum sum; no rate of convergence in $a$ is asserted. Direct $4\times4$ spin-matrix checks verify~\eqref{eq:overlap_free_mass} to numerical precision. This benchmark is vector-like and free; its proof does not control interacting chiral projectors, nonzero external Higgs momenta, non-Abelian measure phases, or the continuum limit of Craig's full Standard Model regulator. The latter remains the open comparison problem of Section~\ref{sec:overlap}.

\subsection{A determinant-variation criterion for interacting matching}
The interacting goal is Eq.~\eqref{eq:overlap_target}. A useful sufficient (not established) analytic condition is the following; it is separated from the physics overview because it does not itself prove regulator equivalence.
\begin{proposition}[Reduction to a determinant-variation estimate]
Let $\mathcal S$ be a compact connected finite-dimensional smooth family of connections with sufficient fixed $C^r$ regularity for the determinant comparison, and suppose both weights are nonzero everywhere on $\mathcal S$, uniformly for sufficiently small $a$. Give tangent variations the $C^r$ norm. Suppose every $A\in\mathcal S$ can be joined to a reference $A_*$ by a path in $\mathcal S$ of length at most $L$, the reference ratio equals $1+o(1)$, and continuous logarithms for the normalized weight ratio can be chosen consistently on $\mathcal S$ (as is automatic for a nonzero ratio on a simply connected family). If local counterterms satisfy
\begin{equation}
\left|\delta\log W_{\rm ov}^{(a)}
 -\delta\log W_{\rm CW}^{(a)}-\delta C_a\right|
\le r_a\|\delta A\|_{C^r},\qquad r_a\to0,
\label{eq:overlap_variation}
\end{equation}
uniformly on $\mathcal S$, then Eq.~\eqref{eq:overlap_target} holds with logarithmic error bounded by $Lr_a+o(1)$.
\end{proposition}
\begin{proof}
Integrate the differential estimate along a path from $A_*$ to $A$, use its bounded length and the reference normalization, then exponentiate. The continuous logarithm makes the integrated comparison independent of the chosen path.\end{proof}
This elementary integration criterion is not itself a regulator-equivalence estimate. The restricted comparison that follows supplies such an estimate under explicit analytic, reconstruction and complex-phase hypotheses; equivalence for arbitrary admissible links remains open.

\subsection{Weak-field comparison of continuum Weyl and overlap determinants}
\label{app:weak_overlap}

\paragraph{Physical setting and matching criterion.}
The free calculation above does not test the interacting complex chiral measure. Here we compare two ultraviolet definitions of massless chiral fermions on smooth prescribed gauge backgrounds. The essential universality question is whether their \emph{nonlocal} contributions to the one-loop gauge-field effective action agree; ultraviolet differences may remain as local counterterms. We work at zero Yukawa coupling with no Higgs background, on a fixed flat four-torus $X$ with at least one antiperiodic fermion direction. The ambient anomaly-free representation is $(\mathbf{10}\oplus\overline{\mathbf5})^{\oplus3}$ of $SU(5)$, subsequently restricted to $\SM$. Only the trivial bundle and an invertible weak-field chart are considered.

The argument has two stages. In the first we compare the overlap weight directly to the BF/Quillen continuum Weyl weight on the same $A$. The Ginsparg--Wilson identity reduces its induced current to a one-loop trace plus a chiral measure term; lattice power counting then shows agreement of the nonlocal amplitudes after subtraction. In the second stage, we compare the continuum determinant at $A$ with that evaluated at the connection $\mathcal R_a(A)$ reconstructed from lattice links. Only this second stage requires differentiated bounds on reconstruction. We keep the complex phase throughout; opposite chiralities are not equivalent merely because gauge anomalies cancel.

For $s>5$ sufficiently large for the fixed finite set of estimates, define the Fourier algebra norm
\begin{equation}
\|A\|_{\mathcal A_s}=\sum_{q\in X^*}(1+|q|)^s\|\widehat A(q)\|,
\qquad\mathcal U_{s,\delta}=\{A:\|A\|_{\mathcal A_s}<\delta\}.
\label{eq:D2_fouriernorm}
\end{equation}
Choose $\delta$ small enough that the continuum Weyl operator has a uniform inverse and the overlap perturbation series converges on exact edge holonomies. Let $\mathcal S\subset\mathcal U_{s,\delta}$ be a compact, connected, finite-dimensional, star-shaped family of smooth fields with bounded path lengths and fixed smooth jets. The exact lattice holonomies are $U_a(A)$; the reconstructed connection, transported to the original bundle by Appendix~\ref{app:family_reconstruction}, is $\mathcal R_a(A)$. Thus $W_{\rm CW}^{(a)}[U_a(A)]=W_{\rm BF}[\mathcal R_a(A)]$ in the chosen reference and normalization convention.

\paragraph{Chirality dictionary and complementary overlap measure.}
The superscript $+$ on the continuum Weyl map $D_A^+:S^+\otimes E_R\to S^-\otimes E_R$ denotes a \emph{positive} $\gamma_5$ domain in this paper. Craig's original overlap weight $W_{{\rm ov},-}^{(a)}$, by contrast, uses $\widehat P_-=(1-\widehat\gamma_5)/2\to P_-$. His Euclidean gamma matrices are Hermitian, while Appendix~\ref{app:gamma_conventions} uses anti-Hermitian matrices: identifying $\gamma_\mu^{\rm Craig}=i\gamma_\mu^{\rm here}$ leaves $\gamma_5$ unchanged because $i^4=1$. In particular, this change of Hermiticity does not exchange $S^+$ and $S^-$. The unmodified Craig weight is therefore \emph{not} the complex determinant to be compared directly with $D_A^+$.

On the connected invertible, fixed-index chart define the positive-chirality comparator using the complementary overlap projector and a conjugate reference measure. With the vacuum value normalized consistently, this gives
\begin{equation}
\frac{W_{{\rm ov},+}^{(a)}[U_a(A)]}{W_{{\rm ov},+}^{(a)}[U_a(0)]}
=\overline{\left(\frac{W_{{\rm ov},-}^{(a)}[U_a(A)]}{W_{{\rm ov},-}^{(a)}[U_a(0)]}\right)}.
\label{eq:D2_chirality_conjugate}
\end{equation}
This is the \emph{complementary-chirality} overlap theory on the same oriented background, not a rephasing of Craig's original weight at that background. Craig's orientation-reversing lattice transformations exchange the complementary projectors and thereby relate this comparator to his original theory evaluated on the \emph{reflected} gauge configuration (with the correspondingly transformed spin structure and reference measure). A reflection is not an internal symmetry of a chiral gauge theory; in the comparison below we retain the same oriented $A$ and use the complementary projector directly. The comparator consequently has the same continuum chiral domain as $D_A^+$.

The dynamical identity is exact at finite spacing. Put $D=D_a$ and $\alpha_\pm=\operatorname{Tr}[P_\pm(\delta D)D^{-1}]$. The inverse Ginsparg--Wilson equation and $\gamma_5$-Hermiticity imply
\begin{equation}
\overline{\alpha_+}=\alpha_-+\frac a2\,\delta\operatorname{Tr}(\gamma_5D_a)
=\alpha_-,\qquad
 a\operatorname{Tr}(\gamma_5D_a)=\operatorname{Tr}\operatorname{sign}(H_a)=\text{constant on the chart}.
\label{eq:D2_chiral_trace}
\end{equation}
The subscript on $\alpha_\pm$ labels the \emph{constant antifermion trace projector}: Craig's $W_{{\rm ov},-}^{(a)}$ has dynamical term $\alpha_+$, whereas $W_{{\rm ov},+}^{(a)}$ has $\alpha_-$. The complementary projector reverses both the curvature and the anomaly pairing,
\[
\Omega_+=-\Omega_-,\qquad
\widehat P_+-P_-=-(\widehat P_--P_+).
\]
The conjugate reference measure therefore obeys the complementary integrability and gauge-invariance conditions, with the measure one-form changing sign. Thus \eqref{eq:D2_chirality_conjugate} holds for the \emph{complete} normalized weights, not just their magnitudes. In the chiral block decomposition, $\alpha_-$ tends to $\operatorname{Tr}_{S^+}[(D_A^+)^{-1}\delta D_A^+]$. The constant factor $i$ relating the two Euclidean Dirac conventions drops out of $D^{-1}\delta D$ and $D^\dagger D$. The complex overlap bounds transfer by conjugation, including parity-odd terms.

\begin{stageproposition}[Overlap-to-continuum matching on smooth weak backgrounds]\label{prop:D2_overlap}
On the open weak-field ball $\mathcal U_{s,\delta}$ specified above, fix the complementary overlap measure built from Craig's current and the positive-chirality BF/Quillen prescription, both normalized at the flat connection and using the chirality dictionary above. There exist local gauge-invariant counterterms $C_a^{\rm ov}[A]$ and an error bound $\eta(a)\to0$ such that
\begin{equation}
\sup_{A\in\mathcal S}\sup_{\|v\|_{C^r}\le1}
\left|\delta_v\log W_{{\rm ov},+}^{(a)}[U_a(A)]
-\delta_v\log W_{\rm BF}[A]-\delta_v C_a^{\rm ov}[A]\right|
\le\eta(a)+Ca^6\longrightarrow0.
\label{eq:D2_overlap_bound}
\end{equation}
The ratio of the vacuum-normalized complex weights, corrected by $e^{-C_a^{\rm ov}[A]+C_a^{\rm ov}[0]}$, converges uniformly to one on $\mathcal S$. No estimate for the reconstructed continuum connection is used.
\end{stageproposition}

\begin{stageproposition}[Matching after gauge-field reconstruction]\label{prop:D2_weak_match}
In addition to the weak-field setup of Proposition~\ref{prop:D2_overlap}, assume for the \emph{actual} normalized reconstruction, including its determinant-one correction and family variations, the uniform mixed-derivative bounds
\begin{equation}
\|\mathcal R_a(A)-A\|_{C^m}
+\sup_{\|v\|_{C^r}\le1}\|D_A\mathcal R_a(v)-v\|_{C^m}
\le C_m a^{2-m},\qquad m=0,1,2,
\label{eq:D2_recon}
\end{equation}
in compatible fixed gauges and across chart transitions. Then the full chirality-matched overlap and continuum-Weyl weights satisfy, for a local gauge-invariant $C_a[A]$,
\begin{equation}
\sup_{A\in\mathcal S}\sup_{\|v\|_{C^r}\le1}
\big|\delta_v\log W_{{\rm ov},+}^{(a)}[U_a(A)]
-\delta_v\log W_{\rm CW}^{(a)}[U_a(A)]-\delta_v C_a[A]\big|
\le\eta(a)+Ca^2(1+|\log a|)+Ca^6\longrightarrow0.
\label{eq:D2_variation_bound}
\end{equation}
After vacuum normalization, the corrected ratio converges uniformly to one on $\mathcal S$ by the determinant-variation criterion above.
\end{stageproposition}
\stepcounter{proposition}

The remainder of this subsection proves Proposition~\ref{prop:D2_overlap}: finite-grid power counting gives the loop convergence, Lemma~\ref{lem:D2_current_uniform} handles the exact measure current and its background derivative, and the BF finite-part comparison identifies the local difference. Proposition~\ref{prop:D2_weak_match} has one additional, genuinely geometric hypothesis, \eqref{eq:D2_recon}, which is stronger than the undifferentiated projector bound in Appendix~\ref{app:family_reconstruction}. \emph{Mixed jets} denotes spatial derivatives of the reconstructed connection together with derivatives under changes of the background gauge field.

\paragraph{Overlap fermion loops and contact vertices.}
We first calculate using Craig's original negative-chirality convention; complex conjugation as in \eqref{eq:D2_chirality_conjugate} gives the positive-chirality result required by Proposition~\ref{prop:D2_overlap}. The variation of a logarithmic fermion weight is the induced gauge current (up to sign). Expanding it in $A$ generates the standard one-loop diagrams, with one probe field $v$. On a lattice there are additional vertices involving several gluons at the same point: they arise because exact link holonomies and the overlap sign function are nonlinear in $A$. These contact diagrams can multiply ultraviolet-divergent sums and cannot be omitted merely because their individual vertices carry powers of $a$.

On the invertible chart Craig's original negative-chirality logarithmic response, in the notation of Ref.~\cite{craig}, is
\begin{equation}
\delta_v\log W_{{\rm ov},-}^{(a)}
=\operatorname{Tr}\big[(\delta_v D_a)\widehat P_{-,a}D_a^{-1}P_+\big]-i\mathfrak L_a(v).
\label{eq:D2_overlap_response}
\end{equation}
The exact Ginsparg--Wilson intertwining $D_a\widehat P_{-,a}=P_+D_a$ eliminates the moving chiral projector from the dynamical trace:
\begin{equation}
\delta_v\log W_{{\rm ov},-}^{(a)}=\alpha_a(v)-i\mathfrak L_a(v),
\qquad \alpha_a(v)=\operatorname{Tr}\big[P_+(\delta_vD_a)D_a^{-1}\big].
\label{eq:D2_alpha}
\end{equation}
This is finite-dimensional trace cyclicity at fixed cutoff; it does \emph{not} set the measure current to zero or assert that $\alpha_a$ is closed.

Write $D_a=D_{0,a}+\sum_{m\ge1}V_{m,a}(A)$, with $V_m$ the homogeneous order-$m$ vertex including its Taylor factorial, and $G_a=D_{0,a}^{-1}$. With the tangent field counted as one insertion, the complete degree-$n$ response is a sum over ordered compositions $n=m_1+\cdots+m_{N_{\rm p}}$ (with $N_{\rm p}$ free propagators):
\begin{equation}
\alpha_{n,a}(A;v)=
\sum_{N_{\rm p}=1}^n(-1)^{N_{\rm p}-1}
\sum_{\substack{m_1+\cdots+m_{N_{\rm p}}=n\\m_j\ge1}}
\operatorname{Tr}\big[P_+(\delta_vV_{m_1})G_aV_{m_2}G_a\cdots V_{m_{N_{\rm p}}}G_a\big].
\label{eq:D2_words}
\end{equation}
The varied vertex $\delta V_{m_1}$ contains $m_1-1$ background fields and one probe field. There are $2^{n-1}$ ordered contributions: these are the one-loop diagrams with all possible multi-gluon contact vertices, expressed compactly without drawing every graph. For example $\alpha_{2,a}=\operatorname{Tr}P_+[(\delta V_2)G_a-(\delta V_1)G_aV_1G_a]$. The three- and four-point terms contain four and eight ordered words respectively, automatically retaining all Wilson-holonomy contact vertices and their noncommutative ordering.

For clarity about the path-order convention, Craig's $U(x,\mu)$ transports a field from $x+a\widehat\mu$ back to $x$. With $\nabla=d+A$, its exact straight-edge transport is
\begin{equation}
U(x,\mu)=1+\sum_{m\ge1}a^m
\int_{0<t_1<\cdots<t_m<1}
A_\mu(x+at_1\widehat\mu)\cdots A_\mu(x+at_m\widehat\mu)
\,dt_1\cdots dt_m,
\label{eq:D2_order}
\end{equation}
with the earlier path parameter on the \emph{left}. Indeed, this solves $\partial_tU(t)=U(t)aA_\mu(x+at\widehat\mu)$ under the transport convention used here. Writing $H_a=\gamma_5(D_{W,a}-1)$, $R_0(z)=(z-H_0)^{-1}$, the $m$th overlap vertex is obtained by summing the ordered Riesz words
\begin{equation}
 aV_{m,a}=\sum_{\substack{i_1+\cdots+i_k=m\\i_j\ge1}}
\frac{2\gamma_5}{2\pi i}\oint_{\Gamma_+}
R_0H_{i_1}R_0\cdots H_{i_k}R_0\,dz.
\label{eq:D2_Riesz}
\end{equation}
The positive and negative Wilson spectra have a uniform gap, so the contour is fixed on a sufficiently small field ball. As a normalization check, for $f(z)=(e^{iz}-1)/(iz)$ the first Hermitian Wilson vertex is
\begin{equation}
H_{1,a}(p+q,p;B)=\frac a2\sum_\mu f(aq_\mu)\gamma_5
\left[(\gamma_\mu-1)e^{iap_\mu}+(\gamma_\mu+1)e^{-ia(p+q)_\mu}\right]R(B_\mu(q)).
\label{eq:D2_H1}
\end{equation}
Its positive-negative spectral blocks yield $V_{1,a}\to\gamma^\mu R(B_\mu)$ at fixed physical momenta, with the conventional continuum sign and no extra factor of two. The ordered holonomy integrals, periodic free Wilson symbols, and resolvent differences imply, for each fixed number $k$ of allowed physical Fourier-label differences,
\begin{equation}
\|\Delta_{h_1}\cdots\Delta_{h_k}V_{m,a}\|
\le C_k B^m a^{m-1}\prod_{j=1}^k(a|h_j|).
\label{eq:D2_vertex_bound}
\end{equation}
The Kronecker momentum-conservation factor is removed before differencing the translation-reduced vertex; free propagators are always routed modulo the lattice reciprocal grid, whereas the exact holonomy factors retain physical external momenta. This avoids an arbitrary off-grid extension of a finite-torus delta function.

\paragraph{Ultraviolet subtraction and continuum loop amplitudes.}
The power counting is familiar from four-dimensional continuum QFT: an $n$-gauge-field fermion loop has superficial divergence degree $4-n$. Thus only the orders through $n=4$ require local polynomial subtractions. Above fourth order the loop sums converge, after lattice contact terms have been controlled. The subtraction at orders $n\le4$ is a finite-grid version of BPHZ renormalization: it removes the required low-degree polynomial in external momenta using values on the allowed Fourier grid, without an arbitrary continuation of lattice momentum-conservation factors.

For the mass-$-1$ free overlap operator,
\begin{equation}
D_{0,a}(p)=a^{-1}\left[1+
\frac{\sum_\mu(1-\cos ap_\mu)-1+i\gamma_\mu\sin(ap_\mu)}
{\sqrt{(\sum_\mu(1-\cos ap_\mu)-1)^2+\sum_\mu\sin^2(ap_\mu)}}\right].
\label{eq:D2_free}
\end{equation}
It has a single zero on the Brillouin torus, avoided by the antiperiodic fermion grid. Hence $g_a(p):=\|G_a(p)\|\le C\langle p\rangle^{-1}$. Discrete H\"older and invariance of the fermion grid under bosonic shifts give, uniformly over arbitrary allowed shifts $Q_j$,
\begin{equation}
\frac1{\operatorname{vol}X}\sum_p\prod_{j=1}^{N_{\rm p}} g_a(p+Q_j)
\le \frac1{\operatorname{vol}X}\sum_pg_a(p)^{N_{\rm p}}
\le C_{N_{\rm p}}\begin{cases}a^{N_{\rm p}-4},&N_{\rm p}<4,\\1+|\log a|,&N_{\rm p}=4,\\1,&N_{\rm p}>4.
\end{cases}
\label{eq:D2_holder}
\end{equation}
By the vertex bound \eqref{eq:D2_vertex_bound}, a word with $N_{\rm p}$ free propagators and external field degree $n$ carries $a^{n-N_{\rm p}}$. The $N_{\rm p}\le4$, $n\ge5$ contact terms sum to $O(a(1+|\log a|))$ in first variation on a small Fourier-algebra ball. For $N_{\rm p}\ge5$, the propagator product is summable. All terms with $n>N_{\rm p}$ vanish, and the surviving $n=N_{\rm p}$ terms converge, by dominated convergence of the high-momentum tail and a geometric sum over $N_{\rm p}$, to the first variation of the continuum $\log\det_5(1+D_0^{-1}V(A))$ in the common chiral convention.

For the remaining complete $n=2,3,4$ complex amplitudes, put $d_n=4-n$ and fix the physical Fourier-grid steps $h_j=2\pi/L_j$. The degree-$d$ multivariate Newton interpolant of a translation-reduced amplitude $I$ is
\begin{equation}
\mathfrak T_d^h I(\mathbf q)=\sum_{|\beta|\le d}
\left(\prod_j\binom{q_j/h_j}{\beta_j}\right)
(\Delta_h^\beta I)(0).
\label{eq:D2_newton}
\end{equation}
It reproduces every polynomial of degree at most $d$, including at negative Fourier labels. Subtract $\mathfrak T_{d_n}^h$ from each ordered word in \eqref{eq:D2_words}, retaining the identical momentum routing in every difference. The inverse-propagator identity
\begin{equation}
\Delta_hG_a(p)=-G_a(p+h)[D_{0,a}(p+h)-D_{0,a}(p)]G_a(p)
\label{eq:D2_diffG}
\end{equation}
shows that a $k$-fold difference produces words with $\ell+1$ propagators and $\ell$ differences of $D_{0,a}$, $1\le\ell\le k$. One must \emph{not} assume that every difference produces an additional propagator; rather, the bound for a constituent term is $Ca^{k-\ell}\prod_{t=0}^{\ell}g_a(p+Q_t)$.

For $Q=1+\max_j|q_j|\lesssim a^{-1}$, split the loop sum at $|p|=cQ$. In the high region, the exact Newton remainder and the difference bounds give for a word with $N_{\rm p}$ propagators and contact-vertex excess $\kappa=n-N_{\rm p}$
\begin{equation}
|(1-\mathfrak T_{d_n}^h)I_{n,N_{\rm p},a}(p;\mathbf q)|
\le C_na^\kappa Q^{5-n}\langle p\rangle^{-5+\kappa}.
\label{eq:D2_high}
\end{equation}
In the low region, use the restricted shifted H\"older inequality. If all shifts satisfy $|Q_j|\le C Q$ and $Q\lesssim a^{-1}$, then the translate of $|p|\le cQ$ lies inside a ball of radius $C'Q$ on the periodic momentum grid, so
\begin{equation}
\frac{1}{\operatorname{vol}X}\sum_{|p|\le cQ}
\prod_{j=1}^{t}g_a(p+Q_j)
\le\prod_{j=1}^{t}\left(\frac{1}{\operatorname{vol}X}\sum_{|p|\le cQ}g_a(p+Q_j)^t\right)^{1/t}
\le C_t\Phi_t(Q),
\label{eq:D2_restricted_holder}
\end{equation}
where the bound for the centered, translated grid ball is
\begin{equation}
\Phi_t(Q)=\begin{cases}\langle Q\rangle^{4-t},&t<4,\\1+\log(2+Q),&t=4,\\1,&t>4.\end{cases}
\label{eq:D2_phi}
\end{equation}
The estimate applies with $t=N_{\rm p}+\ell$ in the finite-difference words. Its third case is essential near coincident poles: an infrared-convergent product of more than four propagators does not gain spurious negative powers of $Q$. The Newton coefficients cost at most the corresponding powers of $Q$; the remaining powers of $a$ satisfy $aQ\lesssim1$. For $Q\gtrsim a^{-1}$, use the global bound \eqref{eq:D2_holder} for the unsubtracted words and for each finite difference. These cases produce the uniform complete-kernel majorant
\begin{equation}
|\mathcal B_{n,a}^{h\text{-}\rm sub}(\mathbf q)|
\le C_n(1+Q)^{4-n}[1+\log(2+Q)],\qquad n=2,3,4.
\label{eq:D2_majorant}
\end{equation}
At each fixed external tuple the terms with $\kappa>0$ vanish (at worst $O(a|\log a|)$), while the $N_{\rm p}=n$ term tends to the continuum Weyl loop after the same subtraction. This is a \emph{complex} limit; no real-part projection has been made. The continuum difference between the finite-grid Newton and derivative Taylor subtractions is itself a finite local polynomial: both reproduce Taylor degrees through $d_n$, and their difference starts at order $j=d_n+1$, whose loop summand is $O(\langle p\rangle^{-n-j})=O(\langle p\rangle^{-5})$.

A simple Fourier reduction turns fixed-tuple convergence into the required uniform first-variation convergence. Indeed, divide the difference of two kernels obeying \eqref{eq:D2_majorant} by $\prod_j\langle q_j\rangle^s$. On any finite set of Fourier tuples it tends to zero. Outside that set at least one $|q_j|=Q$, and the quotient is uniformly bounded by $C\langle Q\rangle^{4-n-s}(1+\log(2+Q))$, which tends to zero. Summation against the weighted absolutely convergent Fourier coefficients, with one tangent field, gives convergence in the dual $\mathcal A_s$ norm. The $n=1$ traceless-color contribution vanishes in the ambient $SU(5)$ theory. The required full-current estimates are addressed next.

\paragraph{The chiral measure contribution and its smooth-field derivatives.}
The chiral measure is essential to Craig's exact finite-cutoff gauge invariance, although its contribution to the induced current becomes small when the lattice links sample a smooth continuum background. We use Craig's exact-current locality and vacuum-cancellation estimates, including their fixed background-parameter derivatives (Ref.~\cite[App.~F.1--F.3]{craig}), rather than treating the current as a new independent assumption.

\begin{lemma}[Uniform exact-current bound on the smooth weak-field ball]\label{lem:D2_current_uniform}
For $A\in\mathcal U_{s,\delta}$, uniformly bounded smooth curvature gives $\|\mathbf1-U_p(A)\|\le C_\delta a^2$ for every sampled plaquette $p$. Hence the exact holonomies lie in Craig's admissible region for all sufficiently small $a$, uniformly on this ball. Using the fixed-order locality estimates of Ref.~\cite{craig}, let $\mathfrak L_a$ denote the measure one-form of the complete exact current, or its negative for the complementary-chirality measure. Then for all sufficiently small $a$ and smooth tangent fields $v,w$,
\begin{equation}
|\mathfrak L_a(A;v)|+\sup_{\|w\|_{\mathcal A_s}\le1}
|D_A\mathfrak L_a(A;v)(w)|
\le Ca^6\|v\|_{\mathcal A_s},
\label{eq:D2_current}
\end{equation}
with $C$ independent of $a$ and $A$ in the fixed ball. The derivative includes the actual moving midpoint frame.
\end{lemma}
\begin{proof}
In a radial gauge about an output-link midpoint, an input-link logarithm a lattice distance $R$ away has the expansion
\begin{equation}
\log U_f=a^2L_{f,2}(F_A)+a^3L_{f,3}(\nabla F_A)
+O\!\left(a^4(1+R)^N\mathcal M_3(A)\right),
\label{eq:D2_midpoint}
\end{equation}
where $\mathcal M_3(A)$ denotes a bounded polynomial in the $C^3$ norms of $A$ and its curvature on the fixed ball. Its first background derivative satisfies the corresponding bound with an additional factor $\|w\|_{\mathcal A_s}$. These are the short radial-gauge expansions of Craig's Appendix F.3; $s>5$ controls the requisite jets. Exponential locality permits summation of the polynomial factors $(1+R)^N$, while contributions outside a fixed physical radial patch are $O(e^{-c/a})$, also after one background derivative.

The complete exact current and its first three original-link derivatives vanish at the vacuum, so its Taylor series in input-link logarithms begins at degree four. The nominal $a^8$ term is a quartic expression in the curvature at the output midpoint. Under full inversion the curvature factors are even whereas the axial current changes sign; the coefficient therefore vanishes as a polynomial, and so does its polarization under $w$. Link-log degrees five and higher begin at $a^{10}$. Hence the midpoint current and its first fixed-frame background derivative are $O(a^9)$ uniformly over the smooth-field ball. Changing to the complementary measure reverses the sign of the current and does not affect these bounds.

For a moving radial frame $g(A)$, differentiation adds $[\Omega_w,j_a]$ with $\Omega_w=(D_Ag(w))g^{-1}$. The radial parallel-transport equation bounds $\Omega_w$ and its first parameter variation uniformly by $C\|w\|_{\mathcal A_s}$ on the fixed physical patch; consequently this commutator is also $O(a^9)$. A field-dependent change of frame also induces a gauge variation of the link tangent, not merely a simultaneous conjugation of the current and the tangent. Craig's covariant summation-by-parts identity (his Appendix F.3, Eq.~(F.10)) controls that term by placing a backward link difference on the current. The difference gains a factor $a$, so no spatial regularity of the gauge generator is needed. Equivalently, the pairing may be evaluated in smooth midpoint radial frames, with the induced gauge variation accounted for. Finally, each smooth link variation is $O(a\|v\|_{\mathcal A_s})$, and the number of links is $O(a^{-4})$ at fixed volume. Differentiating the link pairing gives a current-derivative term and a variation-of-the-link-insertion term, both $O(a^{-4}a^9a)=O(a^6)$ with the appropriate tangent norms. This proves \eqref{eq:D2_current}.
\end{proof}

Define the scalar radial primitive $\Gamma_a[A]=\int_0^1\alpha_a(tA;A)\,dt$. Since $d\alpha_a=i\,d\mathfrak L_a$, the field-space radial homotopy identity gives
\begin{equation}
\Gamma_a[A]-\log\frac{W_{{\rm ov},-}^{(a)}[U_a(A)]}{W_{{\rm ov},-}^{(a)}[U_a(0)]}
=i\int_0^1\mathfrak L_a(tA;A)\,dt,
\label{eq:D2_radial}
\end{equation}
and its first variation has an $O(a^6)$ correction under \eqref{eq:D2_current}. Thus the ultraviolet subtraction can be performed on a scalar functional without silently imposing cyclicity on a regularized trace or dropping the complex phase.

\paragraph{Why different continuum determinant prescriptions differ locally.}
The finite-grid calculation identifies the nonlocal continuum limit of the overlap loops. We must compare it with the \emph{specific} continuum BF/Quillen prescription in this paper. This is analogous to comparing two ultraviolet renormalization schemes for the same fermion effective action: once the nonlocal momentum dependence has been matched, their remaining difference can be a polynomial in the external momenta, hence a local functional of $A$. For a chiral determinant this must hold for the \emph{complex} amplitude. A background-dependent spectral weight and noncyclicity of regularized traces add local residue terms; they do not justify discarding the phase. In the formulas below the chiral trace is consistently taken in one convention at a time, and the comparison to the $D_A^+$ BF weight uses the conjugated positive-chirality version of the original overlap kernels.

Put $K(A)=D_0^{-1}V(A)\in\mathfrak S_5$, $P_0=D_0^\dagger D_0$ and $P_A=D_A^\dagger D_A$. The continuum fixed-positive-weight perturbative functional is
\begin{equation}
\mathcal F[A]=\log\det_5(1+K(A))+
\sum_{n=1}^4\frac{(-1)^{n+1}}n\operatorname{tr}^{P_0}\big(K(A)^n\big).
\label{eq:D2_F}
\end{equation}
For $1\le n\le4$, let $I_n(p;\mathbf q)$ be the continuum $n$-gauge-field loop integrand with the same chiral order and momentum routing. The derivative-BPHZ kernel is the absolutely convergent sum
\begin{equation}
\mathcal B_{n,0}^{\rm sub}(\mathbf q)
=\sum_{p\in X^*_{\rm AP}}(1-t_{\mathbf q}^{4-n})I_n(p;\mathbf q).
\label{eq:D2_contBPHZ}
\end{equation}
The antiperiodic grid avoids the free pole, and after subtraction the ultraviolet summand decays as $O(|p|^{-5})$. For the fixed-weight trace, split $I_n=(1-t^{4-n})I_n+t^{4-n}I_n$ before taking the meromorphic finite part. The first term is absolutely summable, so the finite part leaves it unchanged; the second produces precisely a local momentum polynomial:
\begin{equation}
\operatorname{tr}^{P_0}(K^n)-\mathcal B_{n,0}^{\rm sub}
=P_n[A],\qquad
\deg_{\partial}P_n\le4-n,
\label{eq:D2_localpoly}
\end{equation}
with coefficients
\begin{equation}
c_{n,\beta}=\frac{1}{\beta!}\operatorname{FP}_{z=0}
\sum_p\left.\partial_{\mathbf q}^{\beta}I_n(p;\mathbf q)\right|_{\mathbf q=0}
(|p|^2)^{-z},\qquad |\beta|\le4-n.
\label{eq:D2_polycoeff}
\end{equation}
The equality is understood at the level of the corresponding integrated $n$-linear functional, with the same coefficients and orderings on both sides. It is complex-linear, including parity-odd terms. The finite grid-versus-derivative subtraction difference is another local polynomial of the same degree. These polynomials need not be separately gauge invariant.

The full BF/Quillen response uses the \emph{moving} positive weight $P_A$, not $P_0$:
\begin{equation}
\omega_{\rm BF}(A;v)
=\operatorname{FP}_{z=0}\operatorname{Tr}
[D_A^{-1}(\delta_vD_A)P_A^{-z}].
\label{eq:D2_BF}
\end{equation}
The standard positive-weight trace-change and cyclic-trace defect identities express the difference between \eqref{eq:D2_BF} and $d\mathcal F$ as finite-jet local residue one-forms. On the anomaly-free invertible chart, with the same determinant-line trivialization, their sum has a local radial primitive $L_{\rm BF}$, giving
\begin{equation}
\log\frac{W_{\rm BF}[A]}{W_{\rm BF}[0]}
=\mathcal F[A]-\mathcal F[0]+L_{\rm BF}[A].
\label{eq:D2_BFtrans}
\end{equation}
Hence the \emph{correct} overlap convergence target after local subtraction is $d(\mathcal F-P_{\mathcal F})$, not $d\mathcal F$. The difference from the continuum BF functional is local, including its complex parity-odd part, after the chirality dictionary above has been applied. The normalized comparison on this contractible invertible chart uses the same chiral block and vacuum reference as fixed above; \eqref{eq:D2_BFtrans} makes no statement about relative phases between distinct topological sectors.

\paragraph{Reconstructed gauge fields, weak stability, and Ward identities.}
The second comparison is geometric rather than a new ultraviolet loop calculation. The reconstructed field $\mathcal R_a(A)$ can be close to $A$ in amplitude while retaining variations on the lattice scale. The available target bounds are $O(a^2)$ for the connection, $O(a)$ for its first derivative, and only $O(1)$ for its second derivative. The renormalized determinant is nevertheless stable because it probes integrated momentum-weighted combinations of the error, not the maximum of each derivative. The quadratic vacuum polarization is the borderline case, producing the logarithmic factor $|\log a|$; the higher-point amplitudes satisfy stronger multilinear estimates. All this depends on the differentiated reconstruction bound \eqref{eq:D2_recon}, which is not claimed to follow from the published undifferentiated estimate.

The affine weights in Craig's reconstruction satisfy $\sum_i\rho_i(y)=1$ and $\sum_i\rho_i(y)(i-y)^\mu=0$. In radial gauge, the apparent order-$a^2$ motion of the range projection cancels:
\begin{equation}
(K_2J)_i=\tfrac12\phi_iF_{\mu\nu}(i-y)^\mu
\sum_j\rho_j r_j^\nu=0.
\label{eq:D2_moment}
\end{equation}
More explicitly the final factor is the affine moment $\sum_j\rho_j(j-y)^\nu$, which vanishes; the first displacement is $r_i=i-y$. The cancellation motivates the bound $\Pi_a-Q_a=O(a^3)$ in rescaled fixed-gauge mixed jets, stronger than the undifferentiated $O(a^2)$ estimate of Appendix~\ref{app:family_reconstruction}. Passing to the normalized frame
\begin{equation}
\mathcal T_a=\Pi_aJ_a(J_a^\dagger\Pi_aJ_a)^{-1/2}
\label{eq:D2_frame}
\end{equation}
requires differentiating the remainder, fixing common gauge charts and applying the determinant-one correction. The latter is central and cancels in the traceless spatial and family connection, but this fact does not by itself prove \eqref{eq:D2_recon}. That mixed-jet audit is retained as a condition of Proposition~\ref{prop:D2_weak_match}.

Assuming \eqref{eq:D2_recon}, set $e_a=\mathcal R_a(A)-A$ and $w_a=D_A\mathcal R_a(v)-v$. For either error field $z$, interpolation yields $\|z\|_{H^t}\le Ca^{2-t}$ for $0\le t\le2$, with a factor $\|v\|_{C^r}$ for $w_a$. In the quadratic nonlocal response the high--high contribution is controlled by the logarithmic norm. If
$X_z=\sum_q\langle q\rangle^2|\widehat z(q)|^2\le Ca^2$ and $Y_z=\sum_q\langle q\rangle^4|\widehat z(q)|^2\le C$, Jensen's inequality gives
\begin{equation}
\sum_q\langle q\rangle^2\log(2+|q|)|\widehat z(q)|^2
\le CX_z\log(2+C\sqrt{Y_z/X_z})
\le Ca^2(1+|\log a|).
\label{eq:D2_loginterp}
\end{equation}
For cubic and quartic kernels, the continuum bounds $|\mathcal B_{n,0}^{\rm sub}|\le C\langle Q\rangle^{4-n}(1+\log(2+Q))$ suffice: terms with one rough leg are controlled by $\|\widehat z\|_{\ell^\infty}=O(a^2)$ and smooth Fourier moments; terms with two rough legs use weighted Fourier Cauchy--Schwarz; those with three or four use discrete Young and the $H^{t}$ interpolation bounds. For example, the three-rough-leg triangle has size $O(a^{3-3\varepsilon})=O(a^2)$ for $\varepsilon<1/3$. The $\det_5$ term is $C^0$-Lipschitz in the relevant $\mathfrak S_5$ norm and contributes $O(a^2)$; the finite-jet local residue terms obey the corresponding integrated-by-parts bounds. This proves, under \eqref{eq:D2_recon},
\begin{equation}
\left|d\log W_{\rm BF}[\mathcal R_a(A)](D_A\mathcal R_a(v))
-d\log W_{\rm BF}[A](v)\right|
\le Ca^2(1+|\log a|)\|v\|_{C^r}.
\label{eq:D2_stability}
\end{equation}
In particular, $\|\mathcal R_a(A)-A\|_{C^2}=O(1)$ is sufficient here; no $C^2$ convergence is claimed.

Finally, both \emph{complete} weights are exactly gauge invariant: for the overlap regulator this is Craig's Theorem 4.1, and for the CW regulator it follows from the gauge-covariant reconstruction and the determinant-line construction above. Exact edge holonomies obey $U_a(A^g)=U_a(A)^{g|_{\rm vertices}}$. The difference of the logarithms therefore has zero gauge-orbit variation on the \emph{open} ball $\mathcal U_{s,\delta}$. Once the nonlocal difference is first-variation small there, its residual local polynomial belongs, modulo total derivatives, to the fixed finite-dimensional space of local dimension-at-most-four polynomials respecting the actual spacetime symmetries of the finite-volume subtraction scheme. The BF weight-change and cyclicity residues lie in this same local power-counting class. Finitely many gauge-orbit tests separate the quotient by gauge-invariant polynomials. Hence the non-invariant coefficients tend to zero in the first-variation norm and an invariant representative $C_a[A]$ may be chosen. One must not demand that the individual BF and BPHZ local polynomials be separately gauge invariant. For example, a finite torus or spin structure that breaks axis-permutation symmetry can permit a finite kinetic counterterm $\int_X c_a^{\mu\nu\rho\sigma}\operatorname{tr}(F_{\mu\nu}F_{\rho\sigma})$ with the remaining index symmetries. This does not make the local ultraviolet logarithmic divergence anisotropic: its coefficient inherits the continuum rotational symmetry of the short-distance symbol. If the subtraction prescription preserves full $SU(5)$ and hypercubic covariance, the ordinary dimension-four kinetic term reduces instead to one $SU(5)$ trace coefficient, whose restriction fixes the Standard Model coefficients in GUT normalization. We retain the more general finite-volume counterterm space in the proposition.

\begin{proof}[Proof of Propositions~\ref{prop:D2_overlap} and \ref{prop:D2_weak_match}]
\emph{Overlap to BF at the same field.} Calculate first in Craig's original negative-chirality convention. The exact reduction \eqref{eq:D2_alpha}, the finite-grid bound \eqref{eq:D2_majorant}, higher-order tail estimate and Fourier reduction identify the nonlocal overlap response after local subtraction. Lemma~\ref{lem:D2_current_uniform} and the radial primitive contribute $O(a^6)$. The finite-part/BPHZ identity \eqref{eq:D2_localpoly} and BF transgression \eqref{eq:D2_BFtrans} show that the continuum scheme difference is local, including its complex phase. Conjugating the original overlap amplitude gives the positive-chirality comparator \eqref{eq:D2_chirality_conjugate}; the chiral-block identity and common vacuum reference align it with the normalized BF convention. The exact Ward identity supplies a gauge-invariant representative of the local difference. Uniform convergence follows by the Fourier reduction and geometric resummation above. This proves \eqref{eq:D2_overlap_bound} without using $\mathcal R_a$.

\emph{Reconstruction comparison.} The actual CW variation is evaluated at $\mathcal R_a(A)$ in the direction $D_A\mathcal R_a(v)$. Under \eqref{eq:D2_recon}, the weak-stability estimate \eqref{eq:D2_stability} replaces it by the BF variation at $A$ at cost $Ca^2(1+|\log a|)$. Combining with the previous bound establishes \eqref{eq:D2_variation_bound}. Integrate on the bounded-length paths in the compact finite-dimensional family with a common vacuum normalization, using the determinant-variation criterion of Appendix~\ref{app:free_overlap}, to obtain the uniform normalized weight ratio.
\end{proof}

\paragraph{Scope and limitations.}
Proposition~\ref{prop:D2_overlap} matches the complete chirality-aligned overlap weight to the BF/Quillen weight at the same smooth gauge field, up to local counterterms and a vanishing error. Proposition~\ref{prop:D2_weak_match} requires, in addition, the differentiated reconstruction estimate \eqref{eq:D2_recon}. Both results concern a small invertible chart of the trivial bundle on a fixed torus at zero Yukawa coupling. They do not establish matching on arbitrary admissible lattice links, near determinant zeros, across different topological sectors, or in the dynamical bosonic continuum limit. The complementary overlap weight is defined at the same oriented background; it is not identified with Craig's unmodified weight there.

\section{Relation to the original eigenmode and eta-invariant prescription}\label{app:hsu_comparison}

\subsection{Four-dimensional eigenmode prescription}
The earlier paper also notes a conjugation/reflection pairing of gauge backgrounds under which the \emph{integrated} partition function is real in the stated symmetry setting. That is not a claim that each background fermion determinant is real or that a complex-weight numerical sign problem disappears. In the present Yukawa-coupled theory, such a pairing must be checked for the full action and its chosen CP-violating parameters; it cannot be inferred solely from gauge-anomaly cancellation. The proposal of Ref.~\cite[Secs.~3--5]{hsu} begins with a smooth background gauge field, then expands the fermions in eigenfunctions of the corresponding \emph{vector-like} four-dimensional Dirac operator. The eigenfunctions in the background are compared with a fiducial basis. Denote the latter by $\phi_n^{L,R}$, and the eigenmodes in the background $A$ by $\phi_m^{L,R;A}$. The two overlap matrices have \emph{opposite} bra--ket ordering, as specified in Eqs.~(3.10)--(3.11) of Ref.~\cite{hsu}:
\[
 (C_L^{(N)})_{mn}=\langle\phi_m^{L;A}|\phi_n^L\rangle,
 \qquad
 (C_R^{(N)})_{mn}=\langle\phi_n^R|\phi_m^{R;A}\rangle.
\]
Under a background-eigenmode change $\phi_m^A\mapsto e^{i\alpha_m}\phi_m^A$, the $m$th row of $C_L^{(N)}$ acquires $e^{-i\alpha_m}$ while the corresponding row of $C_R^{(N)}$ acquires $e^{+i\alpha_m}$. Their determinant product is therefore unchanged. The same cancellation holds for unitary rotations within degenerate eigenspaces, when the left- and right-handed bases are transported consistently. In particular, one must \emph{not} take an additional complex conjugate of $\det C_R$ with these definitions. On a patch where the overlap determinants do not vanish, the phase is
\begin{equation}
 e^{i\eta_N(A)}=
 \frac{\det C_L^{(N)}(A)\,\det C_R^{(N)}(A)}
 {\bigl|\det C_L^{(N)}(A)\,\det C_R^{(N)}(A)\bigr|}.
 \label{eq:hsu_overlap}
\end{equation}
This agrees with Eqs.~(3.29)--(3.30) of Ref.~\cite{hsu}; Eq.~(3.31) then writes the determinant as $e^{i\eta_N}$ times the eigenvalue product. The associated chiral determinant is written schematically as $e^{i\eta_N(A)}\prod_{n\leq N}\lambda_n(A)$, where the $\lambda_n$ are eigenvalues of the vector-like operator tracked from the reference background. Sign changes in their product account for eigenvalue crossings. At finite $N$, the magnitude of the overlap matrices is not generally unity and is not gauge invariant; Ref.~\cite{hsu} proposes replacing the magnitude by that inferred from the corresponding vector-like theory, while taking $N\to\infty$ at fixed smooth gauge background before removing the gauge-field regulator. The phase $\eta_N$ in Eq.~\eqref{eq:hsu_overlap} is an \emph{eigenmode-overlap phase}, not itself the spectral eta invariant of a five-dimensional operator.

\subsection{Five-dimensional spectral asymmetry}\label{app:gamma_conventions}
To compare the eigenmode prescription with continuum spectral regularization, consider a smooth path of connections $A(t)$ joining a reference background $A_*$ to $A$ on a fixed gauge bundle. Our Euclidean convention uses anti-Hermitian gamma matrices and a Hermitian Dirac operator:
\[
 \{\gamma^\mu,\gamma^\nu\}=-2\delta^{\mu\nu},\qquad
 (\gamma^\mu)^\dagger=-\gamma^\mu,\qquad
 \gamma_5=\gamma^1\gamma^2\gamma^3\gamma^4
          =\gamma_5^\dagger,\quad \gamma_5^2=1.
\]
Then $\not\!D_4=\gamma^\mu(\partial_\mu+A_\mu)$ is Hermitian when $A_\mu$ is anti-Hermitian. With $\partial_t^\dagger=-\partial_t$, the operator $i\gamma_5\partial_t$ is also Hermitian, and their sum is formally self-adjoint. A self-adjoint domain or suitable boundary condition is additionally required on a noncompact path or a cylinder. In the more familiar Euclidean convention where $\gamma^\mu D_\mu$ is anti-Hermitian, an appropriate factor of $i$ must be included in its definition before using spectral signs. Ref.~\cite{hsu} introduces the five-dimensional operator
\begin{equation}
 \mathcal H_5=i\gamma_5\partial_t+\not\!D_4[A(t)],
 \qquad
 \eta_{\mathcal H_5}(s)=
 \sum_{\lambda\neq0}\operatorname{sign}(\lambda)|\lambda|^{-s}.
 \label{eq:hsu_5d}
\end{equation}
Here $\lambda$ denotes an eigenvalue of a self-adjoint realization of the five-dimensional operator, and $\eta_{\mathcal H_5}(0)$ is obtained by analytic continuation. To define the invariant precisely on a cylinder or an interpolating spacetime with ends, one must also specify the boundary conditions or relative spectral prescription.\footnote{One standard choice is an Atiyah--Patodi--Singer (APS) spectral boundary condition, which selects boundary modes using the sign of the boundary Dirac operator. The boundary and reference conventions enter the eta phase for an open interpolation; the present finite-volume existence proof does not require an open-path eta prescription.} If $h_5=\dim\ker\mathcal H_5$, the customary \emph{reduced eta invariant} is $\bar\eta_5=(\eta_{\mathcal H_5}(0)+h_5)/2$. For nonvanishing endpoint determinants, the relative chiral determinant phase in the anomaly-free setting is represented in the normalization of Eq.~(5.2) in Ref.~\cite{hsu} as
\begin{equation}
 \arg W(A)-\arg W(A_*)
 =\pi\bigl[\eta_{\mathcal H_5}(0)+h_5\bigr]
 =2\pi\bar\eta_5\pmod{2\pi},
 \label{eq:hsu_eta_relation}
\end{equation}
with the boundary prescription, reference phase, orientation, and local renormalization convention fixed as part of that representation. The $h_5$ term records the contribution of five-dimensional zero modes and the corresponding spectral-flow sign changes. Where the four-dimensional determinant vanishes, its argument is undefined, although the determinant section itself remains well defined. Equation~\eqref{eq:hsu_eta_relation} records the spectral-phase convention discussed in Ref.~\cite{hsu}; no identity with our particular $W_0$ in Eq.~\eqref{eq:weight} is claimed without a comparison of that earlier operator, boundary prescription and finite normalization. This differs from the fully specified relative Dai--Freed APS cylinder comparison established in Proposition~\ref{prop:aps_bf_exact}. A path from a single trivial reference field also cannot be assumed to cover gauge bundles in distinct topological sectors; such sectors require independent reference data or a suitable relative construction.

\subsection{Eta invariants and gauge invariance}
For a \emph{closed} loop of four-dimensional backgrounds, the five-dimensional mapping torus computes the holonomy of the Bismut--Freed connection. The Bismut--Freed holonomy theorem expresses this in terms of an adiabatic limit of $\exp(2\pi i\bar\eta_5)$: the circle parameterizing the loop is stretched, so the background varies slowly along it, as in the adiabatic approximation underlying Berry phases. Dai and Freed subsequently gave a determinant-line-valued eta-invariant formulation and a gluing law, providing another derivation of the holonomy formula~\cite{bf2}; see in particular~\cite[Cor.~5.11]{daifreed}.\footnote{Depending on whether the determinant line or its inverse is used, the formula for the holonomy has a conjugated phase. Both conventions give trivial holonomy in the anomaly-free theory. The spin structure on the mapping circle and its index factor must be kept together; for the full representation on $T^4$, $\operatorname{ind}P=-12k$ is even, so that parity factor is one.} Because the anomaly polynomial for $\mathbf{10}\oplus\overline{\mathbf5}$ vanishes and $\Omega_5^{\mathrm{Spin}}(BSU(5))=0$, the exponentiated holonomy is unity for every such loop. Accordingly, phase transport from a chosen reference background is path independent and compatible with gauge transformations, which supplies the $\tau_U$ used in Eqs.~\eqref{eq:weight}--\eqref{eq:gaugeinvariant}.

The distinction between an open path and a closed loop is essential. The five-dimensional eta invariant associated with an \emph{open} interpolation can encode a nontrivial, background-dependent chiral determinant phase. The vanishing global anomaly requires only that the phase obtained after completing a \emph{closed} loop be trivial. It therefore does not remove the physically relevant complex phase of $W_0$.

\subsection{Relation between the two prescriptions}
Ref.~\cite[Sec.~5, following Eq.~(5.3)]{hsu} explicitly identifies the relation between the limiting eigenmode phase $\eta_N$ and the five-dimensional spectral asymmetry as a nontrivial problem and does not prove the required $N\to\infty$ convergence. The present construction instead defines $W_0$ directly through the Quillen determinant section and a flat, gauge-equivariant trivialization. It does not depend on proving convergence of the truncated overlap phases. To establish \emph{equality} with the eigenmode prescription of Ref.~\cite{hsu}, however, one must compare both the phase and the magnitude under the same reference normalization, including possible finite local counterterms and behavior at spectral crossings. Such a comparison remains open. This does not conflict with Proposition~\ref{prop:aps_bf_exact}, which identifies the \emph{Dai--Freed APS cylinder transport} with the particular BF/Quillen phase convention; it does not establish convergence of the finite eigenmode cutoff of Ref.~\cite{hsu}. In particular, $|\det D|e^{i\pi\eta}$ is not, by itself, a complete definition unless the five-dimensional operator, its boundary and eta conventions, and its normalization relative to the determinant section have been specified.

\section{Global gauge conventions and reconstruction details}\label{app:geometry}
\subsection{The global form of the gauge group}
The group $S(U(3)\times U(2))$ consists of block-diagonal $5\times5$ unitary matrices, with a three-dimensional color block and a two-dimensional weak block, constrained to have total determinant one. The Lie algebra is the Standard Model gauge algebra. The generator of the central $\mathbb Z_6$ acts trivially on every Standard Model fermion and on the Higgs. Dividing by it therefore leaves the local interactions unchanged, but restricts the allowed gauge representations. On a representation $(r_3,r_2)_n$, the central generator acts by $\exp[2\pi i(t/3+I+n/6)]$. Its action is trivial precisely when $t/3+I+n/6\in\mathbb Z$, where $t\in\{0,1,2\}$ is color triality and $I$ is weak isospin. For example, $Q=(\mathbf3,\mathbf2)_1$ gives $1/3+1/2+1/6=1$. The quotient also permits additional magnetic fluxes, correlated among the color, weak and hypercharge factors. This global form is selected because it embeds into $SU(5)$ \emph{without identifying distinct group elements} (a faithful embedding). The other global forms are not considered here; their lattice realizations are discussed in Ref.~\cite[Sec.~10]{craig}, including the faithful representations needed to formulate their admissibility conditions.

\subsection{Bundle pairings for the Yukawa action}
The three generation matrices are arbitrary: no symmetry condition is required. The hypercharge sums of the terms in Eq.~\eqref{eq:yukawa} vanish: $1+3-4=0$, $1-3+2=0$, and $-3-3+6=0$. On the reconstructed rank-three and rank-two bundles, with $\det E_3\otimes\det E_2\cong\mathbf1$, the species are represented by
\[
Q=E_3\otimes E_2,\quad u^c=E_3^*\otimes\det E_3,\quad d^c=E_3^*,\quad L=E_2\otimes(\det E_2)^{-1},\quad e^c=\det E_2,\quad\Phi\in E_2.
\]
The antisymmetric weak contraction in the up term supplies $\det E_2$, canceled by $\det E_3$; the other terms use canonical dual pairings. This verifies gauge invariance globally, including on nontrivial bundles.

\subsection{Smooth family reconstruction and its topology}\label{app:family_reconstruction}
The existing geometric comparison bounds are sufficient for the fixed-spacing construction. For the additional smooth weak-field comparison with overlap fermions, Appendix~\ref{app:weak_overlap} states a stronger \emph{differentiated} small-$a$ estimate for the reconstructed connection as a conditional input. Its affine-moment cancellation and central determinant-volume correction are described there; the full mixed spatial/family-jet bound and compatibility across charts are not claimed to follow from the undifferentiated bounds below.

\paragraph{The total family bundle and its source.}
Lemma~\ref{lem:smooth_reconstruction} collects consequences of Craig's explicit connection construction in Appendix~A; it does not invoke the fermion-measure reconstruction theorem. The relevant PDF locations are:
\begin{center}\small
\begin{tabularx}{\textwidth}{@{}>{\raggedright\arraybackslash}p{.34\textwidth}X@{}}
\toprule Required input & Location in Ref.~\cite{craig} \\\midrule
Bundle, frames and determinant volume & A.2, printed pp.~65--66, Eqs.~(A.12)--(A.14).\\
Family connection and derivative bounds & A.2, printed pp.~66--67, Eqs.~(A.15)--(A.18) and accompanying text.\\
Gauge covariance & A.3, printed pp.~67--68, Eqs.~(A.21)--(A.22).\\
Torus descent and $3+2$ splitting & A.3, printed p.~68, Eq.~(A.23) and the final paragraph.\\
\bottomrule
\end{tabularx}
\end{center}
In a fixed cover chart, Craig's smooth weighted transporter matrix $K(y,U)$ has five eigenvalues separated from the remaining spectrum by the admissibility gap. Its contour projection $\Pi(y,U)$ is therefore smooth jointly in $y$ and $U$; degeneracies within the selected five-dimensional cluster do not affect it. The ranges form a smooth rank-five bundle over the product of the link domain and the spacetime chart. Craig's orthonormal frames have exact smooth overlap cocycles. Their deck identifications commute with the contour construction, so these chart bundles descend to the single $\mathscr E\to\calU_\epsilon\times T^4$ in Lemma~\ref{lem:smooth_reconstruction}. Local finite vertex spaces suffice on the cover; no globally chosen frame on $T^4$ is assumed.

The family connection is, explicitly,
\[
 \nabla_{\rm tot}=\Pi d_{\rm tot}
       -\frac15(\sigma,d_{\rm tot}\sigma)\,1_{\mathscr E},
 \qquad d_{\rm tot}=d_y+d_U.
\]
It preserves the Hermitian metric, determinant volume and block splitting. Smooth matrix inverses and phase normalizations stay uniformly away from their singularities on the geometric admissible domain. Their derivatives of any fixed order are bounded on compact regulated families. These are the total-space data needed for the smooth Dirac-family construction. Lattice gauge transformations give smooth intertwining maps with the group composition law. A parameter-dependent change of gauge has the additional family-direction term displayed in Craig's Eq.~(A.22); ordinary spatial gauge covariance remains exact.

\paragraph{Compatibility of family connections.}
For this anomaly-free representation, choosing different smooth connection components in parameter directions does not change the BF connection under the canonical identification of the operator determinant lines. Locally interpolate those components with a parameter $s\in[0,1]$, keeping the spacetime connection fixed and adding no $ds$ component. The vertical Dirac operator and its spectral windows are independent of $s$, so BF's construction has no connection term in the $s$ direction~\cite[Def.~1.17]{bf1}. The degree-six index polynomial vanishes for this enlarged family too, so its BF curvature is zero~\cite{bf2}. The mixed $(s,U)$ curvature then says that the base connection is independent of $s$. Naturality gives the same identification on overlapping patches. Thus the ambient BF pullback used in Appendix~\ref{app:phase_references} agrees with the connection obtained from the smooth family above.

\paragraph{Relation to earlier interpolation constructions.} The use of continuum interpolated connections predates the present regulator~\cite{flume_wyler,goeckeler,kronfeld,hsinterp}. Craig's reconstruction is used here for its specific gauge-covariant smooth-family and bundle properties; employing it does not make gauge interpolation itself a new contribution. 
Earlier geometric studies of lattice gauge topology and continuum reconstruction include Refs.~\cite{luscher_topology,phillips_stone,goeckeler}; subsequently Hern\'andez and Sundrum~\cite{hsinterp,hsinterpreview} developed a cellwise interpolation of non-Abelian lattice gauge fields designed to preserve gauge covariance, lattice symmetries, continuity of the reconstructed potentials across cell boundaries in directions transverse to those boundaries, and a gauge-invariant notion of locality. Their interpolated parallel transports reproduce the original link matrices exactly. By contrast, the reconstruction used here yields a globally smooth bundle connection that varies smoothly with admissible link configurations, at the cost of the $O(\epsilon)$ edge discrepancy stated above. For the finite-volume existence theorem, the decisive properties of the reconstruction in Ref.~\cite{craig} are smoothness in spacetime, smooth dependence on the lattice links, exact gauge covariance, preservation of the $\SM=S(U(3)\times U(2))$ structure group, and consistent treatment of nontrivial gauge-bundle topology. These are precisely the properties needed to apply smooth-family determinant-line theory; exact reproduction of individual lattice-link holonomies is not required. We do not assert that the two interpolations yield identical fermion weights.

Smooth dependence on $U$ implies that the topological charges and fluxes of the reconstructed bundle remain constant along any continuous family of admissible links. These quantized quantities are conventionally expressed as integrals of \emph{characteristic classes}. To relate them to familiar gauge-field expressions, take $F_i$ to be the Hermitian curvature two-form of the rank-$i$ block, $F_i=i(dA_i+A_i\wedge A_i)$ when $A_i$ is anti-Hermitian. Then
\[
\int_X c_1(E_2)^2=\frac{1}{4\pi^2}\int_X
\operatorname{tr}F_2\wedge\operatorname{tr}F_2,
\qquad
\int_X c_2(E_i)=\frac{1}{8\pi^2}\int_X
\left[(\operatorname{tr}F_i)\wedge(\operatorname{tr}F_i)
-\operatorname{tr}(F_i\wedge F_i)\right].
\]
The second formula includes a trace-squared term because the blocks are $U(3)$ and $U(2)$ bundles, not necessarily separate $SU(3)$ and $SU(2)$ bundles. For a traceless curvature it reduces, up to the stated sign convention, to the usual quantized $F\widetilde F$ integral. The converse to topological constancy need not hold: two configurations with identical topological charges may still be impossible to deform into one another without leaving the admissible region.

\paragraph{Which continuum sectors occur?}
The preceding construction permits nontrivial bundle topology, but this alone does not establish which sectors occur at a prescribed lattice spacing. The following proposition proves the stronger asymptotic realization statement. It uses the spectral-projection definition of Craig's bundle itself, rather than a comparison of transition cocycles. Write $E_A=E_{3,A}\oplus E_{2,A}$ for the rank-five Hermitian bundle associated to a prescribed smooth $\SM$-bundle, with its canonical determinant trivialization.

\begin{proposition}[Asymptotic realization of every smooth Standard Model bundle]\label{prop:topological_realization}
Let $P\to T^4$ be any smooth principal $\SM=S(U(3)\times U(2))$ bundle and let $A$ be any smooth $\SM$-connection on $P$. Consider a sequence of periodic hypercubic lattices of spacing $a$ tending to zero with fixed physical periods. Choose $\SM$ frames at the physical lattice vertices and let $U^{(a)}$ be the matrices of $A$-parallel transport along their oriented edges. For sufficiently small $a$:
\begin{enumerate}[label=(\roman*),leftmargin=1.6em]
\item $U^{(a)}$ belongs to the strictly positive region of the bosonic admissibility weight $\chi_\epsilon$;
\item Craig's rank-five reconstructed bundle $E[U^{(a)}]$, with its rank-three/rank-two splitting and determinant trivialization, is isomorphic to $E_A$ as an $S(U(3)\times U(2))$ bundle.
\end{enumerate}
Consequently the reconstructed bundle realizes the entire isomorphism class of $P$, including its integral Chern classes and nonliftable $\mathbb Z_6$ flux sectors. The assertion is for every fixed smooth bundle at sufficiently fine resolution, not for every sector on any one finite lattice.
\end{proposition}
\begin{proof}
Frames can always be chosen independently at the finitely many physical torus vertices. They define the periodic link matrices even when $P$ is topologically nontrivial; no continuous global frame is needed. For the comparison with Craig's projector, use dimensionless coordinates $y=x/a$ on the cover $\mathbb R^4$ and lift these discrete frames compatibly with deck transformations. The vertex indices $i\in\mathbb Z^4$ then correspond to $x_i=ai$.

\emph{1. Admissible sampling.}
Fix the connection $A$. On a small coordinate ball centered at any $x\in T^4$, choose radial gauge. In this gauge its potential satisfies $|A_\mu(x+z)|\le C_A|z|$ when $|z|$ is sufficiently small; the constant can be taken uniform over $T^4$ because $F_A$ is smooth and the torus is compact. The holonomy of each lattice plaquette is therefore $\mathbf1_5+O_A(a^2)$, uniformly over all plaquettes. Thus $\|\mathbf1_5-U_p^{(a)}\|\le C_Aa^2$. For small enough $a$, each plaquette lies in the neighborhood of the identity where the fixed cutoff weight $w_\epsilon$ is positive, proving (i).

\emph{2. A projector for the original bundle.}
Let $\varphi_i(y)$ denote Craig's smooth copying factors (his Eqs.~(A.1)--(A.4)), so that $\sum_i\varphi_i(y)^2=1$ and at most $4^4$ factors are nonzero at each $y$. Let $H=\ell^2(\mathbb Z^4)\otimes\mathbb C^5$ be the vertex-fiber space in the chosen frames. For $v\in(E_A)_{ay}$ define
\begin{equation}
 (J_a(y)v)_i=\varphi_i(y)\,T_A(x_i\leftarrow ay)v,
 \qquad Q_a(y)=J_a(y)J_a(y)^\dagger,
 \label{eq:topo_JQ}
\end{equation}
where $T_A(x_i\leftarrow ay)$ is parallel transport along the straight segment from $ay$ to $x_i$, expressed in the selected vertex frame. We use the pullback of $P$ to $\mathbb R^4$ in this formula. Because transport is unitary and $\sum_i\varphi_i^2=1$, $J_a^\dagger J_a=\mathbf1_5$. Thus $Q_a$ is an exact rank-five orthogonal projector and $J_a$ identifies the original bundle fiberwise with $\operatorname{im}Q_a$.

In vertex components,
\begin{equation}
 (Q_a)_{ij}=\varphi_i\varphi_j\,
 T_A(x_i\leftarrow ay)T_A(ay\leftarrow x_j).
 \label{eq:topo_Qblocks}
\end{equation}
The maps $J_a$ and $Q_a$ depend smoothly on $y$. Under a torus-period translation $\Delta$ in lattice units, the lifted vertex frames identify corresponding fibers, and $J_a(y+\Delta)=D_\Delta J_a(y)$ under the deck identification of its domain, with $D_\Delta$ the shift of vertex components. Hence $Q_a(y+\Delta)=D_\Delta Q_a(y)D_\Delta^\dagger$: $\operatorname{im}Q_a$ descends to $T^4$ with precisely the original bundle topology. This is true even when the bundle does not lift to separate $SU(3)$, $SU(2)$, and $U(1)$ bundles.

\emph{3. Comparing with Craig's spectral projector.}
Craig's Eqs.~(A.8), (A.10), and (A.12) define short-transport matrices $u_{ij}$, the Hermitian matrix $K_a$ and the spectral projector $\Pi_a$ by
\begin{equation}
 (K_a)_{ij}=\varphi_i(y)u_{ij}\varphi_j(y),\qquad
 \Pi_a=\frac1{2\pi i}\oint_{|z-1|=1/2}(z-K_a)^{-1}\,dz.
 \label{eq:topo_KPi}
\end{equation}
For any indices with $\varphi_i\varphi_j\ne0$, $x_i$, $x_j$, $ay$, and every coordinate-ordered short path entering $u_{ij}$ lie in a ball of radius $C a$, with a constant determined solely by Craig's fixed support. In radial gauge centered at $ay$, all straight radial transports in~\eqref{eq:topo_Qblocks} are identity matrices. Each coordinate-ordered lattice path has at most twelve links, and each of its transports differs from identity by $O_A(a^2)$ in that same gauge. Their average, its unitary polar factor, and the determinant-one normalization used in $u_{ij}$ therefore also differ from identity by $O_A(a^2)$; the polar and logarithm operations are smooth near identity. Returning to arbitrary vertex frames gives the invariant estimate
\begin{equation}
 \left\|u_{ij}-T_A(x_i\leftarrow ay)T_A(ay\leftarrow x_j)\right\|
 \le C_A a^2\qquad(\varphi_i\varphi_j\ne0).
 \label{eq:topo_trans_est}
\end{equation}
This estimate is uniform in $y$, including near torus seams. For $v=(v_j)\in H$, the Cauchy--Schwarz inequality and $\sum_i\varphi_i^2=1$ now give
\begin{equation}
 \|(K_a-Q_a)v\|^2\le C_A^2a^4
 \left(\sum_j|\varphi_j|\,\|v_j\|\right)^2
 \le C_A^2a^4\|v\|^2.
 \label{eq:topo_KQ_bound}
\end{equation}
Since $Q_a$ is an exact rank-five projector, the resolvent identity on $|z-1|=1/2$ implies, for small enough $a$,
\begin{equation}
 \boxed{\ \sup_y\|\Pi_a(y)-Q_a(y)\|\le C'_Aa^2<1.\ }
 \label{eq:topo_proj_bound}
\end{equation}
The selected projector has rank five. The constants need not be uniform over different continuum backgrounds or as the curvature of the chosen background increases.

\emph{4. Global bundle isomorphism and the $\SM$ structure.}
The norm bound in~\eqref{eq:topo_proj_bound} makes $J_a^\dagger\Pi_aJ_a$ strictly positive. Define its polar-corrected projection
\begin{equation}
 \mathcal T_a=\Pi_aJ_a(J_a^\dagger\Pi_aJ_a)^{-1/2}:
 E_A\longrightarrow E[U^{(a)}].
 \label{eq:topo_isomorphism}
\end{equation}
Then $\mathcal T_a^\dagger\mathcal T_a=1$; because both bundles have rank five, this is a smooth unitary fiberwise isomorphism. Every term is compatible with deck transformations, so the isomorphism descends globally to $T^4$ without assuming a global gauge potential.

The sampled connection, both short transports, $K_a$, $Q_a$, $\Pi_a$, and $\mathcal T_a$ preserve the $3+2$ decomposition. Thus the rank-three and rank-two blocks are separately identified. To preserve Craig's determinant trivialization as well, let $\sigma_A$ denote the canonical unit determinant volume on $E_A$ and $\sigma_C$ Craig's globally defined volume from his Eq.~(A.14). In an $\SM$ frame at an anchor $i$, his local copying map $W_i$ differs by $O_A(a^2)$ from $J_a$ applied to the frame parallel transported radially from $x_i$; this follows from~\eqref{eq:topo_trans_est}. His projected normalized local frame $w_i$ therefore has the same estimate. The determinants of all the radially transported original frames equal $\sigma_A$. Craig's weighted volume average consequently satisfies
\begin{equation}
 \|\sigma_C-(\wedge^5J_a)\sigma_A\|\le C_Aa^2,
 \qquad
 \|\sigma_C-(\wedge^5\mathcal T_a)\sigma_A\|\le C'_Aa^2.
 \label{eq:topo_volume_bound}
\end{equation}
The ratio defined by $(\wedge^5\mathcal T_a)\sigma_A=f_a\sigma_C$ is therefore a smooth $U(1)$-valued function uniformly near $1$. Its principal logarithm is global. Multiplying $\mathcal T_a$ by $\exp[-(\operatorname{Log} f_a)/5]$ corrects its determinant to one without changing unitarity or either block. The corrected map identifies the associated $S(U(3)\times U(2))$ principal bundles, proving (ii). In particular, all integral characteristic classes and quotient-group lifting obstructions agree.
\end{proof}

\paragraph{Scope.}
Every fixed topological sector therefore has an open, positive-Haar-measure set of admissible link configurations on sufficiently fine lattices: admissibility and the bundle class are stable under sufficiently small changes of the links. This is a geometric statement, independent of whether the zero-source massless determinant vanishes because of an index obstruction. It does not establish that all sectors occur at a given finite lattice spacing, nor does it assert convergence of the full quantum theory as $a\to0$.

\paragraph{Reconstruction of the Higgs section.}
The Higgs field is a \emph{section} of the rank-two weak bundle $E_2$, meaning a doublet whose components in overlapping spacetime patches are related by the corresponding gauge transformation. Ref.~\cite{craig} does not introduce scalar fields. Using its smooth local frames $C_i$ and partition-of-unity weights $\rho_i$, we define a Higgs section by transporting site doublets into the continuum weak bundle, as in Eq.~\eqref{eq:higgsinterp}.
Here $j_2:\mathbb C^2\hookrightarrow\mathbb C^5$ is the fixed inclusion of the weak block. By the block-preservation statement of Lemma~1, each $C_i(y)j_2\phi_i$ belongs to $(E_2)_y$, not merely to the ambient $SU(5)$ bundle. The functions $\rho_i$ are fixed, independent of the link variables, form a \emph{partition of unity} ($\sum_i\rho_i=1$), and have support inside the corresponding local frame domain of $C_i$. This is a \emph{uniform family choice}: Craig constructs the weights on a fixed cover in Eqs.~(A.1)--(A.4), while the corresponding frames vary smoothly with $U$, with bounded derivatives of each fixed order on the compact admissible support (Appendix~A.2 of Ref.~\cite{craig}). Thus the same cover and weights work throughout that family; no background-dependent choice of partition is needed. Gauge covariance follows from the transformation law $C_i\mapsto G C_i g_i^{-1}$ together with $\phi_i\mapsto g_i\phi_i$: each summand transforms by the induced gauge action $G$ on $(E_2)_y$. Here the transformations and maps are restricted to the weak block as above. This verifies the smooth, gauge-covariant Higgs input used in Section~\ref{sec:yukawa_full}. Because it is linear in the site values, it also obeys $\Phi[U,t\phi]=t\Phi[U,\phi]$.

\subsection{Topological phases and their index restrictions}\label{app:phase_references}

\paragraph{Ambient transport and fixed reference data.}
Flatness on admissible-link space alone permits one independent phase on each connected component. The convention in Section~\ref{sec:hsu_eta} fixes these phases by comparison in a larger space of connections. For the induced $SU(5)$ bundle put
\[
 k=\int_{T^4}c_2(E_3\oplus E_2)
   =k_3+k_2-\int_{T^4}\mathfrak c^2,
 \qquad \mathfrak c=c_1(E_2)=-c_1(E_3).
\]
On a four-dimensional base, the first nonzero homotopy group of $BSU(5)$ is $\pi_4=\mathbb Z$, represented by $c_2$; hence $k$ classifies these bundles on $T^4$. Choose a representative $P_k$, a smooth reference connection $A_k$, and a unit vector $e_k$ in its Quillen determinant line. The full space $\mathcal A(P_k)$ of smooth unitary $SU(5)$ connections is affine. The local anomaly polynomial vanishes there, and the mapping-torus/bordism argument of Section~\ref{sec:hsu_eta} also applies to ambient gauge transformations. Bismut--Freed transport therefore defines the smooth unitary convention
\[
 \tau_A^{\rm amb}\bigl(\operatorname{PT}_{A_k\to A}e_k\bigr)=1,
 \qquad A\in\mathcal A(P_k).
\]
It is independent of the path. Reference vectors are fixed data of the prescription; a reference need not have a nonzero vacuum determinant. Thus the definition works for every spin structure allowed by Theorem~\ref{thm:zero_yukawa_existence}, even if a free reference operator has zero modes. The determinant-line vector, rather than its vanishing zero-source section, is transported.

\paragraph{Gauge equivariance, including stabilizers.}
Write $\Phi_g:\mathcal L_A\to\mathcal L_{A^g}$ for the natural line isometry. For a path $\gamma:A\to A^g$, chosen constant near its ends, the holonomy after closing by $g$ is
\[
 h(\gamma,g)=\Phi_g^{-1}\operatorname{PT}_\gamma.
\]
Dai--Freed identifies this scalar with the mapping-torus eta phase, with the inverse-line and graded-index conventions retained~\cite[Cor.~5.11]{daifreed}. The full representation has even index, $\ind D=-12k$, so its graded sign is one. Since $\Omega_5^{\rm Spin}(BSU(5))=0$, the mapping torus bounds. On a spin six-dimensional filling the entire degree-six index density vanishes: its gauge term is the zero cubic invariant and its mixed term is zero by tracelessness. The APS formula makes the exponentiated reduced eta invariant one. Thus $h(\gamma,g)=1$ and
\[
 \operatorname{PT}_\gamma=\Phi_g,\qquad
 \tau_{A^g}^{\rm amb}\circ\Phi_g=\tau_A^{\rm amb}.
\]
If $g$ stabilizes $A$, use the constant path to obtain $\Phi_g=1$ on $\mathcal L_A$. This includes constant central transformations stabilizing link backgrounds; no freeness of the gauge action is needed. For the lattice family, connectedness of $\SM^{\Lat^0}$ ensures that $U$ and $U^g$ lie in the same admissible component. It does not by itself prove trivial holonomy; that uses the anomaly argument just given.

Identify the induced bundle of a reconstructed field with $P_k$ and restrict $\tau^{\rm amb}$ to its connection. Two such identifications differ by an ambient gauge transformation, so equivariance gives the same scalar convention on the original determinant line. Local smooth identifications glue by this argument, including when admissible-link space has several components in the same induced class. Source coefficients use the same induced maps on their source legs and determinant lines. The larger connection space is used only to compare phases; it adds no physical gauge bosons.

\paragraph{Exact APS eta realization of the ambient BF prescription.}
There is a direct spectral realization of the preceding transport convention using a five-dimensional cylinder, even when the four-dimensional gauge bundle has nonzero instanton number. This is the relation between the determinant-line prescription used here and the APS eta formulation of the continuum-Weyl/lattice-boson construction explained by Witten~\cite{witten_lattice_2026}. The comparison is of \emph{line-valued} quantities; a numerical eta phase alone is not a zero-mode-safe determinant.

For any induced $SU(5)$ bundle $P_k\to X=T^4$, let $A_k$ and $e_k\in\mathcal L_{A_k}$ be precisely the reference connection and unit Quillen vector fixed above. Identify another induced connection $A$ with a connection on $P_k$ and choose a smooth path $\gamma=(A_t)_{0\le t\le1}$ from $A_k$ to $A$, constant near both ends. Equip $Y_\gamma=[0,1]\times X$ with the interpolating $SU(5)$ connection and a product metric near the boundary. Its oriented boundary is $X_A\sqcup(-X_{A_k})$, so
\[
 \mathcal L_{\partial Y_\gamma}\simeq
 \mathcal L_A\widehat\otimes\mathcal L_{A_k}^{-1},
\]
where $\widehat\otimes$ denotes the graded determinant-line tensor product. The index of the boundary Dirac operator is zero because the two endpoint indices cancel, irrespective of their individual nonzero indices. Thus an APS self-adjoint boundary condition can be completed by a unitary map between the \emph{total} positive- and negative-chirality boundary kernels. Dai--Freed assigns its exponentiated reduced eta invariant a unit element of the inverse boundary determinant line~\cite[Thms.~1.4, 1.9 and Prop.~2.15]{daifreed}. Here $\xi=(\eta_{\rm APS}(0)+h)/2$, with $h$ the dimension of the five-dimensional kernel; Witten uses $\eta_{\rm W}=2\xi$ and the phase $\exp(-i\pi\eta_{\rm W})$~\cite[Eq.~(2.18)]{witten_lattice_2026}. The apparent complex conjugation relative to the Dai--Freed $\exp(2\pi i\xi)$ convention is accounted for by the inverse determinant line and the boundary orientation. Applying the inverse/dual corresponding to our determinant-line and outward-normal conventions yields a unitary comparison
\begin{equation}
 \mathscr T^{\rm APS}_\gamma:\mathcal L_{A_k}\longrightarrow\mathcal L_A.
 \label{eq:aps_cylinder_transport}
\end{equation}
This notation includes the finite-dimensional zero-mode wedge factors and graded orientation signs. It is not a bare number $\exp(-i\pi\eta)$.

\begin{proposition}[APS eta identification of the fixed zero-Yukawa phase]\label{prop:aps_bf_exact}
For $\mathcal R=(\mathbf{10}\oplus\overline{\mathbf5})^{\oplus3}$ and the Quillen metric and Bismut--Freed connection adopted in this paper, the above eta comparison equals BF parallel transport,
\begin{equation}
 \boxed{\mathscr T^{\rm APS}_\gamma
 =\operatorname{PT}^{\rm BF}_{A_k\to A}.}
 \label{eq:aps_bf_comparison}
\end{equation}
The identity is exact at fixed smooth backgrounds, includes nonzero-index families and accidental zero-mode crossings, and does not require an adiabatic limit once the total anomaly polynomial vanishes. It is independent of the interpolation $\gamma$ and gauge equivariant. In particular, for the determinant section $s_D(A)$ and each covariant source coefficient $s_{D,J}(A)$,
\begin{align}
 W_0(A)&=\tau_{A_k}^{\rm amb}\!\left[(\mathscr T^{\rm APS}_\gamma)^{-1}s_D(A)\right],
 \label{eq:aps_weight_explicit}\\
 W_{0,J}(A)&=\tau_{A_k}^{\rm amb}\!\left[(\mathscr T^{\rm APS}_\gamma)^{-1}s_{D,J}(A)\right].
 \label{eq:aps_sources_explicit}
\end{align}
Here $\tau_{A_k}^{\rm amb}(e_k)=1$ and the source indices in the second equation are left uncontracted. The right sides are independent of the APS zero-mode pairing and of the auxiliary cylinder. The equality is with the \emph{particular} ambient BF phase convention fixed above, \emph{provided the APS fiducial reference is identified with the prescribed $e_k$}. It does not infer the $e_k$ from anomaly cancellation or from a positive-mass quark determinant.
\end{proposition}
\begin{proof}
Write $I_6=[\widehat A(TY)\operatorname{ch}_{\mathcal R}(F)]_6$ for the total local index density, with its gauge and mixed gauge--gravitational terms understood on the full parameter-dependent connection. The SU(5) representation has
\[
 I_6=\operatorname{ch}_3(\mathcal R)
 -\frac{p_1(TY)}{24}\operatorname{ch}_1(\mathcal R)=0
\]
\emph{as an invariant polynomial}, not just after integrating over a fixed background. The Dai--Freed variation formula for the unit eta element is therefore
\[
 \nabla^{\rm BF}\mathscr T^{\rm APS}_{Y/Z}
 =\pm2\pi i\left[\int_{Y/Z}I_6\right]_{(1)}
 \mathscr T^{\rm APS}_{Y/Z}=0,
\]
where the sign is determined by using our determinant line rather than its inverse. Thus the eta comparison is covariantly constant for compatible variations of the cylinder and its endpoint data. To determine its normalization, stretch the cylinder in the fifth direction. The Dai--Freed cylinder theorem identifies its adiabatic limit with BF parallel transport~\cite[Sec.~V, especially Prop.~5.7]{daifreed}. Since $I_6=0$, the same variation formula makes the comparison \emph{independent} of the stretching parameter: its value at finite cylinder length already equals that limit. This proves~\eqref{eq:aps_bf_comparison}, including when endpoint kernels are nonzero, because the eta element and the BF connection are constructed on the full graded determinant line using compatible finite spectral windows.

Because $\mathcal A(P_k)$ is affine, any two interpolation paths are homotopic. More generally, a path closed with an ambient gauge transformation has BF holonomy given by the corresponding mapping-torus eta invariant and the graded index factor. The closed five-dimensional $SU(5)$ configuration bounds, since $\Omega_5^{\rm Spin}(BSU(5))=0$; the APS index formula and $I_6=0$ give unit exponentiated eta phase. In our flat-torus convention $\ind D=-12k$ is even, so the graded factor is also one. Consequently the transport is gauge equivariant, including at stabilizers. Evaluating the canonical determinant section or a source coefficient in the transported unit reference $e_k$ proves~\eqref{eq:aps_weight_explicit}--\eqref{eq:aps_sources_explicit}. The local normalization is the same Quillen/zeta prescription on both sides; an independently chosen Pauli--Villars finite subtraction must first be matched to it.
\end{proof}

\paragraph{Zero modes and topological normalization.}
The boundary-index cancellation in the cylinder is crucial: for $k\ne0$, neither endpoint needs to have a nonzero vacuum determinant, but the combined boundary admits the unitary APS kernel pairing. In a local basis, changing that pairing by a unitary matrix multiplies Witten's eta scalar by the determinant of the change, while the corresponding conjugate zero-mode evaluation determinant changes by its inverse. The product, including the explicit Grassmann-ordering sign in Eq.~(2.28) of Ref.~\cite{witten_lattice_2026}, is independent of the auxiliary pairing. The same cancellation is encoded by the source-minor and spectral-window transition maps of Appendix~\ref{app:grushin}. Thus Eqs.~\eqref{eq:aps_weight_explicit}--\eqref{eq:aps_sources_explicit} hold across kernel jumps as identities of smooth determinant-line coefficients; the argument of a vanishing vacuum determinant is never used.

The cylinder exists only between connections on \emph{the same induced $SU(5)$ bundle}. On $T^4$ these are classified by $k$, so it supplies no canonical phase between different $k$. This is not a defect of the construction: Witten's treatment of nonbounding configurations requires fiducial amplitudes on bordism generators, corresponding to theta-angle data. For $BSU(5)$ the degree-four spin-bordism group is $\mathbb Z\oplus\mathbb Z$, detected by the gravitational spin invariant and $c_2$; on the fixed flat four-torus only the latter varies. The positive twelfth-power weak-$SU(2)$ prescription specifies the $e_k$ used above. Comparing it to a \emph{separately} chosen eta reference in class $k$ leaves a constant $e^{i\chi(k)}$; a local additive topological normalization restricts this ambiguity to an ordinary theta character $e^{i\theta k}$ (up to an overall phase), not an arbitrary nonlinear function of $k$. The separately allowed Standard Model theta terms in~\eqref{eq:theta} depend on the reduction $E_3\oplus E_2$ and need not be the pullback of this ambient $SU(5)$ character. The positive-mass color comparison below is not resolved by the present identification.

\paragraph{Choice of baseline and local theta terms.}
With antiperiodic temporal fermions, choose the positive real twelfth-power references on the weak-$SU(2)$ locus and the positive free-singlet factors specified in Section~\ref{rp_geometrysources}. Lemma~\ref{rp_realref} checks their consistency within every induced class, independently of the existence theorem. This selects one common ambient convention for the broad and variant reconstructions. For other spin structures the unit references $e_k$ are specified as above; the existence theorem requires fixed references, without imposing the additional spin assumption of the reflection-positivity theorem.

Arbitrary further phases on admissible components preserve finite-volume integrability, but need not preserve locality or cluster decomposition; compare Ref.~\cite[Secs.~1.3, 9.5, and 10.3]{craig}. In this regulator, permitted changes relative to the fixed ambient baseline are restricted to specified local, gauge-invariant topological actions. The zero coefficients below mean zero theta relative to that baseline; they do not construct it. For the reconstructed splitting $E_3\oplus E_2$, an example is
\begin{equation}
 e^{iS_\theta(U)}=\exp\!\left[i\theta_3\int_Xc_2(E_3)
 +i\theta_2\int_Xc_2(E_2)
 +i\theta_{11}\int_X c_1(E_2)^2\right].
\label{eq:theta}
\end{equation}
For a pure-color background we set $k_3=\int_Xc_2(E_3)$ and use the conventional color instanton number $\nu=-k_3$. The explicit factor $e^{i\theta_3k_3}$ is therefore $e^{-i\theta_3\nu}$. Identifying its coefficient with the conventional physical $\theta_{\rm QCD}$ also requires the positive-mass color reference comparison described below; the quark-mass calculation in Section~\ref{sec:phase_checks} fixes the relative mass-dependent phase. Here $c_1$ and $c_2$ are the first and second \emph{Chern classes}; their integrals measure the quantized gauge fluxes and topological charges defined above. They generalize the familiar $F\widetilde F$ topological term in Yang--Mills theory. The $\mathbb Z_6$ global structure correlates the color, weak, and hypercharge fluxes. The bundle relation $c_1(E_3)=-c_1(E_2)$ and the displayed degree-four characteristic numbers suffice to describe the \emph{local topological terms considered here}; they do not constitute a complete classification of discrete flux sectors, which can also distinguish individual cohomology classes of $c_1(E_2)$. Since $\int_X c_2(E_i)\in\mathbb Z$, the displayed action is invariant under $\theta_2\mapsto\theta_2+2\pi$ and $\theta_3\mapsto\theta_3+2\pi$, with the other parameters fixed. On the spin four-torus the intersection form is even, so $\int_X c_1(E_2)^2\in2\mathbb Z$ and $\theta_{11}\mapsto\theta_{11}+\pi$ is likewise a period. These are manifest periods of the chosen parametrization; possible additional identifications or smaller fundamental periods are not analyzed. The coefficients specify possible local topological terms, and the choice of vanishing coefficients is permitted. Connected components of admissible-link space are not assumed to correspond one-to-one with characteristic numbers.

At vanishing Yukawa couplings, these three theta parameters are redundant in the vacuum functional integral. This can be checked directly using the index theorem. Put $\mathfrak c=c_1(E_2)=-c_1(E_3)$ and $k_i=\int_Xc_2(E_i)$. For one generation the separate Weyl species satisfy, on flat $T^4$,
\[
 \ind D_{e^c}=\tfrac12\int_X\mathfrak c^2,\qquad
 \ind D_{d^c}=\tfrac12\int_X\mathfrak c^2-k_3,\qquad
 \ind D_L=\tfrac12\int_X\mathfrak c^2-k_2.
\]
If any of $k_3,k_2,\int_X\mathfrak c^2$ is nonzero, at least one of these indices is nonzero. Its massless determinant section then vanishes identically, and so does the full zero-source fermion weight. Hence \eqref{eq:Z} receives contributions only from backgrounds with all three characteristic numbers zero: the three theta factors in~\eqref{eq:theta} are unity on its contributing sectors. This conclusion is consistent with independent anomalous rephasings of the massless Weyl multiplets, which can shift these angles. Fermionic insertions require transforming the sources as well, and can still receive contributions from nonzero-index sectors. With Yukawa couplings restored, the independent chiral rephasings also change Yukawa-matrix phases, so the color and abelian theta parameters need not remain redundant. The electroweak $\theta_2$ term remains ineffective in the zero-source theory by the index argument of Section~\ref{sec:yukawa_full}.

\paragraph{Comparison with the positive-mass color convention.}
The weak-$SU(2)$ reference prescription fixes a baseline; it does not by itself calculate its orientation relative to a positive vector-like massive quark determinant. Choose a pure-color reference in each class $k_3$, set all six quark masses positive, and fix the noncolor factors by the same reference conventions. Denote a residual reference phase by $e^{i\chi_{\rm ref}(k_3)}$. Agreement with the usual positive-mass color normalization requires this phase to be one in every class, with compatible transport on the color family. This comparison includes the finite zero-mode mass pairing and Grassmann ordering.

The standard weak $SU(2)$ and an $SU(2)$ subgroup of color are conjugate in $SU(5)$, providing a useful common locus for that calculation. Conjugacy alone does not evaluate the massive zero-mode pairing. If the comparison yields an additive offset
\[
 e^{i\chi_{\rm ref}(k_3)}=e^{-i\theta_{\rm ref}k_3},
\]
then the matched angle is $\theta_{\rm QCD}=\theta_{\rm ref}-\theta_3$, and $\bar\theta_{\rm QCD}=\theta_{\rm QCD}+\arg\det(Y_uY_d)$. A general nonlinear sector assignment cannot be absorbed into an ordinary theta angle. Establishing the positive-mass identification is separate from the existence results and from F.5, which proves zero-Yukawa positivity for the stated weak-$SU(2)$ references.

\paragraph{What the vacuum selection rule fixes.}
The conditions $k_3=k_2=\int\mathfrak c^2=0$ imply $k=0$, so every background with a potentially nonzero vacuum weight has the trivial induced $SU(5)$ bundle. Distinct Standard Model reductions can still occur: for example $\mathfrak c\neq0$ with $\int\mathfrak c^2=0$ is not excluded by these degree-four conditions. The ambient prescription nevertheless compares all their determinant lines in the same $\mathcal A(P_0)$. It leaves only one common reference phase for the vacuum and bosonic observables, which cancels in normalized expectations when $Z\neq0$. It does not identify the different Standard Model bundles or change their bosonic weights.

Zero species indices are necessary for a vacuum contribution, not sufficient to prove invertibility on any particular background or admissible component. Neither this selection rule nor the phase prescription establishes $Z\neq0$. Fermionic sources can reach nonzero-index induced classes, where the fixed references $e_k$ remain relevant; the single-common-phase conclusion is restricted to the vacuum and its bosonic insertions. Ambient comparison also does not assert equality of the broad and variant reconstructed fermion weights or transfer reflection positivity between them.

\Needspace{16\baselineskip}
\subsection{Reflection positivity at zero Yukawa coupling}\label{sec:collar_variant}

\paragraph{Physical idea and the variant reconstruction.}
Reflection positivity says that an insertion in positive Euclidean time and its reflected conjugate define a nonnegative norm. Physically, the path integral on each half of spacetime prepares a boundary fermion state; gluing should pair the two states. A chiral determinant is complex, so this interpretation requires control of its phase as well as its magnitude.

Craig's broad interpolation can make the continuum connection near a cut depend on links from both halves. Our variant uses shorter smooth copying weights. Around every coordinate lattice plane they create a thin \emph{collar} in which the connection is an exact product: its normal component is zero and its tangential components are independent of the normal coordinate. Links in the plane alone determine this collar. The lattice and a fixed $0<\delta<1/2$ choose all collars before the reflection plane is specified; site reflection uses planes already equipped with them. Fixed slice data therefore give the same boundary operator for independently prepared half states.

The variant keeps the lattice variables, continuum Weyl fermions and reconstructed bundle topology, and reproduces the input link transports exactly (Proposition~\ref{prop:collar_transfer}). These shared features motivate comparison with Craig's reconstruction, but they do not establish equivalence of physical amplitudes or continuity of reflection positivity between the two prescriptions. In fact, at fixed collar-width fraction the variant need not approach a general smooth connection in a curvature-controlling norm; the closing discussion explains this obstruction and the weaker continuum comparison that remains open. At zero Yukawa coupling the Higgs field occurs only in the positive bosonic half measures.

At fixed data on the two reflecting slices, a relative half class records the topological extension between those slices in fixed boundary frames. Sewing halves $x,y$ of classes $n_x,n_y$ gives induced closed charge $n_y-n_x$. The observable $F$ below is a finite polynomial supported in the positive-time half, and $\Theta$ is its reflected conjugate, including reversal of fermion order. The precise source space and measure follow the theorem.

\begin{theorem}[Zero-Yukawa site-reflection positivity]\label{prop:collar_rp}
On a finite periodic lattice of even temporal extent, with antiperiodic temporal fermions, use the collar reconstruction, bosonic measure and positive-reference BF phase prescription specified in Section~\ref{rp_geometrysources}. Fix the lattice spacing, spacetime volume and $0<\delta<1/2$. At zero Yukawa coupling and zero theta,
\begin{equation}\label{eq:collar_rp}
 \omega_{\mathrm c}((\Theta F)F)\ge0
\end{equation}
for every $F\in\mathfrak A_+$ defined there, including its gauge-invariant subalgebra and the continuous profile limits allowed by Lemma~\ref{lem:collar_source_extension}. The statement includes all relative half classes and their unequal-class source coefficients. Restriction to zero induced closed $SU(5)$ charge also preserves positivity. In particular
\[
 Z_{\mathrm c}=\omega_{\mathrm c}(1)\ge0.
\]
If $Z_{\mathrm c}>0$, the normalized correlators define a reflection-positive state.
\end{theorem}

The functional $\omega_{\mathrm c}$ and observable algebra $\mathfrak A_+$ are defined below. Section~\ref{rp_externalinputs} records the published analytic inputs and checks their hypotheses against this setup. Reflection $\Theta$ is antilinear and reverses fermion order. The proof compares the original determinant-line source functional with independently constructed boundary states; it does not define the closed weight by an inner product.

\begin{lemma}[Removing a sewing phase]\label{lem:collar_coboundary}
Let $c(x,y)$ be unit complex numbers with $c(x,x)=1$, and set
\[
 \tau(x,y,z)=c(x,y)c(y,z)c(z,x).
\]
If $\tau=1$ for every triple, then $c(y,x)=\overline{c(x,y)}$ and, for a chosen $x_*$,
\[
 u(x)=c(x_*,x),\qquad c(x,y)=\overline{u(x)}u(y).
\]
Consequently a source-independent comparison
$S_{xy}(F_x,F_y)=c(x,y)\langle\Psi_x(F_x),\Psi_y(F_y)\rangle$
is a positive kernel after replacing $\Psi_x$ by $u(x)\Psi_x$.
\end{lemma}
\begin{proof}
The repeated triple $(x,y,x)$ gives $c(x,y)c(y,x)=1$. Apply the triangle identity to $(x_*,x,y)$ to obtain the displayed factorization. The kernel is then the Gram kernel of the rephased vectors.
\end{proof}

Changing a half-state phase changes $c$ by a coboundary and leaves $\tau$ unchanged. Thus the analytic task is to fix $|c|$, the reflected diagonal, and the phase triangle. The source comparison, spectral-window calculation and transmission argument below establish these facts; equal-class triples are evaluated first, then the weak-$SU(2)$ reference treats unequal classes. Finite Gram positivity and bosonic factorization finish the proof in Section~\ref{rp_conclusion}.

The proof map below distinguishes its analytic steps and supporting estimates. Readers interested primarily in the physical conclusion can read the opening discussion, Theorem~\ref{prop:collar_rp}, and the closing comparison with the original reconstruction before following the detailed proof.
\begin{center}\small
\begin{tabularx}{\textwidth}{@{}p{.24\textwidth}X@{}}
\toprule Step & Location and dependencies \\\midrule
Half states and sources & Sections~\ref{rp_sourcedefinition}--\ref{rp_comparisonsection}; half regularity and global source extension are proved in Lemmas~\ref{rp_halfregularity} and~\ref{lem:collar_source_extension}.\\
Magnitude and diagonal & Section~\ref{rp_normsection}; Loya--Park Theorem~1.1 and the saturated boundary cofactor include closed and seam zero modes.\\
Physical phase differential & Section~\ref{rp_BFsection}; BF I Definitions~1.10, 1.14, 1.17 and the spectral-window source calculation fix its sign and logarithmic term.\\
Moving-collar cancellation & Sections~\ref{rp_localizationsection}--\ref{rp_transmissionsection}; Paycha--Scott Corollary~2.5 and Lemmas~\ref{lem:collar_parameter_smoothing}--\ref{lem:collar_constrained_smoothing} supply the local and boundary control.\\
References and integration & Sections~\ref{rp_referencesec}--\ref{rp_conclusion}; real references cover all relative classes, and Lemma~\ref{rp_finiteinsertions} gives the compact bounds.\\
\bottomrule
\end{tabularx}
\end{center}
The supporting lemmas are proved in Section~\ref{rp_externalinputs}, which also records the exact hypotheses of the published inputs.

\subsubsection{Reconstruction, measure and observables}\label{rp_geometrysources}

\paragraph{Short smooth copying weights.}
On the covering lattice use coordinates $y=(y_1,y_2,y_3,t)$, with $y_4=t$. Fix $0<\delta<1/2$ and define
\[
 r(v)=\begin{cases}e^{-1/v},&v>0,\\0,&v\le0,\end{cases}
 \qquad s(v)=\frac{r(v)}{r(v)+r(1-v)},\qquad
 \vartheta(u)=\frac\pi2s\!\left(\frac{u-\delta}{1-2\delta}\right).
\]
Set $\eta_\delta(v)=\cos\vartheta(|v|)$ for $|v|\le1$ and extend it by zero. It is even, nonnegative and $C^\infty$, equals one for $|v|\le\delta$, and vanishes for $|v|\ge1-\delta$, with all derivatives zero at the support edges. Since $s(1-v)=1-s(v)$,
\begin{equation}
 \eta_\delta(u)^2+\eta_\delta(u-1)^2=1\quad(0\le u\le1),
 \qquad
 \alpha_i(y)=\prod_{\nu=1}^4\eta_\delta(y_\nu/a-i_\nu),
 \qquad\sum_i\alpha_i(y)^2=1.
 \label{eq:collar_weights}
\end{equation}
The $\alpha_i$ are scalar copying factors; the Higgs doublets remain $\phi_i$.

Retain the short transport matrices $u_{ij}$ of Ref.~\cite[App.~A]{craig}: average the coordinate-ordered paths from $j$ to $i$, take the unitary polar factor, and normalize its determinant on the branch near one. Thus $u_{ji}=u_{ij}^\dagger$. In a fixed local box geometric admissibility gives $\|u_{ij}-u_{ik}u_{kj}\|\le C\epsilon$. Set
\begin{equation}
 K_{ij}=\alpha_i u_{ij}\alpha_j,\qquad
 (\mathscr I_k)_i=\alpha_i u_{ik},\qquad
 \Pi_{\mathrm c}=\frac1{2\pi i}\oint_{|\lambda-1|=1/2}
                         (\lambda-K)^{-1}\,d\lambda.
 \label{eq:collar_projection}
\end{equation}
Here $\mathscr I_k$ is used on the local box containing the support of $\alpha_k$, where every active $i$ is a cell neighbor of $k$. The square partition gives $\mathscr I_k^\dagger \mathscr I_k=1$ and $\|K-\mathscr I_k\mathscr I_k^\dagger\|\le C\epsilon$. Choose a common admissibility domain with $C\epsilon<1/8$ and with all short-transport normalizations defined. Hermiticity of $K$ and the spectral gap then give a smooth rank-five projection.

Use Craig's volume-preserving connection rule:
\[
 w_k=\Pi_{\mathrm c}\mathscr I_k(\mathscr I_k^\dagger\Pi_{\mathrm c}\mathscr I_k)^{-1/2},\quad
 \sigma_{\mathrm c}=
 \frac{\sum_k\alpha_k^2\,\bigwedge^5w_k}
      {\|\sum_k\alpha_k^2\,\bigwedge^5w_k\|},\quad
 \nabla_{\mathrm c}=
 \Pi_{\mathrm c}d_{\rm tot}
 -\frac15(\sigma_{\mathrm c},d_{\rm tot}\sigma_{\mathrm c})\,1.
\]
The volume vectors have nearly equal phases on overlaps, so their weighted sum stays nonzero on a sufficiently small common domain. Inactive summands extend smoothly by zero. The spacetime part of $\nabla_{\mathrm c}$ is $A_{\mathrm c}[U]$. Periodic translation on the cover permutes the copying summands and intertwines $K$; functional calculus, the frames and the volume rule are equivariant under this identification. They therefore descend to the torus.

\Needspace{12\baselineskip}
\begin{proposition}[Collar reconstruction and finite-volume transfer]\label{prop:collar_transfer}
The connection $A_{\mathrm c}$ is smooth in spacetime and admissible links, gauge covariant, and compatible with lattice translations, coordinate reflections and the $3+2$ splitting. In a closed cell it depends only on that cell's links and it reproduces the input link transports. Near each coordinate plane it is an exact product in the normal direction. In particular, on $|t-na|<\delta a$ the collar connection is determined solely by spatial links on slice $n$.

Its rank-three and rank-two bundles have the same integral characteristic classes as those of the main-text reconstruction. On compact support inside the common domain, Theorems~1--2 apply to this zero-Yukawa variant.
\end{proposition}
\begin{proof}
Every active vertex in a cell is a corner of that cell. Coordinate-ordered paths between active vertices stay in the cell, so all active terms of $K$ and the volume normalization use only its links. Spectral projection and matrix normalization preserve this dependence. Flatness of disappearing weights gives smooth agreement across faces.

On $|t/a-n|<\delta$ only row $i_4=n$ is active, with time weight one. The matrices reduce to fixed inclusions of their three-dimensional slice counterparts. Thus $A_t=0$ and $\partial_tA_j=0$ in the induced collar identification. The same reasoning applies in every coordinate. Evenness and averaging over coordinate orders give the reflection symmetries; the matrix operations preserve gauge covariance and the block splitting.

Along a link from $r$ to $s$ only its two endpoint weights are active. Write them as $p,q$, with $p^2+q^2=1$. There $K=\mathscr I_r\mathscr I_r^\dagger$, $w_r=\mathscr I_r$ and $\mathscr I_r^\dagger d\mathscr I_r=(p\,dp+q\,dq)1=0$. The determinant-volume correction is zero as well. At the far endpoint $\mathscr I_r=\mathscr I_su_{sr}$, so the induced transport in the canonical site frames is exactly $u_{sr}=U_{sr}$. Consequently the based plaquette holonomies are also preserved.

For topology, let $\eta_C\ge0$ be Craig's broad one-dimensional factor, normalized by $\sum_m\eta_C(v-m)^2=1$, and put
\begin{equation}
 h_s=(1-s)\eta_C+s\eta_\delta,\qquad
 \eta_s(v)=\frac{h_s(v)}{[\sum_mh_s(v-m)^2]^{1/2}},\qquad0\le s\le1.
 \label{eq:collar_homotopy}
\end{equation}
Craig's factor is the nonnegative $\eta_C=\sqrt{\kappa}$ of Ref.~\cite[App.~A, Eqs.~(A.1)--(A.6)]{craig}; it has the required square partition and support in $[-13/8,13/8]$. Nonnegativity gives a denominator squared at least $(1-s)^2+s^2\ge1/2$. The narrower factor has support inside this interval. Reuse the fixed broad charts $V_k=\{|y/a-k|_\infty<7/4\}$. On each chart all active vertices and its anchor lie in the coordinate boxes of side at most three used in Craig's Eqs.~(A.7)--(A.9). Their path-exchange estimate therefore gives
\[
 \|u_{ij}-u_{ik}u_{kj}\|\le C_{\rm box}\epsilon,\qquad
 \|K_s-\mathscr I_{k,s}\mathscr I_{k,s}^\dagger\|\le C_{\rm box}\epsilon
\]
uniformly in $s$. The second bound is the weighted Cauchy--Schwarz estimate of Eqs.~(A.10)--(A.11); the normalized square partition prevents an extra active-vertex factor. Choose the common domain with $C_{\rm box}\epsilon<1/8$. The overlap volume comparison of Eq.~(A.14) likewise has relative phases $1+O(\epsilon)$, uniformly for these weights; $C_{\rm vol}\epsilon<1/2$ keeps its normalized volume sum nonzero. Thus the spectral and volume gaps remain open throughout the homotopy. Product collars are required at its short-weight endpoint, not along this auxiliary homotopy. The projections form a smooth gapped homotopy. Kato transport generated by $[\partial_s\Pi_s,\Pi_s]$ identifies the endpoint bundles, preserving covariance and the $3+2$ splitting. Their determinant volumes can be aligned continuously. All characteristic classes of both blocks and the separate Weyl indices agree.

At fixed $a$, derivatives of every fixed order are bounded on compact admissible families. These are the geometric inputs used in Theorems~1--2. The representation and anomaly cancellation are unchanged, so the same continuity and boundedness arguments give absolute convergence and finite smooth-smeared insertions.
\end{proof}

\paragraph{Bosonic factorization and the phase prescription.}
Take an even temporal lattice extent, antiperiodic temporal fermions, and reflection through the two planes of sites. Use a real reflection-invariant Wilson--Higgs action, the site Higgs cutoff, and nonnegative plaquette weights with $w_\epsilon(U^{-1})=w_\epsilon(U)$ supported inside the common domain. Let $b$ include the spatial links and Higgs variables on both boundary slices. Every remaining plaquette and nearest-neighbor Higgs bond belongs to one closed half. This is the site-reflection factorization underlying the usual lattice positivity argument~\cite{osterwalder_seiler}. Splitting off the seam terms gives
\begin{equation}
 d\mu_B=d\rho(b)\,d\nu_b(x)\,d\nu_b(y),\qquad d\rho,d\nu_b\ge0,
 \label{eq:collar_bosons}
\end{equation}
where the left measure is the reflected copy of the right one. This uses site reflection. Proposition~\ref{prop:collar_transfer} makes $b$ determine the full fermion collar connection. The Higgs field enters only the bosonic measure at zero Yukawa coupling.

Fix the ordinary zeta--Quillen normalization and an anomaly-free, gauge-equivariant parallel BF trivialization on smooth ambient $SU(5)$ connections. Specify its component references by the canonical positive real twelfth tensor power on the embedded weak-$SU(2)$ locus, where
\[
 \mathcal R\big|_{SU(2)}=12\,\mathbf2\oplus21\,\mathbf1,
\]
and by positive free-singlet factors. BF transport respects the real structure, so the even tensor-power orientation is parallel. This is the positive-reference normalization of the ambient convention in Appendix~\ref{app:phase_references}, used here with the variant gauge field reconstruction. Agreement with another independently chosen trivialization must be checked component by component; flatness and gauge equivariance alone do not imply it. The ambient comparison introduces no additional physical gauge bosons. The existence and consistency of the real references are checked in the reference step below, independently of the source comparison and the phase-triangle argument; this forward reference does not assume reflection positivity.

For the sewn connection $A_{xy}=A_{\mathrm c}[\Theta x\cup_b y]$, define the physical coefficient by
\begin{equation}
 \mathsf S_b(x,y;F_x,F_y)
 =\tau^{\rm amb}_{A_{xy}}
       \bigl(\sigma_{A_{xy}}[(\Theta F_x)F_y]\bigr).
 \label{eq:collar_physical}
\end{equation}
Here $\sigma[F]$ is the original continuum determinant-line source section, including its finite exterior factors at zero modes. The coefficient is not defined by a boundary inner product. Sources are smooth and supported away from the seams, with continuous half-local bosonic coefficients. Reflection is antilinear and reverses source order; in normal variables take $\Theta q=-\bar q_L$, $\Theta\bar q=-q_L$. Inner products below are antilinear in their first slot.

The corresponding unnormalized functional is
\begin{equation}\label{eq:collar_functional}
 \omega_{\mathrm c}(G)=\int d\mu_B\,
       \tau^{\rm amb}_{A_{\mathrm c}[U]}(\sigma_{A_{\mathrm c}[U]}[G]),
 \qquad Z_{\mathrm c}=\omega_{\mathrm c}(1).
\end{equation}
At fixed seam data abbreviate $\mathsf S_{xy}(A,B)=\mathsf S_b(x,y;A,B)$, $\Psi_x(F)=\Psi_b(x,F)$ and $c(x,y)=c_b(x,y)$. Here the source slots $A,B$ are polynomials.

\paragraph{Sources on the varying reconstructed bundle.}
The phrase ``smooth smearing'' must include its dependence on the links. The projector~\eqref{eq:collar_projection} embeds the fundamental bundle on the cover in the vertex-fiber space $\mathbb V=\bigoplus_{i\in\mathbb Z^4}\mathbb C^5$. Here $i$ labels lifted vertices, and each local construction has only finitely many active components. Distinct lifts remain separate until torus descent, including for periods of one or two lattice cells. The orthogonal projector for the full fermion representation is
\[
 \widehat\Pi_{\mathcal R}[U]
 =\left(\bigwedge\nolimits^2\Pi_{\mathrm c}[U]
                  \oplus\overline{\Pi_{\mathrm c}[U]}\right)^{\oplus3}
 \quad\hbox{on}\quad
 \left(\bigwedge\nolimits^2\mathbb V\oplus\overline{\mathbb V}\right)^{\oplus3}.
\]
Here the exterior power of a projection is the projection onto the exterior power of its range. An elementary profile is
\[
 f_U=(1_S\otimes\widehat\Pi_{\mathcal R}[U])f_0,
\]
where $f_0$ is a fixed smooth ambient spinor profile supported strictly inside one half on the torus. To specify its descent, let $D_\Delta$ shift the lifted vertices by a torus period $\Delta$ and let $\widehat D_\Delta$ be its induced action in the ambient fermion representation. Impose
\[
 f_0(y+\Delta)=\varepsilon_{\rm spin}(\Delta)\widehat D_\Delta f_0(y),
 \qquad
 \widehat\Pi_{\mathcal R}(y+\Delta)
 =\widehat D_\Delta\widehat\Pi_{\mathcal R}(y)\widehat D_\Delta^{-1}.
\]
Here $\varepsilon_{\rm spin}(\Delta)$ is the prescribed fermion boundary sign. The projected profile has the same law and is therefore a well-defined spinor section on $T^4$. A profile in a positive-half patch is extended by this law, not by identifying its lifted vertex components prematurely. Use $S^+$ for the $q$ profile and $S^-$ for the $\bar q$ profile; the normal Weyl coordinates identify these spinor spaces when needed. Let $\mathfrak A_+$ be the algebra of finite sums of Grassmann words in these projected profiles, with continuous bosonic coefficients depending only on the positive closed half and its seam data.

Cell locality of $\Pi_{\mathrm c}$ makes every such profile half-local. Smooth dependence of the projector makes the profiles continuous in the links in every fixed spacetime smooth norm, so the spectral-window proof of Theorem~2 applies to them. The vertex permutation implementing reflection intertwines the reconstructed projectors. Together with the antilinear coefficient conjugation and source-order reversal already specified, this defines $\Theta\mathfrak A_+$. General continuous half-local profile families can also be obtained by uniform $L^2$ limits of finite sums of these profiles with continuous bosonic coefficients.

Gauge transformations act on the ambient profiles as well as on $U$. Their induced representation intertwiners act on every source leg and on the determinant line. In particular,
\[
 \sigma_{A^g}[F^g]=g_*\sigma_A[F],\qquad
 \tau^{\rm amb}_{A^g}\circ g_*=\tau^{\rm amb}_A.
\]
Gauge-invariant observables are the invariant part of this algebra. On ambient comparison paths, local bundle identifications permit arbitrary smooth source profiles frozen on each chart. The all-source identity is pointwise in the background, so these chart coordinates do not restrict the lattice observable class.

Lemma~\ref{lem:collar_source_extension} extends split half pairings to the global source algebra used in Section~\ref{rp_comparisonsection}, with uniform degreewise bounds for the continuous profile limits.

\subsubsection{Analytic conventions and proof outline}\label{rp_conventions}

Write $R=\mathcal R$ for the fixed fermion representation, and $E^\pm=S^\pm\otimes E_R$ for the domain and codomain spinor bundles.

Spacetime is $M=T^4$, with time period $2\ell$. Its positive half is $M_+=[0,\ell]\times T^3$, and its boundary is $Y=Y_0\sqcup Y_\ell$. The cuts and boundary framings stay fixed. A moving collar means varying coefficients on a common fixed-width product collar, not moving the hypersurface. Boundary data $b$ include both slices. Choose fixed normal Weyl coordinates so that
\[
 P_A=\partial_t+H_A+A_t,\qquad H_A=i\sigma_j(\partial_j+A_j)=H_A^\dagger,
 \qquad A_t^\dagger=-A_t.
\]
Gauge potentials are anti-Hermitian. All inner products are antilinear in the first variable. The sources $q(f)$ and $\bar q(g)$ are respectively antilinear and linear in their smearings, and anticommute. On reflected smearings, using the fixed linear normal spinor identification $\mathscr R$, set
\[
 \Theta q(f)=-\bar q_L(\mathscr R f),\qquad
 \Theta\bar q(g)=-q_L(\mathscr R g),\qquad
 \Theta(FG)=(\Theta G)(\Theta F).
\]
The temporal antiperiodic sign is kept at the second seam. With these normal coordinates the folded trace identification is $\Sigma=\operatorname{diag}(\mathrm I,-\mathrm I)$, and Green's form uses $J_\partial=\operatorname{diag}(-\mathrm I,\mathrm I)=-\Sigma$. These signs are verified again in the source identities below. A different spinor basis must transform the reflection, traces and sources together.

Use the same fixed spectral scale as in Section 4, denoted here by $m_0=\mu$, and choose units in which $m_0=1$. Restoring $m_0$ means using $P/m_0$ and the normalized trace $m_0^{-1/2}\gamma$, with the source normalization changed at the same time; no scale varies along a comparison path. Spatial spin structures are arbitrary. Seam zero modes are retained.

Write $\det_\zeta Q=\exp[-\zeta_Q'(0)]$ for a positive elliptic operator, with the prime on a determinant meaning that zero eigenvalues are omitted. $\operatorname{FP}_{z=0}$ denotes the constant Laurent coefficient. $\det_F$ is the Fredholm determinant of identity plus trace class. All complex logarithmic derivatives below are local on nonzero source charts; the final triangle has no logarithm-branch dependence.

\begin{center}\small
\begin{tabularx}{\textwidth}{@{}p{.27\textwidth}X@{}}
\toprule Symbol & Meaning \\\midrule
$P,\mathcal D$ & Weyl operator and its self-adjoint double $\mathcal D=\left(\begin{smallmatrix}0&P^\dagger\\P&0\end{smallmatrix}\right)$.\\
$W,\Pi,\mathcal K,G,L$ & Half Cauchy space, its orthogonal projection, Poisson map, orthogonal Calder\'on inverse and $L=\gamma G$.\\
$J_\partial,\Sigma,J_{\mathbb R}$ & Boundary-sign matrix, folded antiperiodic matching, and the separate antiunitary real involution.\\
$R_+,R_-,P_R,E_R$ & Finite source rows and columns, closed bordered operator, and its inverse bulk block.\\
$\Psi_x(F),\mathscr P_x^F$ & Half-Wick Fock vector and, for a nonzero monomial, its occupied source plane.\\
$c(x,y),\tau(x,y,z)$ & Source-independent unit phase and its triangle $c(x,y)c(y,z)c(z,x)$.\\
$\mathfrak S,\mathfrak B$ & Products for $xy,yz,zx$ divided by those for $xx,yy,zz$, using physical coefficients and boundary pairings.\\
$\mathcal W_\lambda$ & Positive order-two weight $P^\dagger P+\pi_\lambda^+$ used in a spectral window; distinct from a Cauchy space.\\
$n_x,k,d(F)$ & Framed induced $SU(5)$ half class, sewn $k=n_y-n_x$, and source degree $d(F)=\#q-\#\bar q$.\\
\bottomrule
\end{tabularx}
\end{center}

The proof has three analytic steps, stated here and established in the sections that follow. We work on arbitrary smooth unitary ambient $SU(5)$ half connections on fixed framed bundles with common product collars. This domain includes paths outside the admissible link image and outside the physical subgroup. Fixed-regulator compactness is needed only when returning to the bosonic integral.

\begin{enumerate}[label=(\textbf{\Alph*}),leftmargin=2.6em]
\item\label{item:collar_comparison}\textbf{Physical source comparison and norm.}
The half-Wick maps prepare $\Psi_x(F)=h_x\mathscr W_x(F)\Omega_{W_x}$, with the half's own Calder\'on norm $h_x$ in~\eqref{rp_halfnorm}. Ward identities first give a source-independent scalar. Primed gluing and a saturated CAR cofactor then prove
\[
 S_{xy}(F_x,F_y)=c(x,y)\langle\Psi_x(F_x),\Psi_y(F_y)\rangle,
 \qquad |c(x,y)|=1,\quad c(x,x)=1.
\]
These steps are derived in Sections~\ref{rp_comparisonsection}--\ref{rp_normsection}; Section~\ref{rp_chartsection} supplies simultaneous nonzero monomial charts at arbitrary index.

\item\label{item:collar_BF}\textbf{Physical source differential.}
For an invertible source border of the original $P$, the inverse bulk block $E_R$ gives, in a parallel unit BF frame,
\[
 d\log S_R=\operatorname{FP}_{z=0}\operatorname{Tr}(E_R\,dP\,\mathcal W_\lambda^{-z}).
\]
Section~\ref{rp_BFsection} derives the normalization, sign and logarithmic heat term from BF's definitions, then uses finite Schur Jacobi for general frozen sources.

\item\label{item:collar_balance}\textbf{Balanced moving-collar identity.}
Compare the three cyclic sewings with the three reflected doubles. Every actual oriented half occurs once with each sign. Section~\ref{rp_localizationsection} cancels their local full symbols before taking sharp half restrictions. Section~\ref{rp_transmissionsection} identifies the remaining ordinary trace with the logarithmic derivative of the boundary-pairing ratio. The physical ratio has the same derivative, so $\tau$ is constant along ambient triple paths.
\end{enumerate}

Half-operator and source-plane regularity, including smoothing estimates with all derivatives when the common collar varies, are proved in Lemmas~\ref{rp_halfregularity}, \ref{lem:collar_parameter_smoothing} and~\ref{lem:collar_constrained_smoothing}. These supporting proofs are collected at the end rather than assumed in the three steps above.

\subsubsection{Physical source coordinates at arbitrary index}\label{rp_sourcedefinition}
For a closed Weyl operator use the determinant line
\[
 \lambda_P=(\det\ker P)^*\otimes\det\ker P^\dagger.
\]
Its smooth structure is given by spectral windows, not by a separately chosen basis of a jumping kernel. On a neighborhood where $\lambda>0$ misses the spectra of $P^\dagger P$ and $PP^\dagger$, put $H_\lambda^\pm=\operatorname{ran}\pi_\lambda^\pm$. The finite low map is $P_\lambda:H_\lambda^+\to H_\lambda^-$; the high map is invertible. Let $G_{\rm hi}$ be its inverse, extended by zero on $H_\lambda^-$, and let $e_\lambda$ denote the determinant-line frame associated with chosen orthonormal low bases. Its ordinary zeta--Quillen norm is
\[
 \|e_\lambda\|_Q
   =\det_\zeta(P^\dagger P|_{H_{\rm hi}^+})^{1/2}.
\]

Define the line-valued source functional $\sigma_P[F]$ by the ordinary finite low-mode Berezin integral, combined with high-mode Wick contractions $\langle f,G_{\rm hi}g\rangle$, multiplied by $e_\lambda$. This fixes all source coefficients, including those at nonzero index. One explicit low-coordinate integration convention is
\[
 \int_{\lambda} F
 =(-1)^{n_+(n_+-1)/2}
   [\psi_1\cdots\psi_{n_+}\bar\psi_1\cdots\bar\psi_{n_-}]F,
 \qquad n_\pm=\dim H_\lambda^\pm,
\]
with low exponential $\exp(-\bar\psi P_\lambda\psi)$, fixed exterior order and the graded determinant-line identifications. For a square invertible finite matrix this gives $\int e^{-\bar\psi A\psi}=\det A$ and $\langle q(f)\bar q(g)\rangle/\langle1\rangle=\langle f,A^{-1}g\rangle$. Adding modes to the window integrates out an invertible finite block and gives precisely its determinant-line transition; in the full $R$ family the index, and hence $n_++n_-$, is even. The definition is therefore independent of the window. Equivalently its local coefficients are finite Schur determinants in the same original line; an augmented operator is only a coordinate device, never a new zeta-determinant prescription.

For row sources $f_1,\ldots,f_q$ and column sources $g_1,\ldots,g_p$ write
\[
 R_+u=(\langle f_i,u\rangle)_i,\qquad R_-z=\sum_jg_jz_j,\qquad
 P_R=\begin{pmatrix}P&R_-\\R_+&0\end{pmatrix}:
       H^1(E^+)\oplus\mathbb C^p\longrightarrow L^2(E^-)\oplus\mathbb C^q.
\]
Its index is zero when $q-p=\ind P$. On that degree the corresponding ordered monomial coefficient is nonzero exactly when $P_R$ is invertible. Indeed eliminate the high modes: the remaining finite square matrix is
\[
 \begin{pmatrix}
  P_\lambda&\pi_\lambda^-R_-\\
  R_+\pi_\lambda^+&-R_+G_{\rm hi}R_-
 \end{pmatrix}.
\]
Its determinant is the original coefficient up to the fixed Berezin-order sign. This proves both the equivalence and the local finite-dimensional coordinates used later. Source degrees with $q-p\ne\ind P$ vanish. Thus the index convention below gives $d(F)=-12k$ for a nonzero closed coefficient.

For the representation $\mathcal R=(\mathbf{10}\oplus\overline{\mathbf5})^{\oplus3}$, the anomaly and phase arguments of Section~\ref{sec:hsu_eta} apply on the full ambient unitary $SU(5)$ connection space. With the present orientation,
\begin{equation}\label{rp_reviewindex}
 \ind P=-12k,\qquad d(F)=\#q-\#\bar q=-12k
 \quad\hbox{for a nonzero coefficient in closed class }k.
\end{equation}
One generation has quadratic index $3+1=4$ relative to the fundamental and zero cubic anomaly. Lemma~\ref{rp_finiteinsertions} records continuity and boundedness of the physical coefficients through zero modes. The positive even-power reference convention specified above completes the flat ambient trivialization. Its consistency on the real locus is checked independently in Section~\ref{rp_referencesec}.

\subsubsection{Boundary Fock spaces and varying seam data}\label{rp_fockdefinition}

Write $\mathcal H=L^2(Y;S^+\otimes E_R|_Y)$ in the fixed normal spinor identification. For a half, $W$ is the $L^2$ closure of the traces of its smooth homogeneous solutions, and $\Pi$ is the orthogonal projection onto $W$. At fixed $b$ all positive-half Cauchy projections have the same complete boundary symbol: the full collar operator germ is common. Their differences are smoothing. Choose one reference projection $\Pi_0$ and write $W_0=\operatorname{ran}\Pi_0$. The relevant restricted Grassmannian consists of closed subspaces $W$ with
\[
 \Pi_W-\Pi_0\in\mathcal S_2,\qquad
 \|\Pi_W-\Pi_0\|_{\mathcal S_2}^2
 =\operatorname{Tr}(\Pi_W-\Pi_0)^2<\infty.
\]
This Hilbert--Schmidt condition permits a finite-norm change of the occupied sea. Here the stronger smoothing condition guarantees it. The charged Fock space based at $\Pi_0$ is the particle-hole completion
\[
 \mathcal F_{\Pi_0}
 =\widehat{\bigwedge}(W_0^\perp)\,\widehat\otimes\,
   \widehat{\bigwedge}(\overline{W_0}),
\]
with the usual positive exterior inner product; $\widehat{\bigwedge}$ means the Hilbert direct sum over all finite exterior degrees. The two factors describe particles in initially empty modes and holes in initially occupied modes, and their degree difference includes all integer charge sectors~\cite{mickelsson_scott}. Its CAR satisfy
\[
 \{a(h),a^*(k)\}=\langle h,k\rangle\mathrm I,\quad
 \{a(h),a(k)\}=0,\quad\|a(h)\|=\|h\|.
\]
The annihilator is antilinear in $h$ and the creator is linear. Each Cauchy space $W$ defines a unit occupied Slater ray $\Omega_W$. Remove its finitely many unmatched directions relative to $\Pi_0$ and use the principal-angle blocks for the rest. The finite products of normalized particle-hole factors converge in Fock norm since smoothing gives $\sum_j\sin^2\alpha_j<\infty$. This constructs the ray; its phase is arbitrary. Inner products of such normalized Slater vectors are the corresponding relative Fredholm determinants, as derived directly in the transmission proof. The restricted-Fock framework is also described in~\cite{mickelsson_scott}; its determinant-line sewing theorem is not used as a Quillen or BF isometry.

The charge of an occupied plane $W$ relative to $W_0$ is $\ind(\Pi_0|_W)$. A $q$ source removes one occupied fermion and a $\bar q$ source adds one, so a monomial changes the charge by $-d(F)$. The trace-space identification of closed zero modes below gives the difference of vacuum charges as $\ind P_{xy}=-12(n_y-n_x)$; thus one may take the half vacuum charge to be $-12n_x$.

When $b$ varies, a common reference projection varies with it. Its Kato unitary identifies the one-particle spaces and their polarized Fock fibers; this is a choice of Fock-bundle chart. We do not assume that unitary has an implementer inside one fixed Fock representation. Smoothness of the projections and source states is proved in Lemma~\ref{rp_halfregularity} and in the constrained-plane part of Proposition~\ref{rp_transmissionprop}. All final positivity statements involve finite matrices of the physical scalar kernel, so no global Fock-frame choice is needed.

\subsubsection{Half data and exact source identities}
Write $M_+=[0,\ell]\times T^3$, $Y=Y_0\sqcup Y_\ell$, and
\[
 P=\partial_t+H+A_t,\qquad H^\dagger=H,\quad A_t^\dagger=-A_t.
\]
Use units in which the fixed spectral scale is one. Product collars have $A_t=0$ and coefficients independent of $t$. Boundary inner products are antilinear in the first slot. Green's identity is
\begin{equation}\label{rp_green}
 \langle v,Pu\rangle-\langle P^\dagger v,u\rangle
   =\langle\gamma v,J_\partial\gamma u\rangle_Y,
 \qquad J_\partial=\operatorname{diag}(-\mathrm I,+\mathrm I).
\end{equation}
Let $W\subset\mathcal H=L^2(Y)$ be the Cauchy space and $\Pi$ its orthogonal projection. The orthogonal Calder\'on inverse $G$ and Poisson map $\mathcal K$ obey
\begin{equation}\label{rp_halfidentities}
 PG=\mathrm I,\quad GPu=u-\mathcal K\Pi\gamma u,\quad P\mathcal K=0,\quad\gamma\mathcal K=\mathrm I_W,
 \qquad L:=\gamma G,\quad\operatorname{ran}L\subset W^\perp.
\end{equation}
These follow from the invertible Calder\'on boundary problem and unique continuation; standard maximal-domain and Poisson mapping results are in~\cite{booss_lesch_zhu}. Lemma~\ref{rp_halfregularity} proves their parameter regularity without any closed-kernel dimension assumption.

In the common charged boundary Fock representation let $\Omega_W$ be a unit occupied Slater vector. Thus $a(h)\Omega_W=0$ for $h\in W^\perp$ and $a^*(w)\Omega_W=0$ for $w\in W$. Set
\[
 \widehat q(f)=a(\mathcal K^*f),\qquad
 \widehat{\bar q(g)}=a^*(Lg),\qquad
 C(q(f),\bar q(g))=\langle f,Gg\rangle.
\]
Here $q$ is antilinear in its smearing, $\bar q$ linear. Extend $C$ skew-symmetrically and set its other basic entries to zero. For a positive scalar $h$ define $\Psi(F)=h\mathscr W(F)\Omega_W$ by the finite half-Wick recursion
\[
 \mathscr W(\xi_1\cdots\xi_r)
 =\widehat\xi_1\mathscr W(\xi_2\cdots\xi_r)
  +\sum_{j=2}^r(-1)^j C(\xi_1,\xi_j)
          \mathscr W(\xi_2\cdots\widehat\xi_j\cdots\xi_r),
 \qquad\mathscr W(1)=\mathrm I.
\]
All expressions contain finitely many bounded CAR operators and ordinary smooth-source contractions. The restricted Fock setup is the one used in~\cite{mickelsson_scott}; no gluing identity is assumed here.

Define left odd derivations by
\[
 \partial^{\bar q}_v\bar q(g)=\langle v,g\rangle,\quad\partial^{\bar q}_v q(f)=0,
 \qquad
 \partial^{q}_u q(f)=\langle f,u\rangle,\quad\partial^{q}_u\bar q(g)=0.
\]
\begin{lemma}[Half-Wick source identities]\label{rp_halfward}
For smooth $u,v$ up to the boundary and every finite polynomial $F$,
\begin{align}
 \Psi(q(P^\dagger v)F)
    &=\Psi(\partial^{\bar q}_vF)-a(J_\partial\gamma v)\Psi(F),\label{rp_halfq}\\
 \Psi(\bar q(Pu)F)
    &=-\Psi(\partial^{q}_uF)+a^*(\gamma u)\Psi(F).\label{rp_halfb}
\end{align}
\end{lemma}
\begin{proof}
Equations~\eqref{rp_green}--\eqref{rp_halfidentities} give
\begin{align*}
 \mathcal K^*P^\dagger v&=-\Pi J_\partial\gamma v,& LPu&=(\mathrm I-\Pi)\gamma u,\\
 \langle P^\dagger v,Gg\rangle
   &=\langle v,g\rangle-\langle J_\partial\gamma v,Lg\rangle,&
 \langle f,GPu\rangle
   &=\langle f,u\rangle-\langle\mathcal K^*f,\Pi\gamma u\rangle.
\end{align*}
For~\eqref{rp_halfq}, replace the occupied annihilation leg by $-a(J_\partial\gamma v)$. The added empty annihilation leg kills $\Omega_W$; commuting it through the finite CAR product contributes exactly $-\langle J_\partial\gamma v,Lg\rangle$. The remaining contraction is the contact derivation $\partial^{\bar q}_v$. For~\eqref{rp_halfb}, replace the empty creation leg by $a^*(\gamma u)$. Its added occupied creation leg kills $\Omega_W$; its anticommutators contribute $\langle\mathcal K^*f,\Pi\gamma u\rangle$. The contact term is $-\partial^{q}_u$. Applying the graded Leibniz rule proves both statements by induction in the displayed Wick recursion. This proof uses no determinant, trace or spectral limit.
\end{proof}

\subsubsection{Cancellation at the two seams}

For halves $x,y$ with common seam data define the independently constructed pairing functional by
\begin{equation}\label{rp_pairfunctional}
 \mathcal B_{xy}((\Theta A)B)=\langle\Psi_x(A),\Psi_y(B)\rangle.
\end{equation}
Reflection is antilinear, reverses order and satisfies $\Theta q=-\bar q_L$, $\Theta\bar q=-q_L$. Retain temporal antiperiodicity. Fold the negative half to $M_+$. The boundary matching for a smooth closed test field is
\begin{equation}\label{rp_sigma}
 \Sigma=\operatorname{diag}(\mathrm I,-\mathrm I)=-J_\partial.
\end{equation}
The minus at $Y_\ell$ is the antiperiodic temporal identification. The folded operator relations are $P_L^\dagger\leftrightarrow P_x$ and $P_L\leftrightarrow P_x^\dagger$.

\begin{lemma}[Closed source identities]\label{rp_closedward}
With exactly this reflection and matching dictionary, $\mathcal B_{xy}$ obeys the closed Gaussian identities
\begin{align}
 \mathcal B_{xy}(q(P_{xy}^\dagger v)F)&=\mathcal B_{xy}(\partial^{\bar q}_vF),\label{rp_closedq}\\
 \mathcal B_{xy}(\bar q(P_{xy}u)F)&=-\mathcal B_{xy}(\partial^{q}_uF).\label{rp_closedb}
\end{align}
They hold for arbitrary globally smooth test fields and all finite source polynomials.
\end{lemma}
\begin{proof}
Take $F=(\Theta A)B$ with $A$ homogeneous of degree $r$, and write $X=\Psi_x(A)$, $Y=\Psi_y(B)$. A right-half insertion crosses $r$ left fields. A left insertion is $-\Theta$ of the corresponding opposite species on $x$, and reversal puts that species at the \emph{right} of $A$.

For a closed $v$, denote its right restriction by $v_y$ and its folded left restriction by $u_x$. Then $\gamma u_x=\Sigma\gamma v_y$. Applying~\eqref{rp_halfq}--\eqref{rp_halfb}, the two contributions to $q(P^\dagger v)$ are
\begin{align*}
 -\langle\Psi_x(A\bar q(P_xu_x)),Y\rangle
  &=(-1)^r\bigl[\langle\Psi_x(\partial^{q}_{u_x}A),Y\rangle
                  -\langle X,a(\gamma u_x)Y\rangle\bigr],\\
 (-1)^r\langle X,\Psi_y(q(P_y^\dagger v_y)B)\rangle
  &=(-1)^r\bigl[\langle X,\Psi_y(\partial^{\bar q}_{v_y}B)\rangle
                  -\langle X,a(J_\partial\gamma v_y)Y\rangle\bigr].
\end{align*}
The boundary coefficient is $\gamma u_x+J_\partial\gamma v_y=0$. The contact terms are precisely the derivative of the full word, because
\[
 \partial^{\bar q}_{v_L}(\Theta A)=(-1)^r\Theta(\partial^{q}_{u_x}A).
\]
The conjugation of the contact coefficient comes from antilinearity of $\Theta$ and the first slot of the pairing. This proves~\eqref{rp_closedq}.

For a closed $u$, let its folded left restriction be $v_x$, so $\gamma v_x=\Sigma\gamma u_y$. The contributions to $\bar q(Pu)$ are
\begin{align*}
 -\langle\Psi_x(Aq(P_x^\dagger v_x)),Y\rangle
 &=(-1)^r\bigl[-\langle\Psi_x(\partial^{\bar q}_{v_x}A),Y\rangle
                   +\langle X,a^*(J_\partial\gamma v_x)Y\rangle\bigr],\\
 (-1)^r\langle X,\Psi_y(\bar q(P_yu_y)B)\rangle
 &=(-1)^r\bigl[-\langle X,\Psi_y(\partial^{q}_{u_y}B)\rangle
                   +\langle X,a^*(\gamma u_y)Y\rangle\bigr].
\end{align*}
Now $J_\partial\gamma v_x+\gamma u_y=0$, and
$\partial^{q}_{u_L}(\Theta A)=(-1)^r\Theta(\partial^{\bar q}_{v_x}A)$.
This proves~\eqref{rp_closedb}. Linearity covers arbitrary $A,B$ and sums of split words.

There is no distributional differentiation of a sharp indicator. First differentiate the smooth closed test field, then restrict it. Its half restrictions are smooth up to the cut, and the boundary terms are those in Green's identity. Lemma~\ref{lem:collar_source_extension} supplies the canonical graded extension to global $L^2$ smearings, independent of the split expansion. Thus the pairing is a functional on the same global smooth source algebra used by the uniqueness lemma.
\end{proof}

\subsubsection{A direct all-source comparison at every index}\label{rp_comparisonsection}

\begin{lemma}[Uniqueness for the closed source identities]\label{rp_uniqueness}
Let $P$ be a smooth closed elliptic Weyl operator. The vector space of scalar linear functionals on its finite Grassmann source algebra satisfying~\eqref{rp_closedq}--\eqref{rp_closedb} is one dimensional. No invertibility or index-zero assumption is needed.
\end{lemma}
\begin{proof}
Let $Z_+=\ker P$, $Z_-=\ker P^\dagger$, and let $G_{\rm red}$ be the inverse on their orthogonal complements. For smooth profiles,
\[
 f=\Pi_{Z_+}f+P^\dagger G_{\rm red}^\dagger f,
 \qquad g=\Pi_{Z_-}g+PG_{\rm red}g.
\]
Use the source identities to eliminate every high-mode field. Each elimination replaces a product by contact terms of degree two less, so a finite word needs finitely many steps. On the remaining finite zero-mode algebra, choose $u\in Z_+$ or $v\in Z_-$ in the same identities. The result is $\Lambda(\partial^{q}_uF)=\Lambda(\partial^{\bar q}_vF)=0$. Hence $\Lambda$ vanishes on every zero-mode word which omits any kernel or cokernel generator: multiply such a word by its missing generator and differentiate. Only the single completely saturated coefficient remains free. Its value determines every original coefficient by the preceding recursion.

The usual physical Gaussian source section gives a nonzero example: its saturated coefficient is the nonzero high determinant times the two top zero-mode exterior factors. Multiplication by a background-dependent BF phase does not affect these source identities.
\end{proof}

\begin{proposition}[Source independence without finite boundary windows]\label{rp_allsource}
Assume the Green and reflection dictionary above, unique continuation, and the common charged restricted Fock realization at fixed seam data. Then there is a nonzero scalar $c(x,y)$, independent of the source polynomial, such that
\begin{equation}\label{rp_sourcecomparison}
 \mathsf S_{xy}(A,B)=c(x,y)\langle\Psi_x(A),\Psi_y(B)\rangle.
\end{equation}
This includes nonzero index and identically vanishing vacuum coefficients. It does not yet assert $|c|=1$.
\end{proposition}
\begin{proof}
The physical source functional and $\mathcal B_{xy}$ satisfy the same closed identities by Lemmas~\ref{rp_closedward}--\ref{rp_uniqueness}. It remains to exclude $\mathcal B_{xy}=0$.

For any interior ball $O$, unique continuation makes
\[
 \mathcal K_x^*:C_c^\infty(O)\longrightarrow W_x,
 \qquad L_x:C_c^\infty(O)\longrightarrow W_x^\perp
\]
dense. For $\mathcal K^*$ an annihilator is a homogeneous Poisson solution vanishing on $O$. For $L$ use Green's identity and the adjoint Poisson solution; vanishing on $O$ again implies vanishing everywhere. The full range of $L$ is dense in $W^\perp$: for smooth boundary $h\in W^\perp$, extend it smoothly to $u$; then $GPu=u$, so $LPu=h$.

Finite products of the corresponding occupied annihilators and empty creators span a dense set of particle-hole vectors. Half-Wick products have these CAR products as their top term and only lower-degree contractions. Induction therefore gives the same dense polynomial span. Since all Fock charge components are included, this applies to both halves even when their vacuum charges differ. Their images cannot be mutually orthogonal, so some pairing is nonzero. Expanding the two polynomials shows that a pair of monomials already has a nonzero pairing. Uniqueness now gives a nonzero proportionality scalar.

At a fixed nonzero physical source coefficient, proportionality ensures that the corresponding pairing is nonzero as well. Their ratio supplies a local source-independent coordinate. This identifies transitions through kernel jumps without differentiating a jumping kernel basis.
\end{proof}

\subsubsection{Quillen normalization from minimal saturated sources}\label{rp_normsection}

Write $\mathcal D=\left(\begin{smallmatrix}0&P^\dagger\\P&0\end{smallmatrix}\right)$ and let $\mathcal C_x$ be the orthogonal projection onto the full doubled Cauchy space on a half. The self-adjoint realization $\mathcal D_{x,\mathcal C_x}$ has domain $\{u\in H^1:\mathcal C_x\gamma u=0\}$; its square has the corresponding iterated domain. From now on fix
\begin{equation}\label{rp_halfnorm}
 h_x=\det_\zeta(\mathcal D_{x,\mathcal C_x}^{\,2})^{1/4}>0.
\end{equation}
These boundary problems are invertible: a zero mode has its trace both in the Cauchy space and in its orthogonal complement, hence has zero trace and vanishes by unique continuation.

For a sewing set $T_{xy}=\Pi_x|_{W_y}$ and
\[
 b_{xy}=\det_F\bigl(T_{xy}^*T_{xy}+\Pi_{\ker T_{xy}}\bigr)^{1/2}.
\]
The projection difference is smoothing at common seam data, so the operator inside is identity plus trace class on $W_y$. This definition is the product of the nonzero singular values of $T_{xy}$.

\begin{lemma}[Trace spaces of closed zero modes]\label{rp_tracezeros}
Let $u_\alpha$ and $v_\beta$ be bulk-orthonormal bases of $\ker P_{xy}$ and $\ker P_{xy}^\dagger$. With right-half restriction understood, put
\[
 w_{+,\alpha}=\gamma u_\alpha,\quad
 w_{-,\beta}=J_\partial\gamma v_\beta,\quad
 \mathsf H_\pm=(\langle w_{\pm,a},w_{\pm,b}\rangle).
\]
Their ranges are respectively $W_y\cap W_x^\perp$ and $W_x\cap W_y^\perp$. Both Gram matrices are positive definite.
\end{lemma}
\begin{proof}
Green's identity identifies the adjoint half Cauchy space as $J_\partial W^\perp$. Folding the negative half and imposing $\Sigma=-J_\partial$ gives the first intersection for $P$ and the second for $P^\dagger$. Conversely, a trace in either intersection glues the corresponding Poisson solutions into a closed weak solution, which is smooth by elliptic regularity. Restriction is injective by unique continuation, proving positivity of the Grams. Multiplication by $J_\partial$ is unitary and therefore does not change the adjoint trace Gram.
\end{proof}

\begin{lemma}[Loya--Park specialization]\label{rp_metric}
At every sewing, including nonzero index,
\begin{equation}\label{rp_normglue}
 N_{xy}:=\det{}'_\zeta(P_{xy}^\dagger P_{xy})^{1/2}
 =h_xh_y b_{xy}(\det\mathsf H_+\det\mathsf H_-)^{-1/2}.
\end{equation}
Empty determinants are one. No invertibility of the tangential seam Dirac operator is assumed.
\end{lemma}
\begin{proof}
Loya--Park Theorem 1.1, with their zero-mode convention, states~\cite{loya_park_gluing}
\begin{equation}\label{rp_LPexact}
 \frac{\det{}'_\zeta\mathcal D^2}
 {\det_\zeta\mathcal D_{\mathcal C_x}^2\det_\zeta\mathcal D_{\mathcal C_y}^2}
 =2^{-\zeta_{\mathscr D_Y^2}(0)-h_Y}
   (\det \mathsf L)^{-2}
   \det_F\left(\frac{2\mathrm I+U_b+U_b^{-1}}4\right).
\end{equation}
Here each half uses its own orthogonal Calder\'on projector $\mathcal C_\pm$, with boundary condition $\mathcal C_\pm\gamma u=0$; neither boundary problem depends on the partner. The operator $U_b$ omits the $-1$ eigenspace, $h_Y=\dim\ker\mathscr D_Y$, and $\mathsf L=BB^*$ on its range, where $B$ restricts a bulk-orthonormal basis of $\ker\mathcal D$ to $Y$. The hypotheses are: a compact closed manifold split along a smooth closed hypersurface, a self-adjoint Dirac-type double, product metric and coefficients on a full collar, and the orthogonal Calder\'on projectors. They hold here; $Y$ may have the two components retained throughout this subsection.

For the unitary dictionary, identify the adjoint trace component with $J_\partial$ times itself. The half doubled Cauchy space becomes $W\oplus W^\perp$. For its normal Clifford matrix and eigenspace identifications use
\[
 G_\partial=\begin{pmatrix}0&-\mathrm I\\\mathrm I&0\end{pmatrix},
 \qquad \Gamma_\pm h=2^{-1/2}(h,\mp ih),\qquad S_W=2\Pi_W-\mathrm I.
\]
If $(f,g)\in W\oplus W^\perp$, then its $\pm i$ coordinates satisfy
$h_-=(f-ig)/\sqrt2=S_Wh_+$, with $h_+=(f+ig)/\sqrt2$.
In LP's common normal splitting,
\[
 \mathcal C_\pm=\frac12
 \begin{pmatrix}\mathrm I&\kappa_\pm^{-1}\\\kappa_\pm&\mathrm I\end{pmatrix},
 \qquad \kappa_+=S_{W_y},\quad\kappa_-=-S_{W_x},\quad
 U=-\kappa_-\kappa_+^{-1}=S_{W_x}S_{W_y}.
\]
The minus sign in $\kappa_-$ comes from the reflected negative-half space $W_x^\perp\oplus W_x$. On a two-dimensional principal-angle block its eigenvalues are $e^{\pm2i\alpha}$, so the determinant of the LP block is $\cos^4\alpha$. A fourth root gives $\cos\alpha$. Each unmatched kernel or cokernel trace gives a $-1$ eigenvector and is omitted from $U_b$. Thus the fourth root of the Fredholm factor is $b_{xy}$.

The nonzero eigenvalues of $BB^*$ equal those of $B^*B=\mathsf H_+\oplus\mathsf H_-$. The factor $(\det\mathsf L)^{-2}$ gives the exponent $-1/2$ on the two Grams after the fourth root. Finally the closed three-dimensional seam has no constant local heat coefficient for this Dirac Laplacian, so $\zeta_{\mathscr D_Y^2}(0)=-h_Y$. The power of two is exactly one, with its zero modes retained. Reflection unitarily identifies the negative half's doubled boundary Laplacian with that of $x$, giving its factor $h_x$; the positive half gives $h_y$. Nonzero spectra of $P^\dagger P$ and $PP^\dagger$ agree, so $(\det{}'_\zeta\mathcal D^2)^{1/4}=N_{xy}$. This proves~\eqref{rp_normglue}. The factors $h_x,h_y$ belong to the individual halves, $b_{xy}$ is their genuine overlap factor, and $\mathsf H_\pm$ retain the closed zero-mode trace normalization. The purely seam-dependent power of two has cancelled; no partner-dependent factor has been absorbed into a half norm.
\end{proof}

\begin{lemma}[Saturated CAR cofactor with its bulk normalization]\label{rp_saturated}
Put $r_+=\dim\ker P_{xy}$ and $r_-=\dim\ker P_{xy}^\dagger$. Choose all $r_+$ row sources $f_i$ and all $r_-$ column sources $g_j$ in one interior ball of the right half, so that
\[
 (\mathsf A_+)_{i\alpha}=\langle f_i,u_\alpha\rangle,
 \qquad (\mathsf A_-)_{\beta j}=\langle v_\beta,g_j\rangle
\]
are invertible. Let $F_R$ be this minimally saturated word in any one fixed order. Then
\begin{align}
 |\langle\Psi_x(1),\Psi_y(F_R)\rangle|
 &=h_xh_y b_{xy}
   \frac{|\det\mathsf A_+\det\mathsf A_-|}
        {(\det\mathsf H_+\det\mathsf H_-)^{1/2}},\label{rp_CARnorm}\\
 \|\sigma_{xy}[F_R]\|_Q
 &=N_{xy}|\det\mathsf A_+\det\mathsf A_-|.\label{rp_physicalnorm}
\end{align}
\end{lemma}
\begin{proof}
Unique continuation makes the restrictions of each zero-mode basis linearly independent on the chosen ball; compact interior dual tests therefore give the two invertible evaluation matrices.

The boundary-normalized unmatched frames are $w_+\mathsf H_+^{-1/2}$ and $w_-\mathsf H_-^{-1/2}$. The annihilation legs have the actual evaluations
\[
 \langle \mathcal K_y^*f_i,w_{+,\alpha}\rangle
   =\langle f_i,u_\alpha\rangle.
\]
For the creation legs Green's identity gives the equally important adjoint formula
\[
 \langle w_{-,\beta},L_yg_j\rangle
 =\langle\gamma v_\beta,J_\partial\gamma G_yg_j\rangle
 =\langle v_\beta,P_yG_yg_j\rangle
 =\langle v_\beta,g_j\rangle.
\]
Hence their boundary-normalized evaluation matrices are $\mathsf A_+\mathsf H_+^{-1/2}$ and $\mathsf H_-^{-1/2}\mathsf A_-$.

In the principal-angle decomposition, the right vacuum occupies precisely $r_+$ unmatched modes absent from the left, and lacks precisely $r_-$ unmatched modes occupied on the left. A nonzero matrix element must annihilate every former mode and create every latter mode. The word is minimal, so all of its CAR legs must be used for these tasks. Every internal Wick contraction removes two legs and leaves an unmatched mode unfilled; those terms vanish. Components of the legs in the matched angle blocks vanish for the same reason. The remaining finite exterior coefficient is the product of the two evaluation determinants, up to its fixed order sign. The matched blocks contribute $\prod\cos\alpha=b_{xy}$. This proves~\eqref{rp_CARnorm}.

This is also a direct infinite-wedge argument: remove the finitely many unmatched modes, truncate the principal-angle tail, and let the cutoff grow. The normalized Slater vectors converge in Fock norm because the squared sines are summable; finite CAR products are bounded; and the remaining product of cosines converges. No bulk zeta determinant or inverse is compressed in this limit.

For the physical source, the reduced closed inverse supplies the nonzero-mode Gaussian. Minimal saturation leaves exactly the two top zero-mode exterior evaluations; any high-mode contraction again leaves a zero mode unsaturated. In bulk-orthonormal bases their absolute product is $|\det\mathsf A_+\det\mathsf A_-|$, while the high Quillen factor is $N_{xy}$. This proves~\eqref{rp_physicalnorm}.
\end{proof}

\begin{corollary}[Unit-modulus source comparison]\label{rp_unitcomparison}
The source-independent $c(x,y)$ in~\eqref{rp_sourcecomparison} has $|c(x,y)|=1$ at every sewing. One minimal nonzero coefficient suffices at each point, even where the uninserted vacuum vanishes. Fixed invertible source borders give smooth ratios through kernel jumps.
\end{corollary}
\begin{proof}
Equations~\eqref{rp_normglue}--\eqref{rp_physicalnorm} make the two nonzero magnitudes equal. The specified closed BF trivialization is unitary. Source independence then gives the same $c$ for every coefficient. At a jump choose an invertible fixed source border, rather than differentiating the jumping minimal kernel basis.
\end{proof}

\paragraph{The reflected diagonal.}
Folding a closed zero mode on $\Theta x\cup_b x$ gives $P_xu=0$, $P_x^\dagger v=0$, with $v(0)=u(0)$ and $v(\ell)=-u(\ell)$. Green's identity gives
\[
 0=-\|u(\ell)\|^2-\|u(0)\|^2.
\]
Unique continuation and reflection exclude kernel and cokernel. Thus every reflected double is invertible. In~\eqref{rp_normglue} it has $b_{xx}=1$ and no trace Grams, giving $\|\sigma_{xx}[1]\|_Q=h_x^2$. This also proves smooth dependence of the half norm from an invertible closed determinant family.

To fix its phase, choose an absolute half frame, product near its ends, and reflect it to a global frame on the double. Scale the connection to zero in this frame. The entire path consists of invertible reflected doubles. The linear unitary temporal pullback $\mathscr R$ obeys $\mathscr RP\mathscr R^{-1}=P^\dagger$ and $\mathscr RV\mathscr R^{-1}=V^\dagger$. For large $\Re z$, adjunction, this pullback and intertwining give
\[
 \overline{\operatorname{Tr}(P^{-1}V(P^\dagger P)^{-\bar z})}
 =\operatorname{Tr}((PP^\dagger)^{-z}VP^{-1})
 =\operatorname{Tr}(P^{-1}V(P^\dagger P)^{-z}).
\]
The continued finite part is real. The BF differential proved in Section~\ref{rp_BFsection} therefore leaves the phase constant. Its free value is positive by the specified prescription. Consequently
\begin{equation}\label{rp_diagonal}
 \mathsf S_{xx}(1,1)=h_x^2>0,\qquad c(x,x)=1.
\end{equation}
This uses no spatial antiperiodicity and makes no assumption that the seam operator is invertible.
\subsubsection{Single monomials and their source planes}

Write $V_x[F]=\Psi_x(F)$ for the same normalized half vector.

\begin{lemma}[A monomial prepares one modified source plane]\label{rp_sourceplane}
Let a half monomial have $q$ row sources $R_+$ and $p$ column sources $R_-$. In the orthogonal half coordinates put
\[
 \mathsf R=R_+\mathcal K_x:W_x\to\mathbb C^q,\qquad
 \mathsf T=\gamma G_xR_-:\mathbb C^p\to W_x^\perp,\qquad
 \mathsf C=R_+G_xR_-.
\]
Define $\Lambda=(\mathsf R,-\mathsf C)$ and $\mathsf H=(\mathrm I,-\mathsf T)$ on $W_x\oplus\mathbb C^p$. The following are equivalent:
\begin{enumerate}
 \item $V_x[F_x]\ne0$;
 \item $\Lambda$ is onto and $\mathsf H|_{\ker\Lambda}$ is injective.
\end{enumerate}
When they hold, the half state is a nonzero scalar times the unit Slater ray with occupied plane
\begin{equation}\label{rp_modifiedplane}
 \mathscr P_x^F=\{w-\mathsf Tz:\ \mathsf Rw-\mathsf Cz=0\}.
\end{equation}
Its projection differs from $\Pi_x$ by a finite-rank smoothing operator and
\[
 \ind(\Pi_x|_{\mathscr P_x^F})=p-q.
\]
The reflected constrained trace plane is exactly $(\mathscr P_x^F)^\perp$. For two nonzero monomials of equal resulting Fock charge, the pairing is nonzero exactly when their original closed source border is invertible.
\end{lemma}
\begin{proof}
The forced homogeneous equation and source constraints are
\[
 u=\mathcal K_xw-G_xR_-z,\qquad R_+u=0.
\]
Its trace is precisely~\eqref{rp_modifiedplane}; the plane depends on the actual half and its sources, not on a sewing partner. A nonzero sewn border implies that $\Lambda$ is onto and this constrained trace is injective: prescribe its source rows arbitrarily for surjectivity, and set the other half to zero to test injectivity.

Here is a finite exterior proof of the stronger assertion about the state itself. Let $E=\operatorname{ran}\mathsf R^\dagger\subset W_x$, $E'=\operatorname{ran}\mathsf T\subset W_x^\perp$, and $r=\dim E$. Occupation on the orthogonal complement of $E\oplus E'$ is unchanged; this lemma requires no ultraviolet limit. On this finite window the half-Wick coefficient is, up to its fixed exterior-order sign, the component with $q$ constraint legs in the exterior image of
\begin{equation}\label{rp_finitepush}
 \begin{pmatrix}
  \mathsf R&-\mathsf C\\
  \mathrm I_E&0\\
  0&-\mathsf T
 \end{pmatrix}:
 E\oplus\mathbb C^p\longrightarrow
 \mathbb C^q\oplus E\oplus E'.
\end{equation}
Precisely, apply its $(r+p)$th exterior power to the occupied volume of $E$ wedged with the standard volume of $\mathbb C^p$, then extract the component in $\det\mathbb C^q\otimes\bigwedge^{r+p-q}(E\oplus E')$.
Indeed selecting constraint legs reproduces the contractions $\mathsf C$ and the remaining CAR legs in the half-Wick recursion; this is ordinary finite exterior Cauchy--Binet. If $\Lambda$ is not onto the component vanishes. If it is onto, choose a complement to its kernel. Taking all $q$ constraint legs leaves the wedge of the images of a basis of $\ker\Lambda$ under $\mathsf H$, multiplied by a nonzero constraint determinant. This wedge is nonzero exactly when the constrained trace is injective, and its ray is~\eqref{rp_modifiedplane}. Thus the state is a single Slater ray, rather than a general superposition. This calculation identifies a ray; all coefficients and Berezin signs still come from the original source map above.

For example, take one occupied direction $e$, one empty direction $f$, and one row and column, with $\mathsf Re=\rho$, $\mathsf C=\kappa$, $\mathsf T1=\xi f$. The two columns in~\eqref{rp_finitepush} are $(\rho,e,0)$ and $(-\kappa,0,-\xi f)$. Extracting their one-constraint-leg component gives, up to the common order sign,
\[
 \kappa e-\rho\xi f.
\]
The constraint is $\rho w-\kappa z=0$, whose trace $we-\xi zf$ spans exactly this ray when it is nonzero. If $(\rho,\kappa)=(0,0)$ the constraint is not onto and the extracted component vanishes. This displays both the contraction and the particle-hole leg in the same finite wedge.

The plane contains $W_x\cap\ker\mathsf R$ and lies in $W_x+\operatorname{ran}\mathsf T$. These are finite modifications by smooth boundary legs. It is closed and its projection difference is finite-rank smoothing. The projection onto $W_x$ has kernel dimension $p-\operatorname{rank}\mathsf C$ and cokernel dimension $q-\operatorname{rank}\mathsf C$, since $\operatorname{ran}\mathsf R+\operatorname{ran}\mathsf C=\mathbb C^q$. This proves the index formula.

The reflection assertion has an explicit sign check. In the fixed folded convention the matching identification is $\Sigma=-J_\partial=\operatorname{diag}(\mathrm I,-\mathrm I)$, while Green's boundary form uses $J_\partial=\operatorname{diag}(-\mathrm I,\mathrm I)$. Green's identity gives
\[
 G_L=G_x^\dagger,\qquad \mathcal K_L=-L_x^\dagger,\qquad L_L=\mathcal K_x^\dagger.
\]
Reflection takes the source maps to $R_{-,L}=-R_+^\dagger$ and $R_{+,L}=-R_-^\dagger$. Its three finite blocks are consequently
\[
 (\mathsf R_L,\mathsf T_L,\mathsf C_L)
   =(\mathsf T^\dagger,-\mathsf R^\dagger,\mathsf C^\dagger).
\]
On the other hand, taking the annihilator of $\ker(\mathsf R,-\mathsf C)$ directly gives
\begin{equation}\label{rp_dualplane}
 (\mathscr P_x^F)^\perp
 =\{v+\mathsf R^\dagger\eta:\ v\in W_x^\perp,
                       \mathsf T^\dagger v=\mathsf C^\dagger\eta\},
\end{equation}
exactly the reflected trace plane. This check includes conjugations and holds for $p\ne q$. The reflected rank conditions follow too: $\binom{\mathsf T}{\mathsf C}$ is injective and $(\mathsf R,\mathsf C)$ is onto.

Finally eliminate the two halves' finite source constraints. The closed matching border is invertible exactly when matching on $(\mathscr P_x^F)^\perp\oplus\mathscr P_y^F$ is invertible. For equal charges that is equivalent to invertibility of $\Pi_{\mathscr P_x^F}|_{\mathscr P_y^F}$, hence to a nonzero Slater overlap. Thus simultaneous nonzero source charts can be checked without assuming a Quillen metric comparison. Smoothness on such charts follows from the smooth full-rank constraint maps and invertible matching border; no constant rank of $\mathsf R$ or $\mathsf C$ individually is required.
\end{proof}

\subsubsection{Simultaneous charts for the triangle}\label{rp_chartsection}

Define
\begin{equation}\label{rp_triangle}
 \tau(x,y,z)=c(x,y)c(y,z)c(z,x).
\end{equation}
It is invariant under all local half-phase choices. The crucial point is that \emph{only this triangle} needs transport. Constancy of the individual $c(x,y)$ under interior variations is unnecessary.

Use the six pairs $p=(x,y),(y,z),(z,x),(x,x),(y,y),(z,z)$ with signs $\epsilon_p=(1,1,1,-1,-1,-1)$. Every actual oriented half occurs once with each sign. On a simultaneous nonzero source chart put
\[
 \mathfrak S=\prod_p S_p^{\epsilon_p},\qquad
 \mathfrak B=\prod_p B_p^{\epsilon_p},\qquad
 B_{ij}=\langle V_i[F_i],V_j[F_j]\rangle.
\]
Here $S_{ij}=\mathsf S_{ij}(F_i,F_j)$. Then $\mathfrak S/\mathfrak B=\tau$, since all three diagonal comparison factors equal one.

\begin{lemma}[One-ball monomial charts, with prescribed charge]\label{rp_monomialcharts}
Fix an open ball $O$ strictly inside a connected half and a target smooth Slater plane whose projection differs from $\Pi_x$ by a smoothing operator. A ray for that plane can be approximated in restricted Fock norm by rays of single monomials, with every source supported in $O$ and with the target charge. This holds over the real numbers whenever the half and boundary spaces carry compatible antiunitary involutions commuting with $P_x$ and its Green and Poisson maps. Lemma~\ref{rp_realref} verifies these hypotheses at the references used below.

For finitely many halves at common seam data, possibly in unequal relative classes, there are monomials of one resulting Fock charge for which every pair overlap and every original closed source border is nonzero. These source profiles can be frozen on a neighborhood of the chosen configuration. Finitely many such neighborhoods cover any compact triple path.
\end{lemma}
\begin{proof}
\emph{Complex source charts.} Unique continuation makes both interior leg maps have dense range:
\[
 \mathcal K_x^\dagger:C_c^\infty(O;E^+)\longrightarrow W_x,
 \qquad
 L_x:C_c^\infty(O;E^-)\longrightarrow W_x^\perp.
\]
For the first assertion an annihilator is a homogeneous Poisson solution vanishing on $O$. For the second, Green's identity identifies the adjoint leg map, up to the fixed boundary sign, with the adjoint Poisson solution; the same argument applies. No statement about dense \emph{polynomial spans} is being used.

First approximate the target ray by a finite modification of $W=W_x$. Principal-angle truncation retains all unmatched occupied or empty directions and finitely many angle blocks. The omitted smoothing tail converges in Fock norm, since its squared sines are summable. Smoothness of all finite legs follows from the smoothing projection difference. Thus it suffices to prepare a plane
\[
 (W\ominus E)\oplus\mathscr L,\qquad
 E\subset W,\quad E'\subset W^\perp,\quad
 \mathscr L\subset E\oplus E',\quad
 \dim E=r,\quad \dim\mathscr L=p.
\]
Its charge relative to $W$ is $p-r$. Enlarge $E'$ by fresh smooth empty directions and perturb $\mathscr L$ arbitrarily little so that its projection into $E'$ is injective. Explicitly add distinct small fresh empty components to a basis of $\mathscr L$. This also covers a plane containing pure occupied directions. The case $p=0$ needs no perturbation.

Choose coordinates on $E$ and a basis of the perturbed $\mathscr L$. Write its basis columns as $(\mathsf C,-\mathsf T)$, with $\mathsf T:\mathbb C^p\to E'$ injective, and set $q=r$ and $\mathsf R$ equal to the coordinate projection $W\to E\cong\mathbb C^q$. Then
\[
 \{w-\mathsf Tz:\mathsf Rw=\mathsf Cz\}
       =(W\ominus E)\oplus\mathscr L.
\]
Lemma~\ref{rp_sourceplane} therefore gives its nonzero ray from these finite data. Approximate the rows of $\mathsf R$ and the columns of $\mathsf T$ using the dense interior leg maps. Choose the row approximations independent, as the desired rows are independent. Let $R_+$ and $R_-^0$ be the resulting actual smooth source maps in $O$.

The contraction block can be prescribed \emph{exactly}, without changing those boundary legs. The independent rows make $R_+:C_c^\infty(O;E^+)\to\mathbb C^q$ onto. For each column choose $\chi_j\in C_c^\infty(O;E^+)$ with
\[
 R_+\chi_j=\mathsf C_{\cdot j}-(R_+G_xR_-^0)_{\cdot j},
 \qquad g_j=g_j^0+P_x\chi_j.
\]
Compact support gives $G_xP_x\chi_j=\chi_j$ and $\gamma\chi_j=0$. Consequently
\begin{equation}\label{rp_tunecontraction}
 \gamma G_xg_j=\gamma G_xg_j^0,\qquad
 (R_+G_xR_-)_{\cdot j}=\mathsf C_{\cdot j}.
\end{equation}
The rows and columns remain supported in $O$. The finite half-Wick state is continuous in $\mathsf R,\mathsf T,\mathsf C$ because it is a finite polynomial in contractions and bounded CAR operators. It remains nonzero and its ray approaches the target ray as the leg approximations improve. Its relative charge is still $p-q$, as required. If $q=0$ the contraction adjustment is empty; if $p=0$ the construction uses only rows.

For finitely many halves choose one common smooth target ray of any fixed resulting charge. The collar condition makes their Cauchy projections differ by smoothing operators, so this target is in the restricted Grassmannian for every half. Apply the construction separately and align the approximating unit vectors within distance $1/4$ of the target. Each overlap then differs from one by less than $1/2$, so all overlaps are nonzero. Lemma~\ref{rp_sourceplane} identifies this condition with invertibility of the original source borders, independently of the Quillen comparison.

\emph{Real source charts.} On a real family denote its antiunitary involution by $J_{\mathbb R}$, distinct from the boundary-sign matrix $J_\partial$. Average interior leg preimages with $J_{\mathbb R}$. Use real finite bases and fresh empty directions; the independent real rows have real compactly supported duals. The correction~\eqref{rp_tunecontraction} is then real too, because $J_{\mathbb R}$ commutes with $P_x$ and $G_x$. Thus the algebraic construction works with real sources on families preserved by this involution. It does not use the later triangle or reference-phase conclusion.

Finally freeze the actual profiles at each chart center. Lemma~\ref{rp_halfregularity} and openness of the original closed bordered inverse preserve the nonzero half vectors and borders on a neighborhood, including across a jump of the unbordered kernel. After a common Kato pullback, the half vectors and overlaps are smooth in the restricted Fock charts. Compactness gives a finite path cover. On the Standard Model family the vacuum charge is $-12n_x$ up to a common origin, so target charge zero means $d(F_x)+12n_x=0$. The full monomial degree has even parity. Unequal class labels cause no additional restriction.
\end{proof}

Use Lemma~\ref{rp_monomialcharts} for each triple on the path. Source independence in Proposition~\ref{rp_allsource} and Corollary~\ref{rp_unitcomparison} makes the scalar triangle agree on chart overlaps. Thus the subsequent derivative calculation never divides by a vacuum coefficient that may vanish.

For changing $b$, these are charts of a Fock bundle. Choose a smooth common reference projector and its Kato unitary on the one-particle boundary space. Pulling all half data back by this same unitary gives smoothing differences from a fixed polarization and smooth restricted-Fock charts. No assertion implements that unitary inside one fixed Fock representation. This distinction is used explicitly in the relative determinant below.

\subsubsection{BF normalization and the physical bordered coefficient}\label{rp_BFsection}

BF I Definition 1.10, Eq.~(1.41), uses a real metric parameter denoted there by $\mu$. Rename it $\kappa_{\rm BF}$ here, to distinguish it from the fixed spectral scale. Its high-window metric factor is
\[
 \exp\{-\tfrac12\zeta_\lambda'(0)-\tfrac12\kappa_{\rm BF}A_0\}.
\]
Thus $\kappa_{\rm BF}=0$ is the ordinary zeta--Quillen metric. Definition 1.17 and Theorem 1.18 give the high connection term~\cite{bf1}
\begin{equation}\label{rp_bfhigh}
 \omega_{\rm hi}^\lambda
 =\tfrac12(\gamma_0^\lambda-\delta_0^\lambda)
   +\tfrac12\Gamma'(1)(dA_0-B_0).
\end{equation}
On a real tangent $V=dP$, block multiplication of
$\mathcal D=\left(\begin{smallmatrix}0&P^\dagger\\P&0\end{smallmatrix}\right)$ gives
\[
 a(r;V)=\operatorname{Tr}(G_{\rm hi}Ve^{-rP^\dagger P}),\qquad
 \gamma_r=2\Re a(r;V),\quad\delta_r=-2i\Im a(r;V).
\]
The order $(\nabla\mathcal D)\mathcal D$ is the one in BF Definition~1.14, Eq.~(1.47). Its ordinary diagonal blocks are $V^\dagger P$ and $VP^\dagger$.

For the sign conversion, write forms on the left. The graded product is
\[
 (\alpha\otimes A)(\beta\otimes B)
 =(-1)^{|A|\deg\beta}(\alpha\wedge\beta)\otimes AB.
\]
Since $\mathcal D^{-1}$ is odd, on a one-parameter path it gives
$\mathcal D^{-1}(\nabla\mathcal D)
=-ds\otimes(\mathcal D^{-1}\nabla_s\mathcal D)$.
Consequently the minus ordinary trace and plus supertrace in BF Eq.~(1.49) become respectively the ordinary matrix trace and minus supertrace of $\mathcal D^{-1}\nabla_s\mathcal D$. With $S^+$ positive these give the displayed $\gamma_r,\delta_r$. On a finite band the calibration is
\[
 \frac{\nabla\det P}{\det P}=\operatorname{Tr}(P^{-1}\nabla P);
 \qquad P=re^{i\theta}\ \Longrightarrow\ dr/r+i\,d\theta.
\]
This is BF's finite-band identity, Eq.~(1.53), in ordinary matrix notation. In a parallel frame $e_\parallel$, $\nabla(se_\parallel)=(ds)e_\parallel$, so passing to a scalar coefficient adds no minus sign. In BF I Theorem 1.15, Eq.~(1.50), the logarithmic coefficients of $\gamma_r$ and $\delta_r$ are respectively $dA_0$ and $B_0$. Since $2a=\gamma_r-\delta_r$, the logarithmic coefficient is exactly
$L(V)=\tfrac12(dA_0-B_0)(V)$ in $a(r;V)=\text{powers}+L\log r+C_\lambda+o(1)$.
For $\mathcal W_\lambda=P^\dagger P+\pi_\lambda^+$ the low trace is zero because $G_{\rm hi}$ has high range. Therefore
\begin{align}
 \operatorname{Tr}(G_{\rm hi}V\mathcal W_\lambda^{-z})
  &=\frac1{\Gamma(z)}\int_0^\infty r^{z-1}a(r;V)\,dr\notag\\
  &=-L/z+C_\lambda+\Gamma'(1)L+O(z),\label{rp_mellin}\\
 \omega_{\rm hi}^\lambda(V)
  &=\operatorname{FP}_{z=0}\operatorname{Tr}(G_{\rm hi}V\mathcal W_\lambda^{-z}).\label{rp_bfexact}
\end{align}
Indeed $\int_0^1r^{z-1}\log r\,dr=-1/z^2$ and
$1/\Gamma(z)=z-\Gamma'(1)z^2+O(z^3)$. Negative powers and the large-$r$ tail contribute no constant after this multiplication. The logarithm is retained; choosing a parameter to suppress it would change the metric normalization.

For fixed smooth source borders,
\[
 P_R=\begin{pmatrix}P&R_-\\R_+&0\end{pmatrix},\quad
 q-p=\ind P,\quad E_R=(P_R^{-1})_{\rm bulk},
\]
finite Schur Jacobi in the \emph{original} determinant line gives
\begin{equation}\label{rp_relativejacobi}
 d\log(S_R/S_T)=\operatorname{Tr}((E_R-E_T)dP).
\end{equation}
The high factor cancels, and the inverse difference is finite-rank smoothing. This is not a new determinant prescription for the augmented operator.

At a base point choose auxiliary columns spanning the cokernel and rows evaluating the kernel, both orthonormal and frozen. Choose the cut below the first positive singular value there. Then $E_{R^0}=G_{\rm hi}$ and the finite low/source block is
\[
 F_{R^0}=\begin{pmatrix}0_{p\times q}&\mathrm I_p\\\mathrm I_q&0_{q\times p}\end{pmatrix}.
\]
Its inverse is off diagonal, while the physical low variation is upper left, so its Jacobi trace is zero. Kato low-frame connections vanish at the base point; the two high/low Schur couplings vanish there, so their differentiated product is zero. These vanishings concern the full complex Jacobi trace, not just its imaginary part. Equations~\eqref{rp_bfexact}--\eqref{rp_relativejacobi} consequently give, in a parallel unit BF frame,
\begin{equation}\label{rp_bfsource}
 d\log S_R=\operatorname{FP}_{z=0}\operatorname{Tr}(E_R\,dP\,\mathcal W_\lambda^{-z}).
\end{equation}
Its real part is $d\log\|\sigma_P[F_R]\|_Q$, by metric compatibility; its imaginary part is $d\arg S_R$. The positive sign is also the ordinary finite-band Jacobi sign in BF I Eqs.~(1.51)--(1.53), after the displayed ordinary/supertrace dictionary. This calculation is local on invertible source charts and needs no BF holomorphic theorem. The logarithmic term is retained at $\kappa_{\rm BF}=0$.

\paragraph{Product-cylinder consistency check.}
The separated temporal model for a constant tangential mode checks the overlap and phase conventions. On $[0,\ell]$ let the tangential eigenvalue be $\lambda\in\mathbb R$ and the unitary temporal transport on half $j$ have phase $\varphi_j$. Put $s=\ell\lambda$, $q=e^{-s}$ and $\Delta=\varphi_y-\varphi_x$. Its normalized Cauchy vector is $v_j=(1,q e^{-i\varphi_j})/\sqrt{1+q^2}$. On the antiperiodic sewn circle the temporal spectrum is
\[
 \lambda+\frac{i[2\pi(n+\tfrac12)+\Delta]}{2\ell},
 \qquad n\in\mathbb Z.
\]
The symmetric finite-part prescription, calibrated by $\det_\zeta\partial_t=2$, gives the local coefficient $\mathcal S_\lambda(\Delta)=2\cosh(s+i\Delta/2)$. Directly,
\[
 \frac{|\mathcal S_\lambda(\Delta)|}{\mathcal S_\lambda(0)}
 =\frac{|1+q^2e^{-i\Delta}|}{1+q^2}
 =|\langle v_x,v_y\rangle|,\qquad
 \frac{\mathcal S_\lambda(\Delta)}{\mathcal S_\lambda(0)}
 =e^{i\Delta/2}\langle v_x,v_y\rangle.
\]
In this local holonomy chart the residual phase is a half-phase coboundary. The single mode is not a separately gauge-invariant chiral theory. Using $\sum_{n\in\mathbb Z}^{\rm sym}(n+z)^{-1}=\pi\cot(\pi z)$ in the logarithmic spectral derivative gives, away from the zeros,
\[
 d\log\mathcal S_\lambda
 =\tanh(s+i\Delta/2)\left(\ell\,d\lambda+\frac i2\,d\Delta\right),
\]
with the positive Jacobi sign used in~\eqref{rp_bfsource}. At $\lambda=0$ the Cauchy vector is $(1,1)/\sqrt2$ when $\varphi_j=0$, and the reflected coefficient remains $2$. At $\lambda=0,\Delta=\pi$ the vacuum vanishes and a saturated source chart is required. These mode ratios test the overlap and local phase dictionary; absolute half norms and the full three-dimensional seam factor retain the joint zeta normalization of~\eqref{rp_LPexact}. The check does not replace the moving-collar or four-dimensional logarithmic-term arguments.

\subsubsection{Local cancellation with changing collar data}\label{rp_localizationsection}

Let $p=(xy,yz,zx,xx,yy,zz)$ with $\epsilon=(+,+,+,-,-,-)$. The balance can be read directly as
\[
\begin{array}{c|ccc|ccc}
 p&xy&yz&zx&xx&yy&zz\\ \hline
 \epsilon_p&+&+&+&-&-&-\\
 \text{left half}&\Theta x&\Theta y&\Theta z&\Theta x&\Theta y&\Theta z\\
 \text{right half}&y&z&x&x&y&z
\end{array}
\]
Thus each oriented half occurs once with each sign. Freeze physical source profiles on a chart where all six closed borders are invertible. At each actual oriented half $a$ put
\[
 \Delta_a=\sum_{p\ni a}\epsilon_p E_p^{aa},\qquad V_a=\dot P_a.
\]
\begin{lemma}[Balanced finite parts on fixed cuts]\label{rp_localization}
Suppose the six operators are smooth Dirac-type families with fixed principal symbol and common positive-width product collars. Their finite source borders are smooth and supported in the interiors. Then $\Delta_a$ has a smooth kernel up to the boundary, and
\begin{equation}\label{rp_localtrace}
 \sum_p\epsilon_p\operatorname{FP}_{z=0}\operatorname{Tr}(E_p\dot P_p\mathcal W_p^{-z})
       =\sum_a\operatorname{Tr}(\Delta_a V_a).
\end{equation}
All traces on the right are ordinary. Common collar coefficients may vary.
\end{lemma}
\begin{proof}
An elliptic source border changes the inverse parametrix only by smoothing terms. Thus $E_p\in\Psi^{-1}$, with the full inverse symbol of $P_p$. In a half interior the two occurrences of that actual half have identical complete symbols and opposite signs. Near either cut there is an open full collar on which all six operator germs agree. Their balanced inverse kernel is smooth there. Together with off-diagonal elliptic regularity this proves smoothness up to each half boundary, uniformly on a smaller compact parameter chart.

The weight $\mathcal W_p=P_p^\dagger P_p+\pi_p^+$ is positive invertible of order two. Its complete symbol, and that of $\log\mathcal W_p$, depend only on the local operator germ; low-window modifications are smoothing. Paycha--Scott Corollary 2.5 supplies the exact local density~\cite{paycha_scott}
\[
 \operatorname{FP}\operatorname{Tr}(A\mathcal W^{-z})
   =\int\bigl(\operatorname{TR}_u A
               -\tfrac12\operatorname{res}_{u,0}(A\log\mathcal W)\bigr)du.
\]
Corollary~2.5, Eq.~(2.9), applies to every classical insertion, including the present integer order $-1$. The paragraph following its Eq.~(2.10) states that the corrected combination is a global density independently of the order. Individually $\operatorname{TR}_u A$ need not be invariant here; $\operatorname{res}_{u,0}$ is the log-free degree-$-4$ symbol coefficient. Use matching coordinates, quantizations and smooth cutoffs on half-interior and common-collar patches as a proof convention. In an interior patch pair the two copies of its half; in a collar patch sum all six terms. The paired variations $V=\dot P$ agree with the paired germs, even when the common collar data move. Complete insertion symbols, logarithmic weight symbols and differentiated cutoff symbols therefore cancel. The remaining finite-part density is the ordinary diagonal of the smoothing difference. This also handles all cutoff overlaps: identical differentiated cutoff symbols cancel before a trace is assigned. Equivalently, an additional right cutoff equal to one near the left support changes the operator by an off-diagonal smoothing kernel with zero ordinary diagonal.

Only after this cancellation split the smooth diagonal with the sharp half indicators. On each side of each seam the terms group exactly into the three actual halves present on that side. Smooth partitions and sharp partitions both integrate this same total smooth density; individual half assignments need not agree before regrouping. This proves~\eqref{rp_localtrace}. Changing $b$ alters the common germ on both copies at once, so the argument applies to its tangent too. No sharp cutoff has entered an individual divergent trace.
\end{proof}

For a concrete symbol check, in the flat fixed-symbol setting write $P=p_1(\xi)+p_0(u)$ and $E\sim\sum_{j\ge1}e_{-j}$. The parametrix recursion is
\[
 e_{-1}=p_1^{-1},\qquad
 e_{-(j+1)}=-p_1^{-1}\left(p_0e_{-j}
              +\frac1i\sum_\alpha\partial_{\xi_\alpha}p_1\,
                                      \partial_{u_\alpha}e_{-j}\right).
\]
Every homogeneous term therefore agrees on the identified germs; source borders do not modify it. Matching only $p_1$ would not suffice. This explains precisely what the common-collar hypothesis supplies.

\subsubsection{The transmission trace identity on fixed graph domains}\label{rp_transmissionsection}

This section checks the bulk inverse, source elimination and relative boundary determinant in one set of spaces. The cuts and boundary framings stay fixed while the common collar coefficients vary. In particular, the geometric trace and transmission maps do not acquire a shape derivative.

\begin{proposition}[Transmission identity on a nonzero source chart]\label{rp_transmissionprop}
Let the actual oriented halves be smooth Dirac-type coefficient families with fixed principal symbol and common positive-width product collars. Freeze interior monomial source profiles on a simultaneous nonzero chart. Use the normal reflection dictionary and six signs above. Then the full complex identity~\eqref{rp_compositeidentity} holds, with $\mathfrak B$ the actual six-pairing ratio. All traces in that identity are ordinary. No index-zero restriction is imposed.
\end{proposition}
\begin{proof}[Proof, including the endpoint domains]

\par\medskip\noindent\textbf{The domain of the boundary trace.}\enspace

For an actual oriented half $a$, write $\mathscr H_a^\pm=L^2(M_a;E_a^\pm)$ and put
\[
 \mathscr G_a=\{u\in\mathscr H_a^+:P_{a,0}u\in\mathscr H_a^-\},\qquad
 \|u\|_{\mathscr G_a}^2=\|u\|_{L^2}^2+\|P_{a,0}u\|_{L^2}^2.
\]
Because $P_{a,s}-P_{a,0}$ is zero order, all maximal domains are this same space, with locally uniformly equivalent graph norms. The maximal trace theorem gives a bounded map
\begin{equation}\label{rp_graphtrace}
 \gamma:\mathscr G_a\longrightarrow H^{-1/2}(Y),\qquad
 \|\gamma u\|_{H^{-1/2}}\le C\bigl(\|u\|_{L^2}+\|P_{a,s}u\|_{L^2}\bigr).
\end{equation}
These are exactly~\cite[Eqs.~(2.26)--(2.27)]{booss_lesch_zhu}, at Sobolev order zero. Trace is not used as a bounded map from arbitrary $L^2$ fields. For example, $u_N(t,y)=\sqrt N\eta(Nt)h(y)$, with fixed smooth $h$ and $\eta(0)=1$, has bounded bulk $L^2$ norm but unbounded boundary norm even in $H^{-1/2}(Y)$.

Take $\beta_a$ to be the orthogonal Calder\'on projector at the chart center. The earlier Green identity and unique-continuation argument give its invertible regular half problem. Freeze this projector; the problem stays invertible on a sufficiently small parameter chart. Let $G_a$ be its inverse and let $\mathcal K_a$ be its Poisson map, parametrized by the fixed space $\mathscr B_a=\operatorname{ran}\beta_a$ on boundary $H^{-1/2}$. The endpoint Poisson estimate is~\cite[Thm.~5.6 and Def.~5.9]{booss_lesch_zhu}. All derivatives below are bounded derivatives in the displayed spaces:
\begin{align}
 G_a&:\mathscr H_a^-\longrightarrow H^1\cap\ker(\beta_a\gamma)\subset\mathscr G_a,
 &\mathcal K_a&:\mathscr B_a\longrightarrow\mathscr G_a,\notag\\
 \mathcal K_{a,s}&=\mathcal K_{a,0}-G_{a,s}(P_{a,s}-P_{a,0})\mathcal K_{a,0},
 &\dot G_a&=-G_aV_aG_a,\quad\dot{\mathcal K}_a=-G_aV_a\mathcal K_a.\label{rp_auxvariation}
\end{align}
These are the fixed-auxiliary maps constructed in Lemma~\ref{rp_halfregularity}, not the varying orthogonal maps there. For the last identity, first differentiate the formula for $\mathcal K_{a,s}$, then use the resolvent identity for $G_a$. Compactly supported source columns $R_{-,a}:\mathbb C^{p_a}\to\mathscr H_a^-$ and rows $R_{+,a}:\mathscr H_a^+\to\mathbb C^{q_a}$ are frozen on the chart. Set
\[
 \mathscr D_a=\mathscr B_a\oplus\mathbb C^{p_a},\qquad
 \mathcal A_a(c,z)=\mathcal K_ac-G_aR_{-,a}z.
\]
Thus $\mathcal A_a:\mathscr D_a\to\mathscr G_a$ is bounded, $P_a\mathcal A_a(c,z)=-R_{-,a}z$, and $\dot{\mathcal A}_a=-G_aV_a\mathcal A_a$ in the graph norm.

For a sewing $p$ let $C_p$ stack the fixed trace mismatch and both source-row maps. Its domain is $\mathscr G_L\oplus\mathscr G_R$, and its target is
$\mathscr Y_p=H^{-1/2}(Y)\oplus\mathbb C^{q_L}\oplus\mathbb C^{q_R}$. It is bounded there and independent of the path parameter. Put
\[
 G_0=G_L\oplus G_R,\quad\mathcal A=\mathcal A_L\oplus\mathcal A_R,
 \qquad M_p=C_p\mathcal A:\mathscr D_L\oplus\mathscr D_R\longrightarrow\mathscr Y_p.
\]
Consequently $C_pG_0$ and $C_p\mathcal A$ are bounded. No factorization below applies $C_p$ to an arbitrary $L^2$ field.

\par\medskip\noindent\textbf{The bulk inverse and an ordinary trace.}\enspace

On a nonzero source chart $M_p$ is invertible on these endpoint spaces. To see the Sobolev realization explicitly, at the regular level this is the usual equivalence of the closed source border and transmission. The boundary matching operator is elliptic of order zero between Calder\'on ranges with finite source blocks: on the common collar the two oriented principal Cauchy ranges are complementary, so the mismatch principal symbol is an isomorphism. The source blocks are smoothing. Its kernel and cokernel are smooth, so its Fredholm index and nullities are independent of the boundary Sobolev order. Invertibility at the regular level therefore implies invertibility on $H^{-1/2}$ as well. This uses the full elliptic matching problem, not surjectivity of the unrestricted trace onto $H^{-1/2}$.

For closed bulk forcing $f$ and zero source-row target, solve on the two halves as $u=G_0f+\mathcal A d$. The matching and row equations give $M_pd=-C_pG_0f$. Thus the bulk block $E_p$ of the original closed bordered inverse satisfies the exact bounded identity
\begin{equation}\label{rp_transmissioninverse}
 E_p=G_0-\mathcal A M_p^{-1}C_pG_0,
 \qquad\dot M_p=-C_pG_0V\mathcal A.
\end{equation}
It applies at arbitrary index; the finite source coordinates are already part of $d$. The second identity follows from~\eqref{rp_auxvariation} with $C_p$ fixed on the common graph domain.

The map $\mathcal A_a$ is bounded below into the \emph{bulk} Hilbert space, a stronger fact than its graph-domain bound. Its source columns are independent because the sewn border is invertible. Choose compact interior dual tests $\rho_j$ with $\langle\rho_j,R_{-,a}z\rangle=z_j$. For $u=\mathcal A_a(c,z)$, integration by parts has no boundary term and gives
\[
 |z_j|=|\langle P_a^\dagger\rho_j,u\rangle|\le C_j\|u\|_{L^2}.
\]
Now $w=u+G_aR_{-,a}z=\mathcal K_ac$ is homogeneous. Apply~\eqref{rp_graphtrace} to $w$ and use $\beta_a\gamma \mathcal K_ac=c$. This yields
\begin{equation}\label{rp_boundedbelow}
 \|c\|_{H^{-1/2}}+|z|\le C\|\mathcal A_a(c,z)\|_{L^2}.
\end{equation}
Hence $\mathcal A_a$ has closed range in $\mathscr H_a^+$ and a bounded left inverse
$\mathcal R_a=(\mathcal A_a^*\mathcal A_a)^{-1}\mathcal A_a^*$. It varies smoothly on a smaller chart. This construction of $\mathcal R_a$ uses only bounded Hilbert-space maps and does not extend trace to all $L^2$.

Let $\pi_a$ select the actual-half parameter component and $\iota_a$ insert its graph-domain component. Define the bounded map
\[
 \mathcal N_a=\sum_{p\ni a}\epsilon_p\pi_aM_p^{-1}C_p\iota_aG_aV_a:
                   \mathscr H_a^+\longrightarrow\mathscr D_a.
\]
Balancing the two copies of each actual half in~\eqref{rp_transmissioninverse} cancels its $G_a$ term. The exact identities are
\begin{align}
 \Delta_aV_a&=-\mathcal A_a\mathcal N_a,
 &X_a:=\sum_{p\ni a}\epsilon_p(M_p^{-1}\dot M_p)^{aa}
       &=-\mathcal N_a\mathcal A_a.\label{rp_twoproducts}
\end{align}
Section~\ref{rp_localizationsection} makes $\Delta_aV_a$ smoothing, hence trace class on $\mathscr H_a^+$. The left inverse proves
$\mathcal N_a=-\mathcal R_a\Delta_aV_a\in\mathcal S_1(\mathscr H_a^+,\mathscr D_a)$. Cross-space cyclicity is now an ordinary trace-class identity:
\begin{equation}\label{rp_cyclicity}
 \operatorname{Tr}_{\mathscr H_a^+}(\Delta_aV_a)=\operatorname{Tr}_{\mathscr D_a}X_a.
\end{equation}
No trace of $C_p$, $M_p^{-1}\dot M_p$, or an unbalanced half generator is needed separately. Smooth kernels on compact halves, with their derivatives, give trace-norm differentiability by factoring through a Sobolev embedding of order greater than four.

The endpoint is natural. For $P=\partial_t+B$ on a product cylinder, a positive tangential eigenmode $B e_\lambda=\lambda e_\lambda$ gives $u_\lambda=e^{-\lambda t}e_\lambda$. Then
\[
 \|u_\lambda\|_{L^2}^2=\frac{1-e^{-2\lambda\ell}}{2\lambda},\qquad
 \|\gamma u_\lambda\|_{H^{-1/2}}^2
       =(1+\lambda^2)^{-1/2}(1+e^{-2\lambda\ell}).
\]
Here the Sobolev norm is defined by $1+B^2$. Their ratio tends to two; the boundary $L^2$ to bulk $L^2$ norm ratio instead diverges. Zero tangential modes give finite constants and are retained.

\par\medskip\noindent\textbf{Eliminating the source constraints.}\enspace

The grouped trace in~\eqref{rp_cyclicity} is unchanged by a smooth bounded column isomorphism chosen identically on both occurrences of each actual oriented half. A column change $B_a^{\rm col}$ conjugates its grouped $X_a$ and adds $ (B_a^{\rm col})^{-1}\dot B_a^{\rm col}$ with opposite signs, which cancel \emph{before} tracing. A finite-rank unitriangular row change has zero total Jacobi trace. These assertions concern grouped or finite-rank operators only.

Set $\Lambda_a=R_{+,a}\mathcal A_a$ and $H_a=\gamma\mathcal A_a$. Invertibility of $M_p$ makes $\Lambda_a$ onto: prescribe its source-row target arbitrarily and all other targets zero. Choose fixed smooth preimages $J_{a,0}$ of a finite source-row basis at the chart center. They exist because smooth boundary parameters are dense and $\Lambda_a$ is onto a finite-dimensional target. Then
$J_{a,s}=J_{a,0}(\Lambda_{a,s}J_{a,0})^{-1}$ is a smooth finite right inverse on a smaller chart, with smooth boundary columns. The map $H_a$ on $\ker\Lambda_a$ is bounded below: insert this half and a zero partner into the inverse bound for $M_p$. Its closed range is the intrinsic constrained trace plane $\mathscr P_a^F$ in Lemma~\ref{rp_sourceplane}, independent of the auxiliary boundary domain. Thus
\[
 B_a^{\rm col}(l,\alpha)
   =(H_a|_{\ker\Lambda_a})^{-1}l+J_a\alpha,
 \quad H_aB_a^{\rm col}(l,\alpha)=l+H_aJ_a\alpha,
 \quad\Lambda_aB_a^{\rm col}(l,\alpha)=\alpha
\]
is a bounded half column isomorphism. Maps on the varying trace ranges are written in local bounded smooth trivializations; the common choice below supplies such trivializations explicitly.

Here is a direct regularity check for those constrained ranges. Work first on the regular boundary level
$\mathscr D_a^0=\beta_a L^2(Y)\oplus\mathbb C^{p_a}$. The same elliptic matching inverse is bounded at this level, so $H_a|_{\ker\Lambda_a}$ is bounded below in boundary $L^2$, not just in $H^{-1/2}$. Surjectivity of the finite row map makes
\[
 E_s^{\ker}=\mathrm I-\Lambda_s^*(\Lambda_s\Lambda_s^*)^{-1}\Lambda_s
\]
the smooth orthogonal projection onto its kernel. The rows are smooth boundary functionals; thus $E_s^{\ker}-\mathrm I$ has finite rank and smooth columns. Its Kato transport $U_s^{\ker}$ identifies these kernels smoothly on every Sobolev level. On the fixed kernel at the chart center set
\[
 T_s^F=H_sU_s^{\ker}|_{\ker\Lambda_0},\qquad
 \Pi_s^F=T_s^F\bigl((T_s^F)^*T_s^F\bigr)^{-1}(T_s^F)^*.
\]
The lower bound makes the inverse legitimate. As in~\eqref{rp_projectionregularity}, its elliptic order-zero parametrix gives smoothness on every boundary Sobolev level. Lemma~\ref{lem:collar_constrained_smoothing} supplies the explicit finite-rank Gram projection for these constraints. The combined surjectivity and lower bound have just been verified, and interior source supports make all boundary legs smooth with parameter derivatives. It follows in the full smoothing topology that $\Pi_s^F-\Pi_s$ is a smooth finite-rank smoothing family. This construction never differentiates a basis of $\operatorname{ran}\mathsf R$ or $\operatorname{ran}\mathsf C$, whose ranks may jump. It proves both the needed constrained-range trivializations and the smooth smoothing differences for the relative determinant below.

For a reflected pair these changes put the original matching border into the exact block form
\[
 \begin{pmatrix}
  \mathcal T_{ij}^{\rm src}&H_{\Theta i}J_{\Theta i}&-H_jJ_j\\
  0&\mathrm I&0\\0&0&\mathrm I
 \end{pmatrix},\qquad
 \mathcal T_{ij}^{\rm src}(l,r)=l-r,
 \quad (l,r)\in(\mathscr P_i^F)^\perp\oplus\mathscr P_j^F.
\]
Left multiplication by the finite-rank row matrix with first row
$(\mathrm I,-H_{\Theta i}J_{\Theta i},H_jJ_j)$ and identity diagonal removes the two upper-right blocks. The finite identity blocks have zero derivative. The grouped coordinate invariance therefore reduces~\eqref{rp_cyclicity} to matching on these two source planes. Constant permutations and the fixed minus sign on the right are incorporated identically on every occurrence. No determinant of an infinite sign or half-coordinate map is assigned.

\par\medskip\noindent\textbf{One relative determinant and the six-term rearrangement.}\enspace

Choose as common reference one positive modified plane along the chart. Its projector has a Kato transport $U(s)$ which is unitary on boundary $L^2$ and bounded on every required Sobolev space. In pulled-back coordinates let $\Pi_*$ be the fixed reference projector and $\Pi_i$ the three positive projectors. Their differences are smoothing. Nonzero Slater overlaps and equal charge imply $\|\Pi_i-\Pi_*\|<1$; let $Q_i$ be the corresponding direct rotations. They satisfy
\[
 Q_i=\mathrm I+\Psi^{-\infty},\qquad Q_i\Pi_*Q_i^*=\Pi_i,\qquad\det_FQ_i=1.
\]
The last identity follows by writing a direct rotation as the exponential of an off-diagonal skew-adjoint smoothing operator, whose ordinary trace is zero. In these coordinates the reduced matching maps are
\begin{equation}\label{rp_reducedmatching}
 M_{ij}^{\rm red}=U N_{ij},\qquad
 N_{ij}=Q_i(\mathrm I-\Pi_*)+Q_j\Pi_*=\mathrm I+\Psi^{-\infty}.
\end{equation}
This is an identification of one-particle spaces and their polarized Fock fibers. It does not implement $U$ inside one fixed Fock representation.

Define three-block cyclic and diagonal sums $\mathbb M_{\rm cyc},\mathbb M_{\rm diag}$ from the reduced $M_{xy},M_{yz},M_{zx}$ and $M_{xx},M_{yy},M_{zz}$, respectively, and define $\mathbb N$ similarly. The single relative operator is
\begin{equation}\label{rp_onecycle}
 \mathscr C=\mathbb M_{\rm diag}^{-1}\mathbb M_{\rm cyc}
 =\mathbb N_{\rm diag}^{-1}\mathbb N_{\rm cyc}=\mathrm I+\Psi^{-\infty}.
\end{equation}
The common moving $U$ cancels before differentiation. All remainders and their derivatives are trace class on the endpoint boundary space. On a three-dimensional seam, for instance,
\[
 \|S\|_1\le\|(\mathrm I+\Delta_Y)^{-2}\|_1
                 \|(\mathrm I+\Delta_Y)^2S\|_{L^2\to L^2}.
\]
Conjugate by the fixed Sobolev isometry to obtain the same conclusion on $H^{-1/2}$. The Fredholm determinant is independent of that Sobolev realization: generalized eigenvectors of a smoothing operator with nonzero eigenvalue are smooth, and hence its nonzero eigenvalues and algebraic multiplicities agree on all these spaces.

Compute the source interpretation on boundary $L^2$, using the Sobolev independence just established. Multiplying $N_{ij}$ by $Q_i^*$ gives an operator triangular in $\operatorname{ran}(\mathrm I-\Pi_*)\oplus\operatorname{ran}\Pi_*$, with diagonal blocks $\mathrm I$ and $\Pi_*Q_i^*Q_j\Pi_*$. Therefore
\[
 \det_FN_{ij}=\det_F(\Pi_*Q_i^*Q_j\Pi_*|_{\Pi_*\mathcal H}),
\]
which is the overlap of the normalized Slater lifts $\upsilon_i,\upsilon_j$. Write the actual half-Wick vectors as $V_i[F_i]=s_i\upsilon_i$, with $s_i\ne0$. Every half scalar, including its Quillen normalization, phase and finite source cofactor, cancels once conjugated and once unconjugated in the six-term ratio. Thus
\begin{equation}\label{rp_boundaryratio}
 \det_F\mathscr C
   =\frac{B_{xy}B_{yz}B_{zx}}{B_{xx}B_{yy}B_{zz}}=\mathfrak B.
\end{equation}
All finite exterior-order signs remain in the actual half vectors.

To check the rearrangement without tracing divergent blocks, let indices $i$ run cyclically through $x,y,z$, let $P_-=\mathrm I-\Pi_*$, $P_+=\Pi_*$, and put
$\mathcal L_{ij}=(M_{ij}^{\rm red})^{-1}\dot M_{ij}^{\rm red}$. The grouped blocks for the negative and positive actual half $i$ are
\begin{align*}
 \widetilde X_i^-&=P_-(\mathcal L_{i,i+1}-\mathcal L_{ii})P_-,\\
 \widetilde X_i^+&=P_+(\mathcal L_{i-1,i}-\mathcal L_{ii})P_+.
\end{align*}
These are trace class by~\eqref{rp_cyclicity} and the legitimate coordinate changes. The difference
$\mathbb L_{\rm cyc}-\mathbb L_{\rm diag}$ is trace class too, because every $N_{ij}$ is identity plus smoothing and the common $U^{-1}\dot U$ is bounded. Compress its three diagonal blocks to $P_-$ and $P_+$. In the positive compression the identity
\[
 \sum_iP_+(\mathcal L_{i,i+1}-\mathcal L_{ii})P_+
 =\sum_iP_+(\mathcal L_{i-1,i}-\mathcal L_{ii})P_+
\]
is an equality of finite sums of bounded operators \emph{before} tracing. The negative compression already has the desired ordering. Consequently
\begin{equation}\label{rp_groupedrearrangement}
 \operatorname{Tr}(\mathbb L_{\rm cyc}-\mathbb L_{\rm diag})
   =\sum_i\bigl(\operatorname{Tr}\widetilde X_i^-+\operatorname{Tr}\widetilde X_i^+\bigr)
   =\sum_a\operatorname{Tr} X_a.
\end{equation}
No cyclic permutation of individual non-trace-class traces has occurred.

Finally put $D=\mathbb M_{\rm diag}^{-1}\dot{\mathbb M}_{\rm diag}$. From $\mathbb M_{\rm cyc}=\mathbb M_{\rm diag}\mathscr C$,
\[
 \mathbb L_{\rm cyc}-\mathbb L_{\rm diag}
 =\mathscr C^{-1}\dot{\mathscr C}+(\mathscr C^{-1}D\mathscr C-D).
\]
The second term has zero ordinary trace. Indeed write $\mathscr C=\mathrm I+K$, with $K$ trace class, and cycle only trace-class factors:
\[
 \operatorname{Tr}(\mathscr C^{-1}DK-\mathscr C^{-1}KD)
 =\operatorname{Tr}(DK\mathscr C^{-1}-D\mathscr C^{-1}K)=0.
\]
Combining~\eqref{rp_cyclicity},~\eqref{rp_groupedrearrangement}, Fredholm Jacobi and~\eqref{rp_boundaryratio} proves the required full complex identity
\begin{equation}\label{rp_compositeidentity}
 \boxed{\displaystyle
 \sum_a\operatorname{Tr}(\Delta_aV_a)
     =\operatorname{Tr}(\mathscr C^{-1}\dot{\mathscr C})
     =\frac d{ds}\log\mathfrak B.}
\end{equation}
Only bounded compositions on fixed graph or endpoint spaces and ordinary trace-class cyclicity were used.

Combine~\eqref{rp_compositeidentity} with the local finite-part cancellation and the full BF source differential:
\begin{equation}\label{rp_fullratio}
 \frac d{ds}\log\mathfrak S
 =\sum_a\operatorname{Tr}(\Delta_aV_a)
 =\frac d{ds}\log\mathfrak B.
\end{equation}
Its real part checks the balanced norm ratio; the pointwise modulus comparison for each sewing was proved separately by Loya--Park and saturated sources. Its imaginary part gives $d\tau/ds=0$. Interior and common-collar variations are both allowed. Source independence identifies the scalar on overlaps of frozen source charts; a compact ambient triple path is covered by finitely many such charts. This proves the transmission step needed to transport the triangle to its real reference, within the stated smooth Dirac and Calder\'on framework.

\end{proof}

\subsubsection{Absolute reference phases in all relative classes}\label{rp_referencesec}

Fix boundary framings. An induced $SU(5)$ half bundle over $[0,\ell]\times T^3$ is absolutely trivial, but its framed relative class is $n\in H^4(M_+,\partial M_+;\mathbb Z)\cong\mathbb Z$. This follows because the four-dimensional Postnikov truncation of $BSU(5)$ is $K(\mathbb Z,4)$. The sewn class and index conventions are
\[
 k=n_y-n_x,\qquad \ind P_{xy}=-12k.
\]
One generation contributes $3+1=4$ to the index coefficient. No claim classifies the components of the admissible lattice set by $n$.

\begin{corollary}[Equal-class reflection positivity]\label{cor:collar_equalclass}
At fixed seam data, the source kernel restricted to any one relative class is positive semidefinite. In particular, the zero-induced-charge restriction is reflection positive. The unrestricted functional also has a nonnegative quadratic form for every homogeneous observable of a fixed net source degree; mixed degrees require the all-class argument below.
\end{corollary}
\begin{proof}
If $n_x=n_y=n_z=n$, identify the three fixed framed bundles. Connections with the same product collar form an affine space. Move $y$ and $z$ to $x$ in that space while holding the collar fixed. Triangle transport gives
$\tau(x,y,z)=\tau(x,x,x)=1$. The unit-modulus comparison, the reflected diagonal and Lemma~\ref{lem:collar_coboundary} therefore give a Gram kernel within each class, without a weak-$SU(2)$ reference evaluation. The finite-matrix integration argument in Section~\ref{rp_conclusion} applies to each positive block.

For a nonzero coefficient the selection rule is
\[
 d(F_y)-d(F_x)=-12(n_y-n_x).
\]
Thus an observable of a single net degree has no unequal-class contributions; the vacuum is the degree-zero example. A polynomial mixing net degrees can have unequal-class coefficients and needs the all-class argument below. Index zero alone would not have established the positivity of its equal-class block.
\end{proof}

\paragraph{Unequal classes and their real references.}
First move the common framed collar potential to zero, using $b_s=(1-s)b$ and the same extension cutoff on each half. Preserve a smaller fixed product collar. At $b=0$, interpolate each interior connection affinely on its fixed framed bundle to a weak-$SU(2)$ reference. The ambient connection space with fixed collars is affine. The inclusion $SU(2)\hookrightarrow SU(5)$ induces an isomorphism on $\pi_3$; thus every relative $n$ has a reference
\[
 A_n=-q(t)\,dg_ng_n^{-1},\qquad \deg(g_n:T^3\to SU(2))=n,
\]
where $q$ is zero and one near the two ends and the boundary frames are $1,g_n$. The low-dimensional homotopy assertion follows successively from $SU(j+1)/SU(j)=S^{2j+1}$ for $j\ge2$; the same homotopy calculation gives the stated four-dimensional bundle classification. Both framed collar potentials are zero. These paths cover every class triple. They need not remain in the admissible-link image, and all lemmas above apply on the smooth ambient path.

\begin{lemma}[Real source charts and positive full references]\label{rp_realref}
At the weak-$SU(2)$ references the triangle of the actual comparison phases is one for every relative class triple, in the prescribed ambient BF trivialization.
\end{lemma}
\begin{proof}
The representation restricts as $R|_{SU(2)}=12\,\mathbf2\oplus21\,\mathbf1$. For a single doublet use normal Pauli coordinates and
\[
 J_{\mathbb R}=(i\sigma_2\otimes i\tau_2)K_0,\qquad J_{\mathbb R}^2=\mathrm I.
\]
This real structure is distinct from the boundary-sign matrix $J_\partial$. The necessary conjugation identities are
\[
 (i\sigma_2)\overline{\sigma_j}(i\sigma_2)^{-1}=-\sigma_j,
 \qquad (i\tau_2)\overline{A_\alpha}(i\tau_2)^{-1}=A_\alpha.
\]
Complex conjugation also changes the sign of $i$, so $J_{\mathbb R}$ commutes with $P=\partial_t+i\sigma_jD_j+A_t$ and with its adjoint. The two pseudoreal squares multiply to $+1$. The two endpoint signs in $J_\partial$ and $\Sigma$ are real scalar matrices, so they commute with $J_{\mathbb R}$ as well. The normal antilinear reflection exchanges $q$ and $\bar q$ with real signs; hence its source action commutes with the induced real structure.

The Cauchy spaces, orthogonal projections, Poisson maps and Calder\'on inverses are preserved by $J_{\mathbb R}$, by uniqueness. Choose a common real reference polarization and real exterior bases. Conjugating their wedge coefficients defines an antiunitary real structure on every charged Fock sector, with square one. It sends $a(h)$ to $a(J_{\mathbb R}h)$ and similarly for creation. Real source contractions and real CAR legs therefore give real half vectors, after a real choice of each vacuum lift. Their pairings are real. The real version of the one-ball chart lemma supplies single-doublet monomials of any common resulting charge with all six pairings nonzero, including unequal half classes.

BF transport preserves the real single-doublet determinant line. Indeed the real structure intertwines the heat and low-window constructions with their conjugates, so the connection preserves this real line also across kernel crossings, represented by fixed source borders. A real unit line can have sign holonomy; its even twelfth tensor power has a canonical positive unit orientation, which is parallel. Only the evenness of twelve is used here.

For completeness, the absolute prescription is consistent across the whole weak-$SU(2)$ locus at fixed closed class. On a four-dimensional base $SU(2)$ bundles are classified by $c_2$. Two such reductions with the same class can therefore be identified by a bundle isomorphism; it extends to an $SU(5)$ isomorphism. After that identification their $SU(2)$ connections are connected affinely. Gauge equivariance and flat ambient transport, together with real parallel transport on this affine family, identify the same positive twelfth-power orientation everywhere on the locus. Thus fixing it once per closed class has the meaning required by the prescription; it is not a phase chosen separately for each sewing.

Repeat the chosen single-doublet source list in all twelve identical copies, and leave the free singlets uninserted. Direct sums tensor the original determinant lines and add their BF connections. The nonzero single-copy physical coefficients and half pairings are real; their twelfth powers are positive. Any within-copy exterior sign is raised to the twelfth power. For the cross-copy reordering, each crossing sign repeats $\binom{12}{2}=66$ times, so it is $+1$ for all source parities. The 21 free singlet factors are positive in the chosen convention; temporal antiperiodicity makes their closed operator invertible even for periodic spatial spin structure.

Thus every full physical coefficient and pairing used in this reference chart is positive. Their ratio is also of unit modulus by Corollary~\ref{rp_unitcomparison}, and is therefore one. All three reference comparison phases, and their triangle, are one.
\end{proof}

The phase triangle is independent of the source chart and invariant under local half phases. Source charts cover the ambient path in finitely many neighborhoods, so its transport identity and Lemma~\ref{rp_realref} give
\begin{equation}\label{rp_triangleone}
 \tau(x,y,z)=1
\end{equation}
for every triple at arbitrary seam data and relative classes. This evaluates every component constant, rather than appealing to local constancy alone.

An extra closed-sector rephasing $e^{i\beta(k)^3}$, cubic in the induced charge $k$ rather than linear as a theta term, preserves magnitudes, reflected diagonals, Hermiticity under $k\mapsto-k$ and local derivative identities. On half classes $(0,1,2)$ the sewn classes are $(1,1,-2)$, so its triangle is
\[
 e^{i\beta(1^3+1^3+(-2)^3)}=e^{-6i\beta}.
\]
The specified reference prescription excludes this ambiguity; the trace argument alone does not. An additive phase $e^{i\theta k}$ instead has triangle one and is absorbed by $u(x)=e^{i\theta n_x}$. This distinction, rather than a prohibition of ordinary theta phases, is the point of the example.

\subsubsection{Reflection positivity from the analytic identities}\label{rp_conclusion}
The three analytic steps and the all-class reference evaluation now give the kernel required by the theorem stated at the beginning of this subsection.

\begin{proof}[Proof of Theorem~\ref{prop:collar_rp}]

The source comparison and its unit modulus are established in Sections~\ref{rp_comparisonsection}--\ref{rp_normsection}. The source charts in Section~\ref{rp_chartsection}, the differential in Section~\ref{rp_BFsection}, the trace identities in Sections~\ref{rp_localizationsection}--\ref{rp_transmissionsection}, and the reference evaluation in Section~\ref{rp_referencesec} give~\eqref{rp_triangleone}. For a finite collection choose arbitrary unit Slater representatives and one reference half. Lemma~\ref{lem:collar_coboundary} supplies unit half phases $u(x)$ with
\[
 c(x,y)=\overline{u(x)}u(y),\qquad
 S_{xy}(F_x,F_y)=\langle u(x)V_x[F_x],u(y)V_y[F_y]\rangle.
\]
Every finite matrix of physical fixed-seam source coefficients is therefore positive semidefinite. This conclusion covers arbitrary polynomial sources by the source-independent comparison. The special saturated monomials were only coordinates for the proof.

A bounded measurable positive-definite kernel integrates nonnegatively against a finite measure. We use this fact through finite samples, so integration needs no globally measurable choice of Fock frames or $x_*$. Write $K_b(x,y)=S_{xy}(F_x,F_y)$ and $M_b=\nu_b(\text{half space})$. The physical scalar kernel is measurable and bounded on the compact regulated support by Lemma~\ref{rp_finiteinsertions}. If $M_b>0$, sample $N$ independent halves with law $\nu_b/M_b$. Finite Gram positivity gives $\sum_{r,s}K_b(X_r,X_s)\ge0$ for every sample. Taking its expectation and multiplying by $M_b^2/N^2$ yields
\[
 0\le (1-N^{-1})\iint K_b(x,y)\,d\nu_b(x)d\nu_b(y)
       +\frac{M_b}{N}\int K_b(x,x)\,d\nu_b(x).
\]
Letting $N\to\infty$ proves nonnegativity of the double integral. The case $M_b=0$ is immediate. Integrate against $d\rho(b)$ in~\eqref{eq:collar_bosons}. This proves Theorem~\ref{prop:collar_rp} without a global boundary-frame selection.

Taking $F=1$ proves $Z_{\mathrm c}\ge0$. The vacuum selection rule and the role of mixed source degrees were separated in Corollary~\ref{cor:collar_equalclass}. Positive diagonal kernels do not prove that the integrated half-state is nonzero, so no strict noncancellation claim is made.

Restriction to zero induced closed $SU(5)$ charge retains $n_x=n_y$. The resulting block-diagonal kernel is positive semidefinite, so the same integration proves the restricted assertion. This is a restriction on the induced $SU(5)$ charge, not an arbitrary restriction on all Standard Model flux labels.

\end{proof}

\begin{corollary}[Compatible half phases]\label{cor:collar_halfphases}
Suppose a rephasing of the full source weight at fixed seam data has the form
\[
 S'_{xy}(F_x,F_y)
 =e^{i[L_b(y)-L_b(x)]}S_{xy}(F_x,F_y),\qquad L_b(x)\in\mathbb R,
\]
with measurable half phases independent of the sources. It preserves reflection positivity and finite-volume bounds. In particular an induced $SU(5)$ theta factor $e^{i\theta(n_y-n_x)}$ has this property for every real $\theta$.
\end{corollary}
\begin{proof}
Replace each Gram vector by $e^{iL_b(x)}u(x)\Psi_x(F_x)$. All finite kernel matrices remain positive, and the same bounded-kernel integration applies. The theta example has $L_b(x)=\theta n_x$.
\end{proof}
More generally a real local imaginary action preserves this argument when its reflected negative-half integral is the negative of the positive-half integral and no independent seam phase is added. Its integral on $\Theta x\cup_b y$ is then $L_b(y)-L_b(x)$; these half integrals need not be integers. Locality by itself is insufficient: a constant phase multiplying every closed amplitude leaves the balanced ratio unchanged but can spoil its positive diagonal. The theorem fixes the zero-theta prescription; the corollary states an explicit sufficient test for other phase conventions.

\subsubsection{Supporting estimates and published inputs}\label{rp_externalinputs}

The following estimates supply the regularity used above. Their proofs use fixed-domain inversion, elliptic boundary estimates and finite-rank formulas, and do not assume reflection positivity.

\paragraph{Half operators and common collars.}
\begin{lemma}[Half-operator regularity through closed zero modes]\label{rp_halfregularity}
Let $P_s$ be a smooth family of the stated Dirac-type half operators, with fixed principal symbol on fixed bundles. Locally in $s$, the orthogonal Cauchy projection $\Pi_s$, the orthogonal Calder\'on inverse $G_s$, and the Poisson and source-leg maps depend smoothly in their boundary Sobolev and bulk operator norms. Their classical symbols depend smoothly as well. No assumption on the dimension of any \emph{closed} kernel is needed. For halves with a common collar germ, their Cauchy-projection differences are smooth families of smoothing operators, including when that germ varies with $s$.
\end{lemma}
\begin{proof}
At a chart center freeze $\beta=\Pi_0$. The regular boundary problem with domain
$\{u\in H^1:\beta\gamma u=0\}$ is invertible at the center. The domain is fixed, and $P_s-P_0$ is zero order, so its inverse $G_{\beta,s}$ exists nearby and is smooth by the bounded inverse theorem. In particular
$\dot G_\beta=-G_\beta\dot P G_\beta$. The boundary parametrix gives its usual Sobolev bounds. With $\mathscr B^r=\beta H^r(Y)$, define
\[
 \mathcal K_{\beta,s}=\mathcal K_{\beta,0}-G_{\beta,s}(P_s-P_0)\mathcal K_{\beta,0}.
\]
Then $P_s\mathcal K_{\beta,s}=0$ and $\beta\gamma \mathcal K_{\beta,s}=\mathrm I_{\mathscr B^r}$. This is the unique Poisson solution with that prescribed $\beta$-trace; the endpoint extension $r=-1/2$ takes values in the maximal graph domain. The Poisson and maximal-trace bounds used here are those of~\cite[Sec.~2.3, Thm.~5.6 and Def.~5.9]{booss_lesch_zhu}, together with the classical boundary parametrix.

Put $T_s=\gamma \mathcal K_{\beta,s}:\mathscr B^0\to L^2(Y)$, and let $\iota:\mathscr B^0\hookrightarrow L^2(Y)$ be inclusion. Its range is the Cauchy space and $\beta T_s=\mathrm I$. The displayed resolvent formula and fixed principal symbol give
$S_s:=T_s-\iota$ of order $-1$, smoothly in the boundary calculus. Also $\beta S_s=0$. Consequently
\begin{equation}\label{rp_projectionregularity}
 T_s^*T_s=\mathrm I+S_s^*S_s\ge\mathrm I,\qquad
 \Pi_s=T_s(T_s^*T_s)^{-1}T_s^*.
\end{equation}
Adjoints here use fixed boundary $L^2$. The positive elliptic order-zero operator $\mathrm I+S_s^*S_s$ is invertible on every $\mathscr B^r$: a Sobolev kernel or obstruction is smooth by elliptic regularity and hence excluded by its $L^2$ invertibility. Its inverse is smooth by the bounded inverse theorem on each Sobolev level and the parameter-dependent parametrix. Thus~\eqref{rp_projectionregularity} gives the required smooth orthogonal projection, rather than assuming smoothness of its range.

The actual orthogonal maps in~\eqref{rp_halfidentities}, extended by zero on $W_s^\perp$ when appropriate, are
\begin{equation}\label{rp_orthogonalregularity}
 \widetilde{\mathcal K}_s=\mathcal K_{\beta,s}\beta\Pi_s,\qquad
 G_s=G_{\beta,s}-\mathcal K_{\beta,s}\beta\Pi_s\gamma G_{\beta,s},\qquad
 L_s=\gamma G_s.
\end{equation}
Indeed $P_sG_s=\mathrm I$ and $\Pi_s\gamma G_s=0$; uniqueness identifies these with the orthogonal boundary problem. These formulas prove the asserted norm regularity of $G_s:L^2\to H^1$, of the Poisson map, and of the interior source legs and contractions. No closed inverse or closed kernel projection occurs.

For equal but varying collar germs, Lemma~\ref{lem:collar_parameter_smoothing} gives the full smoothing estimates with every parameter derivative. The fixed-domain formulas above give smoothness of each individual projection in every boundary Sobolev operator norm. In particular $[\dot\Pi_s,\Pi_s]$ is bounded on those spaces, and its Kato equation supplies the common smooth transport used in the transmission argument.
\end{proof}

\begin{lemma}[Common-collar smoothing with parameter derivatives]\label{lem:collar_parameter_smoothing}
Let $P_{i,s}$ be finitely many smooth families of unitary Dirac-type half operators with fixed principal symbol and boundary identifications. Suppose their germs agree on a common fixed positive-width product collar for every $s$. Their orthogonal Cauchy projections satisfy
\[
 \sup_{s\in K}\left\|\partial_s^k(\Pi_{i,s}-\Pi_{j,s})
                    \right\|_{H^{-N}(Y)\to H^M(Y)}<\infty
 \quad\hbox{for every }N,M,k\ge0
\]
on each compact parameter subinterval $K$ in a sufficiently small chart. In particular the differences are smooth in the smoothing topology, even when the common collar coefficients vary.
\end{lemma}
\begin{proof}
Use the product-collar Poisson construction in Ref.~\cite[Sec.~5.3, Eqs.~(5.30)--(5.32) and Remark~5.11]{booss_lesch_zhu}. The common tangential operator is self-adjoint. Locally choose a spectral cut away from its spectrum and a smooth cutoff in the normal coordinate. On the incoming spectral subspace the model solution is $e^{-tH_s}$; at any positive distance from the boundary it is smoothing. The corresponding outgoing model on the doubled operator has the same property. All parameter derivatives are smoothing there: the cut stays in a spectral gap, and the positive-distance spectral multipliers and their derivatives are rapidly decreasing at high frequency. Finite spectral-band changes only add smooth finite-rank terms. This argument does not assume that $H_s$ is invertible at zero.

The defect of the cutoff product solution is supported away from the seam, and is consequently a smooth family of smoothing maps from boundary data into bulk smooth sections. Correct it by an invertible auxiliary double with fixed normal matching $T=(J_0^*)^{-1}$, where $J_0$ is the fixed principal normal coefficient, so $J_0^*T=1$. Its domain is fixed because the principal normal symbol is fixed. Unique continuation excludes its kernel and cokernel; the bounded inverse theorem gives a smooth local inverse. Boundary elliptic estimates give the same assertion on every needed Sobolev level. Composing this inverse with the cutoff defect and taking the regular trace therefore remains smoothing with all parameter derivatives. This is the remainder in the cited product-collar construction. It follows that each Cauchy projection is the same model boundary projection plus a smooth smoothing remainder.

One may equivalently start with nonorthogonal Calder\'on idempotents $C_s$ from that construction. Orthogonalization preserves the conclusion: with
\[
 F_s=C_sC_s^*+(1-C_s^*)(1-C_s),\qquad
 \Pi_s=C_sC_s^*F_s^{-1},
\]
$F_s$ is positive invertible on $L^2$ and elliptic of order zero. Its inverse is smooth on every Sobolev level. The resolvent identity shows that a smoothing difference of two $F_s$ gives a smoothing difference of their inverses, with every parameter derivative. Thus orthogonalization leaves the common complete symbol and the asserted remainder topology unchanged. The estimates follow on compact parameter subintervals.
\end{proof}

\begin{lemma}[Finite source constraints preserve smooth smoothing differences]\label{lem:collar_constrained_smoothing}
Let $W_s=\operatorname{ran}\Pi_s\subset\mathcal H_b=L^2(Y)$ vary smoothly as above. Write a nonzero monomial source plane as
\[
 W_s^F=\{w+T_s\eta:w\in W_s,\ \eta\in\mathbb C^p,
                           \ R_sw+C_s\eta=0\},
\]
where the finite rows of $R_s$ and columns of $T_s$ are smooth boundary legs, $T_s\mathbb C^p\subset W_s^\perp$, and $C_s$ is a finite matrix. Suppose the combined constraint is onto and the trace map on its kernel is bounded below. These conditions hold locally on an invertible source-border matching chart. Then $\Pi_s^F-\Pi_s$ is a smooth family of finite-rank smoothing operators. No constancy of the separate ranks of $R_s$, $T_s$ or $C_s$ is required.
\end{lemma}
\begin{proof}
On $\mathcal H_b\oplus\mathbb C^p$ set
\[
 B_s=\operatorname{diag}(\Pi_s,I_p),\qquad
 \Lambda_s=(R_s\Pi_s,C_s),\qquad
 E_s=B_s-\Lambda_s^*(\Lambda_s\Lambda_s^*)^{-1}\Lambda_s.
\]
Surjectivity makes the finite inverse legitimate. Since $\Lambda_sB_s=\Lambda_s$, $E_s$ is the orthogonal projector onto the constrained kernel inside $W_s\oplus\mathbb C^p$. Put
\[
 A_s(w,\eta)=\Pi_sw+T_s\eta,\qquad
 M_s=E_sA_s^*A_sE_s+1-E_s.
\]
The assumed lower bound makes $M_s$ positive invertible. The projector onto $W_s^F$ is exactly
\[
 \Pi_s^F=A_sE_sM_s^{-1}E_sA_s^*.
\]
Indeed $E_s$ commutes with $M_s$, and direct multiplication proves self-adjointness, idempotence and the stated range. Orthogonality of $T_s$ to $W_s$ gives
\[
 A_s^*A_s-B_s=\operatorname{diag}(0,T_s^*T_s-I_p),\qquad
 M_s-1=E_s\operatorname{diag}(0,T_s^*T_s-I_p)E_s.
\]
This is a finite-rank correction with smooth boundary legs. Its inverse differs from the identity by another such correction; a finite-dimensional inverse formula, or the resolvent identity, gives smooth dependence. Also $E_s-B_s$ is finite-rank smoothing. Substitution in the projector formula leaves $\Pi_s$ plus finite-rank smoothing terms with smooth coefficients. Differentiating these formulas proves the assertion in every smoothing seminorm. Only the combined surjectivity and lower bound are used.
\end{proof}

The lower bound in the constrained-plane lemma follows on a nonzero matching chart by inserting its trace on one half and zero data on the partner into the full matching inverse estimate. Compact interior source supports give smooth boundary legs with every parameter derivative. These estimates establish the smoothing topology used for the relative Fredholm determinant.

\paragraph{Source extensions and compact regulated bounds.}
\begin{lemma}[$L^2$ extension of split source coefficients]\label{lem:collar_source_extension}
At fixed seam data, the half-Wick pairings extend to continuous graded-multilinear functionals in each fixed finite source degree on the global $L^2$ source space. The extension is independent of the expansion into split words and agrees with the original pairing on interior sources. Its restriction to globally smooth smearings has the same domain as the physical source section.
\end{lemma}
\begin{proof}
For a fixed half, its orthogonal Green operator $G:L^2\to H^1$, Poisson map $\mathcal K:L^2(Y)\to L^2(M_+)$ and source leg $L=\gamma G:L^2\to L^2(Y)$ are bounded. Thus the Wick contractions obey $|\langle f,Gg\rangle|\le C\|f\|_2\|g\|_2$, and the boundary CAR legs obey $\|a(h)\|,\|a^*(h)\|\le\|h\|$. A fixed-degree half-Wick word is a finite sum of these bounded products applied to its vacuum vector. Its norm is at most a degree-dependent constant times the product of the $L^2$ norms of its profiles. Pairing the two halves preserves this bound.

The seam has measure zero, so the global $L^2$ space is the orthogonal direct sum of its two half spaces. The algebraic exterior algebra of a direct sum is canonically the graded tensor product of the two exterior algebras. Expand a global generator into its two restrictions and use this identification, including its Koszul signs. This defines the extension independently of a chosen ordering or split expansion. Smooth interior profiles are dense in each half $L^2$ space, giving uniqueness and agreement with the original half maps. The physical section has the same degreewise extension: a spectral window leaves bounded high-mode contractions and finitely many bounded low-mode evaluations. On compact regulated families a finite cover by spectral and half-operator charts makes these degreewise bounds uniform, so the corresponding integrated coefficients are continuous under uniform $L^2$ profile limits. For a globally smooth test field in a Ward identity, differentiate it before restricting to the halves; Green's formula then supplies its actual boundary traces. No derivative of a sharp half indicator is used.
\end{proof}

\begin{lemma}[Continuity and boundedness of finite insertions]\label{rp_finiteinsertions}
For a compact smooth family of reconstructed connections, every fixed finite smooth-smeared source coefficient is continuous through zero modes and bounded. The assertion also holds for coefficients with continuous bosonic prefactors, and after applying a smooth unitary phase trivialization.
\end{lemma}
\begin{proof}
Choose a spectral window near any parameter value. Low-mode dimensions are fixed there, and all low matrices and their source evaluations are smooth. The high map has a uniform local gap, so its inverse, its contractions on fixed smooth sources and its zeta--Quillen factor are smooth. Each finite source coefficient is a finite sum of products of these quantities and polynomial low-mode minors; it contains no inverse of a low matrix. The coefficients are consequently continuous at a closed zero mode. A compact parameter support has a finite cover by such charts, giving boundedness. Multiplication by a unitary trivialization changes no magnitude. Continuous bosonic prefactors are bounded on the same compact support.
\end{proof}

\paragraph{External hypotheses.}
The construction and gluing derivations are contained in F.5. The following published analytic and topological results are the external inputs; their hypotheses and uses are recorded explicitly.

\paragraph{Elliptic boundary theory and unique continuation.}
For a smooth first-order Dirac-type operator and its adjoint on a compact manifold with boundary, use the classical boundary parametrix and unique continuation: a homogeneous solution vanishing on an open interior set, or with full zero boundary trace, is zero. The orthogonal projection onto Cauchy traces is a classical order-zero pseudodifferential projection. Its complete symbol is determined by the full operator germ at the boundary. Regular elliptic boundary problems have the standard Sobolev regularity and smooth kernels for their finite-dimensional obstructions.

The endpoint facts used in the transmission argument are specifically~\cite[Eqs.~(2.26)--(2.27)]{booss_lesch_zhu}
\[
 \gamma:\{u\in L^2:Pu\in L^2\}\longrightarrow H^{-1/2}(Y),\qquad
 \|\gamma u\|_{H^{-1/2}}\le C(\|u\|_{L^2}+\|Pu\|_{L^2}),
\]
and~\cite[Thm.~5.6, Def.~5.9]{booss_lesch_zhu} the Poisson bounds
$\mathcal K:H^r(Y)\to H^{r+1/2}(M_+)$ on Cauchy data for $-1/2\le r\le1/2$, with higher regularity supplied by the boundary parametrix. We do not infer a trace bound on arbitrary bulk $L^2$. Nor do we cite a general continuity theorem as if it proved the needed smoothness: Lemma~\ref{rp_halfregularity} derives parameter regularity for this fixed-principal-symbol, zero-order perturbation family.

All hypotheses are checked here: smooth compact cylinder, smooth Hermitian bundles, fixed Dirac principal symbol and metric, and smooth unitary connections. The two boundary components are kept together. A frozen orthogonal Calder\'on condition is invertible at the chart center; its inverse and fixed domain are used only on its open invertibility neighborhood. The source rows and columns are finite and smooth in the interior.

\paragraph{The primed Loya Park formula.}
The precise external norm identity is~\cite[Thm.~1.1]{loya_park_gluing}, written in our notation as~\eqref{rp_LPexact}. It applies to a compact self-adjoint Dirac operator split along a closed hypersurface, with a product metric and coefficients in an open collar, using the orthogonal Calder\'on conditions on both sides. The closed determinant is primed, the $-1$ eigenspace of the relative boundary unitary is omitted from $U_b$, and the restricted closed-kernel trace operator $\mathsf L$ is retained. The tangential operator need not be invertible.

Here the operator is the self-adjoint double $\mathcal D$, the seam is the disjoint union $Y_0\sqcup Y_\ell$, and Proposition~\ref{prop:collar_transfer} supplies product collars of positive width. Lemma~\ref{rp_metric} gives the normal-Clifford unitary dictionary, the fourth-root principal-angle factor and the identification $B^*B=\mathsf H_+\oplus\mathsf H_-$. Lemma~\ref{rp_saturated} then derives the physical source evaluations and CAR cofactor. The external theorem fixes a norm; it supplies no phase identity or positivity conclusion.

\paragraph{The BF spectral window connection.}
The external BF input is the smooth determinant bundle with zeta--Quillen metric and compatible connection, as constructed for smooth elliptic families in~\cite{bf1}. In the notation used in Section~\ref{rp_BFsection}, the exact pieces are:
\begin{enumerate}
\item Definition~1.10, printed p.~168: the high metric factor is $\exp[-\zeta_\lambda'(0)/2-\kappa_{\rm BF}A_0/2]$. Ordinary zeta normalization therefore means $\kappa_{\rm BF}=0$.
\item Definition~1.14 and Theorem~1.15, printed pp.~170--171: the constant terms of the high heat integrals are $\gamma_0^\lambda,\delta_0^\lambda$, and their logarithmic coefficients are $dA_0,B_0$ in Eq.~(1.50).
\item Definition~1.17 and Theorem~1.18, printed pp.~173--174: the low determinant connection plus the high term~\eqref{rp_bfhigh} define a smooth connection compatible across spectral windows and with that metric.
\end{enumerate}
The fibers are closed compact spin four-tori, the metric and Hermitian structures are fixed, and the family consists of smooth elliptic operators with fixed principal symbol. Finite source borders do not replace this family or its determinant line. Equation~\eqref{rp_bfexact} is derived here by Mellin finite part, and~\eqref{rp_bfsource} by an auxiliary frozen border and finite Schur Jacobi. BF's holomorphic Theorem~1.21 is not used. The sign dictionary $\gamma=2\Re a$, $\delta=-2i\Im a$, and the retained logarithm are explicit in the derivation above.

\paragraph{The weighted trace density.}
For a classical pseudodifferential insertion $A$ and an admissible elliptic weight $Q$ of positive order $q$, Paycha--Scott~\cite[Cor.~2.5, Eq.~(2.9)]{paycha_scott} gives
\[
 \operatorname{FP}_{z=0}\operatorname{Tr}(AQ^{-z})
  =\int_M\left(\operatorname{TR}_u A
               -\frac1q\operatorname{res}_{u,0}(A\log Q)\right)du
       -\operatorname{Tr}(A\Pi_Q).
\]
The formula includes integer-order insertions; here $A=E_p\dot P_p$ has order $-1$ in four dimensions. Here $\operatorname{TR}_u$ is the finite-part symbol integral in a fixed quantization, and $\operatorname{res}_{u,0}$ takes the log-free homogeneous term of degree $-\dim M$ and integrates it on the unit cotangent sphere. The corrected combination is a globally defined density independently of the insertion order, as stated immediately after their Eq.~(2.10); the two summands need not separately be densities. In this subsection $Q=\mathcal W_\lambda$ is positive invertible, so $\Pi_Q=0$, its spectral cut is fixed and $q=2$. Each closed bordered inverse has order $-1$ and differs from a closed parametrix only by smoothing terms. Lemma~\ref{rp_localization} cancels matched full symbols before assigning sharp halves. It does not assign a finite part to a discontinuous half insertion. The full operator germs, weight symbols and all overlap cutoff symbols are matched, rather than only their leading symbols.

\paragraph{Anomaly transport and bordism.}
BF curvature is the family index density. Dai--Freed~\cite[Cor.~5.11]{daifreed} identifies the holonomy of the inverse determinant line with the adiabatic exponentiated reduced eta invariant, including its graded index sign; inversion gives the corresponding formula for our line. Section~\ref{sec:hsu_eta} checks vanishing of the entire degree-six local polynomial and evenness of the index. The bordism input is $\Omega_5^{\rm Spin}(BSU(n))=0$ for $n>2$~\cite[Sec.~3.4]{gem}. The mapping-torus gauge bundle is an induced $SU(5)$ bundle, and its spin structure is retained. Thus the global phase obstruction vanishes. No assertion that $\Omega_5^{\rm Spin}(B\mathsf G_{\rm SM})$ vanishes is needed.

\paragraph{Fock and trace class facts.}
The boundary representation is defined in Section~\ref{rp_fockdefinition} by exterior Hilbert completion and square-summable angle tails. Basic CAR boundedness and ordinary finite exterior identities are the Fock facts used. Mickelsson--Scott~\cite{mickelsson_scott} provides the broader determinant-line/Fock framework. Its Theorem~3.4 is not invoked to identify our metric or BF phase.

The analytic operator-ideal facts are ordinary trace-class cyclicity for bounded factors and trace-norm Fredholm Jacobi. Every use is preceded by the trace-class estimate:~\eqref{rp_boundedbelow}--\eqref{rp_cyclicity} for the bulk/endpoint factorization, and the smoothing Sobolev embedding in Section~\ref{rp_transmissionsection} for the relative determinant. The determinant's independence of boundary Sobolev order is proved there using smoothness of generalized eigenvectors with nonzero eigenvalue.

\paragraph{Relation to the original regulator and its continuum limit.}
The variant preserves exact link and plaquette holonomies, continuum chiral fermions, anomalies and bundle topology. These are reasons to investigate the same positivity mechanism for Craig's reconstruction, whose broader weights also act on distances of order $a$. The gapped homotopy~\eqref{eq:collar_homotopy} identifies bundles; it supplies no identity between the connections or their physical source sections. Even away from zeros,
\[
 \log\left|\frac{W_{0,\mathrm c}(U)}{W_0(U)}\right|
 =\frac12\log
 \frac{\det_\zeta(P_{\mathrm c}^\dagger P_{\mathrm c}/m_0^2)}
      {\det_\zeta(P^\dagger P/m_0^2)}
\]
requires spectral information beyond characteristic classes. Positivity does not propagate along an arbitrary homotopy of weights.

For fixed $\delta$, no collar reconstruction of this type can converge to a non-flat smooth target in a norm controlling curvature. Both collars normal to $\mu$ and those normal to $\nu$ have $F_{\mu\nu}(A_{\mathrm c})=0$. If $\|F_{\mu\nu}(A)\|\ge f_0>0$ on a coordinate box $O$, their union has limiting volume $[1-(1-2\delta)^2]|O|$. Thus, for $1\le p<\infty$ and any gauge identification,
\[
 \liminf_{a\to0}
 \|F_{\mu\nu}(A_{\mathrm c})^g-F_{\mu\nu}(A)\|_{L^p(O)}
 \ge f_0\bigl([1-(1-2\delta)^2]|O|\bigr)^{1/p}>0.
\]
Every non-flat target has such a component and box. This rules out the sufficient strong-smooth-connection route of Section~\ref{sec:continuum_comparison} for this variant; it does not establish different continuum physics. Exact link transports preserve the input plaquette holonomy while curvature is redistributed into cell interiors. A weaker target is convergence of renormalized determinant ratios on smooth external backgrounds, with matched references and local counterterms. Letting $\delta$ shrink with $a$ avoids this fixed-density bound only at the cost of weight derivatives of order $1/(\delta a)$ and separate control of all shrinking-collar estimates. Comparing integrated continuum theories also requires control of the fluctuating bosonic measures.

Theorem~\ref{prop:collar_rp} gives $Z_{\mathrm c}\ge0$. Strict positivity remains a separate question: a normalized reflection-positive state requires $Z_{\mathrm c}>0$. Nonzero Yukawa couplings, Hamiltonian reconstruction, and the continuum comparison are outside the result.

\paragraph{Effective-field-theory outlook.}
Craig's reconstruction and the variant begin with the same lattice links and the same continuum chiral fermions. They provide different prescriptions for filling in the gauge connection between lattice sites. The variant concentrates field variation away from lattice planes, creating the product regions needed for the positivity proof. It introduces no additional dynamical variables. The physical expectation is that the detailed field profile within a cell becomes unresolvable by measurements at fixed physical distances as $a\to0$.

Consider a fermion with fixed physical wavelength $\lambda$. As $a/\lambda\to0$, its wavefunction extends across increasingly many cells. One expects its long-distance propagation to depend on the coarse gauge background and renormalized couplings, with diminishing sensitivity to precisely where the curvature lies inside each cell. The curvature obstruction above does not contradict this picture. Although the collars occupy a fixed fraction of the volume, their widths and separations both shrink with $a$. Rapid oscillations can retain an order-one curvature difference even when the underlying connection perturbation becomes small: differentiation amplifies variation on short distances. Pointwise curvature agreement is therefore a stronger requirement than agreement of long-distance fermion propagation.

Short-wavelength virtual fermions can still resolve the cell profile. The Wilsonian expectation is that integrating out these modes records that sensitivity in local terms in the effective action, as in the effective-action approach to lattice-spacing effects~\cite{symanzik}; see also Ref.~\cite[Sec.~2.2]{luscher_symanzik}. At zero Yukawa coupling, the leading background matching terms concern vacuum normalization and gauge kinetic terms. Their coefficients can differ between reconstructions, so the comparison must hold measured couplings fixed rather than bare couplings. After this matching, remaining effects are expected to be suppressed by powers of $a/\lambda$, possibly accompanied by logarithms. This supplies a physical reason to expect a common continuum theory despite different microscopic curvature profiles.

A first concrete test is to compare the two renormalized determinants to quadratic order in a weak, slowly varying external gauge field. The coefficient of the gauge kinetic term would determine the required coupling matching; the remaining momentum dependence would test whether the long-distance responses agree. Its dependence on the cell profile $\delta$ and on $a\mu$, including possible finite and logarithmic terms, must be calculated. The curvature bound alone does not determine that coefficient.

Establishing this expectation requires showing that the reconstruction change produces no surviving additional low-energy effects, that its ultraviolet contribution admits the stated local matching, and that the full chiral phase agrees in the prescribed zero-theta convention. Matching observables and controlling the fluctuating bosonic measures are also necessary. These steps remain open; the argument motivates continuum equivalence without establishing it or reflection positivity for Craig's reconstruction at finite $a$.

\section{Mathematical terminology}\label{app:terminology}

\begingroup\small
\begin{description}[style=nextline,leftmargin=2em,itemsep=0pt,topsep=0pt,parsep=0pt]
\item[Bundle and connection] Gauge fields described in overlapping patches, with transition functions relating local gauge potentials; a connection defines parallel transport.
\item[Section] A field defined consistently in those patches, with components related by gauge transformations on the overlaps.
\item[Spin structure] The choice of periodic or antiperiodic fermion boundary conditions on noncontractible spacetime cycles.
\item[Characteristic class] A quantized topological charge or flux, such as an integral of $F\wedge F$ or of products of abelian field strengths.
\item[Elliptic/Fredholm operator] A kinetic operator invertible in its leading momentum dependence for $p\ne0$; on compact spacetime it has finite-dimensional spaces of zero modes for itself and its adjoint (kernel and cokernel).
\item[Determinant line, Quillen metric] The one-dimensional space recording the basis dependence of a chiral determinant, including zero modes, together with its zeta-regularized magnitude. In the Standard Model Yukawa theory without Majorana masses, the Pfaffian reduces to this determinant.
\item[Pfaffian] The Gaussian integral for a skew fermion bilinear; its square is the determinant of the skew matrix. Field ordering fixes its overall sign.
\item[Finite-part trace] The constant term of a meromorphically regularized operator trace. It defines a spectral subtraction scheme for ultraviolet divergent loop traces.
\item[Bismut--Freed connection] The analogue of a Berry connection for the determinant phase as the gauge background varies.
\item[Holonomy] The net phase accumulated around a closed loop of gauge backgrounds, analogous to a Berry phase.
\item[Eta invariant and spin bordism] The eta invariant measures regularized spectral asymmetry. Two closed five-dimensional gauge/spin configurations are bordant if together they bound a six-dimensional one. After cancellation of the perturbative anomaly, the exponentiated eta phase is constant on bordism classes; if the relevant fifth spin-bordism group vanishes, this residual phase is $1$.
\item[Spectral filter] A smooth function $f(H/\Lambda^2)$ of a positive kinetic operator that leaves low spectral modes intact and suppresses modes above a cutoff, without choosing a gauge-dependent basis.
\item[Cutoff-to-Quillen normalization] A finite, real, local change of the regulated fermion weight needed to make its magnitude agree with the positive spectral (Quillen) metric. It does not change its complex phase.
\item[Symbol and nested filter] A symbol is an operator's high-momentum expansion. A nested filter is a separate factor on each side of each interaction insertion in a loop, so its momentum derivatives affect local counterterms.
\item[Principal symbol and smoothing] The principal symbol is the leading homogeneous momentum term; the complete symbol includes its lower-order asymptotic terms. Equal complete symbols imply a smoothing difference, with a smooth kernel and rapid high-momentum decay. On compact spacetime such differences are trace class.
\item[Spectral cut (Agmon cut)] A ray avoided by the high-momentum principal symbol and the relevant spectrum of a complex elliptic operator, used to choose a logarithm branch. It is a phase convention for complex powers, not an energy cutoff.
\item[Riesz projection and Kato transport] A contour integral of the resolvent selects a finite spectral subspace separated by a gap. Kato transport identifies these smoothly varying subspaces without choosing eigenvectors individually.
\item[Semiclassical calculus] An expansion at small $h=\Lambda^{-1}$ with scaled momentum $\eta=h\xi$; operator products and their trace remainders can then be estimated uniformly on a momentum shell.
\item[Calder\'on projection] The orthogonal boundary projection onto traces of bulk solutions of the homogeneous Dirac equation. Its range is the Cauchy-data space; it is different from imposing a spectral boundary condition.
\item[Restricted Fock space and CAR] The particle-hole Hilbert space relative to a chosen filled fermion sea. Changes of occupied subspace with Hilbert--Schmidt projection difference define normalizable Slater rays. CAR are the canonical anticommutation relations of creation and annihilation operators.
\item[Principal angles] Angles between two closed subspaces; their cosines are the nontrivial singular values of the projection from one subspace onto the other and determine the Slater overlap.
\item[APS boundary condition] A boundary condition selecting one spectral half of the tangential Dirac operator, with an explicit convention for its zero modes; it enters the Atiyah--Patodi--Singer index and eta formulas.
\item[Weighted trace and residue] The weighted trace is the finite part of a trace with a complex power of a positive elliptic reference operator; the noncommutative (Wodzicki) residue extracts the coefficient of symbol order $-4$ in four dimensions.
\item[Local counterterm] A regulator-dependent integral of a local gauge-invariant expression in fields and finitely many derivatives, used to define renormalized parameters.
\item[Variation estimate] A uniform bound on the difference of derivatives of two log weights (equivalently, effective actions with the opposite sign); with reference phases fixed, integrating the bound compares their full weights.
\item[Grushin extension] Add finitely many auxiliary rows and columns to a Fredholm operator so that the enlarged operator is invertible; its finite effective block retains the zero-mode and source information.
\item[Schatten classes and trace norm] $T\in\mathcal S_p$ means $\sum_j s_j(T)^p<\infty$ for the singular values $s_j$ of $T$. The trace norm $\|T\|_1=\sum_j s_j(T)$ is finite for trace-class operators. In four dimensions the relative Yukawa operator lies in $\mathcal S_5$.
\item[Regularized Fredholm determinant] The fifth-order determinant $\det_5$ retains the fifth and higher Yukawa-loop contributions; the complete renormalized first four amplitudes must also be included.
\item[Common spectral filter] The same gauge-covariant function of the zero-Yukawa positive operators is used on both sides of the Yukawa multiplication map, in every index sector.
\item[Finite-dimensional reduction] Separate potentially vanishing modes into a finite-dimensional block, invert the complementary operator, and evaluate the low-mode Grassmann integral directly.
\end{description}
\endgroup

\Needspace{14\baselineskip}
\paragraph{Finite-dimensional model of the determinant line.}
For an operator $M:V\to W$ between two $n$-dimensional complex spaces, a change of bases by $g$ in $V$ and $h$ in $W$ sends its matrix to $h^{-1}Mg$; the numerical determinant changes by $\det g/\det h$. The invariant object is the induced map $\det M:\det V\to\det W$, or an element of $\det V^*\otimes\det W$, where $\det V$ is the one-dimensional top exterior power of $V$. For the chiral Dirac operator, a standard dual convention identifies the line locally with
\[
\mathcal L_U=\bigl(\bigwedge^{\mathrm{top}}\ker D_A^+\bigr)^*
\otimes\bigwedge^{\mathrm{top}}\ker((D_A^+)^\dagger).
\]
The regularized nonzero modes supply the determinant section and its Quillen norm. These local descriptions glue across kernel crossings, using finite spectral windows; the inverse-line convention is obtained by dualization. This is the coordinate explanation for why a smooth, globally gauge-equivariant phase convention is nontrivial in a chiral gauge theory.

\end{document}